\documentclass[11pt]{article}
\ifdefined\pdfminorversion\pdfminorversion=7\fi

\usepackage[margin=1in]{geometry}
\usepackage[T1]{fontenc}
\usepackage[utf8]{inputenc}

\usepackage{amsmath,amssymb,amsthm,mathtools}
\usepackage{newpxtext,newpxmath}
\usepackage{microtype}

\usepackage{booktabs}
\usepackage{array,tabularx}
\usepackage{enumitem}
\usepackage{graphicx}
\usepackage{float}
\usepackage{flafter}
\usepackage{placeins}
\usepackage{caption}
\usepackage{subcaption}
\usepackage{needspace}

\usepackage{xcolor}

\usepackage[
    backend=biber,
    style=authoryear,
    natbib=true,
    maxcitenames=3,
    maxbibnames=99,
    giveninits=true,
    uniquename=init
]{biblatex}

\usepackage[
    colorlinks=true,
    linkcolor=blue!55!black,
    citecolor=blue!55!black,
    urlcolor=blue!55!black
]{hyperref}

\usepackage[nameinlink,noabbrev]{cleveref}

\allowdisplaybreaks
\numberwithin{equation}{section}

\theoremstyle{plain}
\newtheorem{theorem}{Theorem}
\newtheorem{proposition}{Proposition}
\newtheorem{lemma}{Lemma}
\newtheorem{corollary}{Corollary}

\theoremstyle{definition}
\newtheorem{assumption}{Assumption}

\theoremstyle{remark}

\crefname{assumption}{Assumption}{Assumptions}
\Crefname{assumption}{Assumption}{Assumptions}

\crefname{proposition}{Proposition}{Propositions}
\Crefname{proposition}{Proposition}{Propositions}

\crefname{lemma}{Lemma}{Lemmas}
\Crefname{lemma}{Lemma}{Lemmas}

\newcommand{\R}{\mathbb{R}}
\newcommand{\E}{\mathbb{E}}
\newcommand{\dd}{\mathop{}\!\mathrm{d}}
\newcommand{\cx}{\mathrm{cx}}

\DeclareMathOperator{\CS}{CS}
\DeclareMathOperator{\TS}{TS}
\DeclareMathOperator{\supp}{supp}
\DeclareMathOperator{\cav}{cav}
\DeclareMathOperator*{\argmax}{arg\,max}

\hypersetup{
    pdftitle={Shaping the Menu: Information Design under Screening},
    pdfauthor={Panagiotis Kyriazis},
    pdfsubject={Information design, nonlinear pricing, profit weights, and market segmentation}
}
\title{Shaping the Menu: Information Design under Screening\thanks{I am grateful to Zeinab Aboutalebi, \"Ozlem Bedre-Defolie, Giacomo Calzolari, Pietro Dall'Ara, Piotr Dworczak, Vasundhara Mallick, Matthew Mitchell, Zvika Neeman, and Alessandro Pavan for valuable comments and suggestions. I also thank the participants in the 4th Italian Junior Workshop in Economic Theory and the 2026 European Economic Theory Conference, as well as seminar audiences at Athens University of Economics and Business and the Micro Working Group at the EUI. During the course of this research, I received financial support from the Deutsche Forschungsgemeinschaft (DFG, German Research Foundation) under project number 504870592, which I gratefully acknowledge.}}
\author{Panagiotis Kyriazis\thanks{European University Institute. Email: \href{mailto:panagiotis.kyriazis@eui.eu}{panagiotis.kyriazis@eui.eu}.}}
\date{November 24, 2026}

\begin{document}
\maketitle
\begin{center}
\href{https://www.dropbox.com/scl/fi/87lwyddbzlso2rxkezpzt/Shaping_the_Menu__Information_Design_under_Screening.pdf?rlkey=eezc62t5e6r55qg72zknwdpmh&st=9mo4ixxf&raw=1}{\color{blue}{\Large Link to Latest Version}}
\end{center}

\begin{abstract}

How should consumers learn when a seller chooses a
menu of product qualities and prices? The seller
observes the learning policy, but not consumers'
private signals. Under a cost curvature condition,
every optimal learning policy for weights strictly
favoring consumers sustains a continuum of differentiated
products extending to efficient quality at the top,
even with only two underlying valuations. The feasibility constraints
on learning can make it optimal to offer an isolated
basic product alongside this range or to exclude some
buyers. The basic product lowers prices for higher
qualities, but keeping it profitable can leave fewer
buyers purchasing those qualities. Exclusion sacrifices
this alternative but can allow more buyers to earn
rents within the differentiated range. Maximizing consumer surplus
can therefore require leaving gains from trade unrealized.
Joint design of buyers' private information and a public
signal---equivalently, joint design of learning and market segmentation---eliminates
inefficient exclusion and can raise both consumer surplus and profit
relative to buyer-optimal private learning in a single market. Even random segmentation with different
learning policies and separate menus can give consumers
more surplus than any single market yielding the same
profit, without revealing information about valuations
through segment labels. Our results highlight the power
of intermediaries to influence consumers' terms of
trade, making their incentives central to consumer
protection.

\end{abstract}



\newpage
\section{Introduction}
\label{sec:introduction}
Who controls what consumers learn can shape both
the gains from trade and their division between
consumers and firms.
Firms choose their products and prices based on
consumers' willingness to pay, which depends on what consumers know or learn about
their needs. Prior experiences, product reviews, and advice all
contribute to consumers' learning and understanding of what they value most. Intermediaries that guide this learning and, in
some cases, also control access to offers can influence
the products and prices consumers face.
They can
use this influence to secure better terms for consumers
or higher profits for sellers.
Their incentives are therefore central both to consumer
protection and market efficiency.

Consumers can benefit from a product they never buy.
Consider a website that helps households understand
whether faster broadband would be worth paying for.
A cheaper but slower plan gives households considering
faster service an alternative and limits what the
provider can charge them. Whether that alternative
is offered, and at what speed and price, depends on
demand from other households. What those households
learn about their needs can therefore affect the
prices paid by households choosing faster service.
Information shapes both the choices consumers make
and the choices firms find profitable to offer. How should
learning and access to offers be designed to benefit
consumers when they also shape the alternatives
consumers face?

Our results show that learning designed to benefit
consumers induces greater product variety than learning
designed to maximize profit. Product variety can
therefore provide evidence about whose interests an
intermediary serves. Yet when all buyers face the same
offers, maximizing consumers' overall gains can require
leaving some buyers unserved, even if serving them
would create value. The intermediary can change this
trade-off by choosing which buyers share the same
offers as well as what they learn. Coordinating these
decisions can make further price discrimination work
in consumers' interests. Even random assignment to
segments with different information policies and
separate menus can give consumers more at the same
seller profit. Because segment membership then reveals
nothing about buyers' needs, these gains remain possible
when privacy considerations rule out such disclosure.
Allowing assignment to depend on valuations can go further.
Coordinating learning and segmentation can then increase
both consumers' gains and the seller's profit while
bringing previously excluded buyers into the market.

We begin with a seller offering a menu of qualities
and prices to buyers who are initially uncertain about
their valuation of quality. Valuations are drawn from a prior distribution with finite support. An information designer
publicly commits to an information policy. Buyers privately
observe their signals before the seller offers its
profit-maximizing menu. The seller observes the policy
but neither valuations nor signals. Buyers then choose
a product or decline trade. Utility is linear in
valuation, so each buyer's posterior mean determines
her willingness to pay for quality. We therefore work directly with the induced distribution
of posterior means, which we call the \emph{posterior
distribution}. Such a distribution is feasible if and only if it is a mean-preserving contraction of the prior, that is, it has the same mean and is no more dispersed than the prior. Production costs
are strictly convex and assumed to satisfy a curvature condition which, for example, holds whenever marginal
cost is convex. The designer maximizes a weighted sum
of consumer surplus and profit. Our main structural
results concern weights strictly favoring consumers,
while the analysis of attainable payoffs also allows
weights favoring profit.

Even with only two underlying valuations, every optimum
for weights strictly favoring consumers induces a
continuous range of posterior values over which quality
strictly increases. A positive share of buyers learn
that they have the highest valuation and receive
efficient quality. Below this range there may be a
single lower atom, separated from it by a support gap.
Buyers at this atom either purchase an isolated
basic product, separated from the differentiated
range by a gap in quality, or learn that they
have the low valuation and remain unserved. The two lower configurations cannot coexist.
Demand at intermediate posterior values sustains
products that constrain prices higher in the menu.
Buyers with the same underlying valuation can
therefore receive different information and
purchase different qualities. By contrast,
\citet{bergemannheumannmorris2026} show that a seller
choosing both information and its menu offers finitely
many products even with a continuous prior and, under
additional conditions, offers only one.

The lower group may appear because it allows the designer to satisfy the
prior mean without reducing posterior values throughout
the differentiated range. Buyers in that range,
including those at the top, may have an average
posterior value above the prior mean. We give sufficient
conditions for both inclusion and exclusion. Exclusion can help by allowing more buyers to purchase
from the differentiated range and enjoy the information rents created within it. Lower posterior values
for the excluded buyers make this possible without
changing the prior mean. Those values, however, cannot
fall below the lowest underlying valuation. When that
valuation is high and the prior mean is low, exclusion
offers limited room to expand the differentiated group.
Under our inclusion conditions, keeping the basic
product and the rents it generates for buyers of higher-quality products is preferable, so every optimum
serves everyone. When the positive lowest valuation is
sufficiently small relative to the low mean, exclusion
can allow a larger expansion, while the basic product
generates little rent. Yet sustaining it requires some
high-valuation buyers to purchase it. Under our exclusion
conditions, moving more buyers into the differentiated
range is preferable to retaining this alternative,
so every optimum excludes some buyers, despite the fact that serving them would create gains.
Thus, maximizing consumer surplus can require
leaving gains from trade unrealized. This contrasts
sharply with the single-item framework in \citet{roeslerszentes2017}, where
buyer-optimal learning induces efficient trade under
posted pricing.

Participation can be nonmonotone as high valuations
become more common, raising the prior mean and improving
the prior distribution in the sense of first-order stochastic
dominance. For each fixed weight strictly favoring
consumers, this occurs when the lowest valuation is
positive and sufficiently small. When the high
valuation is sufficiently rare, every optimum includes an isolated basic product and serves
everyone.
At some intermediate probabilities, every optimum
excludes a positive share of buyers. When the high
valuation is sufficiently common, every optimum
again serves everyone, now through the differentiated
product range.

When the lowest valuation has no gains from trade,
exclusion again allows the designer to satisfy the
prior mean while preserving the differentiated range
for the remaining buyers. But only buyers who learn
that they have this lowest valuation are excluded,
so no gains from trade are lost. For weights strictly
favoring consumers, optimal learning at low prior
means keeps the distribution of posterior values
among trading buyers fixed. As the high valuation
becomes more common, more buyers receive these
posterior values and trade. The seller continues
to offer the same qualities and prices. Participation
thus rises continuously until everyone trades and
remains full at higher means. In this case, and under a separate curvature condition,
the optimal posterior distribution is unique for
every nonnegative profit weight. This allows us to characterize the consumer surplus frontier, that is, to
identify the greatest consumer surplus
attainable under every feasible level of profit, and the unique posterior distribution
generating it.

With more than two underlying valuations, Bayes plausibility
imposes additional constraints on the dispersion of posterior
means. The optimal posterior distribution, however, retains a local binary structure.
Within an interval between adjacent underlying valuations,
redistributing posterior means preserves feasibility as long as
the share of buyers in that interval and their average
willingness to pay remain unchanged. The arguments used
for binary priors therefore continue to apply on
such an interval, but need not remain feasible when they redistribute
posterior means across intervals. Even if the overall mean is
preserved, the resulting distribution may be too dispersed to
arise from the prior. We show that optimal learning can induce several separating
regions and trading atoms. A separating region can end at
an atom, with separation resuming above the support gap
that follows it.
Exclusion can also coexist with a lower trading atom.
 As in the binary case,
distinct trading posterior types receive distinct qualities,
and the uppermost separating region ends at an efficiently served
group of buyers who learn that they have the highest valuation.

The preceding results describe what learning can
achieve when all buyers face a common menu. In the
modern digital economy, intermediaries can increasingly control
both the way consumers learn and the offers
they can access.\footnote{
AppDirect, for example, offers
\href{https://helpcenter.appdirect.com/hc/en-us/articles/28251322311831-Start-with-the-Right-Expertise}{advice on businesses' software needs},
while its marketplace supports
\href{https://support.appdirect.com/hc/en-us/articles/20852023784211-Can-I-hide-or-make-certain-products-invisible-to-certain-users-or-companies}{product eligibility rules}
and
\href{https://www.appdirect.com/pricing}{prices that vary across customer segments}.}
These capabilities give intermediaries additional
power over how firms price discriminate, how much
surplus trade creates, and how that surplus is divided.
Motivated by this, we also study the joint design of public and private information.
The designer assigns buyers to segments whose labels are observed
by both parties and chooses private learning within each. In the broadband example, the provider chooses
the plans and prices for each segment. The website
shows each household its assigned menu, and the
household can purchase only from that menu.

Even assignment independent of valuations can give
consumers more surplus than any single-market policy
yielding at least as much profit. Every segment
then has the same prior distribution, but different
learning policies can generate different posterior
distributions and induce different menus. Keeping
these menus separate preserves each policy's payoffs,
which could change if buyers could choose products
offered in other segments. We show that this random
segmentation can generate new Pareto-efficient outcomes
even with a binary prior. Two segments suffice to
attain every efficient pair available under independent
segmentation. Because segment labels reveal nothing
about underlying valuations, these gains remain
possible even when privacy concerns restrict such disclosure.
Random segmentation cannot, however, raise maximal consumer
surplus relative to private learning alone.

Valuation-dependent assignment allows segments to
contain different mixes of underlying valuations
and can raise maximal consumer surplus itself.
Joint design can also increase both consumer surplus
and profit relative to optimal private learning.
These results contrast with \citet{KartikZhong2026},
where a seller posts a price for a fixed indivisible
object. There, every
attainable payoff pair can be generated by designing
consumer learning while keeping the seller uninformed.
We also show that learning and segmentation are not
separable instruments. With a binary prior, whenever
joint design improves on private learning alone,
even the best segmentation of any posterior distribution
optimal for a single market falls strictly short
of the joint optimum. Realizing the full gains from
segmentation therefore requires changing what buyers learn.

Joint design also eliminates inefficient exclusion.
For weights strictly favoring consumers, every joint
optimum serves all buyers with gains from trade.
The designer can place otherwise excluded buyers
with only some trading buyers. A basic product can
then become profitable for the seller to offer because
it requires price reductions for fewer existing
customers. Those customers receive the same qualities
at lower prices, while the excluded buyers are brought
into the market.
In addition, for any finite prior and weights strictly
favoring consumers, buyers with the same posterior value
receive the same quality across segments at a joint
optimum, although prices may differ.
This extends the common-quality property established
by \citet{bergemannheumannwang2026} for a given
distribution of privately known valuations to the case
where learning and segmentation are chosen jointly.

With a binary prior, two segments suffice to achieve the joint optimum.
When segmentation improves the objective, the two
segments have the same distribution of posterior
values among buyers of differentiated products.
Exactly one menu includes a basic product, and
buyers of differentiated products in that segment
pay less for the same qualities. With $N$ underlying
valuations, at most $N$ segments are needed, and some priors
require all $N$.  Joint optimality also requires buyers at every posterior atom
to learn their underlying valuation.

We also identify cases in which optimal private learning makes segmentation redundant. For binary priors with a payoff-neutral lowest
valuation, neither form of segmentation changes
the single-market Pareto frontier of the attainable surplus set. For each fixed
weight strictly favoring consumers, segmentation
is also redundant at sufficiently high prior means,
even when the lowest valuation has gains from trade.

Methodologically, the main difficulty in our analysis is that changing what buyers learn
can affect contracts throughout the seller's menu.
Learning can generate distributions with atoms, support gaps, continuous parts or combinations of these. To handle this, we order buyers by posterior value and work in quantile
space.
The concavified revenue curve determines the quality
supplied at each rank and the seller's
profit.\footnote{At each population rank, the revenue curve multiplies the corresponding posterior
value by the share of buyers at higher ranks.
Its concavification is the smallest concave curve
lying above it.}
Regularization replaces the posterior distribution
by one whose revenue curve equals its concavified
revenue curve. It preserves qualities and profit
while weakly increasing consumer surplus
\citep{kyriazis2026markets}. Whenever it changes the
distribution, however, it raises the mean. Feasible
learning must preserve the prior mean, so this
increase requires lower posterior values elsewhere.

We therefore characterize all attainable means and
payoffs at a fixed concavified revenue curve. Starting
from a regularized distribution, posterior values
can be lowered without changing the curve only on
its linear parts, where buyers receive a common
quality. We determine the largest reduction possible
on each part and show that all reductions within
these bounds can be implemented together. Each
reduction's cost in terms of consumer surplus is exactly
that reduction multiplied by the quality supplied
on the corresponding part. Restoring the mean is therefore a
linear allocation problem, and its solution is to lower
posterior values at lower qualities first. We then
use this characterization to perturb
learning in a way that alters the curve while preserving the
prior mean. By identifying which distributions can
be improved upon after the seller reoptimizes, we
establish the structure of optimal learning.

We then characterize the separating range and its
connection to any lower group. We express the upper
distribution, including its top atom, in virtual values.
Holding its probability, mean, and lowest posterior
value fixed, we construct interior variations that
preserve monotonicity of virtual values. Quality
therefore remains nondecreasing, as required by
incentive compatibility. These variations yield
Lagrangian optimality conditions without assuming
a differentiable virtual-value profile. The curvature
condition then establishes smoothness and yields a
differential system characterizing the separating range.

Joint design adds the choice of segment composition
to the choice of learning. We formulate an aggregate
program over posterior and virtual values, without
first specifying their allocation across segments.
Bayes plausibility restricts the posterior distribution,
while seller optimality yields inequalities for buyers
above each posterior threshold. These
constraints give an upper bound on the designer's
payoff. We prove that the bound is attainable by a
feasible joint policy by choosing segment compositions
that make the prescribed qualities optimal for the
seller. The aggregate program therefore gives the exact
value of joint design, and every joint optimum must also
solve this program. This allows us to establish
properties shared by all optimal policies.

\subsection{Related literature}
\label{subsec:related-literature}

Our model builds on the screening models of
\citet{mussarosen1978,maskinriley1984}. Early work on information
provision under monopoly includes
\citet{LewisSappington1994,OttavianiPrat2001}.
More recently, \citet{roeslerszentes2017}
characterize optimal learning under monopoly pricing, while
\citet{debroesler2024} show that pure bundling can implement a
buyer-optimal outcome under exchangeable multidimensional priors.
In \citet{vonWangenheim2025}, the seller's response also includes refund terms when buyers
learn their values after contracting.
We instead study how optimal learning, chosen alone or jointly with market segmentation, shapes a vertically differentiated menu
under strictly convex production costs, changing both the
qualities supplied and the information rents left to consumers.

Related work endogenizes demand through advertising and product
design \citep{JohnsonMyatt2006,BarIsaacCaruanaCunat2023}, costly
valuation changes \citep{DaiKoh2025}, or choice of type distribution
before monopoly pricing \citep{CondorelliSzentes2020}.
Information design with pricing or screening is studied by
\citet{WeiGreen2024} and \citet{ShiZhangChen2026}.
\citet*{bergemannheumannmorris2026} jointly choose information
and the menu to maximize profit, obtaining finitely many signals
and qualities for a continuous prior. Under our curvature condition, weights strictly
favoring consumers instead require a continuous range of differentiated products even with two underlying valuations.

Our analysis of joint learning and segmentation relates to
\citet{bergemannbrooksmorris2015,haghpanahsiegel2022,
haghpanahsiegel2023,bergemannheumannwang2026}. The common-quality property of
\citet{bergemannheumannwang2026} continues to hold
when learning and segmentation are chosen jointly, but optimally segmenting a buyer-optimal
posterior distribution need not attain the joint optimum.
Related models show that partial disclosure of
imperfect group information can maximize total surplus
\citep{HausleEtAl2025Privacy} and derive welfare bounds for
finite-type markets from two-type markets \citep{HausleEtAl2026}.
These analyses hold buyers' information fixed. Work allowing
information on both sides includes
\citet{KartikZhong2026}, who, in a single-object environment, show that seller
information does not enlarge the attainable payoff set when
the buyer is better informed and does not update from prices. In our screening
environment, public information can instead raise maximal consumer
surplus relative to private learning alone.
Given a privately informed seller who commits to a mechanism,
\citet{StrackYang2024} show that a buyer-optimal signal can be
chosen independent of the seller's signal, while
\citet{DebPaiRoesler2024} show that seller information can benefit
consumers in an informed-principal setting with interdependent
values. We jointly choose private learning and segmentation, making both the posterior types being screened
and their division across menus endogenous.

Other related literature differs in the control of information and the timing of
learning. \citet{CarmonaLaohakunakorn2025} study information
structures determined by the conflicting choices of buyers and
sellers, while \citet{ChopraEly2026Auctions} study bidders who
choose how to learn noncooperatively before the seller designs a
revenue-maximizing auction. When information is supplied or
acquired after prices or menus are chosen, sellers instead
anticipate how their offers affect subsequent learning
\citep{xuyang2026,menschravidthereze2026}.
Observability also matters: privately chosen learning can sustain
inefficient trade even as information costs vanish
\citep{ravidroeslerszentes2022}.
Here, a single designer chooses learning, alone or jointly with segmentation,
and commits publicly before the seller chooses its menu or menus.

Finally, our framework is in the spirit of mean persuasion
\citep{kolotilin2018,dworczakmartini2019}: only the posterior mean matters for the ensuing screening problem and the feasible
distributions over posterior means are the mean-preserving contractions of the prior.
The extreme-point characterizations of
\citet{KleinerMoldovanuStrack2021,ArieliEtAl2023} simplify
optimization of objectives linear in a distribution or quantile
function. Majorization also links information and mechanism design
\citep{BergemannHeumannMorris2026Integrated}:
\citet{bergemannheumannmorris2026} exploit linearity conditional
on the quality distribution, while
\citet{bergemannheumannwang2026} reduce segmentation to a nonlinear
problem over inverse hazard rates. Our objective incorporates
the seller's optimal menu response and need not be linear or
convex in the posterior distribution.
We build on and extend the variational methods of \citet{kyriazis2026markets}, who studies unrestricted
choice of the type distribution under screening with type-dependent
costs and characterizes the set of feasible surplus outcomes and their limits without information constraints. In our
learning problem, the information constraint can make lower trading atoms and exclusion part of the optimum
and create scope for gains from segmentation.

\section{Model}
\label{sec:model}

A buyer can learn about her valuation before a monopolist offers a
menu of qualities and prices. A contract specifies quality
$x\in[0,\bar Q]$ and a transfer $t\in\R$ from buyer to seller. The buyer's realized utility
is $Vx-t$, and the seller's profit is $t-c(x)$. Both maximize
expected payoffs, and the buyer can choose the outside option $(x,t)=(0,0)$.

For $N\geq2$, the prior distribution of the buyer's valuation for quality is
\begin{equation}
 H=\sum_{i=1}^{N}p_i\delta_{v_i},
 \qquad 0\le v_1<\cdots<v_N=1,
 \qquad p_i>0,\quad \sum_{i=1}^{N}p_i=1,
 \label{eq:primitive-prior}
\end{equation}
where $\delta_v$ denotes a point mass at $v$, and
$m:=\E_H[V]$. The highest valuation is normalized to one.

\paragraph{The designer.}
An information designer chooses the buyer's learning
policy to maximize $\CS+\beta\Pi$, anticipating the
seller's optimal menu. Here $\beta\ge0$ is the weight
on profit relative to consumer surplus. Consumer-oriented
weights put at least as much weight on consumer surplus
as on profit. Our main characterizations concern
$\beta\in[0,1)$ while the frontier analysis also allows
$\beta>1$.

\paragraph{Timing.} Before the seller posts a menu, the designer publicly
chooses an information structure $\sigma(\dd s\mid v)$, a
conditional distribution of signals for each primitive valuation.
The signal space may be infinite.\footnote{Formally, the signal space
is standard Borel and $\sigma$ is a probability kernel.}
Nature draws $V$ from $H$ and generates the signal $S$ according
to $\sigma$. The buyer observes $S$ but not $V$. The seller
observes the chosen information structure, but neither the valuation nor the signal
realization, and posts a profit-maximizing menu. The buyer then
chooses a contract or the outside option.

\vspace{0.1in}
Throughout the paper, we assume the following about the seller's production cost.

\begin{assumption}[Cost]
\label{ass:cost}
The function $c:[0,\bar Q]\to\R_+$ is continuously differentiable on
$[0,\bar Q]$ and twice continuously differentiable on $(0,\bar Q]$,
with $c(0)=c'(0)=0$ and $c''(x)>0$ for $x>0$. There is a
$\bar q\in(0,\bar Q)$ such that $c'(\bar q)=1$.
\end{assumption}

For $z\le1$, define
$q(z):=\argmax_{x\in[0,\bar Q]}\{zx-c(x)\}$ and $\pi(z):=\max_{x\in[0,\bar Q]}\{zx-c(x)\}$.
For a buyer with valuation $y$, $q(y)$ is the efficient quality and
$\pi(y)$ is the associated total surplus, before dividing it between
buyer and seller. When $y$ is the buyer's posterior mean, $q(y)$ is efficient given
her information and $\pi(y)$ is the corresponding expected total
surplus. Under \Cref{ass:cost}, strict convexity gives
$q(z)=(c')^{-1}(z)$ for $z\in(0,1]$ and
$q(z)=\pi(z)=0$ for $z\le0$. Both functions are continuous, and
$\bar q=q(1)$ is the highest efficient quality.

\vspace{0.1in}
For $0\le\beta<1$, we also use the following curvature condition.

\begin{assumption}[Cost curvature]
\label{ass:curvature}
The cost function is three times continuously differentiable on
$(0,\bar Q]$, and
\[
 (2-\beta)[c''(x)]^2+[1-c'(x)]c'''(x)>0
 \qquad\text{for }0<x<\bar q.
\]
\end{assumption}

The assumption is satisfied, for example, if marginal cost is convex, including
power cost functions $c(x)=\kappa x^p/p$, with $\kappa>0$ and
$p\ge2$. The condition also permits $c'''<0$ when its magnitude is small
enough relative to $c''$.

\paragraph{Gains from trade.}
A type $y$ is \emph{payoff neutral} if $q(y)=\pi(y)=0$. There are no gains from trade at that type. A type is \emph{productive} if $q(y)>0$,
or equivalently $\pi(y)>0$. The prior has a \emph{productive lowest
value} when $v_1$ is productive. Every feasible posterior type then
has positive gains from trade.
The requirement $c'(0)=0$ in \Cref{ass:cost} is a normalization and is without loss. In this case, zero is the only payoff-neutral
value in $[0,1]$, and the lowest value is productive exactly when
$v_1>0$. When
$c'(0)=d\in[0,1)$, $q(y)=\pi(y)=0$ for $y\le d$, while $q(y)>0$ and
$\pi(y)>0$ for $y>d$. A payoff-neutral type is feasible when $v_1\le d$ while the lowest value is productive when
$v_1>d$. Unless stated otherwise, we maintain $c'(0)=0$.

\subsection{Feasible learning}
\label{subsec:feasible-learning}

After observing $S$, the buyer has posterior mean
$Y:=\E[V\mid S]$. Her expected utility from a contract is
$Yx-t$, so the posterior mean $Y$ is her type in the screening problem. Let $G$
be the distribution of these posterior means. We refer to
$G$ as the \emph{posterior distribution}, and to a realization
$y$ as a \emph{posterior value} or \emph{posterior type}. Signals with the same posterior mean induce the same preferences
over contracts. Since cost does not depend on $V$, the seller's
problem and both parties' payoffs depend on information only
through $G$, which we take as the designer's choice. The seller knows $G$, while the buyer's
realized type $Y$ remains private. Uniqueness always refers to
the posterior distribution, not the information structure generating it.

A posterior distribution is feasible if and only if it is dominated
in convex order by $H$, or equivalently, is a mean-preserving
contraction of the prior \citep{Blackwell51}.
We write $G\preceq_{\cx}H$ in this case. Formally,
\[
 G\preceq_{\cx}H\quad\Longleftrightarrow\quad
 \int f(y)\,G(\dd y)\le\int f(v)\,H(\dd v)
 \quad\text{for every convex }f:\R\to\R.
\]
The set of feasible posterior distributions is therefore
\[\mathcal G(H):=
 \{G\in\Delta([v_1,1]):G\preceq_{\cx}H\},\]
where $\Delta([v_1,1])$ denotes the Borel probability distributions
on $[v_1,1]$.

Every feasible posterior distribution has mean $m$. For a binary prior
on $\{a,1\}$, with $0\le a<1$ and $a<m<1$, this and the support
restriction are the only requirements. Therefore,
\begin{equation}
 \mathcal G_{a,m}
 =\left\{G\in\Delta([a,1]):\int y\,G(\dd y)=m\right\}.
 \label{eq:binary-feasible-set}
\end{equation}

\subsection{Screening and information design}
\label{subsec:screening}

Fix a feasible posterior distribution $G$. A direct mechanism assigns
quality $x(y)$ and transfer $t(y)$ to each posterior type\footnote{We can focus on direct mechanisms without loss since the revelation principle applies.}. Incentive
compatibility and individual rationality require
\begin{align}
 yx(y)-t(y)&\ge yx(\widehat y)-t(\widehat y)
       &&\text{for all }y,\widehat y\in\supp G,
       \label{eq:ic}\\
 yx(y)-t(y)&\ge0
       &&\text{for all }y\in\supp G.
       \label{eq:ir}
\end{align}
The buyer reports truthfully and participates when indifferent.
Individual rationality holds for each realized posterior type. The seller maximizes
$\E_G[t(Y)-c(x(Y))]$ subject to \eqref{eq:ic}--\eqref{eq:ir}.

Incentive compatibility requires quality to be nondecreasing in the
posterior mean. Buyers with higher posterior values can choose
contracts intended for buyers with lower values,
so the seller must leave them enough surplus to
prefer their own. The optimal qualities and rents
therefore depend on how posterior values are
distributed.

Because learning may generate atoms, continuous components, gaps in the
support, or combinations of these, we formulate the seller's problem in quantile space so that
the same characterization covers all these possibilities.\footnote{ Quantile formulations of
screening appear in \citet{BergemannHeumannMorris2026Integrated}
and \citet{kyriazis2026markets}.}

Let \[Q(u):=\inf\{y\in[0,1]:G(y)\ge u\},\qquad 0<u<1\] be the lower
quantile of $G$ with $Q(0):=\lim_{u\downarrow0}Q(u)$ and $Q(1):=\lim_{u\uparrow1}Q(u)$, where $G(y)$ also denotes its distribution function.
Higher $u$ corresponds to a higher posterior mean. Let $Q^+$ denote
the right-continuous version. Define the \emph{raw revenue curve}
$\widehat R_Q$ and its \emph{concave envelope} $R_Q$ by\footnote{Precisely,
$Q^+(u)=\inf\{y\in[0,1]:G(y)>u\}$ for $0\le u<1$.
It agrees with $Q$ almost everywhere. At a support gap it selects the
next higher posterior value, making the raw revenue curve upper
semicontinuous. Choosing the lower value at an interior quantile jump would
not change the concave envelope since a finite concave majorant is
continuous there and must also dominate the upper one-sided
limit.}
\begin{equation}
 \widehat R_Q(u):=(1-u)Q^+(u)\quad \text{for} \quad 0\le u<1,
 \qquad \widehat R_Q(1):=0,
 \qquad R_Q:=\cav\widehat R_Q.
 \label{eq:revenue-envelope}
\end{equation}
Here $\cav$ denotes the least concave majorant. We refer to $R_Q$
as the \emph{concavified revenue curve}. Its negative right derivative,
\[
 \phi_Q(u):=-R'_{Q,+}(u),\qquad 0<u<1,
\]
is the \emph{ironed virtual value}. It is nondecreasing and satisfies
$\phi_Q(u)\le Q(u)\le1$ almost everywhere.

The seller offers quality $q(\phi_Q(u))$ at rank $u$.
When quality is positive, marginal cost equals the
ironed virtual value. Types with nonpositive virtual
values receive zero quality and are excluded.\footnote{When
$c'(0)=d>0$, the same formulas apply, but a type receives
zero quality when its virtual value is at most $d$.}
Concavification yields
marginal revenues that reflect incentive compatibility, including
the rent requirements across support gaps. On the interior of a
linear interval of $R_Q$, the ironed virtual value is constant.
We call a posterior distribution \emph{regular} when
$\widehat R_Q=R_Q$, and refer to $R_Q>\widehat R_Q$ as
\emph{strict concavification}.

An atom at $y$ occupies an interval of ranks even
though its buyers share the same posterior value.
Replace $\widehat R_Q$ on each such interval
$[G(y-),G(y)]$ by the straight line joining its
endpoint values, leaving the curve unchanged
elsewhere. Denote the resulting \emph{linearly
interpolated revenue curve} by $\widetilde R_Q$.
For a finite posterior distribution, this is the
polygon joining successive cutoff-revenue points.
Its negative segment slopes are the usual discrete
virtual values.

Each interpolating chord lies weakly above the raw
curve and below its concave envelope. Hence
$\widehat R_Q\le\widetilde R_Q\le R_Q$ and $\cav\widetilde R_Q=R_Q$.
Concavifying the raw curve therefore performs both
the interpolation across atoms and any further
ironing. For a finite posterior distribution,
$\widetilde R_Q$ is already concave when the discrete
virtual values are weakly increasing. Our definition
of regularity refers to the raw curve, so strict
concavification can arise from interpolation alone, and need not indicate nonmonotonic discrete virtual
values or bunching of distinct posterior types.\footnote{For
$G=(1-\tau)\delta_b+\tau\delta_1$, with $0<\tau<b<1$, the
discrete virtual values $z_b=(b-\tau)/(1-\tau)$ and $1$ are
positive and strictly ordered. Thus $\widetilde R_Q=R_Q$, while
$R_Q(u)-\widehat R_Q(u)=u\tau(1-b)/(1-\tau)>0$ for
$0<u<1-\tau$. These ranks all have posterior value $b$.
Under our cost assumptions, the seller assigns the two types the
positive and distinct qualities $q(z_b)$ and $q(1)$.
\citet[Section~2]{MaSivan2020} discuss this distinction between
conventional discrete virtual values and a construction based on the
pointwise revenue curve.}
Interpolation keeps the posterior distribution fixed.
Figure~\ref{fig:interpolation-concavification} illustrates the
distinction between interpolation and further ironing.

\begin{figure}[t]
 \centering
 \captionsetup{font=small,skip=4pt}\includegraphics{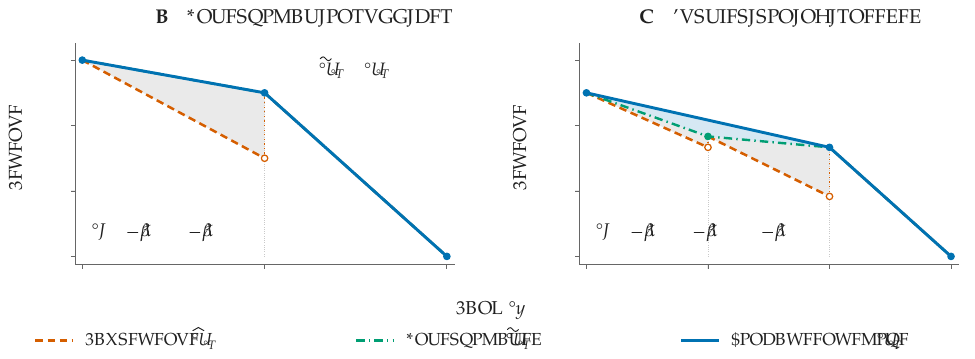}
 \caption{Interpolation (gray shading) and further ironing
(blue shading) for fixed posterior distributions. In (b),
the two lower types receive the same positive quality.}
 \label{fig:interpolation-concavification}
\end{figure}

Strict convexity of the cost makes optimal quality
unique for $G$-almost every posterior type, for every
$G\in\Delta([0,1])$. The seller
extracts the lowest type's surplus and leaves the rents
required by incentive compatibility to higher types. All optimal
menus therefore give the same expected profit and consumer surplus
\citep[Proposition~A.2]{kyriazis2026markets}:
\begin{align}
 \Pi(G)&=\int_0^1\pi(\phi_Q(u))\,\dd u,
       \label{eq:seller-profit}\\
 \CS(G)&=\int_0^1[Q(u)-\phi_Q(u)]q(\phi_Q(u))\,\dd u.
       \label{eq:consumer-surplus}
\end{align}

Total surplus is $\TS(G):=\CS(G)+\Pi(G)$.
Information thus changes both quality and the rents the seller
must leave. Reducing virtual values can help consumers by limiting
surplus extraction, but also weakly reduces quality.

The designer chooses the posterior distribution to maximize
\[
 J_\beta(G):=\CS(G)+\beta\Pi(G),\qquad
 W_H(\beta):=\max_{G\in\mathcal G(H)}J_\beta(G),\qquad \beta\ge0.
\]
For a binary prior, write
\begin{equation*}
 V_{a,\beta}(m):=\max_{G\in\mathcal G_{a,m}}J_\beta(G),
 \qquad V_a(m):=V_{a,0}(m).
\end{equation*}
The main characterization concerns $0\le\beta<1$, so consumers
receive greater weight than the seller. Buyer-optimal learning is the
case $\beta=0$. For weights between zero and one,
$J_\beta=(1-\beta)\CS+\beta\TS$.
A higher $\beta$ places more weight on efficiency relative
to consumer rent. At $\beta=1$, the designer maximizes total surplus, while above one, it places greater weight on profit than on consumer surplus. Both payoff functions are continuous under weak convergence
\citep[Proposition~A.7]{kyriazis2026markets}. Since $\mathcal G(H)$ is
nonempty and weakly compact, a maximizer exists for every finite
$\beta$.

\paragraph{An equivalent formulation of \Cref{ass:curvature}.}
For a posterior mean $y$ and virtual value $z$, define the local
contribution to the designer's objective by
\begin{equation}\label{eq:local_objective_contribution}
F_{y,\beta}(z):=(y-z)q(z)+\beta\pi(z).
\end{equation}
For any fixed $0\le\beta<1$, \Cref{ass:curvature} is equivalent to
\begin{equation}\label{eq:rent-curvature}
(2-\beta)q'(z)-(y-z)q''(z)>0
\qquad\text{for }0<z<y\le1.
\end{equation}
The condition makes $F_{y,\beta}$ strictly concave in $z$.
To interpret it, write $h=y-z$. The term $h q(y-h)$ is the
local contribution to aggregate consumer surplus, termed
the local information rent in \citet{bergemannheumannwang2026}.
On a regular separating region, $h$ is the inverse hazard
rate. It measures the mass of buyers with posterior values
above $y$ relative to the density at $y$. A larger $h$
increases the rent contribution at any given quality.
But it also lowers the virtual value, inducing the seller
to supply less quality. This reduces both the rents
generated by that quality and the local profit contribution.
The curvature condition requires that, holding $y$ fixed,
the marginal return to increasing $h$, accounting for these
effects and the designer's weight on profit, strictly
decreases over $0<h<y$.

\section{Optimal Learning with a Binary Prior}
\label{sec:binary}

With a binary prior on $\{a,1\}$, the designer can choose any
posterior distribution on $[a,1]$ with mean $m$. We first characterize the structure of optimal distributions.

\subsection{The optimal posterior distribution}
\label{subsec:binary-structure}

Call a type \emph{excluded} if the seller assigns it zero quality. A \emph{lower trading atom} is a posterior atom at a
valuation strictly below one, whose buyers receive
positive quality. In the optima characterized below,
these buyers receive less than their efficient quality.
We call this configuration \emph{distorted inclusion}. A \emph{separating interval}
is an interval of posterior types on which the optimal quality schedule is strictly increasing. The following theorem describes the structure of every optimum and
gives sufficient conditions for a basic product and for exclusion. Figure~\ref{fig:binary-structures} illustrates these possibilities.

\begin{theorem}[Optimal learning with a binary prior]
\label{thm:binary-structure}
Fix $0\le\beta<1$ and a cost function satisfying
\Cref{ass:cost,ass:curvature}. Consider binary priors on $\{a,1\}$
with $0\le a<m<1$.
\begin{enumerate}
 \setlength{\itemsep}{3pt}
 \setlength{\parsep}{0pt}
 \item[(i)] \emph{Structure of every optimum.}
 Every posterior distribution maximizing $J_\beta$ has one of the
 following lower parts.
 \begin{enumerate}
  \setlength{\itemsep}{0pt}
  \setlength{\parsep}{0pt}
  \item[(a)] An excluded atom at $a$, with no lower trading atom.
  \item[(b)] One lower trading atom, with no excluded types.
  \item[(c)] Neither excluded types nor a lower trading atom.
 \end{enumerate}
 Above this lower part, the posterior distribution has a
 nondegenerate connected separating interval,
 followed by a positive atom at one served at efficient quality
 $\bar q$. On the separating interval, posterior means, virtual values, and
 qualities are continuous and strictly increasing in rank, and the
 posterior density is positive and continuous on the interior.
 In cases (a) and (b), a positive support gap separates the lower
 atom from the separating interval. In case (b), the atom's virtual
 value is strictly below the limiting virtual value at the bottom
 of that interval.

 \item[(ii)] \emph{When a basic product is necessary.}
 There exist constants $0<\overline a_E<\underline m_S<1$,
 depending only on the cost function and $\beta$, such that if
 $\overline a_E<a<m<\underline m_S$, every posterior distribution
 maximizing $J_\beta$ has a lower trading atom and full participation.
 If, in addition, $m\ge a(1-\log a)$, the lower trading atom is
 at a posterior valuation $v>a$.

 \item[(iii)] \emph{Exclusion despite gains from trade.}
 There exists a bound $\bar m_E\in(0,1/2]$, depending only on the
 cost function and $\beta$, such that if $0<a<m<\bar m_E$ and
 $m\ge 2a\bigl(1-\log(2a)\bigr)$, every posterior distribution
 maximizing $J_\beta$ has an excluded atom at $a$.
\end{enumerate}
\end{theorem}

\begin{figure}[p]
 \centering
 \captionsetup{font=small,skip=4pt}
 \includegraphics[width=\textwidth]{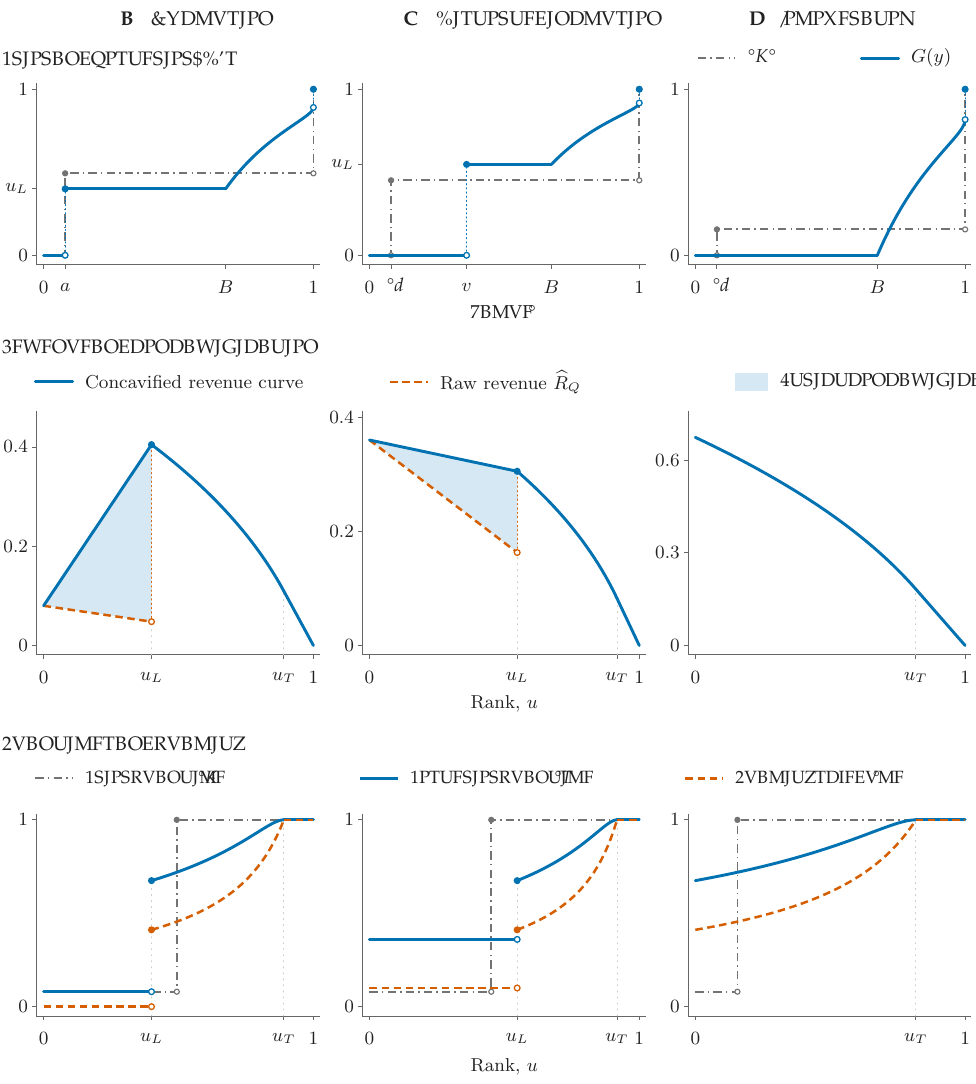}
 \caption{The possible structures of optimal posterior distributions in \Cref{thm:binary-structure}(i). Gray dash--dotted curves show the binary prior $H$ in the first
 row and its quantile $H^{-1}$ in the third. The prior and
 posterior have the same mean in each column.
 The lower atom has mass $u_L$ and value $a$ in (a) or $v$ in (b).
 $B$ is the bottom of the separating region, and $u_T=G(1-)$. Shading marks strict concavification of the raw revenue curve.  The terminal linear part $R_Q=\widehat R_Q=1-u$
 represents the efficiently served top atom.}
 \label{fig:binary-structures}
\end{figure}
Lower-quality products constrain what the seller can charge buyers
with higher posterior values, who can choose those products instead.
A buyer's payoff therefore depends partly on what other buyers learn. Their demand can make it profitable for the seller to offer a
lower-quality product that raises her rent even though she does
not purchase it.

The differentiation in part~(i) serves this purpose. First, no learning leaves every buyer with posterior value
$m$, so the seller supplies efficient quality and
extracts the entire surplus. Some
learning can instead create information rents while making
any loss in total surplus small relative to those
rents. This improves the designer's objective
when $\beta<1$. Full revelation of the two underlying values is also
suboptimal. A buyer who learns
that her valuation is highest can benefit when other
buyers remain uncertain because their demand can make a
product close to the highest quality profitable.
That product gives her a better alternative and
forces the seller to lower the price of the highest
quality.

Moreover, if buyers with
different posterior values receive the same positive quality, a
change in learning can induce the seller to supply more quality to
those who value it more and less to those who value it less, while
preserving the prior mean. For sufficiently small changes, the gain
in consumer surplus outweighs any profit loss in the designer's
objective.

Optimal learning therefore sustains a continuous range of products
extending to the highest quality. A positive share of buyers learn
that they have the highest valuation and purchase that quality.
Their quality is efficient because there are no buyers with higher
posterior values whose rents the seller could reduce by degrading it.

This contrasts sharply with \citet{bergemannheumannmorris2026},
where the seller chooses both information and the menu, pools
valuations to reduce buyers' rents, and offers finitely many
products even with a continuous prior.\footnote{A designer who
maximizes profit solves the same problem as a seller who
chooses both information and the menu. With a binary prior, maximal profit can be attained
with at most two posterior values and hence at most
two products.} Here, information is chosen
to favor consumers, and a continuous range is required even with
a binary prior. Holding other market conditions fixed, our results
imply greater product variety when information is designed to serve
consumers than when it is chosen to maximize
profit. In markets where intermediaries shape consumer learning,
product variety can therefore provide
analysts and regulators with evidence about whose
interests these intermediaries serve.

The comparison with \citet{kyriazis2026markets} helps explain
why a separate lower group may be optimal, whether its buyers
purchase a basic product or are excluded. In that analysis, there is no Bayes plausibility constraint and the designer is free to choose any distribution in $\Delta([a,1])$.
The optimum then
has full participation and a separating range ending
at an efficiently served atom at the highest valuation.
There is no lower atom or support gap.

In our problem, the designer must also preserve the prior mean. Higher posterior values for some buyers require lower ones for others. A differentiated range that
generates rents may have an average posterior value
above the prior mean. The designer cannot then sustain
it for all buyers. A separate group with lower posterior
values allows that range to be sustained for the
remaining buyers, without lowering posterior values
throughout it.

 Along a separating
range, posterior values cannot simply be lowered while
preserving the seller's quality schedule. A group of buyers
receiving a common quality can provide room to reduce
the mean without changing qualities and profit.
But if that quality is positive, buyers in the group
must share a posterior value at an optimum. Otherwise,
buyers who value quality differently would receive
the same quality, leaving the opportunity for beneficial
differentiation described above. The lower group
therefore forms an atom. If it is excluded, its buyers
must instead all have posterior value $a$ since higher
values would contribute more to the mean without
generating surplus.

A key step in our analysis is to show that, holding
qualities and profit fixed, a given reduction in
the mean costs less consumer surplus when it falls
on buyers receiving lower quality. Under the curvature
condition, moving the reduction to these buyers
improves the designer's objective even after the
seller adjusts the menu. The lower group, when present, must therefore
lie below the separating region.

Inducing some buyers to purchase a basic product or excluding them are different ways the designer can benefit consumers. A basic product leaves its own
buyers no rents, but gives buyers of higher qualities
an alternative that limits what the seller can charge
them. The seller weighs the revenue from the basic product
against its production costs and the additional rents that must be
left to buyers of higher-quality products, and may restrict
its quality or leave the group unserved. Distorted inclusion preserves this alternative.
Exclusion can instead allow more buyers to purchase
from the differentiated range and earn rents within
it. Parts~(ii) and~(iii) identify priors for which
each choice is preferable.

In part~(ii), the mean is low enough to require a
separate lower group, but the lowest valuation is
relatively high. Even excluded buyers would contribute
$a$ to the prior mean without generating surplus.
This limits how much exclusion can help sustain
more favorable beliefs elsewhere. Under the stated
conditions, the designer does better by keeping
these buyers in the market and using their product
to support rents. The additional condition concerns
the information needed to do so. At posterior value
$a$, revenue from this group would not cover the
additional rents its product requires elsewhere.
Some buyers with the high underlying valuation must
therefore be pooled into the group, raising its
posterior value above $a$.

In part~(iii), consider a candidate with full
participation and a lower trading atom. Its buyers
have a posterior value below the already low prior
mean, so the seller supplies only low quality.
The basic product therefore generates little rent
for buyers of higher-quality products. Yet, with
$a$ sufficiently small relative to the mean,
sustaining it requires pooling some buyers with
the high underlying valuation into this group.

The designer can instead reveal $a$ to some buyers
and induce exclusion. Their lower posterior value
allows more buyers to move from the basic product
into the differentiated range while preserving the
prior mean. The conditional distribution of posterior
values and the qualities within that range can remain
unchanged. Buyers already purchasing from it lose
the basic product as an alternative and pay more.
But additional buyers now earn the rents generated
within that range. Thus, fewer buyers trade, but
more receive positive information rents. Under the
conditions in part~(iii), this change improves the
designer's objective.

This is in stark contrast with how optimal learning operates in the posted-price problem in \citet{roeslerszentes2017}. There, buyer-optimal learning induces trade with
probability one and achieves ex post efficiency.
This is because an outcome with exclusion can be improved
for the buyer by changing information to restore
trade while preserving seller profit. Here, even
maximizing consumer surplus alone can require
excluding some buyers and leaving gains from
trade unrealized.

The next subsection develops the arguments that
establish the structural results in \Cref{thm:binary-structure}(i). We then characterize
the separating range and the conditions linking its
mean and lowest posterior value to any lower group.

\subsection{The cost of satisfying the prior mean}
\label{subsec:mean-adjustment}

We can think of the designer's problem as choosing both a concavified revenue curve and a posterior distribution
that induces it while respecting the prior mean. The revenue curve
determines the seller's quality schedule through the virtual value. We
first characterize the distributions compatible with a fixed concavified revenue
curve and their attainable means and payoffs. This yields a linear allocation problem. We then use it to compare different curves under mean-preserving
changes in posterior demand. Write $R=R_Q$, suppressing the
subscript when the curve is fixed. We also write
$\operatorname{CS}(Q)$, $\Pi(Q)$, and $J_\beta(Q)$ for the payoffs
of the distribution with quantile $Q$.

Define the \emph{regularization} of $Q$ by
\[
Q^r(u):=\frac{R(u)}{1-u},\qquad 0\le u<1.
\]
This is a regular quantile on $[a,1]$ whose raw revenue curve
equals $R$. It induces the same qualities and profit as $Q$,
while $Q^r\ge Q$ almost everywhere \citep[Proposition~A.3]{kyriazis2026markets}. By
\eqref{eq:seller-profit}--\eqref{eq:consumer-surplus},
\[
\Pi(Q^r)=\Pi(Q),\qquad
\CS(Q^r)-\CS(Q)
=\int_0^1[Q^r(u)-Q(u)]q(\phi_Q(u))\,\dd u\ge0.
\]
Whenever $Q^r$ differs from $Q$, its mean is strictly higher.
We characterize all reductions from $Q^r$ that preserve $R$
and restore the prior mean, together with their costs in
consumer surplus.

Regularization changes the posterior distribution,
whereas interpolation leaves it fixed. At every
optimum in \Cref{thm:binary-structure}, interpolation
across the lower atom, if present, fills any gap
between the raw revenue curve and its envelope,
giving $\widetilde R_Q=R_Q$. Thus, any strict concavification arises within a group of
buyers with the same posterior value and no genuine ironing
is needed.

\paragraph{Where posterior values can be lowered.}
Posterior values can be lowered from $Q^r$ while preserving $R$
only within its linear intervals. This follows because the
concavified revenue curve is linear wherever it lies strictly
above the raw revenue curve. We call a rank $u$ a \emph{contact} if $\widehat R_Q(u)=R(u)$. In words, concavification leaves the raw revenue unchanged there.

Consider a maximal linear interval $I=[\alpha_I,\beta_I]$ of
positive length, with slope $-z_I$. Buyers at interior ranks
receive the same quality $q(z_I)$. Its endpoints are
contacts because the raw revenue curve equals $R$ there.
Let $b_I:=Q^r(\alpha_I)$. Contact at the lower endpoint and
monotonicity prevent posterior values on $I$ from falling
below $b_I$. They cannot exceed $Q^r$, since the raw revenue
curve must remain below $R$. Thus,
$b_I\le Q(u)\le Q^r(u)$ almost everywhere on $I$.

Define
\[
 d_I:=\int_I[Q^r(u)-Q(u)]\,\dd u,
 \qquad
 \kappa_I:=\int_I[Q^r(u)-b_I]\,\dd u.
\]
Here $d_I$ is the reduction in mean on $I$, and $\kappa_I$
is its \emph{mean-reduction capacity}: the largest reduction
possible while preserving $R$. Thus $0\le d_I\le\kappa_I$.
The upper bound is attained by setting $Q=b_I$ throughout
the interior of $I$, leaving the quantile unchanged elsewhere.
This preserves monotonicity and the endpoint contacts, and
hence leaves $R$ unchanged
(Appendix~\ref{app:mean-accounting}).

The adjustment is \emph{at capacity} when $d_I=\kappa_I$,
so all buyers on the interval share posterior value $b_I$.
Otherwise, the unused capacity is $\kappa_I-d_I$.
Conversely, raising posterior values toward $Q^r$ can undo
the reduction and restore up to $d_I$ units of mean.
A terminal linear interval, ending at rank one, has zero
capacity. Since $R(1)=0$, linearity gives $R(u)=z_I(1-u)$
there, so $Q^r=z_I=b_I$ and $\kappa_I=0$.
All sums below run over the maximal linear intervals.

\begin{proposition}[Mean and surplus at a fixed concavified revenue curve]
\label{prop:mean-accounting}
Under \Cref{ass:cost}, a posterior quantile $Q$ on $[a,1]$ and
its regularization $Q^r$ satisfy
\begin{align}
 \int_0^1 Q(u)\,\dd u
   &=\int_0^1 Q^r(u)\,\dd u-\sum_I d_I,
   \label{eq:mean-accounting}\\
 \CS(Q)&=\CS(Q^r)-\sum_I q(z_I)d_I.
   \label{eq:surplus-accounting}
\end{align}
The sums converge absolutely, and $\Pi(Q)=\Pi(Q^r)$.
Conversely, every collection $0\le d_I\le\kappa_I$ can be
implemented by a posterior quantile on $[a,1]$ with concavified
revenue curve $R$ and with the mean and surplus given above. It
is feasible for the binary prior when
\[
 \sum_I d_I=\int_0^1 Q^r(u)\,\dd u-m.
\]
\end{proposition}
On a linear interval $I$, buyers receive the same quality
$q(z_I)$. Reducing their contribution to the mean by $d_I$
therefore lowers total surplus by $q(z_I)d_I$. Buyers now value the same qualities less. Production
costs and profit are unchanged, so the entire loss
in total surplus falls on consumers. The loss depends only on how much the mean falls on the interval.
Which ranks bear the reduction makes no difference.
The converse shows that any reductions within the
capacities $\kappa_I$ can be made on all intervals at
once, while keeping the quantile monotone and the
concavified revenue curve unchanged.

\paragraph{Implementing a reduction in mean.}
Figure~\ref{fig:mean-accounting} illustrates the argument. Start
from $Q^r$. On an interval $I$ with $\kappa_I>0$, lower an
initial portion of the quantile to $b_I$, leaving the rest unchanged.
The area between the quantiles is $d_I$. Lengthening the constant
portion increases this area continuously from zero to $\kappa_I$.
A horizontal portion represents a posterior atom, and the jump back
to $Q^r$ leaves a gap in posterior support. A partial reduction,
$0<d_I<\kappa_I$, leaves distinct posterior types receiving the
same quality. The full reduction, $d_I=\kappa_I$, makes the
entire interval one posterior type.

\begin{figure}[!t]
 \centering
 \captionsetup{font=small,skip=4pt}
 \includegraphics[width=\textwidth]{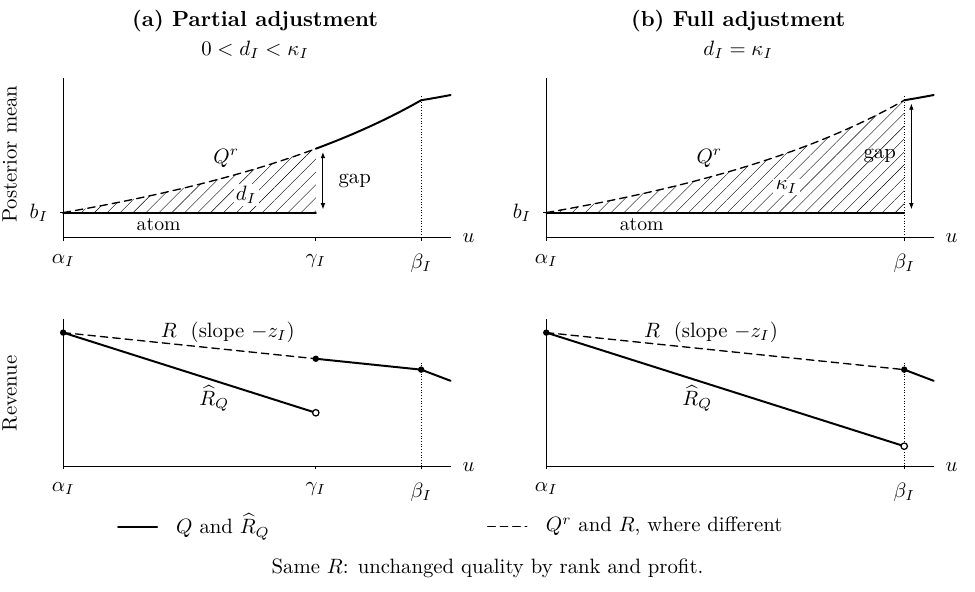}
 \caption{Implementing a reduction in mean at a fixed concavified
 revenue curve. Outside $I$, both quantiles remain equal to $Q^r$.
 In (a), the constant portion ends at $\gamma_I$ and the hatched
 area is $d_I$. In (b), it extends across $I$ and the area is the
 maximum reduction $\kappa_I$. The lower row shows the corresponding
 raw revenue curves.
 Relative to $Q^r$, the two distributions
 remove $d_I$ and $\kappa_I$ units of mean, respectively.}
 \label{fig:mean-accounting}
\end{figure}

\paragraph{The least costly adjustment at a fixed curve.}
Once $R$ is fixed, so are $Q^r$, the qualities $q(z_I)$, and
the bounds $\kappa_I$. The remaining choice is how much mean to
remove on each interval.

A distribution with mean $m$ can induce this curve if and only if
\[
 0\le\int_0^1Q^r(u)\,\dd u-m\le\sum_I\kappa_I.
\]

When these inequalities hold, maximizing $J_\beta$ among distributions
that induce $R$ is equivalent to
\begin{equation}
 \begin{aligned}
 \min_{(d_I)_I}\quad &\sum_I q(z_I)d_I\\
 \text{subject to}\quad
 &0\le d_I\le\kappa_I\quad\text{for every }I,\\
 &\sum_I d_I=\int_0^1Q^r(u)\,\dd u-m.
 \end{aligned}
 \label{eq:conditional-mean-problem}
\end{equation}
The program has one variable $d_I$ for each maximal
linear interval of $R$, and hence finitely or countably
many variables. Appendix~\ref{app:accounting-completion}
shows that the minimum is attained whenever the
problem is feasible.

Its solution assigns reductions to lower qualities first. Suppose
$q(z_I)<q(z_J)$, $d_I<\kappa_I$, and $d_J>0$. The designer
can increase the reduction on $I$ by a small $\varepsilon>0$
and decrease it on $J$ by the same amount. Their mean effects cancel, and both preserve $R$ and profit. The net gain in consumer surplus is
$[q(z_J)-q(z_I)]\varepsilon>0$.
Thus an optimal choice has $d_I=\kappa_I$ on every lower-quality
interval whenever a higher-quality interval has a positive reduction.
Intervals assigned the same quality have the same unit cost, so
reallocating reductions among them leaves the designer's objective unchanged.
Figure~\ref{fig:mean-accounting-exchange} illustrates this comparison.
The allocation rule is independent of $\beta$ conditional on $R$,
although the optimal curve generally depends on the profit weight.

\begin{figure}[!t]
 \centering
 \captionsetup{font=small,skip=4pt}
 \includegraphics[scale=0.9]{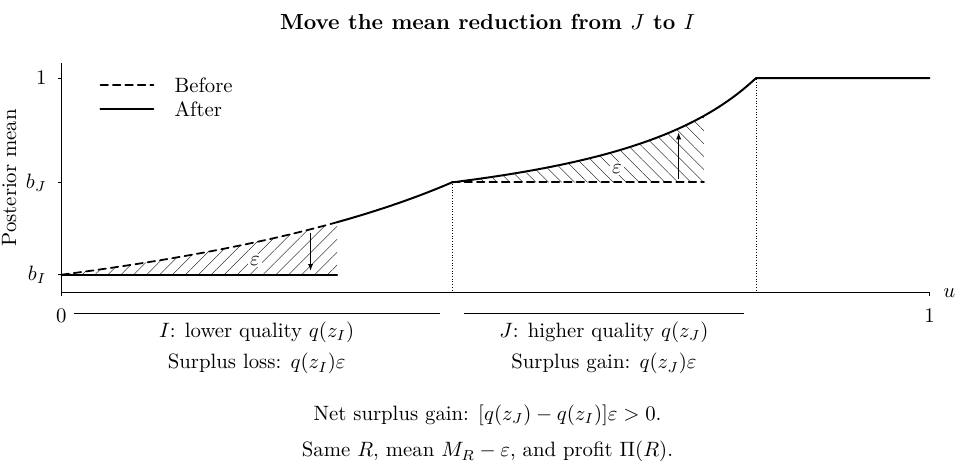}
 \caption{Reallocating a reduction in mean. The dashed quantile
 places the entire reduction $\varepsilon$ on $J$ and the solid
 quantile places it on $I$, where $q(z_I)<q(z_J)$. The two
 quantiles coincide on the portions drawn only as solid curves.
 The hatched areas are equal: moving from dashed to solid removes
 mean on $I$ and restores it on $J$. Both distributions remove
 $\varepsilon$ units of mean from $Q^r$ and induce the same $R$.
 Here $M_R$ is the mean of $Q^r$, and $\Pi(R)$ is their common profit.
 Consumer surplus rises by $[q(z_J)-q(z_I)]\varepsilon$.
 The drawing uses the same $R$ as Figure~\ref{fig:mean-accounting},
 with $\varepsilon<\min\{\kappa_I,\kappa_J\}$.}
 \label{fig:mean-accounting-exchange}
\end{figure}

If $q(z_I)=0$, lowering posterior values on $I$ has no surplus
cost. These buyers are already excluded. Holding $R$ fixed also
fixes quality at every rank and hence the probability of exclusion.
Changing participation requires
changing $R$, which can affect qualities, rents, and profit elsewhere.

\paragraph{Choosing the concavified revenue curve.}
The linear program solves the designer's problem conditional on
$R$. To choose the curve, we construct mean-preserving deviations
and evaluate them under the seller's reoptimized menu. The
fixed-curve result tells us how to restore the prior mean
and what this costs consumers. Together, these
arguments establish the structural results in \Cref{thm:binary-structure}(i).

First, suppose buyers with different posterior values receive
the same positive quality. A mean-preserving contraction within
this group leaves $R$ and payoffs unchanged and produces an
interval on which $Q=Q^r$, with posterior values still varying.
A further mean-preserving change alters $R$, inducing the seller
to supply more quality to buyers with higher posterior values
and less to those with lower values. Total quality supplied is
unchanged to first order, as are production costs because these
buyers initially receive the same quality. To first order, total surplus rises because quality shifts
toward buyers who value it more. The first-order increases and decreases in virtual
values cancel in aggregate. At the common initial
virtual value $z$, $\pi'(z)=q(z)$ equals the quality
initially supplied to every buyer in the group.
Profit is therefore unchanged to first order, so
consumers receive the gain. Hence, at an optimum, buyers on every
nonterminal positive-quality linear interval share posterior
value $b_I$, and $d_I=\kappa_I$.

Second, a nonterminal trading atom cannot lie above the separating region.
We can raise posterior values at the atom toward $Q^r$, keeping
qualities and profit fixed, and offset the increase in mean
by lowering posterior values on a short separating block
below it to that block's lowest value. At the original qualities, the gain in consumer
surplus at the atom exceeds the loss in the block.
The same amount of mean is added where quality is
higher and removed where quality is lower. The seller responds to this pooling by
replacing the block's virtual values with their average.
Under \Cref{ass:curvature}, the resulting quality adjustment
weakly offsets the direct loss in the designer's objective.
The combined change therefore improves the objective.
The terminal atom has $d_I=\kappa_I=0$, so its posterior
values cannot be raised while preserving $R$.

Third, consider two adjacent lower trading atoms.
Raise posterior values on part of the higher atom to
$Q^r$, preserving $R$, and lower the remaining buyers'
posterior values to that of the lower atom. Choose
the two parts so that the mean is unchanged. Both parts initially receive the same
quality, so the direct gain from raising posterior values
on one part exactly offsets the loss from lowering them
on the other, evaluated at their original quality.
The seller responds to pooling by replacing the two
virtual values on the pooled block with their weighted
average and supplying a common quality between the
original qualities. Under \Cref{ass:curvature}, this
adjustment strictly raises the designer's objective.

These arguments leave at most one lower trading atom that must lie below any separating region. To see why the separating range must
be nonempty, notice that if only a single atom lay above the lower
groups, pooling a small mass from it and the group
immediately below it could create a posterior value just
below the highest one and sustain a product close to
the highest quality. This gives buyers who still purchase
the highest quality a better alternative and lowers their
price. Any loss in total surplus can be made small relative
to the additional rents, so the change improves the objective
when $\beta<1$. No learning is also suboptimal since a small
mean-preserving spread around the prior mean creates
rents to first order, with a smaller-order loss in total surplus.

Any excluded buyers must have posterior value $a$.
Otherwise, lowering their posterior values to $a$ allows
more buyers to trade while preserving the prior mean.
The conditional distribution among trading buyers, their
menu, and payoffs per trading buyer can remain unchanged.
Lowering excluded buyers' posterior values has no direct
surplus cost, so expanding the trading group strictly
raises the designer's objective.

Finally, exclusion and a lower trading atom cannot coexist.
When the excluded group's virtual value is strictly negative,
adjusting group shares allows the designer to vary the basic
group's posterior value and the quality the seller offers it,
while preserving the prior mean and the conditional distribution
in the differentiated range. The proof shows that, at a candidate optimum with both
groups, the objective must remain unchanged
for every sufficiently small change in this posterior
value at fixed quality. The designer can therefore
raise the posterior value without a payoff loss.
This increases the marginal return to basic quality,
so a further small increase in quality improves
the objective. If the excluded group's virtual
value is zero, the preceding adjacent-atom argument gives
a strict improvement instead.

\subsection{The separating interval}
\label{subsec:separating-interval}
We now characterize the regular upper distribution through
variations that preserve its contribution to the prior mean and
its attachment to any lower group. It consists of a separating
region and the atom at one. When no lower group is present, this
is the whole posterior distribution. Otherwise, we condition on
the upper part and normalize its probability to one. Then, $u$
denotes rank within that conditional distribution. Its probability
in the full market remains part of the designer's choice. Write
$\varphi$ for its virtual-value profile and $Q$ for its quantile.
Since its raw revenue curve is already concave, each posterior
mean is the average virtual value above its rank:
\begin{equation}
 Q(u)=\frac{1}{1-u}\int_u^1\varphi(t)\,\dd t.
 \label{eq:continuation-quantile}
\end{equation}
Two things govern how this upper distribution fits into the full
posterior distribution. Its lower support point determines the contact with the
lower group and its mean determines its contribution to the prior mean, after multiplying by the probability of the upper part. They are given by
\[
 B:=Q(0)=\int_0^1\varphi(u)\,\dd u,
 \qquad
 M:=\int_0^1Q(u)\,\dd u
   =\int_0^1[-\log(1-u)]\varphi(u)\,\dd u.
\]
We first vary the upper distribution while keeping its probability, lowest support point $B$ and mean $M$ fixed. The perturbations we consider are confined to the separating
interior and leave the endpoint virtual values unchanged, preserving
the attachment to any lower group. We now use $J_\beta$
to denote the weighted objective of the normalized upper distribution.

\paragraph{Two multipliers.}
A change in virtual values generally changes both $B$ and $M$.
Two multipliers account for the payoff effects of offsetting these
changes elsewhere. With constant terms omitted, we use the Lagrangian
$J_\beta-\theta M+\xi B$. For admissible interior perturbations
at an optimum, the first-order payoff change satisfies
\[
 \dd J_\beta=\theta\,\dd M-\xi\,\dd B.
\]
Thus, $\theta$ is the marginal value of increasing $M$
while holding $B$ fixed, and $-\xi$ is the marginal value
of increasing $B$ while holding $M$ fixed. These values
refer to the normalized upper distribution, with its
probability in the full market fixed.

Both $B$ and $M$ are linear in $\varphi$. $B$ integrates
virtual values with unit weight, whereas $M$ weights them
by $-\log(1-u)$. Perturbations at different quantiles change $B$
and $M$ in different proportions. We can therefore
choose two interior perturbations with linearly
independent effects and vary their sizes to keep
both quantities fixed after any other small interior
change. We preserve monotonicity by using small
increasing transformations of virtual values.
Since the combined change is feasible in either
direction, it must leave payoff unchanged to first
order at an optimum. This gives the multiplier
identity and interior stationarity condition below.

\paragraph{The cumulative marginal effect.}
Define
\[
 A(u):=\xi+\int_0^u\frac{q(\varphi(t))-\theta}{1-t}\,\dd t.
\]
To understand this term, raise virtual values over a short interval
around rank $u$, before making the compensating changes. By
\eqref{eq:continuation-quantile}, this raises posterior means at
all ranks below that interval, even though their virtual values
and qualities are unchanged. At a lower rank $t$, the increase in $Q(t)$ equals the integrated
increase in $\varphi$ divided by $1-t$. In the payoff formula, each unit increase in
$Q(t)$ contributes $q(\varphi(t))$ to consumer surplus and
nothing to profit. It also increases the mean, which is valued
at $\theta$. As the altered interval shrinks to $u$, these net effects give the
integral in $A(u)$. The same change enters
$B=\int_0^1\varphi$ with coefficient one, giving the remaining
term $\xi$. Thus, $A(u)$ records the marginal effect on rents at lower ranks,
after accounting for the changes needed to preserve $B$ and $M$.\footnote{The unadjusted cumulative rent term appears
in \citet{kyriazis2026markets}. Here the two multipliers
account for preserving the mean $M$ and the lower
support point $B$.}

\begin{proposition}[The optimal separating distribution]
\label{prop:separating-conditions}
Under the hypotheses of \Cref{thm:binary-structure}, each normalized
optimal upper distribution has unique multipliers $\theta,\xi\in\R$
and a cutoff $b\in(0,1)$ such that, for $0<u<b$,
\begin{equation}
 A(u)+[Q(u)-\varphi(u)]q'(\varphi(u))-(1-\beta)q(\varphi(u))=0.
 \label{eq:continuation-euler}
\end{equation}
On this interval, $Q>\varphi>0$, $q(\varphi)\ge\theta$, and the
three functions are continuously differentiable, with
\begin{align}
 Q'(u)&=\frac{Q(u)-\varphi(u)}{1-u},
 \label{eq:continuation-Q}\\
 \varphi'(u)&=
 \frac{q(\varphi(u))-\theta+[Q(u)-\varphi(u)]q'(\varphi(u))}
 {(1-u)\{(2-\beta)q'(\varphi(u))-[Q(u)-\varphi(u)]q''(\varphi(u))\}}>0.
 \label{eq:continuation-phi}
\end{align}
At $b$, they reach $(A,Q,\varphi)=((1-\beta)\bar q,1,1)$ continuously and
$Q=\varphi=1$ on $[b,1]$. For a given $\theta$, the
solution traced backward from this terminal state is unique up to
its stopping point and rank normalization.
\end{proposition}
Equation~\eqref{eq:continuation-euler} balances the indirect effect
$A(u)$ against two local effects of raising virtual value.
The term $[Q(u)-\varphi(u)]q'(\varphi(u))$ is the gain from
the induced increase in quality. At the seller's choice,
$Q(u)-\varphi(u)$ is the marginal total surplus from additional
quality. The term $-(1-\beta)q(\varphi(u))$ captures rent
extraction. Holding quality fixed, an increase in virtual value
reduces the consumer-surplus integrand and raises the profit
integrand by the same amount, while the designer places weight one on
the former and $\beta$ on the latter. These two terms and $A(u)$ must sum to zero because
an interior change that preserves $B$ and $M$ cannot
improve the objective.

The inequality $q(\varphi)\ge\theta$ compares the cost of
reducing the mean locally with the value of restoring it
elsewhere. Lower posterior values on a short separating
interval to its lowest value. As the interval shrinks,
the consumer-surplus loss per unit of mean reduced
approaches the quality supplied there, while the profit
change becomes negligible by comparison. Restoring the
mean elsewhere in the separating region, with $B$ fixed,
changes the objective by $\theta$ per unit restored, to
first order. If the original quality were below $\theta$,
the combined change would improve the objective.
Together with \Cref{ass:curvature}, this inequality makes
the derivative in \eqref{eq:continuation-phi} strictly
positive.

If a lower trading atom is present, with virtual value $z$, and
$\underline\varphi:=\lim_{u\downarrow0}\varphi(u)$, the same
argument implies
$q(z)\le\theta\le q(\underline\varphi)$, $z<\underline\varphi$.
Consider raising some posterior means at the lower atom while
keeping the concavified revenue curve fixed. Qualities and profit
remain unchanged, and consumer surplus rises by $q(z)$ per unit
increase in the overall mean. To preserve the prior mean, this
increase must be offset by lowering posterior means in the
separating region, at a marginal cost of $\theta$ to the
objective. Optimality therefore requires $q(z)\le\theta$.

To see why quality must jump above the basic product,
suppose the two endpoint qualities were equal, so
$\theta=q(z)$. We can move a small share of buyers from the
atom into the differentiated range, while keeping $B$ and
the atom's posterior value fixed, and we can lower posterior
values within the range to preserve the prior mean.
The basic product now sells to fewer buyers while
requiring rents for more buyers, so the seller reduces
its quality. Additional buyers obtain rents within
the differentiated range, but the ones already in the range
lose some rents from the basic product. Since $\theta=q(z)$, adjusting the mean costs $q(z)$
per unit. Nevertheless,
expanding the differentiated range raises the
objective at rate
$(1-\rho)(B-z)^2q'(z)>0$,
where $\rho<1$ is the share of buyers purchasing
from that range. The basic product must therefore
have strictly lower quality than the quality at
which separation begins.

To see why there must be an atom at the top,
consider raising virtual values near rank one.
This also raises posterior values farther down
the distribution without directly changing the
qualities supplied there. The term $A(u)$ captures
the benefit of these increases after the adjustments
needed to keep $B$ and $M$ fixed. If the separating region
continued all the way to rank one, this marginal
benefit would grow without bound. The local cost
of extracting more rent, however, would remain
bounded. Raising virtual values further would then
improve the objective, so such a distribution
cannot be optimal. The separating region
must end while a positive share of buyers remains.
Their posterior value must be one, since otherwise, a further
increase near the top, with appropriate changes to preserve
$B$ and $M$, would improve the objective. These buyers
learn that they have the highest valuation and receive efficient quality.

A jump in quality below the top atom would also be
suboptimal. Lowering virtual values slightly on the first ranks
of the top atom gives a local gain from reduced rent
extraction. At the initial posterior values, the small
reduction in efficient quality has no first-order
efficiency cost. The Euler condition just below a jump
implies $A(b)<(1-\beta)\bar q$, so the change raises
the objective even after accounting for the indirect
effects and the adjustments preserving $B$ and $M$.
Quality must therefore approach the highest quality
continuously. \Cref{lem:binary-top-fit} establishes
these terminal conditions.

On the separating region, \Cref{ass:curvature} makes
the left-hand side of \eqref{eq:continuation-euler}
strictly decreasing in $\varphi$ for given $(A,Q)$.
The equation therefore determines $\varphi$ locally
as a continuously differentiable function of $(A,Q)$.
For fixed $\theta$, setting $\tau=-\log(1-u)$ gives
the autonomous system:
\[
 \frac{\dd A}{\dd\tau}=q(\varphi)-\theta,
 \qquad
 \frac{\dd Q}{\dd\tau}=Q-\varphi.
\]
These equations determine a unique
path backward from the terminal state. Where we stop
the path determines $B$, $M$, and the probability of the top atom within this distribution.
The designer must still choose $\theta$, the stopping point,
the probability of this upper component, and any lower
group. These choices must satisfy the prior mean and
the seller's incentives. Comparing payoffs across the
feasible choices then determines the optimal posterior
distributions.

\subsection{Two special cases}
We now consider two cases in which the lower part has a simpler
description. The first case permits a regular distribution among buyers who
trade, with its probability adjusting to the prior, while the
second yields no lower atom or exclusion.

\subsubsection{A payoff-neutral lowest value}
\label{subsec:neutral-optimum}

The first case we consider allows the designer to satisfy the prior mean
by varying the probability of a regular active market, rather than
introducing a lower trading atom. An \emph{active distribution} is
one under which the seller supplies positive quality to almost every
buyer. The simplification arises when the remaining buyers can be
assigned a posterior value at which even efficient quality is zero.

Allow marginal cost  at zero to be positive, that is, $c'(0)=d\in[0,1)$, and retain the other
conditions in \Cref{ass:cost}. The lowest value is payoff neutral
when $a\le d$. Under the normalization $d=0$, this
requires $a=0$.
Suppose $a\le d$, and let $G_A$ be a regular active
distribution on $[a,1]$, with mean $M_A$. Buyers at $a$
cannot be served profitably. They receive zero quality, so their presence creates
no additional rents for the rest of buyers. Changing
their share changes the size of the trading population
but leaves its composition unchanged. The profit from
every menu then changes in the
same proportion.
The seller consequently offers the same qualities and
prices \citep[Lemma~B.1]{kyriazis2026markets}.
For $0<\lambda\le1$,
\[
 \begin{gathered}
 G_\lambda=(1-\lambda)\delta_a+\lambda G_A,
 \qquad m=a+\lambda(M_A-a),\\
 \CS(G_\lambda)=\lambda\CS(G_A),\qquad
 \Pi(G_\lambda)=\lambda\Pi(G_A).
 \end{gathered}
\]

We can prove that every feasible
posterior distribution is weakly improved by one of this form, and
that every optimum already has this form. Once $G_A$ is chosen,
Bayes plausibility fixes $\lambda=(m-a)/(M_A-a)$. The mixture $G_{\lambda}$ is a valid
probability distribution precisely when $M_A\ge m$, so the designer's problem
reduces to
\begin{equation*}
 V_{a,\beta}(m)=(m-a)
 \max_{\substack{G_A\text{ regular and active}\\M_A\ge m}}
              \frac{J_\beta(G_A)}{M_A-a}.
\end{equation*}
The ratio measures the weighted payoff per unit of mean
above $a$. A higher $M_A$ may raise the payoff per trading
buyer, but leaves room for fewer buyers to trade at the
given prior mean. When $a<d$, the denominator also includes
the distance from $a$ to the threshold $d$ at which
positive quality becomes efficient. Raising posterior
values over this distance increases the mean without
generating surplus.

To prove that it is without loss to restrict attention to regular active distributions and a possible excluded atom at $a$, we do not need any curvature assumption. For the characterization and uniqueness result below, we impose
the following condition.
\begin{assumption}[Cost curvature for payoff-neutral learning]
\label{ass:neutral-curvature}
The cost function is three times continuously differentiable on
$(0,\bar Q]$, and
\[
 \frac{x c'''(x)}{c''(x)}>-2,\qquad 0<x\le\bar q.
\]
\end{assumption}
This condition holds whenever marginal cost is convex and is satisfied by power
cost functions, for which the ratio equals $p-2$. It is distinct from
\Cref{ass:curvature} and permits subquadratic powers.
\begin{proposition}[Optimal learning with a payoff-neutral lowest value]
\label{prop:neutral-participation}
Fix $0\le\beta<1$ and suppose $0\le a\le d<1$. Maintain the
technology specified in this subsection and
\Cref{ass:neutral-curvature}. The optimal posterior distribution is unique
at every $a<m<1$. It consists of a regular active distribution and,
possibly, an excluded atom at $a$. The active distribution has a
connected strictly separating interval followed by a positive,
efficiently served atom at one. There is no lower
trading atom, and exclusion may occur only at a payoff-neutral type. Moreover, there are a critical mean $m_{a,\beta}^c\in(d,1)$ and a regular
active distribution $G_{a,\beta}^c$ of that mean such that, for
$a<m\le m_{a,\beta}^c$,
\[
 G_{m,\beta}^*=\left(1-\frac{m-a}{m_{a,\beta}^c-a}\right)\delta_a
                +\frac{m-a}{m_{a,\beta}^c-a}G_{a,\beta}^c.
\]
For $m>m_{a,\beta}^c$, the optimal distribution is regular and every buyer
trades. Finally, $V_{a,\beta}$ is concave in the prior mean.
\end{proposition}

For low prior means, the best choice is $G_{a,\beta}^c$,
which maximizes payoff per unit of mean above $a$.
An increase in $m$ then expands the share of buyers who trade without
changing the conditional distribution or the menu.
The objective is therefore linear in $m-a$. At
$m_{a,\beta}^c$, everyone trades. Above this threshold,
participation cannot increase further. Learning must
change so that the active distribution itself has mean $m$.

\Cref{prop:neutral-weighted-value} establishes that the
value function is strictly concave above $m_{a,\beta}^c$.
At the threshold, a higher mean still raises the
designer's payoff. Near one, it lowers it.
Posterior values cannot exceed one and must average
to the prior mean, so a mean close to one requires
almost all buyers to have posterior values close
to one. Buyers then differ little in their willingness
to pay, leaving less scope for differentiation and
information rents.
At $m=1$, all posterior values equal one and the seller
extracts the entire surplus: $V_{a,\beta}(1)=\beta\pi(1)$. For every $\beta<1$,
the objective therefore reaches its unique maximum
strictly between $m_{a,\beta}^c$ and one, then declines.
Figure~\ref{fig:mean-value-common-optimum} illustrates
this dependence on the prior mean.

This also identifies the connection with the unrestricted
market composition problem in \citet{kyriazis2026markets}.
Holding the cost, type interval, and profit weight fixed, allowing
the designer to also choose the mean gives
\[
 \max_{a<m<1}V_{a,\beta}(m)
 =
 \max_{G\in\Delta([a,1])}J_\beta(G).
\]
Every distribution on $[a,1]$ is feasible in the binary problem
at its own mean. At the unique mean maximizing $V_{a,\beta}$, the optimal posterior
distribution coincides with the optimal composition in \citet{kyriazis2026markets}. For
every other mean there is a strict loss. This mean differs from $m_{a,\beta}^c$ because
the critical distribution maximizes payoff per unit of mean above $a$,
whereas unrestricted composition maximizes payoff itself.

 For
a fixed prior $H$ and cost function, the \emph{surplus set} is
$\mathcal P(H):=
 \{(\CS(G),\Pi(G)):G\in\mathcal G(H)\}$.
Its Pareto frontier consists of the pairs $(C,P)\in\mathcal P(H)$
for which there is no $(C',P')\in\mathcal P(H)$ satisfying
$C'\ge C$ and $P'\ge P$, with at least one inequality strict.
With a payoff-neutral lowest type, we characterize this frontier
and its generating posterior distributions. Fix $a<m<1$.
For $0\le\beta<1$, let $G_{m,\beta}^*$ be the unique optimum
in \Cref{prop:neutral-participation}, and write
$(C_\beta,P_\beta)=(\CS(G_{m,\beta}^*),\Pi(G_{m,\beta}^*))$.
At $\beta=1$, set $G_{m,1}^*=H$, the full-revelation
posterior distribution. It maximizes both total surplus and profit,
and gives
\[
 (C_1,P_1)=\left(0,\frac{m-a}{1-a}\pi(1)\right).
\]
The low type receives zero quality, and the high type receives
efficient quality with full surplus extraction.

\begin{proposition}[The payoff-neutral Pareto frontier]
\label{thm:neutral-frontier}
\label{cor:neutral-frontier}
Under the hypotheses of \Cref{prop:neutral-participation}, fix
$a<m<1$. A feasible posterior distribution generates a
Pareto-efficient surplus pair if and only if it is $G_{m,\beta}^*$
for some $\beta\in[0,1]$. Each efficient pair is generated by a
unique posterior distribution.
For $\beta<1$, these distributions have the structure in
\Cref{prop:neutral-participation}. The distribution at $\beta=1$ is $H$.
The payoff path $(C_\beta,P_\beta)$ is continuous on $[0,1]$,
with consumer surplus weakly decreasing and profit weakly increasing
in $\beta$. It corresponds to the Pareto frontier of the surplus set, and
every feasible surplus pair is weakly dominated by a point on this path.
\end{proposition}

Uniqueness makes the payoff continuous in the profit weight. The
path therefore contains every profit between $P_0$ and $P_1$.
At each such profit, the corresponding weighted objective bounds
consumer surplus for all other feasible distributions. Payoffs with
profit below $P_0$ are weakly dominated by the buyer-optimal payoff.
Thus, the path covers the whole frontier. Uniqueness also identifies each
generating distribution, although
the supporting weight need not be unique.

Figure~\ref{fig:global-fixed-mean-frontiers} compares
the fixed-mean and unrestricted frontiers for $a=0$
and $c(x)=x^2/2$. Each broken curve fixes the prior
mean and varies the profit weight. The solid curve
allows the designer to choose the mean as well.
At each marked point, the unrestricted optimum has
the specified prior mean, so both problems select
the same posterior distribution. Point A is the
unrestricted buyer optimum. Points B and C are tangencies
at positive profit weights.
Figure~\ref{fig:mean-value-common-optimum} shows
how the weighted payoff varies with the prior mean
at the weight corresponding to C.

\begin{figure}[tbp]
\centering
\captionsetup{font=small,skip=4pt}
\includegraphics[
        width=\linewidth,
        height=0.4\textheight,
        keepaspectratio
    ]{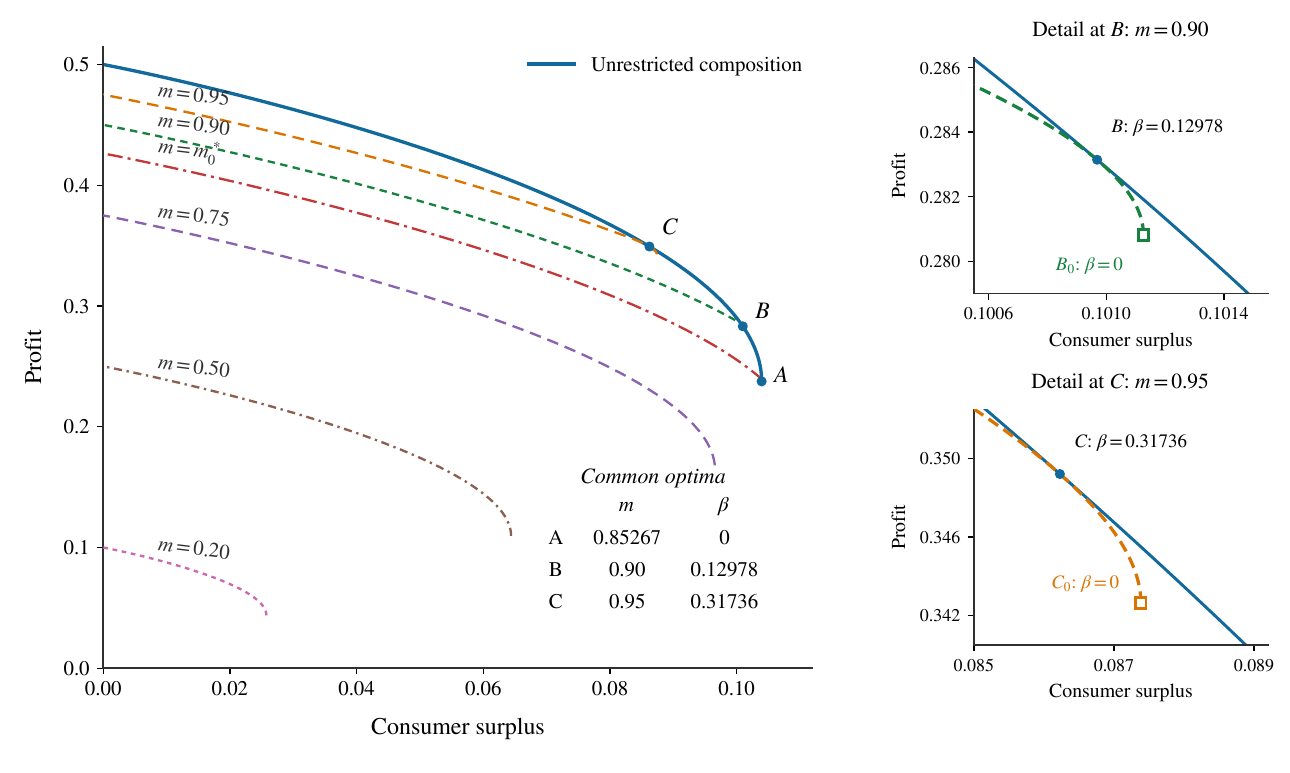}
\caption{Fixed-mean and unrestricted Pareto frontiers under quadratic
cost, with $a=0$. Each broken curve gives the single-market frontier
for the binary prior $(1-m)\delta_0+m\delta_1$ at its labeled mean.
The solid curve gives the frontier when the
distribution on $[0,1]$ can be chosen freely \citep{kyriazis2026markets}.
Points A, B, and C mark common optimal surplus pairs, with the
corresponding means and weights reported in the table. Point A is
the unrestricted buyer optimum. Points B and C are tangencies at
strictly positive profit weights. The enlarged panels show the
continuations beyond B and C to the distinct buyer-optimal endpoints
for the respective fixed means, $B_0$ and $C_0$, where $\beta=0$.
These endpoints are marked by open squares.}
 \label{fig:global-fixed-mean-frontiers}
\end{figure}

\begin{figure}[htbp]
 \centering
 \captionsetup{font=small,skip=4pt}
 \includegraphics[
        width=\linewidth,
        height=0.3\textheight,
        keepaspectratio]{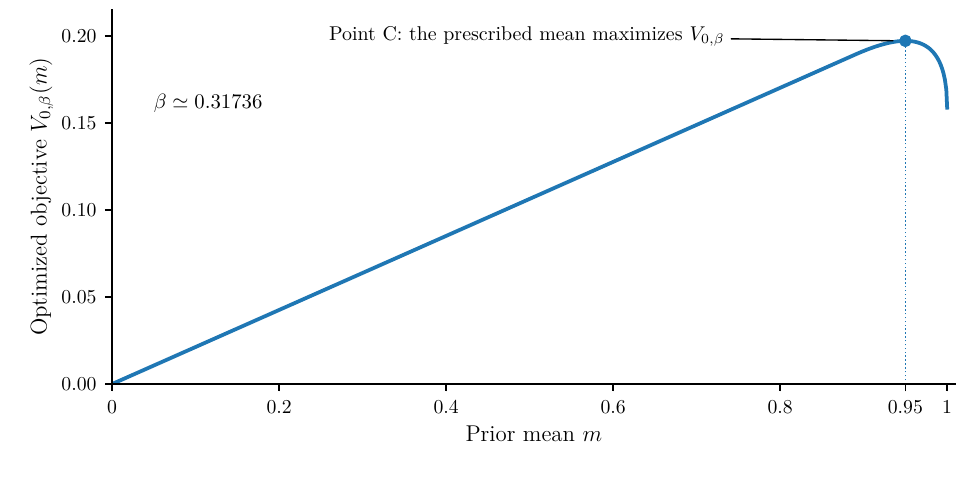}
 \caption{With $a=0$,
 $c(x)=x^2/2$, and $\beta\simeq0.31736$, the value
 $V_{0,\beta}(m)$ is uniquely maximized at $m=0.95$. Choosing
 the posterior distribution at this mean therefore attains the
 unrestricted market-composition value. This corresponds to  point C in
 Figure~\ref{fig:global-fixed-mean-frontiers}.}
 \label{fig:mean-value-common-optimum}
\end{figure}

\subsubsection{High prior means}
\label{subsec:high-mean}

Return to $c'(0)=0$ and the assumptions of
\Cref{thm:binary-structure}. Even with a productive lowest value,
neither exclusion nor a lower trading atom is needed when the prior
mean is sufficiently high, uniformly over any closed interval of profit weights bounded away from one.

\Needspace{12\baselineskip}

\begin{proposition}[Optimal learning with high prior means]
\label{prop:high-mean}
Fix $0\le a<1$ and $0\le\bar\beta<1$. Maintain
\Cref{ass:cost} and suppose \Cref{ass:curvature} holds at weight
$\bar\beta$. There is $\bar m\in(a,1)$ such that, for every
$m\in(\bar m,1)$ and every $\beta\in[0,\bar\beta]$, the
posterior distribution maximizing $J_\beta$ in $\mathcal G_{a,m}$
is unique and regular and every buyer trades. Its lowest posterior value is strictly above
$a$, and it consists of a nondegenerate connected separating
interval followed by a positive, efficiently served atom at one.
This distribution also uniquely maximizes $J_\beta$ in
$\mathcal G_{0,m}$. In particular,
$V_{a,\beta}(m)=V_{0,\beta}(m)$ for $m\in(\bar m,1)$ and $\beta\in[0,\bar\beta]$.
For each such $m$, the optimal distribution is weakly continuous
in $\beta$.\footnote{Weak continuity means
that the expectation of every continuous function of
the posterior value varies continuously with the
parameter. Formally, writing $G_{m,\beta}^*$ for the
optimal distribution, $\beta_n\to\beta$ implies
$G_{m,\beta_n}^*\Rightarrow G_{m,\beta}^*$, where
$\Rightarrow$ denotes weak convergence of probability
distributions.} Consumer surplus is weakly decreasing and
profit weakly increasing in $\beta$.
\end{proposition}

To understand the result, first consider binary priors
on $\{0,1\}$. Near mean one, Bayes plausibility leaves
little scope for the differences in willingness to
pay that generate information rents. For the weights
in the proposition, the proof shows that
$V_{0,\beta}(m)$ is strictly decreasing sufficiently
near one. Exclusion raises the average posterior
value among trading buyers above $m$. It therefore
leaves fewer buyers trading and lowers the maximal
weighted payoff per trading buyer.

An optimum also cannot have a lower trading atom.
Section~\ref{subsec:mean-adjustment} shows that raising
its posterior value can increase consumer surplus
without changing qualities or profit. This would
give the designer a higher payoff at a higher mean,
contradicting the decrease in $V_{0,\beta}$.

The proof then establishes that the regular optimum
is unique and that its lowest posterior value tends
to one as $m$ tends to one. For any fixed $a<1$,
this value eventually exceeds $a$. The same distribution
is then feasible under the prior on $\{a,1\}$.
Restricting posterior values to $[a,1]$ therefore
leaves both the optimal distribution and its payoff
unchanged.

The curvature condition for $\bar\beta$ also holds
at every smaller weight. The distributions constructed in the proof
vary continuously with $\beta$. Compactness of
$[0,\bar\beta]$ then allows us to choose the mean
high enough for the result to hold at all these
weights at once. For such a prior, changing the
weight on profit never makes a separate lower
group optimal.

For a prior satisfying the assumptions of the proposition, write $G_{m,\beta}^*$ for the unique optimum and
$(C_\beta,P_\beta)=(\CS(G_{m,\beta}^*),\Pi(G_{m,\beta}^*))$,
as in the payoff-neutral benchmark. The following result characterizes all efficient pairs with profit between $P_0$ and $P_{\bar\beta}$.

\begin{corollary}[The Pareto frontier at high prior means]
\label{cor:high-mean-frontier}
Under the hypotheses of \Cref{prop:high-mean}, fix
$m\in(\bar m,1)$. The path $(C_\beta,P_\beta)$,
$0\le\beta\le\bar\beta$, is exactly the part of the
Pareto frontier of the surplus set with profit in $[P_0,P_{\bar\beta}]$.
For every $p$ in this interval, some $\beta\in[0,\bar\beta]$
satisfies $P_\beta=p$, and
\[
 \max_{\substack{G\in\mathcal G_{a,m}\\\Pi(G)\ge p}}\CS(G)
   =C_\beta=V_{a,\beta}(m)-\beta p.
\]
Each pair on this path is generated by a unique posterior
distribution. For any profit $p\le P_0$, maximal
consumer surplus is $C_0$.
\end{corollary}

The corollary characterizes only part of the frontier.
When $a>0$, no mean threshold below one works for every
$\beta<1$. To see why, compare priors on $\{0,1\}$ and
$\{a,1\}$ with the same mean $m$. The prior on $\{0,1\}$ puts more
probability on valuation one than a prior with a
positive lowest valuation. It is therefore more
dispersed and yields strictly greater full-information
surplus, since efficient surplus is strictly convex
in valuation. With only valuations zero and one,
full revelation induces efficient provision and
leaves the entire surplus to the seller. As $\beta$
approaches one, the designer puts nearly equal weight
on consumer surplus and profit, so the objective
approaches total surplus. This strict surplus
advantage therefore implies that, at every fixed
$0<a<m<1$, $V_{a,\beta}(m)<V_{0,\beta}(m)$ for weights
sufficiently close to one.

\paragraph{A sufficient condition at a fixed weight.}
A more explicit sufficient condition is available at a fixed profit
weight. Fix $0\le\beta<1$ and maintain
\Cref{ass:cost,ass:curvature}, with $c'(0)=0$. Consider the regular
active distribution on $[0,1]$ that maximizes
$\frac{J_\beta(G)}{\E_G[Y]}$.
This distribution exists and is unique under these assumptions.
Denote it by $G_{0,\beta}^c$ and its mean by $m_{0,\beta}^c$,
consistent with the notation in Section~\ref{subsec:neutral-optimum}. For $0<a<m<1$, if
$m\ge m_{0,\beta}^c$ and $a\le\min\supp G_{0,\beta}^c$, then the unique optimum has no lower atom and
full participation. The ratio identifies the active market that generates
the greatest weighted payoff per unit of prior mean.
For the binary prior on $\{0,1\}$, once
$m\ge m_{0,\beta}^c$, the unique optimum serves everyone
and has no lower atom. Its lowest posterior value is
at least $\min\supp G_{0,\beta}^c$. If
$a\le\min\supp G_{0,\beta}^c$, this same distribution
is therefore feasible for the binary prior on
$\{a,1\}$ with the same mean. It remains uniquely
optimal since the higher lower bound restricts the
designer's choices but does not rule out the
original optimum.

\section{Exclusion and Participation}
\label{sec:participation}

Section~\ref{sec:binary} characterizes the possible optimal
posterior distributions for a fixed binary prior. We now ask how participation changes
as the probability of the high valuation rises. Define participation as
$T(G):=\Pr_G\{x(Y)>0\}$ under the seller's optimal menu. Throughout
this section, we hold $a$ and the profit weight $0\le\beta<1$
fixed while increasing $m$. This raises the prior probability
$(m-a)/(1-a)$ of the high valuation, and the designer reoptimizes
the posterior distribution at each prior.

Under the hypotheses of \Cref{prop:neutral-participation}, the
payoff-neutral benchmark gives\footnote{For $a=c'(0)=0$, the same formula also follows under
\Cref{ass:cost,ass:curvature} from
\Cref{lem:neutral-original-law,lem:fixed-weight-critical-market}.}
\[
 T^*_{a,\beta}(m)=
 \min\left\{\frac{m-a}{m_{a,\beta}^c-a},1\right\}.
\]
Buyers at $a$ cannot support a profitable basic product.
Their share can therefore change without altering
the seller's optimal menu, as long as the conditional distribution
among trading buyers stays fixed. Below the critical
mean, optimal learning keeps this distribution fixed
and expands the set of buyers who trade as $m$ rises.
Above the critical mean, everyone trades and learning
changes to match the higher mean.

\subsection{A productive lowest value}
\label{subsec:productive-exclusion}

Return to $c'(0)=0$ and fix $a>0$. Every feasible
posterior type now has positive gains from trade.
Selling a basic product to buyers with posterior
value $a$ brings in revenue, but also entails
production costs and lowers the prices the seller
can charge for higher qualities. Whether the product
is worth offering therefore depends on how many
buyers purchase it relative to how many benefit
from those lower prices.

With $a$ and one fixed, increasing $m$ makes the
high valuation more common, improving the prior
in the sense of first-order stochastic dominance.
Yet the next result shows that participation can
fall at every optimum, despite gains from serving
everyone. At sufficiently high means, every optimum
again serves all buyers.

\begin{theorem}[Nonmonotonicity of participation in the prior mean]
\label{thm:general-participation}
Fix $0\le\beta<1$. For every cost satisfying \Cref{ass:cost,ass:curvature}, there exist
$0<\bar a<m_H<1$ and $\varepsilon>0$ with the following property.
For each $a\in(0,\bar a)$, there is $\eta_a>0$ such that
$a+\eta_a<m_H$. Choose any $m_L\in(a,a+\eta_a)$, and let
$G_L$ and $G_H$ be any distributions maximizing $J_\beta$ at the priors $(a,m_L)$
and $(a,m_H)$, respectively. The distribution $G_L$ has a lower
trading atom at $a$, while $G_H$ has an excluded atom at $a$. Moreover,
\[
 \begin{gathered}
 T(G_L)=1,\qquad T(G_H)\le1-\varepsilon,\qquad
 \Pi(G_H)>\Pi(G_L),\qquad \TS(G_H)>\TS(G_L).
 \end{gathered}
\]
\end{theorem}
A higher prior probability of the high valuation can therefore reduce the probability of trade
at every optimum, even when information is chosen in consumers'
interests. Profit and total surplus nevertheless rise.

\paragraph{A basic product at low means.}
The designer cannot choose participation independently of the
seller's incentives. Let $s$ be the probability of the active group,
$B$ its lowest posterior value, and $M$ its mean. Serving an excluded
atom at $a$ with a small positive quality generates revenue from
that group, but also gives active buyers a new contract to imitate.
Per unit of quality, the direct revenue is $(1-s)a$, while the
additional rents required by active buyers are $s(B-a)$. Since
$c'(0)=0$, the seller gains from extending service if $a-sB>0$.
The proof also rules out equality, so any optimum with exclusion
must satisfy $sB>a$. In addition, the
optimality conditions governing the active part of the distribution bound $B$ away from $a$. Exclusion thus
requires both a sufficiently large share of buyers who trade and
posterior means for those buyers bounded away from $a$.
Since $m-a=s(M-a)$ and $M\ge B$, these requirements impose a
positive lower bound on $m-a$. Every optimum therefore has full
participation when the prior mean is sufficiently close to $a$.

For fixed $a>0$, almost all buyers purchase the basic
product as $m$ approaches $a$. The average posterior value among buyers of
differentiated products cannot converge to $a$.
As the prior mean approaches $a$, the share buying
those products must therefore go to zero.
Buyers of the basic product receive no rent. Giving them a posterior value
above $a$ would require a lower average posterior
value among buyers of differentiated products.
The designer can instead lower the posterior value
of buyers purchasing the basic product and raise
the mean in the differentiated range. Each unit of
mean transferred there raises its contribution to
the objective by at least the basic quality.
When almost all buyers purchase the basic product,
the profit lost per unit of mean is close to that
quality, but receives only weight $\beta<1$ in
the objective. Lower basic quality also reduces
rents, but only for the few buyers of differentiated
products. The gain therefore exceeds both losses
when the prior mean is sufficiently close to $a$.
Buyers of the basic product then learn $V=a$
with certainty.

Under the conditions of \Cref{thm:binary-structure}(ii),
the additional bound $m\ge a(1-\log a)$ instead requires
the basic product's buyers to have a posterior value
above $a$. That group must therefore include some buyers
with the high underlying valuation.

\paragraph{Why exclusion becomes optimal at a higher mean.}
Fix a mean $m_H$ at which optimal learning with $a=0$
leaves a positive share of buyers unserved. When $a$ is positive but small, serving these buyers
creates little surplus. The seller must earn enough
from their product to cover the additional rents it
leaves to buyers of higher qualities. Revenue and
these rents are both proportional to quality. If
revenue does not cover the rents, cutting quality
cannot make the product profitable. At the fixed
mean $m_H$, full participation requires the lowest
posterior value to remain bounded away from zero.
Giving the lower group sufficiently high posterior
values therefore requires either fewer buyers in
the differentiated range or a lower average posterior
value among those buyers. Such policies cannot converge to an optimum at $a=0$,
which assigns positive probability to posterior value zero.
Exclusion therefore persists at every optimum for small
positive $a$.

At a fixed mean below $\bar m_E$,
\Cref{thm:binary-structure}(iii) guarantees exclusion in
every optimum once $a>0$ is sufficiently small. If instead $a>0$
is fixed and the mean is sufficiently close to $a$, every optimum
offers a basic product to a positive mass of buyers who learn that
their valuation is $a$. A low prior mean alone therefore does not
determine whether the lowest group trades.
Figure~\ref{fig:exclusion-incentive} illustrates the seller's marginal
comparison, holding the conditional upper distribution fixed.

\begin{figure}[htbp]
 \centering
 \captionsetup{font=small,skip=4pt}
 \includegraphics[scale=0.9,width=\textwidth]{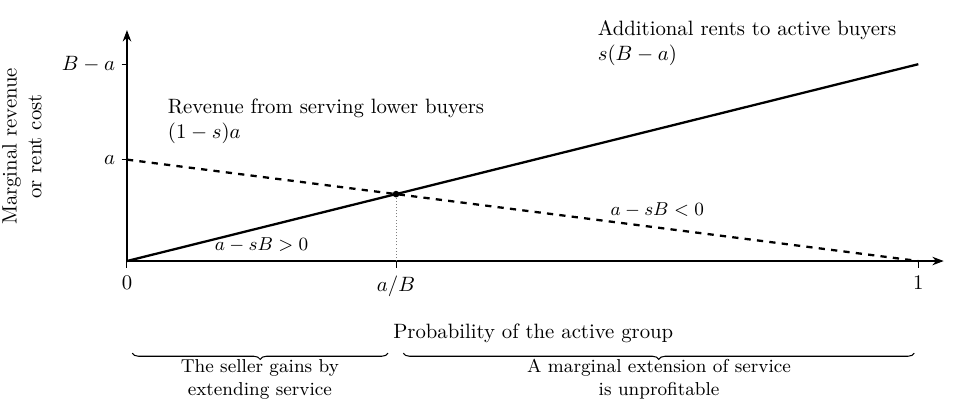}
 \caption{The seller's incentive to serve the lower group. Consider
 a regular active distribution with lowest posterior value $B>a$
 and probability $s$, with the remaining buyers at $a$. Starting
 from a menu that excludes the lower group, a small quality for
 these buyers generates marginal revenue $(1-s)a$ but requires
 additional rents $s(B-a)$ for active buyers. Marginal
 production cost at zero is zero. To the left of $s=a/B$, extending
 service is profitable, so exclusion cannot be seller optimal.
 To the right, a marginal extension is unprofitable.}
 \label{fig:exclusion-incentive}
\end{figure}

Note that the results above concern different priors.
When both $a$ and $m_L-a$ are
small, even first-best surplus is small, bounding profit and total
surplus under every feasible posterior distribution. At $m_H$, the
optimal active market instead generates profit bounded away from
zero. It therefore yields higher profit and total surplus despite
less frequent trade.

\paragraph{The return to full participation.}
Near $m=1$, exclusion raises the average posterior
value among trading buyers and further limits the
scope for information rents. As \Cref{prop:high-mean}
shows, it then leaves fewer buyers trading and lowers
the maximal weighted payoff per trading buyer.

For $a$ as in \Cref{thm:general-participation}, every
optimum therefore serves everyone near $m=a$,
excludes some buyers at $m_H$, and serves everyone
again near $m=1$. Near $m=a$, the basic product's
buyers learn $V=a$. Near $m=1$, no buyer learns
$V=a$ and there is no isolated basic product.
The quadratic case below identifies a discontinuous
switch from full participation to exclusion and
shows how its location depends on $a$.

\subsection{A discontinuous switch under quadratic cost}
\label{subsec:participation-comparative-statics}

Set $\beta=0$ and $c(x)=x^2/2$. Near the switch point
constructed below, only two candidates compete. The distorted inclusion candidate serves everyone,
has a lowest posterior value above $a$, and gives
consumer surplus $C_I(m)$.
The exclusion candidate assigns posterior value $a$ to
buyers who do not trade. Let $M_E(a)$ be the mean posterior
value among trading buyers, and let $\rho_E(a)$ be their
consumer surplus per buyer divided by $M_E(a)-a$.
The mean constraint determines the trading probability,
giving
\[
 s_E(a,m)=\frac{m-a}{M_E(a)-a},
 \qquad C_E(a,m)=\rho_E(a)(m-a).
\]
For each conditional trading distribution, consumer
surplus is $m-a$ times its surplus per unit of mean
above $a$. This common factor does not affect which
distribution is best. Since serving the excluded
group remains strictly unprofitable in this
neighborhood, the maximizing conditional distribution
remains feasible as $m$ changes. Only the trading
share changes.

\begin{proposition}[A switch from inclusion to exclusion]
\label{prop:participation-switch}
For $\beta=0$ and $c(x)=x^2/2$, there are open intervals $\mathcal A$ of
positive lower valuations and $\mathcal M$ of prior means, and a
continuously differentiable, strictly increasing curve
$m^\dagger:\mathcal A\to\mathcal M$, such that
$V_a(m)=\max\{C_I(m),C_E(a,m)\}$
on $\mathcal A\times\mathcal M$.
Below the curve, the unique optimal posterior distribution features distorted inclusion, and all buyers trade. Above it, the unique optimal distribution features exclusion, and $0<s_E<1$. On the curve, only these two distributions
are optimal. Consumer surplus is continuous at the switch, while
seller profit and total surplus jump upward as $m$
crosses $m^\dagger(a)$ from below.
\end{proposition}

\begin{figure}[htbp]
 \centering
 \captionsetup{font=small,skip=4pt}
 \begin{minipage}[t]{0.48\textwidth}
  \centering
  \includegraphics[width=\linewidth]{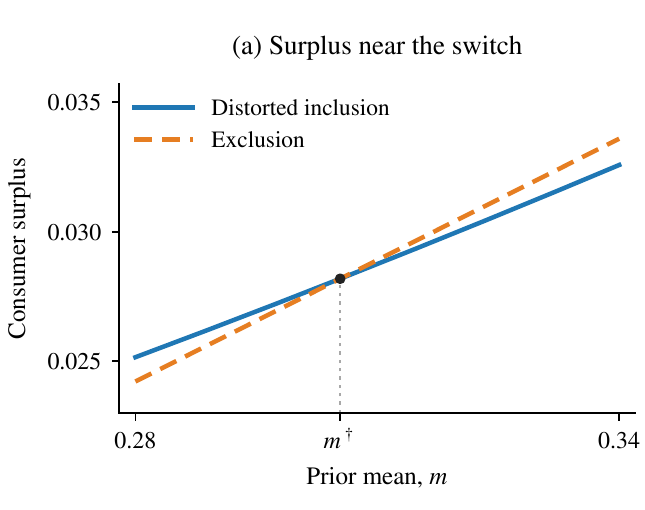}
 \end{minipage}\hfill
 \begin{minipage}[t]{0.48\textwidth}
  \centering
  \includegraphics[width=\linewidth]{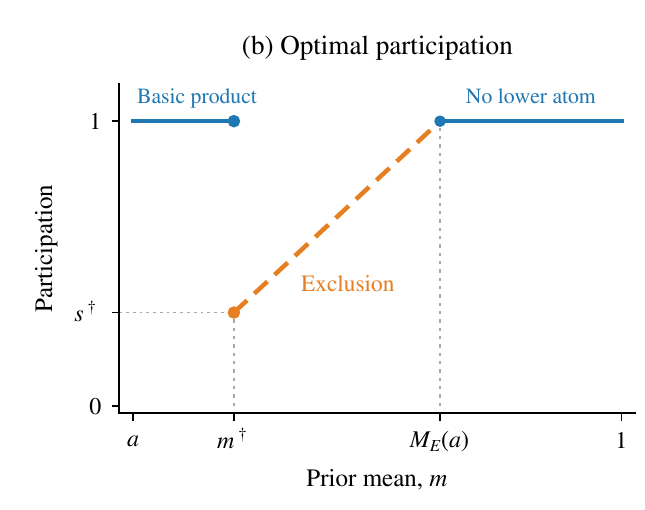}
 \end{minipage}
 \caption{A switch to exclusion and the return to full
    participation. Both panels are computed for $a=1/8$, $\beta=0$, and
$c(x)=x^2/2$. The left panel compares consumer surplus near
the switching mean and the right panel shows optimal participation
over the full range of prior means. At $m^\dagger(a)\simeq0.3054$, the two optimal
    participation rates are one and $0.3280$.
    Participation then rises linearly, reaches one continuously
    at $M_E(a)\simeq0.6749$, and remains one for every higher
    mean. The basic product does not reappear.}
 \label{fig:participation-switch}
\end{figure}

Under exclusion, a higher prior mean allows more buyers
to purchase from the same differentiated menu. Consumer
surplus therefore rises at rate $\rho_E(a)$. Under
inclusion, everyone already trades, so the gain comes
from changing learning and the menu. The proof shows
that this yields a smaller marginal gain at the crossing:
$C_I'(m^\dagger(a))<\rho_E(a)$.
Exclusion therefore becomes optimal as $m$ rises,
even though participation increases with $m$ within
the exclusion regime.

At the switch, fewer buyers trade, but more purchase
from the differentiated range. The seller also no longer
has to leave these buyers the additional rents generated
by the basic product. For the priors in the proposition,
the increase in profit from the differentiated range
outweighs the lost profit on basic-product sales, after
accounting for the accompanying changes in the menu.
Since the two candidates give the same consumer surplus
at the switch, total surplus rises by exactly the
increase in profit.

Increasing the lowest primitive value has the opposite
effect when the mean is fixed. The inclusion distribution
remains feasible and unchanged. Under exclusion, the
optimal distribution among trading buyers generally
changes with $a$. Allowing for this adjustment, the
envelope condition gives
\[
 \frac{\partial C_E(a,m)}{\partial a}
       =-\rho_E(a)[1-s_E(a,m)]<0.
\]
Raising $a$ uses additional mean for excluded buyers,
who generate no surplus. The resulting cost is the
excluded share, $1-s_E(a,m)$, multiplied by the marginal
value of mean under exclusion, $\rho_E(a)$. Exclusion
therefore becomes optimal only at a higher prior mean:
\begin{equation}
 (m^\dagger)'(a)=
 \frac{\rho_E(a)[1-s_E(a,m^\dagger(a))]}
      {\rho_E(a)-C_I'(m^\dagger(a))}>0.
 \label{eq:participation-boundary-slope}
\end{equation}

\begin{corollary}[The value of mean under distorted inclusion]
\label{cor:inclusion-convexity}
On the inclusion side of the neighborhood in
\Cref{prop:participation-switch}, $V_a(m)=C_I(m)$ is twice
continuously differentiable and $C_I''(m)>0$.
\end{corollary}

Consumers can therefore gain from a public division into
two equally sized segments with nearby prior means averaging
to $m$. If learning is chosen optimally in each segment
and the seller offers separate menus, the gain in the
segment with the higher mean exceeds the loss in the other.
Both means can remain in the inclusion region, so everyone
still trades. Such a division requires assignment to
depend on the underlying valuation: valuation-independent
assignment leaves the prior unchanged in both segments.
The next section examines both forms of segmentation.

\section{Market Segmentation}
\label{sec:segmentation}
We now allow the designer to assign buyers to public segments and
choose a private information policy for each. The seller then offers separate menus for each segment,
combining second-degree price discrimination within segments with
third-degree discrimination across them. With gains from
trade at every valuation, even valuation-independent assignment to
two segments can generate efficient surplus pairs unavailable in
a single market. It cannot raise maximal consumer surplus, but
can yield strictly more consumer surplus than any single-market
policy meeting the same profit requirement. Valuation-dependent
assignment also discloses information about consumers' valuations to the seller. Coordinating this disclosure
with private learning can raise both consumer surplus and profit
relative to buyer-optimal single-market learning.

\subsection{Observable segments and attainable payoffs}
\label{subsec:segmentation-payoffs}

Fix the primitive prior $H$. In a \emph{state-independent segmentation},
the designer assigns the buyer to segment $k$ with probability $\lambda_k$
independently of $V$, and commits to a private information structure
$\sigma_k$ for that segment. There are finitely many segments, with
$\lambda_k\ge0$ and $\sum_k\lambda_k=1$. Both parties observe the
segment label. The buyer privately observes her signal, while the seller knows
$\sigma_k$ and the resulting posterior distribution $G_k$, but not
the signal or posterior realization. It chooses a profit-maximizing
menu for each segment. Buyers can choose only from their own segment's menu. Costs are additive, so production in one
segment does not constrain production in another.\footnote{A general public segment signal independent of $V$ gives
the same payoff set and finitely many segments suffice. Additionally, public assignment can equivalently be viewed as a public lottery selecting an information policy before the seller chooses its menu. The buyer then receives the selected policy’s private signal and faces the associated menu, which depends on the public draw. Under this interpretation, the different menus need not
be offered simultaneously.}

Independence makes $H$ the prior in every segment, so
$G_k\in\mathcal G(H)$. Their posterior distributions can differ, although
their underlying valuation distributions coincide. The label reveals how
the buyer was informed, but conveys no information about her primitive valuation.
The designer can equivalently choose
$\{(\lambda_k,G_k)\}_{k=1}^J$. Because each segment is screened
separately, aggregate payoffs are
\begin{equation}
 (C,P)=\sum_{k=1}^J\lambda_k\bigl(\CS(G_k),\Pi(G_k)\bigr).
 \label{eq:seg-payoff-average}
\end{equation}

If only the buyer observed the segment label, the seller
would instead screen the mixture
$\sum_k\lambda_kG_k$ with one menu. All contracts
would then be available to buyers from every segment.
The seller might therefore choose different qualities
and prices, and aggregate payoffs need not equal
the averages under separate menus. Public segmentation
can thus improve the division of surplus even though
the segment label conveys no information about the
buyer's primitive valuation.

\paragraph{The attainable division of surplus.}
Recall that
$\mathcal P(H)=\{(\CS(G),\Pi(G)):G\in\mathcal G(H)\}$
is the single-market surplus set, and write
$\mathcal P^{\mathrm{seg}}(H)$ for its segmented counterpart.
Equation~\eqref{eq:seg-payoff-average} implies
\begin{equation}
 \mathcal P^{\mathrm{seg}}(H)=\operatorname{co}\mathcal P(H),
 \label{eq:seg-payoff-hull}
\end{equation}
where $\operatorname{co}$ denotes the convex hull. Both sets are compact, and every Pareto-efficient segmented
pair can be attained with at most two segments
(Appendix~\ref{app:seg-geometry}).

Segmentation can expand the available payoff combinations, but cannot
raise the maximum of a linear objective. For every
$\beta\ge0$,
\begin{equation*}
 \max_{(C,P)\in\mathcal P^{\mathrm{seg}}(H)}(C+\beta P)
 =W_H(\beta).
\end{equation*}
Each segment's posterior distribution is feasible from $H$, so neither its
weighted payoff nor their average can exceed $W_H(\beta)$. In particular, state-independent segmentation cannot improve
on buyer-optimal private learning when consumer surplus is the sole
objective. Its potential benefit is to leave more surplus to consumers
while maintaining a given level of seller profit.

For a binary prior, write $\mathcal P^a_m$ and
$\mathcal P^{a,\mathrm{seg}}_m$ for the two surplus sets. Let
$\mathcal I^a_m$ be the interval of profits attainable in a single
market, and define
\[
 \widehat\Gamma^a_m(p):=
 \max\{\CS(G):G\in\mathcal G_{a,m},\ \Pi(G)\ge p\},
 \qquad p\in\mathcal I^a_m.
\]
This is the most consumer surplus attainable while leaving the seller
profit at least $p$. Thus, a segmented pair
$(C,p)$ with $C>\widehat\Gamma^a_m(p)$ cannot be attained or
weakly dominated by any single-market pair. Additionally, for a fixed prior, increasing the profit weight weakly raises optimal
profit and lowers consumer surplus. If $(C_i,P_i)$ is any optimal
payoff at $\beta_i$ and $0\le\beta_1<\beta_2$, optimality gives
\begin{equation}
 P_2\ge P_1,\qquad
 \beta_1(P_2-P_1)\le C_1-C_2\le\beta_2(P_2-P_1).
 \label{eq:weight-payoff-monotonicity}
\end{equation}
Adding the two optimality inequalities gives $P_2\ge P_1$ and
rearranging them gives the bounds on $C_1-C_2$. Total surplus also
weakly increases with the profit weight on $[0,1]$. These
arguments do not, however, imply that the optimal payoff varies
continuously with the weight.

\paragraph{When separate menus add efficient payoffs.}
Fix $\beta>0$. The equation $C+\beta P=W_H(\beta)$ defines
a \emph{supporting line}. Every feasible pair lies on or below it,
and a distribution maximizes $J_\beta$ exactly when the payoff pair it induces lies
on the line. Suppose two distributions maximize the objective and yield profits
$P_L<P_H$. The one yielding lower profit gives higher consumer
surplus. Assigning the two distributions to separate segments combines
these different divisions of surplus without reducing the weighted
objective. Giving the low-profit segment probability
$(P_H-p)/(P_H-P_L)$ and the remaining probability to the other
segment attains
$\bigl(W_H(\beta)-\beta p,\ p\bigr)$, $p\in[P_L,P_H]$.
Every such pair is Pareto efficient since improving either payoff without
lowering the other would exceed the maximum.

If no single-market maximizer of $J_\beta$ delivers
profit $p$, the corresponding segmented pair cannot
be attained in one market. Profit $p$ remains feasible,
but only with less consumer surplus. A policy earning
more than $p$ must also give less consumer surplus,
since the supporting line permits less consumer
surplus at a higher profit. Multiple maximizers alone
are therefore insufficient for new efficient pairs.
If they attain every profit between $P_L$ and $P_H$,
segmentation adds no new payoff on that segment of
the supporting line.

Appendix~\ref{app:seg-support-gaps} characterizes supported payoffs and the minimum-profit boundary
under state-independent segmentation. At any supporting weight, every positive-probability
segment must use a single-market optimal information policy at that
same weight. The segmented minimum-profit boundary is the least concave nonincreasing
majorant of the single-market boundary. The next result shows that this
majorant is strictly above the original boundary on a set of
priors with a productive lowest valuation.

\paragraph{A payoff-neutral benchmark.}
Under the hypotheses of \Cref{thm:neutral-frontier}, the continuous
single-market payoff path already contains every efficient division
of surplus. Neither independent nor valuation-dependent segmentation
improves this frontier, as \Cref{cor:seg-neutral} below establishes.

\subsection{New efficient payoffs}
\label{subsec:seg-gains}
With a binary prior, a productive lowest value permits gains from independent segmentation. For the next result, suppose $c(x)=\kappa x^\nu/\nu$, with $\kappa>0$,
$\nu\ge2$, and $\bar Q>\kappa^{-1/(\nu-1)}$.

\begin{theorem}[New efficient payoffs from independent segmentation]
\label{thm:segmentation-gain}
There is a nonempty open set
$\mathcal U\subset\{(a,m):0<a<m<1\}$ such that, at each prior
in $\mathcal U$, there are a weight $\beta>0$, two distributions
$G^L,G^H$ maximizing $J_\beta$ in $\mathcal G_{a,m}$, and profits
$P_L<p_-<p_+<P_H$, where $P_i:=\Pi(G^i)$,
with the following property. For every $p\in[p_-,p_+]$, valuation-independent
segmentation between these two distributions attains the Pareto-efficient
pair $\bigl(V_{a,\beta}(m)-\beta p,p\bigr)$, and
\begin{equation*}
 V_{a,\beta}(m)-\beta p>\widehat\Gamma^a_m(p).
\end{equation*}
Only the segment probabilities change with $p$. The single-market
boundary $\widehat\Gamma^a_m$ is nonconcave at these priors.

The supporting weights and participation regimes
can be specified as follows.
\begin{enumerate}
\item[(i)] $\mathcal U$ can be
chosen with $1<\beta<2$. At the selected weight,
the only maximizers are full revelation and no
learning. At every prior in this set, every optimal
distribution has full participation at every finite
nonnegative profit weight.

\item[(ii)] Under quadratic cost, the segmentation gains also hold on another nonempty
open set of priors, with a weight $\beta\in(0,1)$
chosen at each prior. Here, $G^L$ features distorted
inclusion and $G^H$ features exclusion.
For $c(x)=x^2/2$, these conclusions hold at every
mean sufficiently close to and below $m^\dagger(a)$,
for each $a\in\mathcal A$ from
\Cref{prop:participation-switch}.
\end{enumerate}
\end{theorem}

The two cases concern different priors and show that
segmentation can generate new efficient payoffs both
when the policies differ in participation and when
both serve everyone. Both constructions leave an interval $[p_-,p_+]$
of profits attained by no single-market maximizer
of $J_\beta$. Segmentation between $G^L$ and $G^H$
attains the corresponding efficient payoffs, as
Figure~\ref{fig:segmentation-frontier} illustrates.

\begin{figure}[ht]
 \centering
 \captionsetup{font=small,skip=4pt}
\includegraphics[width=0.7\textwidth]{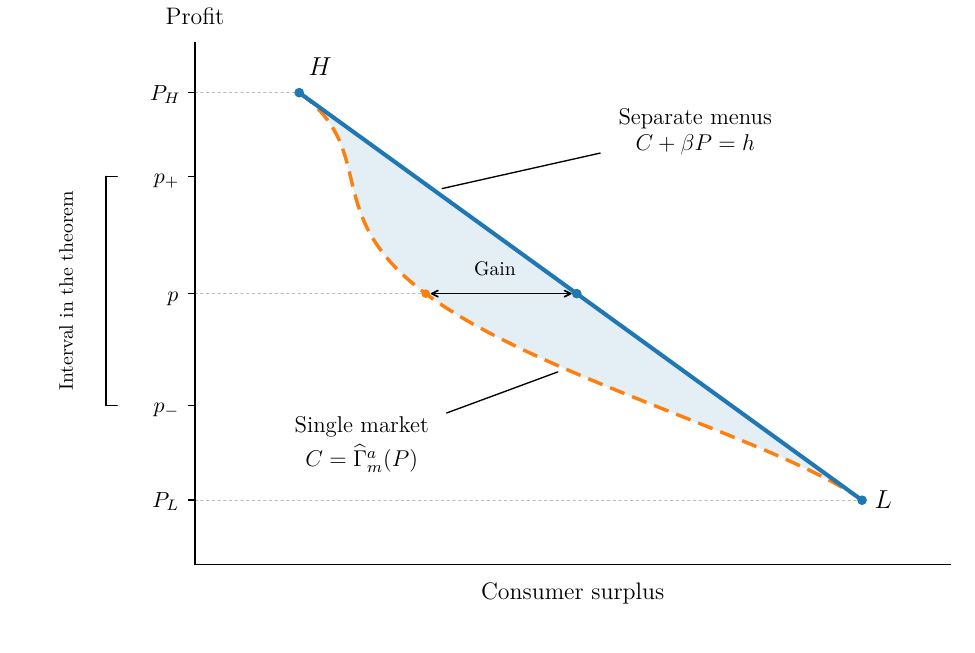}
 \caption{Valuation-independent segmentation and the efficient frontier.
The single-market surplus pairs $H$ and $L$ maximize $C+\beta P$
at the same prior and weight $\beta>0$, with common objective
value $h$. Separate menus
 attain the solid line between them. The bracket
 marks the theorem's closed interval $[p_-,p_+]$, strictly inside
 that gap. Gains remain positive at both $p_-$ and $p_+$, but vanish
 at $P_L$ and $P_H$.}
 \label{fig:segmentation-frontier}
\end{figure}
\paragraph{Inclusion and exclusion.}
The quadratic construction starts from the buyer-optimal
participation switch in \Cref{prop:participation-switch}.
At $m^\dagger(a)$, inclusion and exclusion give the same
consumer surplus, while exclusion yields higher profit.
Any strict average of their payoffs therefore gives
the same consumer surplus and less profit than
exclusion alone. Segmentation between these policies
does not add an efficient payoff.

Just below $m^\dagger(a)$, inclusion yields slightly
more consumer surplus, while exclusion retains its
profit advantage. A small positive weight on profit
can therefore offset the consumer-surplus loss under
exclusion. The proof holds the prior fixed and
reoptimizes learning as this weight changes.
For sufficiently small weights, all optima remain
in two separate profit ranges: lower-profit optima
have distorted inclusion, and higher-profit optima
have exclusion. At some $\beta\in(0,1)$, both regimes
are globally optimal. The gap between these profit
ranges is where independent segmentation yields
new efficient payoffs.

\paragraph{Learning and no learning.}
Segmentation can also help when every optimum has full
participation. For the power cost in the theorem, fix
$1/\nu<a<2/(\nu+1)$ and take $m>a$ sufficiently close
to $a$. At a weight $\beta\in(1,2)$, full revelation
and no learning give the same weighted payoff.
Full revelation generates positive consumer surplus
but lower profit. With no learning, every buyer has
posterior value $m$, so the seller offers efficient
quality and extracts the entire surplus, earning $\pi(m)$.

To establish that these are the only optima, the proof
shows that among fully participating distributions with
lowest posterior value $b$, the unique best is supported
on $\{b,1\}$. Raising $b$ progressively pools high-valuation
buyers with low-valuation buyers, moving from full
revelation at $b=a$ to no learning at $b=m$. Strict convexity in $b$
then implies that every intermediate choice is strictly
worse.

Valuation-independent assignment to public segments
combines these payoffs while restricting buyers to
their own segment's menu. In a common menu, buyers with
posterior value $m$ could instead choose the basic
product offered under full revelation and obtain
positive surplus, preventing full extraction.

\subsection{State-dependent segmentation}
\label{subsec:state-dependent-segmentation}

We now allow segment assignment to depend on the primitive
valuation. Formally, the designer publicly commits to a joint experiment
$\tau(k,\dd s\mid v)$ before the seller chooses its menus.
Both parties observe the public signal $K$, but only the buyer
observes the private signal $S$. We can write
\[
 \tau(k,\dd s\mid v)
 =\Pr(K=k\mid V=v)\sigma_k(\dd s\mid v).
\]
Here $\Pr(K=k\mid V=v)$ is the segment assignment rule, and
$\sigma_k$ is the private learning policy in segment $k$.
Jointly choosing public and private information therefore amounts
to choosing segmentation and learning within each segment.\footnote{
Finite segmentation is without loss for the joint optimum.
With $N$ primitive valuations, at most $N$ segments suffice
(\Cref{thm:finite-joint-design}(iii)).}
The public signal induces a conditional prior $H_k$ in segment $k$.
The buyer's posterior mean is $Y=\E[V\mid K,S]$, whose conditional
distribution given $K=k$ is $G_k$.
The designer can therefore choose these distributions directly,
subject to
\begin{equation}
 \sum_k\lambda_kH_k=H,\qquad G_k\preceq_\cx H_k,
 \qquad \lambda_k=\Pr(K=k).
 \label{eq:sd-prior-feasibility}
\end{equation}
The first condition preserves the population's valuation composition.
The second imposes feasible learning within segments.
Buyers choose only from their segment's menu, so payoffs average
as in \eqref{eq:seg-payoff-average}. Conditional priors can now differ across segments.
We write
\[
 W_H^{\mathrm{SD}}(\beta)
 :=\sup_{\{(\lambda_k,H_k,G_k)\}}
       \sum_k\lambda_k J_\beta(G_k),
\]
where the supremum is over finite collections of segment probabilities, conditional priors, and posterior distributions satisfying \eqref{eq:sd-prior-feasibility}. This is the maximal level of the objective attainable by jointly choosing private learning and state-dependent segmentation.

\paragraph{Efficiency and full surplus extraction.}
When profit receives at least as much weight as consumer surplus,
joint design has a simple solution attaining full efficiency and yielding full surplus extraction. Let $Y$ denote the buyer's
posterior mean after observing both the public label and her private
signal.

\begin{corollary}[Full-surplus extraction under joint design]
\label{prop:sd-full-extraction}
Under \Cref{ass:cost}, for every finite prior $H$ and every
$\beta\ge1$, every optimal joint policy yields
$\CS=0$ and $\Pi=\E_H[\pi(V)]$.
A feasible joint policy is optimal if and only if
$Y=V$ almost surely and $\bigl|\supp(H_k)\cap(0,1]\bigr|\le1$ for every positive-probability segment $k$.
If all primitive valuations are positive, each public segment
contains a single valuation, so the segment label reveals $V$.
If zero is a primitive valuation, a public segment may contain
both zero-valued buyers and buyers of one positive valuation.
In any such mixed segment, the private signal must reveal
which valuation the buyer has.
\end{corollary}
Full public revelation attains the upper bound
$\CS+\beta\Pi
 \le \beta\TS
 \le \beta\E[\pi(Y)]
 \le \beta\E_H[\pi(V)]$.
Strict convexity of $\pi$ and the pointwise surplus bound
imply that every optimum must reveal the primitive valuation
to the buyer and supply efficient quality.
The public label must therefore provide enough information for the seller to optimally offer these efficient qualities.
 If a segment contains more than one positive valuation,
the seller can raise profit by slightly reducing quality
for buyers with the lowest positive valuation. Starting
from efficient quality, this reduces the rents left to
buyers with higher valuations to first order, while
the loss in total surplus is only second order.
With zero and just one positive valuation, the seller
can supply efficient quality and extract the entire
surplus. Buyers with valuation zero receive zero quality,
so their contract gives the other buyers no rent. Thus, full
extraction is necessary even at $\beta=1$, when the designer
maximizes total surplus.
From now on, therefore, we consider $0\le\beta<1$.

For a binary prior on $\{a,1\}$, a conditional prior is determined
by its mean $m_k$. Public feasibility reduces to
$\sum_k\lambda_km_k=m$, and optimal private learning within each
segment gives
\begin{equation}
 V^{\mathrm{SD}}_{a,\beta}(m)
 :=\max_{\substack{\lambda_k\ge0,\ \sum_k\lambda_k=1\\
                   m_k\in[a,1],\ \sum_k\lambda_km_k=m}}
       \sum_k\lambda_kV_{a,\beta}(m_k)
 =\cav V_{a,\beta}(m).
 \label{eq:sd-concavification}
\end{equation}
Two segments suffice. This is concavification in the \emph{prior mean},
after private learning and the seller's response have already been
optimized \citep[Corollary~2]{kamenicagentzkow2011}.\footnote{At the
degenerate conditional priors the seller extracts the entire surplus,
so $V_{a,\beta}(a)=\beta\pi(a)$ and
$V_{a,\beta}(1)=\beta\pi(1)$. With more than two primitive values,
the conditional distributions must reproduce the entire prior.
Averaging their means is not sufficient. Appendix~\ref{app:sd-implementation}
gives the binary assignment probabilities.}
The segment policies need not be individually feasible under the prior. Public assignment can therefore raise the maximum objective
value, rather than merely fill gaps between maximizing
payoff pairs.

To state the characterization, recall that the integrand in the objective can be written as
$F_{y,\beta}(z)=(y-z)q(z)+\beta\pi(z)$, where $y$ is the posterior value and $z$ is the virtual value (see equation~\eqref{eq:local_objective_contribution}). For $y>0$,
define
\begin{equation}
 \mathcal R_\beta(y):=\max_{0\le z\le y}F_{y,\beta}(z),\qquad
 z_{y,\beta}^{\dagger}:=\argmax_{0\le z\le y}F_{y,\beta}(z).
 \label{eq:sd-rent-function}
\end{equation}
Under \Cref{ass:curvature}, the maximizer is unique and interior:
\begin{equation}
 (y-z_{y,\beta}^{\dagger})q'(z_{y,\beta}^{\dagger})
 =(1-\beta)q(z_{y,\beta}^{\dagger}),
 \qquad 0<z_{y,\beta}^{\dagger}<y.
 \label{eq:sd-basic-quality}
\end{equation}
Set $\mathcal R_\beta(0)=z_{0,\beta}^{\dagger}=0$.
At $\beta=0$, $\mathcal R_0(y)$ is the largest rent contribution
$(y-z)q(z)$.

\begin{theorem}[Joint learning and segmentation with a binary prior]
\label{thm:joint-consumer-segmentation}
Fix $a\in(0,1)$ and $0\le\beta<1$, and maintain
\Cref{ass:cost,ass:curvature} at that weight.
\begin{enumerate}
 \setlength{\itemsep}{3pt}
 \setlength{\parsep}{0pt}
 \item[(i)] At every $m\in(a,1)$, a joint optimum exists and two
 public segments suffice. Every positive-probability segment in an
 optimum has full participation. Its posterior distribution consists
 of a nondegenerate separating interval and an efficiently served
 atom at one, possibly preceded by a trading atom at $a$.

 \item[(ii)] Either $V^{\mathrm{SD}}_{a,\beta}=V_{a,\beta}$ at
 every prior mean, or there is a single interval
 $(m_-,m_+)\subset(a,1)$ on which
 $V^{\mathrm{SD}}_{a,\beta}>V_{a,\beta}$, with equality elsewhere.

 \item[(iii)] If the interval $(m_-,m_+)$ is nonempty, there are a
 unique regular upper distribution $G_A$, with bottom posterior value
 $B>a$ and mean $M$, and a unique mean price $\theta_{a,\beta}>0$
 such that the two optimal segment distributions are
 \begin{equation}
  G_L=(1-r)\delta_a+rG_A,\qquad G_H=G_A,\qquad
  r=\frac{a-z_{a,\beta}^{\dagger}}
          {B-z_{a,\beta}^{\dagger}}.
  \label{eq:sd-joint-markets}
 \end{equation}
 Their means are $m_-=a+r(M-a)$ and $m_+=M$. For
 $m\in[m_-,m_+]$, assign probability
 $\lambda(m)=(M-m)/(M-m_-)$ to $G_L$ and the remainder to $G_H$.
 Then
 \begin{equation}
  V^{\mathrm{SD}}_{a,\beta}(m)
  =\theta_{a,\beta}(m-a)+\mathcal R_\beta(a).
  \label{eq:sd-joint-value}
 \end{equation}
 At every $m\in(m_-,m_+)$, these are the only optimal segment
 distributions, up to relabeling or replication of identical segments.
 The pair, the price $\theta_{a,\beta}$, and the interval may depend
 on both $a$ and $\beta$.

   \item[(iv)] Whenever
$V^{\mathrm{SD}}_{a,\beta}(m)>V_{a,\beta}(m)$,
the supremum of the designer's objective over all finite
segmentations whose aggregate posterior distribution
attains $V_{a,\beta}(m)$ is strictly below
$V^{\mathrm{SD}}_{a,\beta}(m)$.
\end{enumerate}
If $a=0$, state-dependent segmentation does not increase the designer's objective:
$V^{\mathrm{SD}}_{0,\beta}(m)=V_{0,\beta}(m)$, $0<m<1$, under the same assumptions at the specified weight.
\end{theorem}
Suppose optimal private learning within a public segment
leaves some buyers unserved. The designer can split
this segment, placing the excluded buyers with a suitable
share of the trading population and assigning the rest
to a separate segment in a way that preserves the aggregate
posterior distribution. Serving the excluded buyers
then requires price reductions for fewer existing
customers, making a basic product profitable. Buyers
who purchase it receive zero rent. Consumer surplus
rises among the previously served buyers who share
the same menu because they receive the same qualities
at lower prices. The remaining buyers retain their
original menu. The new payments more than cover
production costs and price reductions, so profit
also rises. Exclusion therefore cannot occur at
a joint optimum.

A lower trading atom at a joint optimum must also lie
at $a$. If its posterior value were higher, the designer
could lower it in one segment and raise it in another,
adjusting learning and segment shares to preserve the
prior. Such a split can move more buyers into the
differentiated range while retaining a basic product
in both segments and raising the designer's objective.
At $a$, the downward adjustment is infeasible.
Buyers of the basic product therefore learn that
their valuation is $a$.

Now consider an optimum using two distinct segment
designs, one with a basic product and one without.
Public assignment allows the designer to give one
segment a higher mean by lowering the mean in another.
An extra unit of mean must therefore have the same
marginal value, $\theta$, in both segments. The
designer compares their weighted payoffs per buyer
after deducting this opportunity cost for each unit
of mean above $a$.

The optimal segment without a basic product identifies
the best distribution for differentiated buyers, after
accounting for the mean it uses. The designer cannot
choose a distribution with a lower payoff in the other
segment and make up the difference through the basic
product. If that were possible, adding an optimally
chosen basic product to the best differentiated
distribution would yield an even higher payoff.
The segment without a basic product would then not
be optimal. The differentiated distribution must
therefore itself be optimal in both segments.
At the common price of mean, this regular optimum
is unique, so both segments use $G_A$. The seller
supplies the same qualities on the differentiated
range, but charges less in the segment where the
basic product is available.

The designer chooses basic quality by changing how
many buyers share the menu in which this product
is offered. Moving some buyers of differentiated
products to the other segment reduces the rent cost
of the basic product, so the seller raises its quality.
Buyers of differentiated products who remain in
that segment gain a better alternative and pay less.
Those who move lose that alternative and pay more.
These changes can preserve the aggregate posterior
distribution. The designer therefore weighs the
benefit of improving the basic product against the
loss from making it available to fewer buyers,
taking the increase in profit into account.
This determines the quality in
\eqref{eq:sd-basic-quality}.

At this optimal quality, the basic product's best
weighted contribution per buyer at $a$,
$\mathcal R_\beta(a)$, depends only on $a$ and $\beta$.
The best weighted payoff from a differentiated
distribution, after deducting the cost of mean above
$a$, falls as $\theta$ rises. Both segment designs
can be optimal only when these two quantities are
equal. This can happen at only one price of mean.

That price determines the two distributions and their
means, $m_-$ and $m_+$. Between these means, the designer
can match the prior by changing only the segment shares,
leaving learning and both menus unchanged. A single
market at an intermediate mean must depart from these
optimal designs and gives a lower payoff. Since no
other price permits two distinct optimal designs,
segmentation helps on this single interval.

\paragraph{Joint information design and the order of optimization.}
Whenever joint design improves on private learning alone,
\Cref{thm:joint-consumer-segmentation}(iv) shows that
first choosing privately optimal learning and then
segmenting the resulting posterior distribution falls
strictly short of the joint optimum. Conversely, if
this sequential procedure attains the joint optimum,
private learning with one market already does so.
To see why, notice that we can screen the aggregate posterior distribution
required by joint optimality with one menu.
If the lowest type trades, lowering the bottom of
the upper distribution while keeping its mean fixed
raises basic quality without a first-order loss in
the upper component's value. If the lowest type is
excluded, reducing the upper mean allows more buyers
to trade and raises the designer's objective. A privately optimal
policy therefore selects a different posterior
distribution, which subsequent segmentation must
preserve.

The scope for gains from public information contrasts with
\citet[Theorem~2]{KartikZhong2026}. In their single-object, posted-price
environment, every payoff pair attainable under price-independent
beliefs can be replicated with an uninformed seller by redesigning
buyer information.\footnote{\citet{KartikZhong2026} allow the seller to have information
unavailable to the buyer. Price-independent beliefs rule out
further learning from prices. Here the seller's signal is public,
and the buyer evaluates every menu using the same posterior mean
$Y=\E[V\mid K,S]$.}
Here, public information changes which buyers share a menu,
and hence both the qualities the seller finds profitable to
offer and the alternatives that constrain its prices.
These changes can raise consumer surplus beyond what any
private learning policy can achieve with a common menu.

At $\beta=0$, \Cref{thm:joint-consumer-segmentation} implies
that joint design can raise consumer surplus relative to
a single-market optimum with a lower trading atom
at a posterior value strictly above $a$. The following proposition shows that a two-segment policy
can increase profit at the same time, yielding a strict
gain for both consumers and the seller.

\begin{proposition}[A Pareto-improving state-dependent split]
\label{prop:state-dependent-gain}
For a binary prior on $\{a,1\}$, maintain
\Cref{ass:cost,ass:curvature} at $\beta=0$.
Suppose a buyer-optimal single-market solution has
a lower trading atom at a posterior value strictly
above $a$. Then there
is a two-segment state-dependent policy that strictly raises
both consumer surplus and profit relative to that solution.
The two conditional means can be chosen arbitrarily close
to the original prior mean, and both segments retain full
participation.
\end{proposition}

Segmentation allows the basic product to remain
available while reducing the share of buyers who
purchase it. The seller's incentive to offer it
depends on these buyers' willingness to pay and
their proportion relative to buyers of higher qualities.
In a segment with a larger share purchasing the basic
product, a lower posterior value can suffice because there
are more sales to support it and fewer other buyers
requiring price reductions. Where fewer buyers
purchase it, a higher posterior value helps sustain
the product. Separate menus allow the seller to
charge different basic prices in the two segments.
Since the original posterior value exceeds $a$,
learning can move it in both directions.

The designer can combine these segments while preserving
the prior and increasing the share purchasing from
the same differentiated range. Both menus retain a
basic product. Buyers of the basic product receive
no rent, whereas the additional buyers of differentiated
products earn positive rents. The changes in basic
prices can reduce rents for buyers already purchasing
differentiated products. But learning and segment
composition can be chosen so that the additional
rents earned within that range dominate this loss.

Profit rises as well. After accounting for the price
reductions that the basic product requires elsewhere,
its profit per buyer is lower than that generated
in the differentiated range. Moving more buyers into
that range therefore also benefits the seller.

\paragraph{Quadratic cost.}
For $c(x)=x^2/2$, the isolated quality in
\eqref{eq:sd-joint-markets} is $a/(2-\beta)$. At $\beta=0$, the
strict-gain interval in \Cref{thm:joint-consumer-segmentation} is
nonempty if and only if
\[
 0<a<a^{\mathrm{SD}}:=\frac{2}{1+e^{\pi/2}}\simeq0.3442058.
\]
For each such $a$, maximal consumer surplus rises exactly when
$m\in(m_-(a),m_+(a))$, with equality at the endpoints and outside
this interval. For $a=0$ or $a\ge a^{\mathrm{SD}}$, it is unchanged
at every mean.

The cutoff marks the point at which the two segment
designs merge. As $a$ approaches $a^{\mathrm{SD}}$
from below, the lowest posterior value $B$ in $G_A$
approaches $a$, and the share purchasing the basic
product in $G_L$ tends to zero. The price reduction
that this product provides to the other buyers also
vanishes. The two segments approach the same posterior
distribution, and their means $m_-$ and $m_+$ converge
to the same value. The interval on which segmentation
raises consumer surplus therefore collapses.

The following table illustrates the order of optimization at $a=0.1$,
$m=0.4$, and $\beta=0$, where buyer-optimal learning in a
single market induces exclusion. The middle column maximizes
consumer surplus through segmentation while preserving the
aggregate posterior distribution from the first column.
The final column optimizes learning and segmentation jointly.

\begin{center}
\small
\begin{tabularx}{\textwidth}{
    l
    >{\centering\arraybackslash}X
    >{\centering\arraybackslash}X
    >{\centering\arraybackslash}X
}
\toprule
& Consumer-optimal learning
& Optimal segmentation of consumer-optimal posterior distribution
& Optimal joint design \\
\midrule
Consumer surplus
& $0.0448751$ & $0.0461147$ & $0.0461260$ \\
Profit
& $0.0702773$ & $0.0708971$ & $0.0722426$ \\
Participation
& $0.5041634$ & $1$ & $1$ \\
\bottomrule
\end{tabularx}
\end{center}

Segmentation of the original posterior distribution
already restores full participation and raises both
payoffs. Joint design also changes learning. The share
of buyers who learn $V=a$ and purchase the basic product
rises from $49.58\%$ to $50.46\%$, while the remaining
buyers have higher posterior values on average.
Basic quality remains $0.05$. The increase in
rents generated by the basic product more than offsets
a small fall in total rents generated within the
differentiated range. Consumer surplus therefore rises
a further $0.025\%$. Relative to the original single
market, joint design raises consumer surplus by
$2.79\%$ and profit by $2.80\%$.

\paragraph{A payoff-neutral lowest value.}
When $a$ is payoff neutral, buyers with posterior
value $a$ cannot sustain a basic product that creates
rents for other buyers. Under the hypotheses of
\Cref{thm:neutral-frontier},
\Cref{prop:neutral-weighted-value} establishes that
$V_{a,\beta}$ is concave over the full range of
prior means for every $\beta<1$. Splitting the
prior into different conditional means therefore
cannot raise the average weighted payoff.
At $\beta=1$, the full-revelation surplus bound
is linear in mean. The continuous path of
single-market optima attains every profit level
between its endpoints, so these bounds also rule
out new efficient payoff pairs from segmentation.

\begin{corollary}[Segmentation with a payoff-neutral lowest value]
\label{cor:seg-neutral}
Under the hypotheses of \Cref{thm:neutral-frontier}, neither
state-independent nor state-dependent segmentation changes the
single-market Pareto frontier. Every segmented pair is weakly
dominated by a single-market pair on the path $(C_\beta,P_\beta)$,
$0\le\beta\le1$. For every profit requirement feasible in one market,
maximal consumer surplus is the same under either assignment rule and
in one market.
\end{corollary}

\paragraph{High prior means.}
For the weights covered by \Cref{prop:high-mean},
its proof gives an upper bound on each segment's
weighted payoff that is affine in its mean and valid
throughout $[a,1]$, including outside the high-mean region.
Since segment means average to $m$, their combined payoff
cannot exceed what the single-market optimum already achieves.

\begin{corollary}[Segmentation at high prior means]
\label{cor:seg-high-mean}
Under the hypotheses of \Cref{prop:high-mean}, its threshold
$\bar m$ can be chosen so that, for every $m\in(\bar m,1)$,
\begin{equation}
 V^{\mathrm{SD}}_{a,\beta}(m)=V_{a,\beta}(m),
 \qquad 0\le\beta\le\bar\beta.
 \label{eq:seg-high-mean-value}
\end{equation}
Under either assignment rule, any segmentation attaining this value
uses $G_{m,\beta}^*$ in every positive-probability segment. Neither rule
increases maximal consumer surplus subject to a profit requirement in
$[P_0,P_{\bar\beta}]$. The path $(C_\beta,P_\beta)$,
$0\le\beta\le\bar\beta$, is therefore exactly that part of either
segmented Pareto frontier. For profit requirements $p\le P_0$, maximal
consumer surplus remains $C_0$.
\end{corollary}

Only the original posterior distribution attains
the bound, so every segment in an optimal segmentation
must use that distribution. All segments consequently
have mean $m$. Since the prior is binary, their
conditional priors are identical, and the public
label conveys no information about valuation.

\section{General Finite Priors}
\label{sec:finite-priors}

We now return to a general finite prior and let $a=v_1$.
The additional information constraints can rule out changes
in learning that improve the objective in the binary problem. As a result, optimal posterior distributions can
combine several separating regions and atoms, and
exclusion can coexist with a trading atom. Thus, the optimal menu may need to combine a continuum of products differentiated by quality with several isolated items.

\subsection{Bayes plausibility with general discrete priors}
\label{subsec:finite-constraints}

To describe the additional constraints Bayes plausibility imposes, consider
\[\E_G[(Y-t)_+]=\int_t^{1}[1-G(y)]\dd y,\] where $(y-t)_+:=\max\{y-t,0\}$. This
quantity measures how often posterior values exceed $t$
and by how much. Since
$Y=\E[V\mid S]$, Jensen's inequality gives
$\E_G[(Y-t)_+]\le \E_H[(V-t)_+]$.
The intuition is that pooling primitive values on opposite sides of
$t$ offsets some of the excess above the threshold with values
below it. Learning can therefore reduce this quantity relative
to the prior, but cannot increase it. The function
$t\mapsto\E_G[(Y-t)_+]$ is the stop-loss transform of $G$.
Together with equality of means, these inequalities characterize
convex order and hence feasible learning
\citep{shaked2007stochastic}.

For a finite prior, it is enough to check the inequalities at the
interior primitive values. Define
\[
 D_j(G):=\E_G[(Y-v_j)_+],\qquad
 \kappa_j:=\E_H[(V-v_j)_+]
          =\sum_{k>j}p_k(v_k-v_j),\qquad 2\le j<N.
\]
Then, for $G\in\Delta([a,1])$,
\begin{equation}
 G\in\mathcal G(H)
 \quad\Longleftrightarrow\quad
 \quad \E_G[Y]=m,\quad \text{and}\quad
 D_j(G)\le\kappa_j,\ \ 2\le j<N.
 \label{eq:finite-feasibility}
\end{equation}
If the stop-loss inequality holds at two adjacent primitive
valuations, it also holds at every threshold between them.
Over this interval, the prior's stop-loss transform is affine.
The posterior's transform is convex, so it lies below the
line joining its endpoint values. Since the inequality holds
at both endpoints, this line also lies below the prior's stop-loss
transform throughout the interval.

The key observation is that some mean-preserving changes leave
every $D_j$ unchanged. Within any interval $[v_i,v_{i+1}]$ between adjacent
primitive values, each function $(y-v_j)_+$ is affine. It equals
$y-v_j$ for $j\le i$ and zero for $j\ge i+1$. The contribution
of this interval to each $D_j$ therefore depends only on the
probability assigned to the interval and the conditional mean there.
Redistributing a fixed amount of mass inside
$[v_i,v_{i+1}]$, while preserving its mean and leaving the rest
of the distribution unchanged, preserves every $D_j$, irrespective of whether some or all constraints bind.

In contrast, a mean-preserving change across a primitive value
can violate feasibility. For example, consider splitting probability $r$
at $b$ equally between $b-\epsilon$ and $b+\epsilon$, for
sufficiently small $\epsilon>0$. This preserves the mean but raises $\E_G[(Y-b)_+]$ by
$r\epsilon/2$. If the
constraint at $b$ binds, this change is infeasible. Thus, mean-preserving changes remain feasible within each
primitive interval, but may fail across primitive values. This generates the piecewise-binary structure below.

\subsection{The piecewise-binary structure}
\label{subsec:finite-structure}

We say that a support gap $(w,y^+)$ following a trading atom
at $w$ \emph{crosses} a primitive value $v_j$ if
$w<v_j<y^+$. A support gap ending at a primitive value does not
cross it.

\begin{theorem}[Optimal learning with finitely many primitive values]
\label{thm:finite-structure}
Fix $0\le\beta<1$. Under \Cref{ass:cost,ass:curvature}, every
maximizer of $J_\beta$ in $\mathcal G(H)$ has the following properties.
\begin{enumerate}
 \setlength{\itemsep}{0pt}
 \setlength{\parsep}{0pt}
 \item[(i)] Excluded ranks, when present, form one lower interval.
 Distinct trading posterior values receive distinct qualities.
 Each maximal nonterminal positive-quality linear interval of the
 concavified revenue curve corresponds to one posterior atom,
 followed by a nonempty support gap.
 \item[(ii)] Each $[v_i,v_{i+1}]$ contains at most one trading
 atom whose following support gap is also contained in that interval.
 Such an atom lies below every separating range in
 $(v_i,v_{i+1})$. There is at most one separating range in each
 open primitive interval.
 \item[(iii)] Gaps crossing primitive values have disjoint open
 spans, and at most one starts in each $[v_i,v_{i+1})$.
 A gap ending at a primitive value may be followed by an atom there.
 \item[(iv)] A nondegenerate separating region reaches a positive
 atom at one. Virtual value converges to one along that region,
 and the top atom receives efficient quality $q(1)$.
\end{enumerate}
There are at most $N-1$ connected separating regions and $2N-3$
trading atoms below one. On each separating region, posterior means
are continuous and strictly increasing in rank, and qualities are
strictly increasing. On a separating region within each open primitive interval, the
posterior and virtual-value profiles are continuously differentiable,
with a positive continuous posterior density.
\end{theorem}
Buyers with different posterior values still receive different
qualities whenever they trade. At the top, a continuous range
approaches the efficient quality purchased by a positive share
of buyers who learn that their valuation is one. Each isolated
product serves buyers with a common posterior value, with a
gap before the next higher values in the support. The difference
from the binary case lies in where these products can appear.

The structure is piecewise binary because the arguments used in the binary case and the interior conditions characterizing separation
continue to apply within each primitive interval.
Redistributing a fixed probability mass within that
interval while preserving its mean leaves every
additional information constraint unchanged.
Raising posterior values at an atom can place buyers in the
gap above it. If the atom and its entire following gap lie within
one primitive interval, the binary arguments place the atom below
any separation there and rule out a second such atom. Since all buyers still face one menu,
products purchased in one interval continue to
constrain prices in others.

For a prior on $\{a,b,1\}$, only posterior values
above $b$ contribute to the additional constraint.
When it binds, raising values above $b$ cannot
be offset by lowering values already at or below
$b$, even if the changes preserve the mean.

Consider first two consecutive trading atoms at $a$ and
$b$ below the upper separating region. The binary argument raises
posterior values for some buyers at $b$ and assigns
the remaining buyers at that atom the value $a$,
preserving the mean. The gain comes from the seller's
quality adjustment, as explained in
Section~\ref{subsec:mean-adjustment}. This change violates the binding
constraint and is therefore infeasible. Two isolated products
can therefore remain optimal even when the first
gap ends at $b$ and the second lies entirely above it.

An atom at $b$ can also interrupt separation.
Ruling this out in the binary case required raising values at the atom
and offsetting the increase below it, where the
reduction costs less because quality is lower.
A binding constraint at $b$ again prevents this change.
The atom can therefore remain above one separating
range, with separation resuming above the support gap that comes after it.

Exclusion can also coexist with a trading atom.
Consider a prior on $\{0,b,1\}$, with excluded buyers
at zero, a basic product for buyers with posterior
value $b$, and separation above it. In the binary
problem, the designer can vary the posterior value of buyers
purchasing the basic product and the quality the seller chooses
for them, adjusting the excluded share to preserve the mean.
A binding constraint at $b$ limits these adjustments.
Holding the conditional distribution in the separating
range fixed, increasing the share of buyers in that
range raises the constraint's left-hand side.
Reducing the share purchasing the basic product and
increasing exclusion can restore the mean, but cannot
offset the increase in the constraint's left-hand side because buyers at zero and at $b$
contribute nothing to it. The basic
product can therefore remain worth sustaining for
the rents it provides to buyers of higher-quality
products, while buyers at zero remain unserved.

On a separating range in $(v_i,v_{i+1})$, the
interior equations in \Cref{prop:separating-conditions}
hold with a local mean multiplier $\theta_i$. This marginal value can
change across a primitive value even when separation
continues through it. It cannot be higher in a lower
range than in a higher one. Transferring mean from
the higher range to the lower would then improve
the objective while relaxing every intervening information
constraint. If those constraints are all slack, a small
transfer in either direction is feasible, and the
marginal values must agree.

Quality must nevertheless be continuous wherever posterior
values cross a primitive value without an atom or support
gap. A jump would imply that buyers with nearly identical posterior
values receive distinctly different qualities. Lowering the values
of a small group around the threshold to the group's lowest
value induces the seller to supply a common intermediate
quality. The mean can be restored by raising posterior values elsewhere
below the threshold without violating the information constraints. As the group becomes smaller, its posterior values are arbitrarily close,
while the quality jump remains. Under
\Cref{ass:curvature}, the gain from the seller's quality
adjustment outweighs the cost of changing posterior
values and restoring the mean. Quality does not jump at the threshold. If the constraint
binds, however, a marginal increase in posterior value
can lead to a smaller increase in quality just above
the threshold than just below it. If the constraint
is slack, the marginal increase in quality is the
same on both sides.

There are $N-1$ primitive intervals, each containing at
most one lower trading atom together with its entire
following gap. Any remaining lower trading atom has a
gap crossing at least one interior primitive value.
Since the gaps are disjoint, they cannot cross the same
value, allowing at most $N-2$ further atoms. Each open
primitive interval also contains at most one separating
range. With three primitive values, there can therefore
be at most two connected separating regions and three
trading atoms below one.

Finally, interpolation again suffices on
trading quantiles. Each nonterminal positive-quality
linear interval corresponds to a single posterior
atom. Genuine ironing may remain among excluded
buyers, but a mean-preserving contraction of their
posterior values can eliminate it while preserving
feasibility, the concavified revenue curve, and both
payoffs. An optimum
can therefore be chosen with
$\widetilde R_Q=R_Q$.

\begin{figure}[htbp]
 \centering
 \captionsetup{font={small,stretch=1}}
 \includegraphics[width=\textwidth]{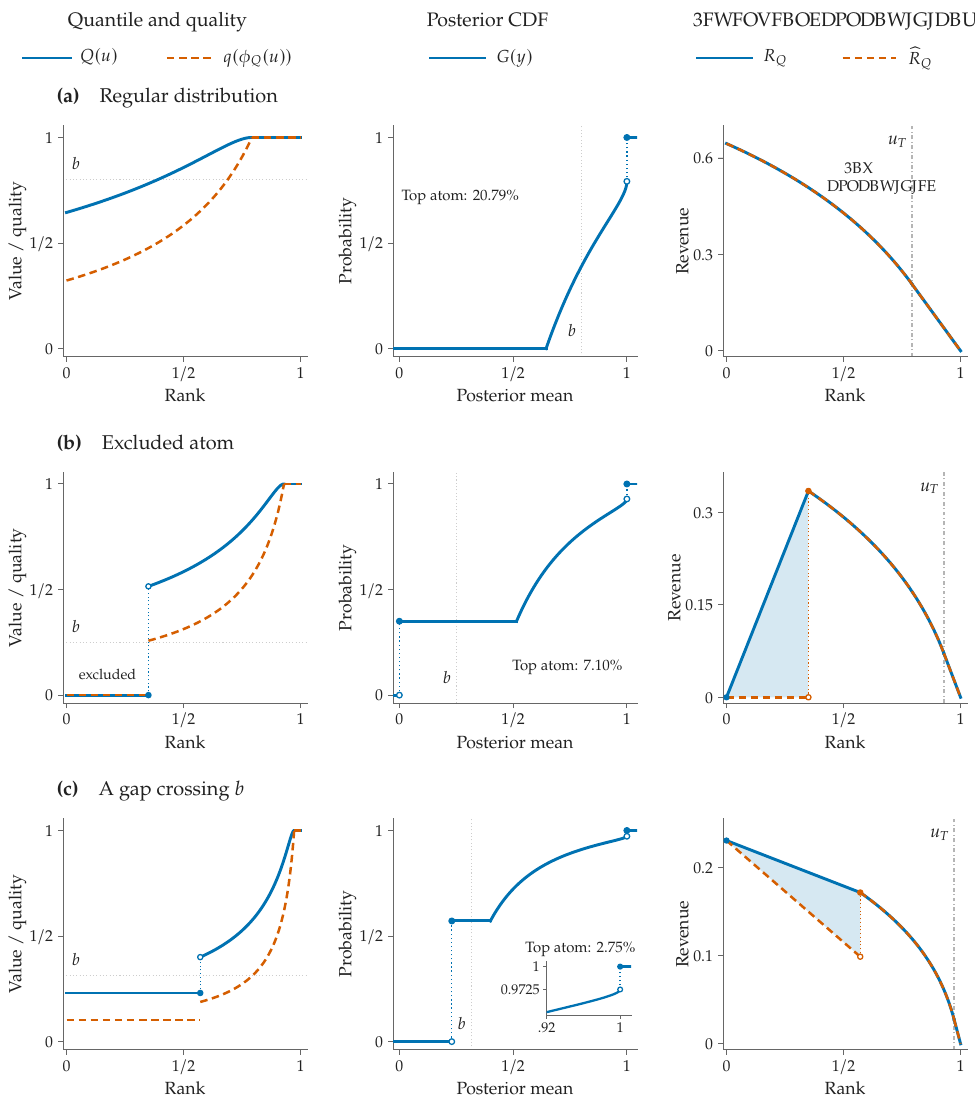}
 \caption{Illustrations of different optimal structures identified in Theorem~\ref{thm:finite-structure}, with $c(x)=x^2/2$ and $\beta=0$. Each row uses a different prior on $\{a,b,1\}$ and
 shows the quantile and assigned quality, the CDF, and the raw and concavified revenue curves. Shading marks strict concavification,
 $R_Q>\widehat R_Q$.
 The gray reference line marks prior value $b$ in the first two columns. In the revenue curve
 panels, the dash--dotted vertical line marks $u_T=G(1^-)$, where the
 top atom starts. The inset
 enlarges the top CDF jump in (c) and percentages indicate atom masses.}
 \label{fig:three-point-structures}
\end{figure}

\begin{figure}[p]
 \centering
 \captionsetup{font={small,stretch=1}}
 \includegraphics[width=\textwidth]{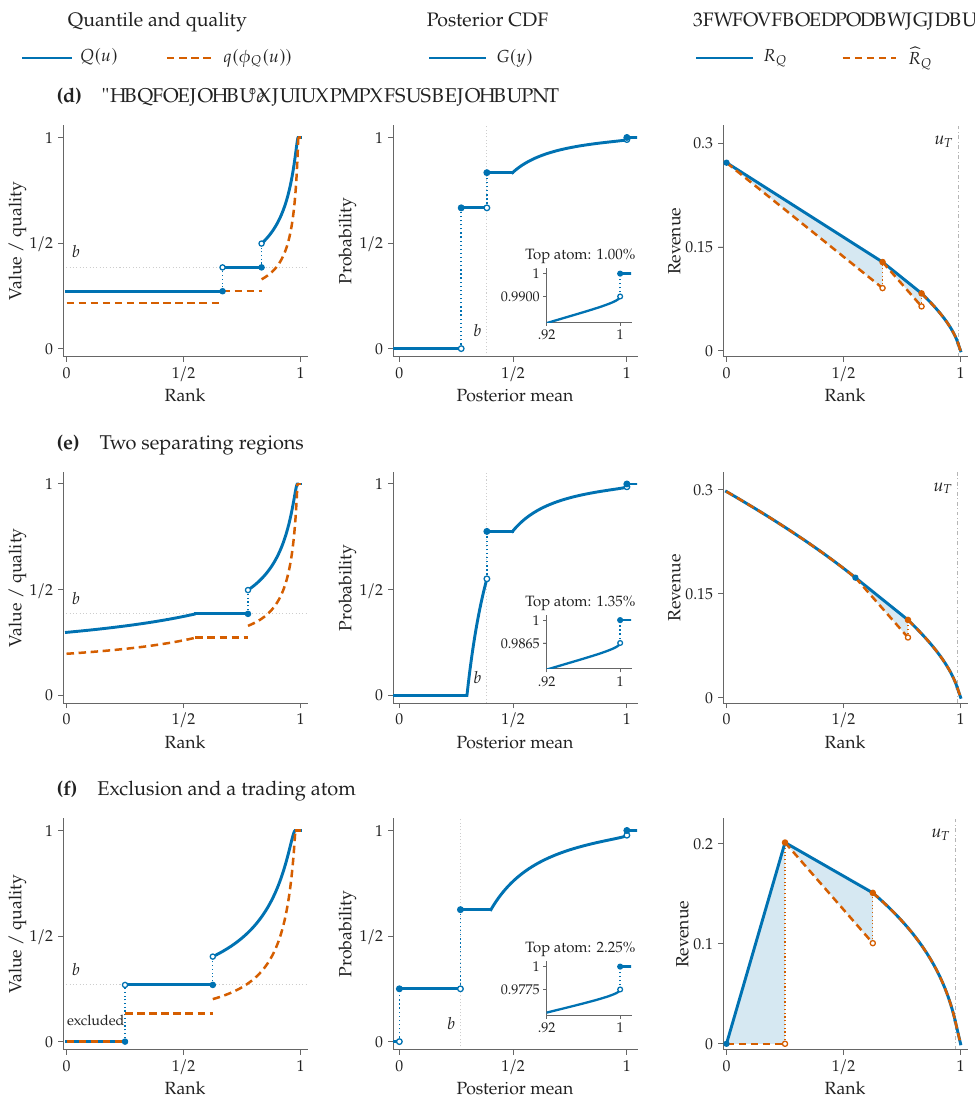}
 \caption{Illustrations of different optimal structures identified in Theorem~\ref{thm:finite-structure}, continuing
 Figure~\ref{fig:three-point-structures}. Panel (d) has two lower
 trading atoms, (e) has two connected separating regions, and (f)
 combines exclusion with a lower trading atom. The insets enlarge
 the top CDF jumps. There is always an efficiently served atom
 at one.}
 \label{fig:three-point-structures-def}
\end{figure}

Figures~\ref{fig:three-point-structures} and
\ref{fig:three-point-structures-def} give examples under
different priors on $\{a,b,1\}$. In (a)--(c), the constraint
at $b$ is slack and the configurations of the binary case remain optimal:
full participation without a lower atom in (a), exclusion
of zero-valued buyers in (b), and full participation with
a lower atom whose gap crosses $b$ in (c). In (d)--(f), the constraint binds. There are two consecutive
lower trading atoms in (d), two separating regions in (e),
and exclusion alongside a lower trading atom in (f).
Example (e) attains the bound of two connected separating
regions. Example (f) shows that a payoff-neutral lowest
valuation no longer rules out an isolated product.

\paragraph{When the binary solution applies.}
The binary prior supported on $\{a,1\}$ with mean $m$ permits
every posterior distribution supported on $[a,1]$ with
the same mean. It therefore provides an upper
bound. Under \Cref{ass:cost}, for every finite $\beta\ge0$,
\[
 W_H(\beta)\le V_{a,\beta}(m),
\]
with equality if and only if some binary optimum satisfies
the additional constraints.
If that optimum is unique, it remains unique
under $H$.
More generally, a mean-preserving spread of the prior
weakly raises the designer's value by expanding the
set of feasible learning policies.

\paragraph{What binding constraints reveal.}
A binding constraint also restricts what buyers can remain
uncertain about. Pooling primitive values strictly on
opposite sides of $v_j$ reduces $D_j(G)$ relative to
its value under the prior. If $D_j(G)=\kappa_j$, no such reduction
can occur and almost every posterior belief must place
all its probability on valuations weakly below $v_j$
or weakly above it. For a prior on $\{a,b,1\}$, a binding
constraint at $b$ therefore rules out pooling the lowest
and highest valuations. A buyer with posterior value
$b$ must know that her valuation is $b$.
If all interior constraints bind, almost every posterior
belief is supported on at most two adjacent primitive
values. The posterior mean then determines the
probabilities on those values. This still permits
a continuum of posterior values between the same two
primitive values.

\paragraph{Independent segmentation.}
A payoff-neutral lowest valuation no longer suffices
to make independent segmentation redundant. Under
quadratic cost, take a binary prior with two positive
valuations for which separate menus generate new
efficient payoffs, and add a sufficiently small
probability of valuation zero. The two productive
valuations remain, and continuity of the payoff sets
preserves the strict advantage of separate menus.
Thus, even with a payoff-neutral lowest valuation,
segmentation can give consumers more surplus than
any single market meeting the same profit requirement.
As before, independent segmentation cannot raise the
unconstrained maximum of consumer surplus.

\subsection{Joint learning and state-dependent segmentation}
\label{subsec:finite-joint-design}
We now choose private learning and valuation-dependent
segmentation jointly for a general finite prior.
A finite collection $\{(\lambda_k,G_k)\}$ of segment
probabilities and posterior distributions can be induced
by public assignment and private learning if and only if
\begin{equation*}
 \overline G:=\sum_k\lambda_kG_k\preceq_\cx H.
\end{equation*}

To see why this condition suffices, consider feasible learning
that generates $\overline G$. At each posterior value,
divide buyers among the segments in the proportions needed
to obtain the distributions $G_k$, keeping the same mix
of underlying valuations among buyers with that posterior
value in each segment. The segment label then leaves the
buyer's belief unchanged once her private signal is known.
Segments can nevertheless have different conditional priors
because they contain different mixes of posterior values.
These priors average to $H$, and learning within each segment
generates its prescribed $G_k$. This construction can be
implemented by public assignment followed by private
learning within segments.
Thus, the information constraints in
\eqref{eq:finite-feasibility} restrict aggregate posterior
demand, while public assignment determines which buyers
face the same menu.

To describe the seller's response, let $\mu$ be
the joint distribution, across segments, of posterior means $y$ and ironed virtual values $z$,
with negative virtual values replaced by zero.  This replacement
leaves quality and payoffs unchanged. Write
$\mathcal T_a=\{(y,z):a\le y\le1,\ 0\le z\le y\}$, and let
$\mu_Y$ denote the marginal distribution of posterior means under $\mu$.

\begin{theorem}[Value and structure of joint learning and segmentation]
\label{thm:finite-joint-design}
Fix $0\le\beta<1$ and maintain
\Cref{ass:cost,ass:curvature} at that weight.
A joint optimum exists, and its value is the attained maximum
\begin{equation}
\begin{aligned}
 W_H^{\mathrm{SD}}(\beta)
 &=\max_{\mu\in\Delta(\mathcal T_a)}
       \int F_{y,\beta}(z)\,\mu(\dd y,\dd z)\\
 &\text{subject to}\quad \mu_Y\preceq_\cx H,\\
 &\hspace{13mm}
   \int_{\{y\ge t\}}(y-z)\,\mu(\dd y,\dd z)
   \le
   \int(y-t)_+\,\mu_Y(\dd y)
   \quad\text{for every }t\in[a,1].
\end{aligned}
\label{eq:joint-value-reduction}
\end{equation}
Every optimal joint policy induces a maximizer of
\eqref{eq:joint-value-reduction}. Moreover:
\begin{enumerate}
 \setlength{\itemsep}{4pt}
 \setlength{\parsep}{0pt}

 \item[(i)] \emph{Common quality and participation.}
 In every joint optimum, there is a common quality rule $x$
 such that, in every positive-probability segment $k$,
 $x_k(y)=x(y)$ for $G_k$-almost every $y$.
 The rule satisfies
 $x(y)<x(y')$ for
 $\overline G\otimes\overline G$-almost every pair $y<y'$.
 Moreover,
 $q(z_{y,\beta}^{\dagger})\le x(y)\le q(y)$ for $\overline G$-almost every $y>0$.
 Thus, only buyers with posterior value zero can be excluded.
 Payments may differ across segments.

 \item[(ii)] \emph{Information at posterior atoms.}
 In every joint optimum, each atom of the aggregate
 posterior distribution lies at a primitive valuation
 and fully reveals that valuation to the buyer\footnote{If \(c'(0)=d>0\), the revelation conclusion applies to every trading atom. Excluded atoms may pool several payoff-neutral primitive valuations and need neither coincide with a primitive valuation nor reveal the buyer’s valuation.}:
 \begin{equation*}
  \overline G(\{w\})>0
  \quad\Longrightarrow\quad
  w\in\{v_1,\ldots,v_N\}
  \quad\text{and}\quad
  \Pr(V=w\mid Y=w)=1.
 \end{equation*}
 If an atom at an interior primitive value $v_j$ receives
 quality strictly below $q(v_j)$, the corresponding
 information constraint binds:
 $D_j(\overline G)=\kappa_j$.
 Each segment has at most $N-1$ nonterminal
 trading atoms.

 \item[(iii)] \emph{The number of public segments.}
 Some joint optimum can be implemented with at most $N$
 segments. This bound is sharp: for every $N\ge2$,
 there are priors with $N$ primitive values for which,
 at $\beta=0$ and under quadratic cost, every joint optimum requires at least
 $N$ distinct segments.
\end{enumerate}
\end{theorem}
The screening constraints come from a quality increase
available to the seller in every segment. Fix $t$ and
suppose the seller adds $\varepsilon$ units of quality
to all contracts intended for buyers with $y\ge t$,
charging $t\varepsilon$ more for each. A buyer with
posterior value $y$ gains $\varepsilon(y-t)$ from
any upgraded contract. Her ranking of these contracts
is unchanged. They become more attractive if $y>t$
and less attractive if $y<t$. The original choices
therefore remain optimal. The seller earns $t$ per
additional unit of quality, while marginal production
cost is $z$. If the average $z$ among these buyers
were below $t$, the change would increase profit.
Optimality rules this out in each segment and hence
in the aggregate, so we obtain the inequality in
\eqref{eq:joint-value-reduction}.

Every learning and segmentation policy induces a feasible
$\mu$ with the same objective value, so the program gives
an upper bound. To establish equality, we must find segments
whose optimal menus produce the qualities chosen by the
program.

The construction uses the seller's trade-off
between revenue from a lower product and the rents it
creates above. Raising its quality brings additional
revenue, net of production costs, from its buyers,
but requires larger rents for buyers of higher-quality
products. The relative numbers of these buyers sharing
a menu determine the seller's preferred quality.
For a fixed aggregate posterior distribution with finite
support, the proof optimizes virtual values and then
chooses these proportions. Each segment has an efficiently
served highest type, with lower types added in proportions
that make the seller supply their prescribed qualities.
Together, the segments reproduce the aggregate posterior
distribution and the consumer surplus and profit obtained
in the program.

An approximation argument gives equality of the optimal values for
general posterior distributions. Since a joint optimum
exists, the upper bound is attained. All optimal
learning and segmentation policies must consequently
maximize the aggregate program, which we use below
to derive their common properties.

\paragraph{Common quality and participation.}
Buyers with the same posterior value can pay different
prices because they face different alternatives within
their segments. At a joint optimum, however, they receive
the same quality. To see why, fix their posterior value.
Under \Cref{ass:curvature}, the marginal return to
increasing virtual value falls as virtual value rises.
Bringing different virtual values closer together,
while preserving their average, therefore raises the
designer's objective. This leaves the posterior
distribution and all aggregate screening constraints
unchanged. Since joint design attains the maximum of
the aggregate program, such differences cannot remain
at an optimum.

This parallels the common-quality result in
\citet{bergemannheumannwang2026}, although here the
distribution being screened is itself endogenously chosen through
learning.\footnote{Their common-quality result holds
for a continuous distribution of privately known
valuations. With discrete values, common quality
need not hold in general but can be established
under a sufficient condition on cost. Here,
\Cref{ass:curvature} gives this conclusion both on
separating regions and at posterior atoms.}
On any separating range of posterior values common
to both segments, the price difference for buyers
with the same posterior value is constant.
Since these buyers receive the same quality,
those paying the lower price receive the same
additional rent throughout that range.

Higher posterior values also receive strictly higher
qualities across segments, even though buyers cannot
choose contracts from another segment's menu. If
buyers with higher posterior values received weakly
lower quality, a small transfer of virtual value
toward them would raise the objective while respecting
the aggregate constraints. The ordering across
segments therefore follows from optimal design.

Joint design also limits how far quality can be
distorted. Below $q(z^\dagger_{y,\beta})$, the benefit
of inducing more quality outweighs the cost of the
accompanying rent extraction in the designer's
objective. Raising virtual value to
$z^\dagger_{y,\beta}$ leaves posterior values unchanged
and relaxes the aggregate screening constraints.
Such low quality therefore cannot occur at an
optimum, and every positive posterior type trades.
For quadratic cost, the bound is
$x(y)\ge y/(2-\beta)$. When only consumer surplus
matters, almost every buyer with a positive posterior
value receives at least half the quality that is
efficient given her information.

\paragraph{Isolated products and informed buyers.}
At a joint optimum, buyers at a posterior atom know
their underlying valuation. They can still receive
inefficiently low quality since the seller may degrade
their product to reduce the rents it leaves to
buyers with higher posterior values.

Consider an atom receiving less than efficient
quality, with a posterior value between two adjacent
primitive values. A small change in learning can
raise posterior values for some of its buyers and
lower them for others, preserving the mean.
At their original common quality, the gains and
losses cancel. Differentiating their qualities then
raises the objective by directing more quality
toward those who value it more. For a sufficiently
small change within the same primitive interval,
the information and aggregate screening constraints
allow this improvement. Such an atom cannot be
optimal.

At an interior primitive value $v_j$, the information
constraint there must bind to prevent this change.
As explained above, binding rules out uncertainty
over valuations on opposite sides of $v_j$.
A buyer with posterior value $v_j$ must then know
that her valuation is $v_j$. At the endpoints,
revelation follows directly from the support bounds.

For efficiently served atoms,
learning within each segment must be optimal for that segment's
conditional prior. The private-learning characterization
in \Cref{thm:finite-structure} places an efficiently
served atom at that prior's highest valuation.
Its buyers therefore know their valuation as well.

\paragraph{The number of public segments.}
The bound in part~(iii) comes from matching the
probabilities of the $N$ primitive valuations.
Keep each segment's conditional prior and private
learning fixed. If an optimum uses more than $N$
segments, their shares can be varied while leaving
all $N$ probabilities unchanged. The objective is
linear in these shares, and small changes in either
direction are feasible. Optimality therefore requires
its value to remain unchanged. We can continue varying the shares until one
segment has zero share, and repeat until at most $N$
remain.

Fewer segments need not attain the same value. In the
quadratic examples establishing sharpness, only
finitely many segment compositions can occur at an
optimum. They share the same conditional distribution
in the differentiated range, but the seller offers
different combinations of lower products across segments.
For each combination, inducing the optimal qualities
fixes the relative numbers of buyers purchasing these
products. The designer cannot change these proportions
to accommodate the prior without sacrificing consumer
surplus. Instead, the prior's valuation probabilities
must be matched by choosing the shares assigned to
the different segment compositions. Some priors require
$N$ such compositions. With fewer segments, at least
one segment must depart from an optimal composition,
reducing consumer surplus.

\paragraph{The joint Pareto frontier.}
Segmentation makes the set of joint payoffs convex. No weight above one is needed as $\beta=1$ already
selects the full-extraction outcome,
and every $\beta>1$ selects the same payoff pair. Thus, the problems with
$0\le\beta\le1$ exhaust the joint Pareto frontier. If
\Cref{ass:curvature} holds at every $\beta\in[0,1)$, then
\Cref{thm:finite-joint-design} describes the structure of all
supporting optimizers away from the full-extraction endpoint, while
\Cref{prop:sd-full-extraction} describes that endpoint. At
$\beta=0$, the Pareto-efficient endpoint is the
consumer-surplus-maximizing pair with the highest profit.

\section{Conclusion}
\label{sec:conclusion}

What consumers learn about their needs affects the products sellers offer, the prices they charge, and the division of gains from trade. Under our curvature condition and for weights strictly favoring consumers, optimal private learning sustains a continuum of differentiated products extending to efficient quality at the top. An isolated basic product, when offered, generates rents for buyers of higher qualities while its buyers’ relatively pessimistic beliefs also permit more optimistic beliefs elsewhere. Making this lower group still more pessimistic allows greater optimism among the remaining buyers, but can lead the seller to withdraw the basic product. Inclusion and exclusion therefore offer competing ways to benefit consumers, and maximizing consumer surplus can require exclusion even when trade with every buyer creates gains.

Market segmentation changes this trade-off by changing which buyers share a menu. Even assignment independent of valuations can increase consumer surplus without reducing profit or disclosing information about underlying valuations through segment labels. Valuation-dependent segmentation can raise maximal consumer
surplus and benefit both consumers and the seller relative
to optimal private learning. Under the same conditions, every joint optimum serves all posterior types with gains from trade. Achieving these gains can require redesigning learning, since even optimally segmenting the posterior distribution generated by buyer-optimal learning can fall short of the joint optimum.

Our results highlight how intermediaries can use their increasing control
over consumer learning and access to offers to improve
consumers' terms of trade or help sellers extract more
surplus. Their incentives should therefore be a central
concern for consumer protection, especially in the modern digital economy.

\clearpage

\printbibliography

\newpage

\section*{\LARGE Appendix}
\appendix

\section{Mean and surplus at a fixed concavified revenue curve}
\label{app:accounting-proof}

Fix a posterior quantile $Q$ on $[a,1]$, write $R=R_Q$, and
let $Q^r=R/(1-u)$. With the notation of
Section~\ref{subsec:mean-adjustment}, we characterize the quantiles
that induce $R$ and prove \Cref{prop:mean-accounting}. The
least-concave-majorant properties used below follow from
\citet[Lemma~A.1 and Proposition~A.3]{kyriazis2026markets}.

\subsection{Quantiles with a common concavified revenue curve}
\label{app:mean-accounting}

The raw revenue curve $\widehat R_Q$ is upper semicontinuous.
Its least concave majorant $R$ is continuous, agrees with
$\widehat R_Q$ at ranks zero and one, and satisfies
$a(1-u)\le R(u)\le1-u$. Thus,
$Q^r(u)=R(u)/(1-u)$ is continuous on $[0,1)$ and takes
values in $[a,1]$. Concavity of $R$, together with $R(1)=0$,
also implies that $Q^r$ is nondecreasing.

On each connected component of $\{R>\widehat R_Q\}$,
the majorant $R$ is linear. Every endpoint of a maximal
linear interval is therefore a contact: otherwise,
$R$ would be linear on a neighborhood of that endpoint,
and the interval could be extended. The interiors of the
maximal linear intervals are disjoint, so there are at most
countably many of positive length.

For each such maximal linear interval $I=[\alpha_I,\beta_I]$, let $-z_I$
be its slope and set $b_I=Q^r(\alpha_I)$.
Contact at $\alpha_I$ gives $Q^+(\alpha_I)=b_I$.
Monotonicity then implies $Q\ge b_I$ almost everywhere
on $I$, while $\widehat R_Q\le R$ gives $Q\le Q^r$
almost everywhere. Outside the interval interiors, every rank is a contact,
so the two quantiles agree almost everywhere. Hence
\begin{equation}
 \begin{gathered}
 Q^+(\alpha_I)=b_I,\qquad
 b_I\le Q\le Q^r\quad\text{a.e. on }I,\\
 Q=Q^r\quad\text{a.e. outside the interiors of these intervals}.
 \end{gathered}
 \label{eq:face-value-bounds}
\end{equation}

These bounds are also sufficient. For each
maximal linear interval $I=[\alpha_I,\beta_I]$, choose a
nondecreasing function $\widetilde Q_I$ on the open interval
$(\alpha_I,\beta_I)$ satisfying
$b_I\le\widetilde Q_I\le Q^r$ almost everywhere.
Set $\widetilde Q=\widetilde Q_I$ on each
$(\alpha_I,\beta_I)$ and $\widetilde Q=Q^r$
elsewhere in $[0,1)$. The lemma below shows that
the choice on one interval imposes no further
constraints on the choices on the others, even
when there are countably many intervals.

\begin{lemma}[Preserving the concavified revenue curve]
\label{lem:fixed-envelope-fillings}
The function $\widetilde Q$ has a nondecreasing quantile
representative with values in $[a,1]$ whose concavified
revenue curve is $R$. Conversely, every quantile inducing
$R$ agrees almost everywhere with a function constructed
in this way.
\end{lemma}

\begin{proof}
Choose right-continuous versions of the functions
$\widetilde Q_I$. Since there are at most countably many
intervals, this changes $\widetilde Q$ only on a null set.
Continuity of $Q^r$ ensures that the bounds
$b_I\le\widetilde Q_I\le Q^r$ then hold at every rank
in $(\alpha_I,\beta_I)$.

To check monotonicity across intervals, suppose
$u<v$ do not lie in the same interval interior.
If $v\in (\alpha_I,\beta_I)$, then $u\le\alpha_I$, so
$\widetilde Q(u)\le Q^r(u)\le Q^r(\alpha_I)
 =b_I\le\widetilde Q(v)$.
If $v$ lies outside all interval interiors, then
$\widetilde Q(u)\le Q^r(u)\le Q^r(v)=\widetilde Q(v)$.
Monotonicity within each interval holds by construction.
Thus $\widetilde Q$ is nondecreasing on $[0,1)$ and takes
values in $[a,1]$, so it agrees almost everywhere with
the quantile of a distribution on $[a,1]$.

Because the raw revenue curve uses the right-continuous quantile,
we must also check the right limits at ranks outside the interval
interiors. Fix such a rank $x<1$.
If $t>x$ lies in an interval $(\alpha_I,\beta_I)$, the interval
starts at or above $x$. Hence
$\widetilde Q(t)\ge b_I=Q^r(\alpha_I)\ge Q^r(x)$.
If $t$ lies outside all interval interiors, then
$\widetilde Q(t)=Q^r(t)\ge Q^r(x)$. Thus,
$Q^r(x)\le\widetilde Q(t)\le Q^r(t)$ for every $t>x$.
By continuity of $Q^r$, letting $t\downarrow x$
gives $\widetilde Q^+(x)=Q^r(x)$. This argument
also applies when $x$ is an accumulation point
of the intervals.

Taking right limits in $\widetilde Q\le Q^r$
also gives $\widetilde Q^+\le Q^r$. The new raw
revenue curve therefore lies below $R$ and equals
it outside the interval interiors. Both curves
are zero at rank one, so $R$ is a concave majorant
of the new raw revenue curve.

To show that $R$ is the least such majorant,
take any concave majorant of the new raw
revenue curve. It must dominate $R$ outside the interval
interiors, in particular at both endpoints of every $I$.
On $I$, concavity requires it to lie above the chord of
$R$, which agrees with $R$ by linearity on $I$.
Thus $R$ is the least concave majorant. Finally, the converse follows from
\eqref{eq:face-value-bounds} and monotonicity of any
quantile inducing $R$.
\end{proof}

\subsection{Proof of \texorpdfstring{\Cref{prop:mean-accounting}}{Proposition~\ref{prop:mean-accounting}}}
\label{app:accounting-completion}

\begin{proof}
The bounds in \eqref{eq:face-value-bounds} give
$0\le Q^r-Q\le Q^r-b_I$ almost everywhere on each
maximal linear interval $I$. Integration yields
$0\le d_I\le\kappa_I$. The intervals have disjoint
interiors and $b_I\ge a$, so
\begin{equation}
 \sum_I\kappa_I
 =\sum_I\int_I[Q^r(u)-b_I]\,\dd u
 \le\int_0^1[Q^r(u)-a]\,\dd u
 \le1-a.
 \label{eq:summable-capacities}
\end{equation}
Since $Q=Q^r$ almost everywhere outside the interval
interiors,
$\int_0^1[Q^r(u)-Q(u)]\,\dd u=\sum_I d_I$.
This is the mean identity \eqref{eq:mean-accounting}.

The two quantiles induce the same concavified revenue
curve $R$ and therefore the same ironed virtual values
and quality at almost every rank. Their profits coincide
by \eqref{eq:seller-profit}. On the interior of each $I$,
the common virtual value is $z_I$ and the common quality
is $q(z_I)$. Hence
\eqref{eq:consumer-surplus} gives
\[
 \CS(Q^r)-\CS(Q)
 =\int_0^1[Q^r(u)-Q(u)]q(\phi_Q(u))\,\dd u
 =\sum_I q(z_I)d_I,
\]
which is \eqref{eq:surplus-accounting}.
Since $0\le q(z_I)\le\bar q$, the sums for the mean
reduction and surplus loss are bounded by $1-a$ and
$\bar q(1-a)$, respectively. Both therefore converge
absolutely, including when there are countably many intervals.

Conversely, take any reductions $d_I$ satisfying
$0\le d_I\le\kappa_I$. We obtain them without changing
$R$ as follows. On each interval
$I=[\alpha_I,\beta_I]$ with $\kappa_I>0$, replace
$Q^r$ by $b_I$ from $\alpha_I$ up to a cutoff,
leaving the rest of the interval unchanged.
As the cutoff moves from $\alpha_I$ to $\beta_I$,
the reduction in mean increases continuously
from zero to $\kappa_I$. We can therefore choose
$\gamma_I\in[\alpha_I,\beta_I]$ such that
\[
 \int_{\alpha_I}^{\gamma_I}[Q^r(u)-b_I]\,\dd u=d_I.
\]
On the open interval $(\alpha_I,\beta_I)$, define
\[
 \widetilde Q(u)=
 \begin{cases}
  b_I, & \alpha_I<u<\gamma_I,\\
  Q^r(u), & \gamma_I\le u<\beta_I.
 \end{cases}
\]
Since $Q^r$ is nondecreasing and $Q^r\ge b_I$ on $I$,
this filling is nondecreasing and satisfies
$b_I\le\widetilde Q\le Q^r$.

If $\kappa_I=0$, then $d_I=0$, and we leave $Q^r$
unchanged on the interval interior. This covers every
interval ending at rank one: linearity and $R(1)=0$
give $R(u)=b_I(1-u)$ there, so $Q^r=b_I$ and the
capacity is zero. At all remaining ranks below one, set
$\widetilde Q=Q^r$.

By \Cref{lem:fixed-envelope-fillings}, these fillings
define a single quantile with values in $[a,1]$ and
concavified revenue curve $R$, even when adjustments
are made on countably many intervals. By construction,
$\int_I[Q^r(u)-\widetilde Q(u)]\,\dd u=d_I$ on every
$I$. The accounting identities give the claimed mean and
consumer surplus, while profit remains unchanged.

Finally, by \eqref{eq:binary-feasible-set}, a distribution
supported on $[a,1]$ is feasible for the binary prior
if and only if its mean is $m$. For the constructed
quantile, this condition is equivalent to
$\sum_I d_I=\int_0^1Q^r(u)\,\dd u-m$.
\end{proof}

\paragraph{The conditional linear program.}
By the proposition, \eqref{eq:conditional-mean-problem}
minimizes the loss in consumer surplus among posterior
distributions with mean $m$ and concavified revenue curve $R$.
Since $R$ fixes profit, minimizing this loss also maximizes $J_\beta$.

Let $D=\int_0^1Q^r(u)\,\dd u-m$ be the required
reduction in mean and $K=\sum_I\kappa_I$ the total
capacity. Summing the bounds $0\le d_I\le\kappa_I$
shows that feasibility requires $0\le D\le K$.
This condition is also sufficient. If $K>0$, choosing
$d_I=(D/K)\kappa_I$ for every $I$ satisfies the bounds
and gives $\sum_I d_I=D$. If $K=0$, the only feasible
case is $D=0$, with $d_I=0$ for every $I$.

To prove attainment, equip $\prod_I[0,\kappa_I]$
with the topology of coordinatewise convergence.
This product is compact. By
\eqref{eq:summable-capacities} and the bound
$q(z_I)\le\bar q$, the series $\sum_I d_I$ and
$\sum_Iq(z_I)d_I$ converge uniformly on the product.
Their finite partial sums are continuous, so both
sums are continuous as well. The constraint
$\sum_I d_I=D$ therefore defines a closed subset
of the compact product. This set is nonempty when
$0\le D\le K$, and the continuous objective attains
its minimum on it.

\section{Binary structure}
\label{app:binary-proofs}

Using Appendix~\ref{app:accounting-proof}, we identify the lower
groups, derive the separating conditions, and prove
\Cref{thm:binary-structure}(i).

\subsection{Curvature and the allocation of mean}
\label{app:binary-preliminaries}

Unless stated otherwise, fix $0\le\beta<1$ and $0\le a<m<1$ and maintain \Cref{ass:cost,ass:curvature}.

\subsubsection{Cost curvature in virtual-value coordinates}
\label{app:cost-curvature}

At $x=q(z)$, inverse differentiation gives $q'(z)=1/c''(x)$ and
$q''(z)=-c'''(x)/[c''(x)]^3$ for $0<z<1$. Thus,
\[
 (2-\beta)q'(z)-(y-z)q''(z)
 =\frac{(2-\beta)[c''(x)]^2+[y-c'(x)]c'''(x)}{[c''(x)]^3}.
\]
For $z<y\le1$, the numerator is positive if $c'''(x)\ge0$.
Otherwise, it is bounded below by its value at $y=1$, which is
positive under \Cref{ass:curvature}. Setting $y=1$ gives the
converse, establishing equivalence with \eqref{eq:rent-curvature}.

\subsubsection{Upper-tail ranks and linear parts}
\label{app:binary-reconstruction}

Write $s=1-u$ for the upper-tail rank, so smaller $s$ corresponds
to higher posterior values, and set
\[
 v(s)=Q(1-s),\qquad r(s)=R(1-s),\qquad
 U_r(s)=\frac{r(s)}s=Q^r(1-s),\quad 0<s\le1.
\]
The quantile $v$ is nonincreasing and right-continuous, with
$v(1)=v(1-)$. Its raw revenue curve is $f_v(s)=sv(s-)$ for $s>0$,
with $f_v(0)=0$, and $r=\cav f_v$. Write $z(s)=r'(s)$ for the ironed virtual
value almost everywhere. Since $v$ and $Q$ describe the same
posterior distribution, this change of coordinates leaves payoffs
unchanged.

A \emph{face} is a maximal interval of positive length on
which $r$ is affine. It corresponds to a linear interval
of $R$ in Section~\ref{subsec:mean-adjustment}, with ranks
reversed. Write $I=[\ell_I,h_I]$, where
$\ell_I=1-\beta_I$ and $h_I=1-\alpha_I$. In these coordinates,
\[
 r(s)=A_I+z_Is,\qquad b_I=U_r(h_I),\qquad
 d_I=\int_I[U_r(s)-v(s)]\,\dd s,\qquad
 \kappa_I=\int_I[U_r(s)-b_I]\,\dd s.
\]
Here $A_I=h_I(b_I-z_I)\ge0$, and $A_I>0$ if and only if
$\kappa_I>0$. If $A_I=0$, maximality implies $\ell_I=0$,
so the face corresponds to a terminal linear interval
ending at $u=1$ in the original ranks.
We call a rank $s$ a \emph{contact} if $f_v(s)=r(s)$.

A \emph{filling} of a face $I$ is a choice of posterior
values $v(s)$ on $(\ell_I,h_I)$.
By \Cref{lem:fixed-envelope-fillings}, it preserves
$r$ if and only if it is nonincreasing and satisfies
$b_I\le v\le U_r$ almost everywhere on $I$.
We call the face \emph{pure} if $v=b_I$ almost
everywhere on it. All its buyers then have the same
posterior value, and the full capacity for reducing
the mean is used, so $d_I=\kappa_I$.

For a binary prior, consider two faces with
$q(z_I)<q(z_J)$. Suppose $d_I<\kappa_I$ and $d_J>0$.
We can increase $d_I$ by a small $\varepsilon>0$
and decrease $d_J$ by the same amount. This preserves
the total reduction in mean and respects both capacity
bounds, so the resulting distribution remains feasible.
The curve $r$ and profit are unchanged, but consumer
surplus increases by
$\varepsilon[q(z_J)-q(z_I)]>0$.
Such a choice cannot be optimal with $r$ fixed.
In particular, if any face with positive quality
has a positive reduction, every face with zero
quality must be at capacity.

Before considering changes to $r$, we establish
a contraction result that holds for any finite prior.

\subsection{The lower groups and the regular upper component}
\label{app:binary-lower-structure}

\subsubsection{No bunching among positively served types}
\label{app:binary-purity}

\label{app:face-contraction}

Changing the filling of a face while preserving $r$
leaves the seller's quality schedule unchanged.
For a given mean on the face, the next lemma identifies
the least dispersed filling in convex order.
It gives buyers a common posterior value except
where that value would exceed $U_r$. At those ranks,
it sets $v=U_r$.

\Needspace{5\baselineskip}
\begin{lemma}[A contraction within a constant-quality group]
\label{lem:face-contraction}
Fix a face $I=[\ell,h]$ with intercept $A>0$, and let
$b=U_r(h)$. There is a unique $k_I\in[b,U_r(\ell)]$ such that
\begin{equation}
 \int_I\min\{U_r(s),k_I\}\,\dd s=\int_Iv(s)\,\dd s.
 \label{eq:finite-cap-mean}
\end{equation}
Set $w(s)=\min\{U_r(s),k_I\}$ on $I$.
Among all admissible fillings with the same average
posterior value on $I$, the distribution induced by $w$
is least in convex order.
Replacing $v$ by this filling, on one face or simultaneously
on any countable collection of faces with positive intercept,
preserves $r$, the mean, quality at almost every rank, and
both payoffs. The resulting posterior distribution is a
mean-preserving contraction of the original, so the
replacement also preserves Bayes plausibility.
\end{lemma}

\begin{proof}
Since $r(0)=0$ and $A>0$, we have $\ell>0$.
On $I$, the function $U_r(s)=z+A/s$ is continuous and
strictly decreasing, with minimum $b=U_r(h)$.
The map $k\mapsto\int_I\min\{U_r(s),k\}\,\dd s$ is therefore
continuous and strictly increasing on $[b,U_r(\ell)]$.
Its endpoint values are $b(h-\ell)$ and
$\int_IU_r(s)\,\dd s$. The bounds $b\le v\le U_r$
place $\int_Iv(s)\,\dd s$ between these endpoint values,
establishing existence and uniqueness of $k_I$.

The new filling $w$ is nonincreasing and satisfies
$b\le w\le U_r$. By \Cref{lem:fixed-envelope-fillings},
it preserves the concavified revenue curve $r$, and hence
the seller's quality schedule and profit.
By \eqref{eq:finite-cap-mean}, the mean within the face is
unchanged, as is $d_I$. Hence consumer surplus is also preserved
by \Cref{prop:mean-accounting}.

For the convex-order comparison, take any admissible filling
$\widetilde v$ with
$\int_I\widetilde v(s)\,\dd s=\int_Iw(s)\,\dd s$.
Fix a finite convex function $f$ on $\mathbb R$ and
choose a supporting slope $p$ at $k_I$.
Where $w=k_I$, the supporting-line inequality gives
$f(\widetilde v)-f(w)\ge p(\widetilde v-w)$.
Where $w<k_I$, we instead have $w=U_r$ and
$\widetilde v\le w$. Convexity gives
$f(w)-f(\widetilde v)\le p(w-\widetilde v)$, since
the secant slope between two values below $k_I$ cannot
exceed a supporting slope at $k_I$.
After rearranging, this is the same inequality as in the first case.
Integrating over $I$ gives
\[
 \int_I f(\widetilde v(s))\,\dd s
 -\int_I f(w(s))\,\dd s
 \ge p\int_I[\widetilde v(s)-w(s)]\,\dd s=0.
\]
Thus $w$ is least in convex order among admissible fillings
with the prescribed mean.

For a countable collection of faces, apply
\Cref{lem:fixed-envelope-fillings} to the chosen fillings simultaneously. The curve $r$ and every $d_I$ remain unchanged,
so \eqref{eq:summable-capacities} ensures that the mean and
payoff identities continue to hold. We may also sum the
convex-order inequalities:
the faces have disjoint interiors, their total length
is at most one, and $f$ is bounded on $[a,1]$.
Outside the chosen faces, posterior values are unchanged.
The resulting distribution is therefore below the original
in convex order and has the same mean.
Transitivity of convex order then preserves Bayes
plausibility under the original prior.
\end{proof}

The next lemma describes how the concave envelope changes when a
point on a face is raised. Its estimate lets us use the slopes of
the chords from the raised point to the original face endpoints
in the subsequent first-order payoff calculation, even when the
envelope changes outside the face.

\Needspace{5\baselineskip}
\begin{lemma}[Concavification after raising a point on a face]
\label{lem:binary-one-point-lift}
Let $[\ell,h]$, with $0<\ell<h\le1$, be a maximal face
of $r$ with slope $z$, and fix $c\in(\ell,h)$.
For $\epsilon>0$, let $R_\epsilon$ be the least concave
function above $r$ satisfying
$R_\epsilon(c)\ge r(c)+\epsilon$.

There are $a_\epsilon\in[0,\ell]$ and
$b_\epsilon\in[h,1]$ such that $R_\epsilon=r$ outside
$(a_\epsilon,b_\epsilon)$,
$R_\epsilon(c)=r(c)+\epsilon$, and $R_\epsilon$ is
affine on $[a_\epsilon,c]$ and $[c,b_\epsilon]$.
As $\epsilon\downarrow0$,
$a_\epsilon\to\ell$, $b_\epsilon\to h$, and
\begin{equation}
 \int_0^1
 |R_\epsilon'(s)-\widehat z_\epsilon(s)|\,\dd s
 =o(\epsilon),
 \label{eq:binary-lift-slope-estimate}
\end{equation}
where $\widehat z_\epsilon$ equals
$z+\epsilon/(c-\ell)$ on $(\ell,c)$,
$z-\epsilon/(h-c)$ on $(c,h)$,
and $r'$ outside $[\ell,h]$.
\end{lemma}

\begin{proof}
We obtain the new curve from two supporting lines to $r$
through the raised point. Write $y_\epsilon=r(c)+\epsilon$ and define
\[
 \lambda_\epsilon
 =\min_{0\le s<c}
   \frac{y_\epsilon-r(s)}{c-s},
 \qquad
 \mu_\epsilon
 =\max_{c<s\le1}
   \frac{r(s)-y_\epsilon}{s-c}.
\]
The first quotient tends to $+\infty$ as $s\uparrow c$,
and the second tends to $-\infty$ as $s\downarrow c$.
Continuity of $r$ ensures that both extrema are attained.
Let $a_\epsilon$ minimize the first quotient and
$b_\epsilon$ maximize the second.

On the original face, the two quotients are
$z+\epsilon/(c-s)$ and $z-\epsilon/(s-c)$.
The first is minimized at $\ell$ on $[\ell,c)$,
and the second is maximized at $h$ on $(c,h]$.
Hence $a_\epsilon\le\ell$ and $b_\epsilon\ge h$.
The line with slope $z$ along the original face lies
weakly above $r$ everywhere. Using this bound, and
evaluating the quotients at $\ell$ and $h$, gives
\[
 z<\lambda_\epsilon
 \le z+\frac{\epsilon}{c-\ell},
 \qquad
 z-\frac{\epsilon}{h-c}
 \le\mu_\epsilon<z.
\]

By definition of $\lambda_\epsilon$, the line through
$(c,y_\epsilon)$ with slope $\lambda_\epsilon$ lies
above $r$ to the left of $c$ and touches it at
$a_\epsilon$. The line with slope $\mu_\epsilon$ likewise lies
above $r$ to the right of $c$ and touches it at
$b_\epsilon$. Replace $r$ by the first line on
$[a_\epsilon,c]$ and the second on $[c,b_\epsilon]$,
leaving it unchanged elsewhere. Since the lines
support $r$ at the contact ranks, the slope does
not increase where either line meets $r$.
At $c$, it falls from $\lambda_\epsilon$ to
$\mu_\epsilon$. The constructed curve is therefore
concave and lies above $r$ and the raised point.

Any concave majorant of $r$ and the raised point
must lie above the endpoints of both segments.
Concavity then requires it to lie above the
segments themselves. Our construction therefore
gives $R_\epsilon$.
Since $r+\epsilon$ is also a concave majorant of
$r$ and the raised point, we have
$0\le R_\epsilon-r\le\epsilon$.
Thus $R_\epsilon$ converges uniformly to $r$.
Suppose a subsequence of $a_\epsilon$ converged
to some $a<\ell$. Since $R_\epsilon$ is affine
on $[a_\epsilon,c]$, its limit $r$ would be affine
on $[a,c]$. This would extend the original face
to the left, contradicting maximality.
Hence $a_\epsilon\to\ell$. The same argument
gives $b_\epsilon\to h$.

We prove the slope estimate by first bounding
how much $R_\epsilon$ exceeds $r$ at $\ell$ and $h$. Let
$H_L=R_\epsilon(\ell)-r(\ell)$ and
$H_R=R_\epsilon(h)-r(h)$ denote these increases.
The left chord supports $r$ at $a_\epsilon$, so concavity gives
$r'\le\lambda_\epsilon$ to the right of this contact.
Before the face, the old curve has slope at least $z$.
Thus $z\le r'\le\lambda_\epsilon$ almost everywhere on $(a_\epsilon,\ell)$.
Similarly, beyond the face the old curve has slope at most $z$.
The right chord supports $r$ at $b_\epsilon$, so the slope to
the left of this contact is at least $\mu_\epsilon$.
Thus $\mu_\epsilon\le r'\le z$ almost everywhere on $(h,b_\epsilon)$.
Since the curves agree at $a_\epsilon$ and
$b_\epsilon$, these bounds imply
$H_L=\int_{a_\epsilon}^{\ell}
          [\lambda_\epsilon-r'(s)]\,\dd s$ and $H_R=\int_h^{b_\epsilon}
          [r'(s)-\mu_\epsilon]\,\dd s$.
The slope bounds established above therefore yield
\[
 0\le H_L
 \le\frac{\epsilon(\ell-a_\epsilon)}{c-\ell}
 =o(\epsilon),
 \qquad
 0\le H_R
 \le\frac{\epsilon(b_\epsilon-h)}{h-c}
 =o(\epsilon).
\]

Each endpoint increase contributes twice to the slope error.
The two extension intervals contribute $H_L$ and $H_R$ to the integral in
\eqref{eq:binary-lift-slope-estimate}, since
$\widehat z_\epsilon=r'$ there.
On $[\ell,c]$, the actual slope is
$\lambda_\epsilon
 =z+(\epsilon-H_L)/(c-\ell)$,
whereas the comparison slope is
$z+\epsilon/(c-\ell)$.
Their integrated absolute difference is therefore
$H_L$. Similarly, on $[c,h]$, the actual slope is
$\mu_\epsilon
 =z-(\epsilon-H_R)/(h-c)$,
so the integrated difference from
$z-\epsilon/(h-c)$ is $H_R$.
Outside $[a_\epsilon,b_\epsilon]$, the slopes agree.
Adding the four contributions gives
\[
 \int_0^1
 |R_\epsilon'(s)-\widehat z_\epsilon(s)|\,\dd s
 =2(H_L+H_R)=o(\epsilon),
\]
as required.
\end{proof}

\begin{lemma}[Absence of bunching with a finite prior]
\label{lem:active-face-purity}
Fix a finite primitive prior and a finite weight $\beta\ge0$,
and maintain \Cref{ass:cost}. At every maximizer of $J_\beta$,
each face $I=[\ell,h]$ with $r(s)=A+zs$ and $A,z>0$
satisfies $v(s)=b:=U_r(h)$ almost everywhere on $I$.
Thus all buyers on the face share one posterior value,
and $d_I=\kappa_I$.
\end{lemma}

\begin{proof}
Suppose an optimum has a face $I=[\ell,h]$ with
$r(s)=A+zs$, where $A,z>0$, but $v$ is not equal to
$b=U_r(h)$ almost everywhere on $I$.
By \Cref{lem:face-contraction}, replace its filling by
$w=\min\{U_r,k\}$, preserving feasibility and both payoffs,
and retain the original filling elsewhere.
The resulting distribution is still optimal.

The original filling was not constant at $b$, so its average
posterior value on $I$ exceeds $b$ and hence $k>b$.
Consequently, on an interval immediately to the left of
$h$, we have $w=U_r=z+A/s$. The raw revenue curve
equals $r$ there, and posterior values strictly decrease
with $s$.

Because the prior has finitely many primitive values,
we can choose a closed interval $K$ of positive length
within this contact interval such that
$U_r(K)\subset(v_i,v_{i+1})$ for two adjacent primitive
values. Fix $x$ in the interior of $K$.
We will raise raw revenue at $x$, compensating elsewhere
in $K$ to preserve the mean.

Define the continuous piecewise linear function
\[
 T_x(s)=
 \begin{cases}
  \dfrac{s-\ell}{x-\ell}, & \ell\le s\le x,\\[1mm]
  \dfrac{h-s}{h-x}, & x<s\le h.
 \end{cases}
\]
It is zero at the face endpoints and one at $x$.
Since $r$ is affine on $I$, the curve $r+\epsilon T_x$
consists of the two line segments joining
$(\ell,r(\ell))$, $(x,r(x)+\epsilon)$, and $(h,r(h))$.

Choose $\delta>0$ sufficiently small that
$\eta_+(s)=(1-|s-x|/\delta)_+$ is supported in the
interior of $K$. This function satisfies
$\eta_+(x)=1$ and $0\le\eta_+\le T_x$.
Choose a nonzero, nonnegative Lipschitz function
$\eta_-$ supported on a disjoint subinterval of the
interior of $K$, and set
\[
 \eta(s)=\eta_+(s)
 -\frac{\int_K\eta_+(t)/t\,\dd t}
        {\int_K\eta_-(t)/t\,\dd t}\,\eta_-(s).
\]
Then $\eta$ is Lipschitz, is supported inside $K$, and
satisfies $\eta(x)=1$, $\eta\le T_x$, and
$\int_K\eta(s)/s\,\dd s=0$.

Set $w_\epsilon=w+\epsilon\eta/s$ on $K$ and leave
$w$ unchanged elsewhere. On $K$, the derivative
$w'=-A/s^2$ is bounded away from zero, while $\eta/s$
is Lipschitz. Thus, for sufficiently small $\epsilon>0$,
the perturbed quantile remains nonincreasing.
It also remains within $(v_i,v_{i+1})$ on $K$.
The mean is unchanged by construction.

Each function $(y-v_j)_+$ is either zero or $y-v_j$ throughout this
primitive-value interval. Its expectation is therefore unchanged. By
\eqref{eq:finite-feasibility}, the perturbed distribution
remains Bayes plausible, and every binding information
constraint continues to bind.

Let $R_\epsilon$ be the curve in
\Cref{lem:binary-one-point-lift}, with the point raised
at $x$. We show that it is the new concavified revenue
curve. On $K$, raw revenue becomes $r+\epsilon\eta$.
Concavity places $R_\epsilon$ above both chords
forming $r+\epsilon T_x$ on $I$. Since $\eta\le T_x$,
it also lies above the new raw curve on $K$.
Elsewhere, raw revenue is unchanged and lies below
$r\le R_\epsilon$.

Conversely, take any concave majorant of the new
raw curve. At the original face endpoints, raw
revenue still equals $r$. The majorant must therefore
lie above their chord, which is $r$ on $I$.
Outside $I$, the raw curve is unchanged.
The majorant thus lies above the entire old raw
curve, and hence above $r$. It must also lie above
the raised point $(x,r(x)+\epsilon)$, so it cannot
lie below $R_\epsilon$. Thus $R_\epsilon$ is indeed
the new concavified revenue curve.

Since $A,z>0$, we have $0<z<b\le1$.
The supporting-line bounds in
\Cref{lem:binary-one-point-lift} place the affected
actual and comparison slopes between
$z-\epsilon/(h-x)$ and $z+\epsilon/(x-\ell)$.
For small $\epsilon$, these lie in a fixed compact
subinterval of $(0,1)$. On this interval,
$F_{y,\beta}$ is uniformly Lipschitz in its second
argument for $y\in[a,1]$.
Thus \eqref{eq:binary-lift-slope-estimate} gives
\[
 J_\beta(w_\epsilon)
 =\int_0^1
 F_{w_\epsilon(s),\beta}(\widehat z_\epsilon(s))\,\dd s
 +o(\epsilon).
\]
Outside $I$, this integrand agrees with
the original one.

On $I$, the direct change in posterior values contributes
$q(z)\epsilon\int_K\eta(s)/s\,\dd s=0$ to first order.
The derivative with respect to virtual value is
$\partial_zF_{w(s),\beta}(z)
 =q'(z)w(s)-zq'(z)-(1-\beta)q(z)$.
The last two terms are constant across the face.
Their contributions cancel: the two comparison slope changes
have integrals $\epsilon$ and $-\epsilon$.
The first-order expansion therefore reduces to
\[
 J_\beta(w_\epsilon)-J_\beta(w)
 =\epsilon q'(z)\left[
 \frac{1}{x-\ell}\int_\ell^x w(s)\,\dd s
 -\frac{1}{h-x}\int_x^h w(s)\,\dd s
 \right]+o(\epsilon).
\]
The expression in brackets is strictly positive:
$w$ is nonincreasing throughout $I$ and strictly
decreasing near $x$, so its average to the left of
$x$ exceeds its average to the right.
Since $q'(z)>0$, the feasible perturbation strictly
raises the objective for sufficiently small $\epsilon$,
contradicting optimality.
\end{proof}

\subsubsection{Ordering and eliminating trading atoms}
\label{app:binary-ordering}

The following arguments pool lower posterior types and restore
the lost mean by raising values on a higher trading
face $I$ with positive intercept. By \Cref{lem:active-face-purity}, buyers on this face
share the posterior value $b_I$. Thus $v=b_I$ almost
everywhere and $d_I=\kappa_I>0$, leaving room to raise
posterior values toward $U_r$. This preserves $r$
and profit, while increasing consumer surplus by $q(z_I)$
per unit increase in the overall mean.

\begin{lemma}[The seller's response to pooling a block of types]
\label{lem:exact-chord-collapse}
Let $0<\alpha<\gamma\le1$ be contact ranks, and set
$b=U_r(\gamma)$ and
$\bar z=[r(\gamma)-r(\alpha)]/(\gamma-\alpha)$.
Pooling the posterior values on $(\alpha,\gamma)$ at $b$
preserves monotonicity. The resulting concavified revenue
curve equals the chord joining the two endpoint contacts
on $[\alpha,\gamma]$ and equals $r$ elsewhere.
The overall mean falls by
$D=\int_\alpha^\gamma(v-b)\,\dd s\ge0$.
Buyers in the pooled block receive quality $q(\bar z)$,
and the block contributes
$(\gamma-\alpha)F_{b,\beta}(\bar z)$ to the designer's
objective, where
$F_{b,\beta}(z)=(b-z)q(z)+\beta\pi(z)$ and
$q(z)=\pi(z)=0$ for $z\le0$.
\end{lemma}

\begin{proof}
Contact at $\gamma$ gives $b=v(\gamma-)$, and
monotonicity of $v$ places posterior values before
$\gamma$ at or above $b$ and those after it at or
below $b$. Replacing the values on $(\alpha,\gamma)$
by $b$ therefore preserves monotonicity and reduces
the mean by $D$.

The new raw revenue curve is $sb$ on the pooled block.
Both endpoint contacts survive since the left limit at
$\alpha$ is unchanged, and the new raw revenue curve at
$\gamma$ is $\gamma b=r(\gamma)$.
Let $L$ be the chord joining these contacts.
Since $r(\alpha)=\alpha v(\alpha-)\ge\alpha b$,
its slope satisfies $\bar z\le b$, and
$L(s)-sb=(\gamma-s)(b-\bar z)\ge0$ for $s\in[\alpha,\gamma]$.
Thus the chord lies above the new raw curve.
Concavity of $r$ places the chord's slope between
the slopes on either side of the block. Replacing
$r$ by $L$ on this block therefore gives a concave
majorant of the new raw curve.

To see that this is the least concave majorant, consider any
concave majorant of the new raw curve. The new raw curve
agrees with $L$ at $\alpha$ and $\gamma$, so concavity requires the majorant
to lie above $L$ throughout $[\alpha,\gamma]$.
Outside the block, the raw curve is unchanged
and retains all its contacts with $r$. Each connected component of
$\{r>f_v\}$ has contact endpoints, with $r$ equal to the
chord between them. Because $\alpha$ and $\gamma$ are
contacts, no noncontact interval crosses either endpoint of
the block. Thus, for each such interval outside the block,
both raw revenue endpoint values are unchanged, forcing any
new concave majorant to dominate the chord between them.
Every new concave majorant therefore lies above $L$ on the
block and above $r$ elsewhere, as claimed.

On the pooled block, the posterior value is $b$
and the virtual value is $\bar z$ almost everywhere.
The seller therefore supplies quality $q(\bar z)$,
and integrating $F_{b,\beta}(\bar z)$ over the block
gives the stated objective contribution. This also
covers $\bar z\le0$, when the pooled buyers are excluded.
\end{proof}

\begin{lemma}[The cost of pooling lower types]
\label{lem:binary-chord-cost}
Maintain \Cref{ass:curvature}. In the setting of
\Cref{lem:exact-chord-collapse}, write $z(s)=r'(s)$
for the original virtual values and let $\widetilde v$
be the pooled profile.
Suppose $0<z(s)<b$ and $z(s)\le z_A$ almost everywhere
on $(\alpha,\gamma)$. If $v(s)>b$ and $z(s)<z_A$ hold
together on a set of positive measure, then
$J_\beta(v)-J_\beta(\widetilde v)<q(z_A)D$.
If instead the original virtual values are nonpositive
almost everywhere on the block, its contributions to
consumer surplus and profit are zero both before and
after pooling.
\end{lemma}

\begin{proof}
Absolute continuity of $r$ gives
$\bar z=(\gamma-\alpha)^{-1}
\int_\alpha^\gamma z(s)\,\dd s$.
Pooling therefore replaces the original virtual values by
their average. Concavity of $F_{b,\beta}$ on $(0,b)$
and Jensen's inequality give
$\int_\alpha^\gamma F_{b,\beta}(z(s))\,\dd s
 \le
 (\gamma-\alpha)F_{b,\beta}(\bar z)$.
By \Cref{lem:exact-chord-collapse}, the right-hand side
is the pooled block's contribution to the objective.

To bound the loss from pooling, consider lowering the posterior
value from $v(s)$ to $b$ at each rank while holding the
virtual value fixed. This lowers the integrand by $(v(s)-b)q(z(s))$.
The inequality above shows that replacing the
virtual values by their average cannot lower
the block's contribution to the objective. Since posterior values and virtual values
outside the block are unchanged, we obtain
$J_\beta(v)-J_\beta(\widetilde v)
 \le
 \int_\alpha^\gamma (v(s)-b)q(z(s))\,\dd s$.
On the block, $v(s)\ge b$ and $q(z(s))\le q(z_A)$.
The product inequality is strict wherever $v(s)>b$ and
$z(s)<z_A$, since $q$ is strictly increasing over the
relevant positive virtual values. This set has positive
measure by hypothesis, so
\[
 \int_\alpha^\gamma (v(s)-b)q(z(s))\,\dd s
 <
 q(z_A)\int_\alpha^\gamma(v(s)-b)\,\dd s
 =q(z_A)D.
\]

If the original virtual values are instead nonpositive,
so is their average $\bar z$. Quality and profit are
then zero throughout the block both before and after
pooling, and consumer surplus is zero as well.
\end{proof}

\begin{lemma}[No separation below a nonterminal trading atom]
\label{thm:no-reentry-rigorous}
At a binary optimum, let $I_A=[\ell_A,h_A]$ be a
positive-intercept trading face with slope $z_A>0$.
Then the set
$\{s\in(h_A,1):v(s)=U_r(s)\}$ has measure zero.
Thus, no separating region of positive measure can lie
below the posterior atom on $I_A$.
\end{lemma}

\begin{proof}
We show that a positive-measure set of separating ranks below
the atom would allow us to pool a short block and restore the
lost mean on the trading face at a strict gain.

A zero-intercept face must extend to $s=0$. Its affine
extension is a supporting line through the origin.
Concavity and $r(0)=0$ also place $r$ above this line
between zero and the face. Hence $r$ coincides with
the line throughout that interval. Every face below
$I_A$ therefore has positive intercept.
By \Cref{lem:active-face-purity}, every such face with
positive slope satisfies $v=b_I$ almost everywhere,
as does $I_A$. In particular,
$d_{I_A}=\kappa_{I_A}>0$.
A lower face with nonpositive slope must also satisfy
$d_I=\kappa_I$: otherwise, increasing $d_I$ and
decreasing $d_{I_A}$ by the same small amount would
preserve the mean and profit while increasing consumer
surplus. Posterior values would fall where quality is
zero and rise where quality is $q(z_A)>0$.
Consequently, on the interior of every lower face,
$v=b_I<U_r$ almost everywhere.

Suppose now that the set in the statement has positive
measure. Removing face interiors, discontinuities of
$v$, and nondifferentiability points of $r$ leaves a
positive-measure set $E$ of contacts. We may also
discard ranks with $z=0$, since a positive-measure zero
level of the nonincreasing slope lies in an affine interval
whose interior has already been removed.
Choose a left density point $\gamma\in E$, so every
sufficiently short interval immediately to its left
contains a positive-measure part of $E$.

Set $b=U_r(\gamma)$. We have $z(\gamma)<z_A$:
equality would make $r$ linear between $I_A$ and
$\gamma$, contradicting the maximality of $I_A$.
Also $b>z(\gamma)$, since equality would make the supporting
line at $\gamma$ pass through the origin and force
$r$ to be linear on $[0,\gamma]$, contrary to
the positive intercept of $I_A$.

Since $r$ is concave and differentiable at $\gamma$,
its slopes approach $z(\gamma)$ from the left.
If $z(\gamma)>0$, a sufficiently short left
neighborhood therefore satisfies
$0<z(s)<\min\{z_A,b\}$ almost everywhere.
If $z(\gamma)<0$, the slopes in that neighborhood
are negative. In either case, after shortening the
neighborhood if necessary, $U_r(s)-z(s)$ is bounded
away from zero. Absolute continuity and
$U_r'(s)=[z(s)-U_r(s)]/s$ then imply that $U_r$
is strictly decreasing there.

Choose another contact $\alpha\in E$ in this
neighborhood, sufficiently close to $\gamma$ that
$\gamma-\alpha<\kappa_{I_A}$.
The density-point property gives
$|E\cap(\alpha,\gamma)|>0$. On this set,
$v(s)=U_r(s)>U_r(\gamma)=b$.
Pooling the block at $b$ therefore reduces the mean by
$0<D=\int_\alpha^\gamma(v(s)-b)\,\dd s
 \le\gamma-\alpha<\kappa_{I_A}$.

To restore the mean, use \Cref{lem:exact-chord-collapse}:
pooling leaves the concavified revenue curve unchanged on
$I_A$. We can therefore raise posterior values on $I_A$
while preserving the new curve, reducing $d_{I_A}$
from $\kappa_{I_A}$ to $\kappa_{I_A}-D$.
By \Cref{lem:fixed-envelope-fillings,prop:mean-accounting},
this restores the original mean and increases consumer
surplus by $q(z_A)D$, without changing profit.

If the pooled block has positive slopes,
\Cref{lem:binary-chord-cost} bounds its objective loss
strictly below $q(z_A)D$. If its slopes are negative,
pooling leaves its contribution to the objective at
zero. The combined change therefore strictly raises
$J_\beta$. It is feasible because both changes preserve
the support bounds $[a,1]$ and together preserve the
prior mean, contradicting optimality.
\end{proof}

\begin{lemma}[Immediate successors and shared contacts]
\label{lem:binary-successor}
Let $I_A=[\ell_A,h_A]$ be a positive-intercept trading
face at a binary optimum. If a trading face lies below
$I_A$ in posterior value, then there is a trading face
$I_B=[h_A,h_B]$ with $h_B>h_A$.
Thus, $I_A$ has an immediate lower trading successor,
and the two faces meet at the contact $h_A$.
\end{lemma}

\begin{proof}
Suppose a trading face lies below $I_A$, but none
begins at $h_A$. Fix one of these lower trading faces.
Concavity of $r$ implies that slopes are positive
between $I_A$ and this face. Every affine face in
this region therefore trades, and has positive
intercept because a zero-intercept face must extend
to $s=0$. If the left endpoints of these trading faces were
bounded away from $h_A$, the intervening interval
would contain no face interiors. Almost every rank
there would then satisfy $v=U_r$, contradicting
\Cref{thm:no-reentry-rigorous}.
Trading faces must therefore accumulate at $h_A$
from the right. We will pool a short block of these faces
and restore the lost mean on the original face.

The positive intercept of $I_A$ gives
$U_r(h_A)>z_A$ and $\kappa_{I_A}>0$.
By continuity, choose $\delta>0$ so that
$[h_A,h_A+\delta]$ lies before the fixed lower face,
$\delta<\kappa_{I_A}$, and $U_r>z_A$ throughout
this interval. Here $0<z(s)\le z_A$ almost everywhere.
Since $U_r$ is absolutely continuous and
$U_r'(s)=[z(s)-U_r(s)]/s<0$ almost everywhere,
it is strictly decreasing on this interval.

Because the accumulating faces have disjoint interiors,
we can choose two of them, $J_1=[\ell_1,h_1]$ and
$J_2=[\ell_2,h_2]$, with
$h_A<\ell_1<h_1\le\ell_2<h_2<h_A+\delta$.
Both $h_A$ and $h_2$ are contacts. Pool the block
$(h_A,h_2)$ at $b=U_r(h_2)>z_A$.
Its original virtual values satisfy
$0<z(s)\le z_A<b$ almost everywhere.

By \Cref{lem:active-face-purity}, buyers on $J_1$
share the posterior value $U_r(h_1)$, which is
strictly greater than $b$ because $U_r$ is strictly
decreasing. The slope on $J_1$ is also strictly
below $z_A$: two distinct maximal faces of a
concave curve cannot have the same slope.
These two strict inequalities on $J_1$ give the conditions
for the strict bound in \Cref{lem:binary-chord-cost}.
The mean reduction satisfies
\[
 0<D=\int_{h_A}^{h_2}(v(s)-b)\,\dd s
 \le h_2-h_A<\delta<\kappa_{I_A},
\]
and the pooling loss is strictly less than
$q(z_A)D$.

The pooled block's slope is the average of its
original slopes and is strictly below $z_A$.
By \Cref{lem:exact-chord-collapse}, $I_A$ therefore
remains a face of the new revenue curve, with its
slope and capacity unchanged.
Moreover, \Cref{lem:active-face-purity} gives
$d_{I_A}=\kappa_{I_A}$, so we can raise posterior values
on $I_A$ by reducing $d_{I_A}$ to $\kappa_{I_A}-D$.
By \Cref{lem:fixed-envelope-fillings,prop:mean-accounting},
this adjustment restores the original mean without
changing the new revenue curve or profit. Consumer
surplus rises by $q(z_A)D$.

This gain exceeds the loss from pooling, so the
two changes together increase $J_\beta$. The
resulting distribution has support in $[a,1]$
and mean $m$, making it feasible for the binary
prior. This contradicts optimality.
\end{proof}

\begin{lemma}[Eliminating two adjacent posterior atoms]
\label{lem:binary-adjacent-elimination}
Suppose $r$ has adjacent positive-intercept faces
$[\ell,h]$ and $[h,k]$, with slopes $z>Z\ge0$.
The posterior value is constant on each face, equal
to $v$ and $w$, respectively, with $v>w$.
There is a modification confined to $[\ell,k]$
that preserves the mean and the support bounds $[a,1]$,
strictly increases $J_\beta$, and leaves the concavified
revenue curve unchanged outside this interval.
\end{lemma}

\begin{proof}
We raise posterior values on part of the first face
and pool the remaining part with the second atom.
The cutoff between these changes is chosen to preserve
the mean.

On the first face, $U_r(s)=z+h(v-z)/s$.
For $\alpha\in(\ell,h)$, replace $v$ by $U_r$
on $(\ell,\alpha)$, leaving other posterior values
unchanged. By \Cref{lem:fixed-envelope-fillings},
this preserves $r$ and makes $\alpha$ a contact:
the modified raw revenue there is
$\alpha U_r(\alpha)=r(\alpha)$.
The mean increases by
$X(\alpha)=\int_\ell^\alpha[U_r(s)-v]\,\dd s$.
Quality on these ranks remains $q(z)$ and profit
is unchanged, so the objective increases by
$q(z)X(\alpha)$.

Next, pool $(\alpha,k)$ at $w$.
Only the remaining ranks of the first atom change
posterior value, so this reduces the mean by
$(h-\alpha)(v-w)$.
Since the first face has positive intercept,
$U_r(s)>v$ for $\ell<s<h$.
Thus $X$ increases continuously and strictly from
$X(\ell)=0$ to $X(h)>0$, while
$(h-\alpha)(v-w)$ decreases strictly to zero.
There is therefore a unique $\alpha_*\in(\ell,h)$
satisfying
$X(\alpha_*)=(h-\alpha_*)(v-w)$.
The two changes together preserve the mean.

Write $a_1=h-\alpha_*$, $a_0=k-h$, and
$\rho=a_1/(a_1+a_0)$.
By \Cref{lem:exact-chord-collapse}, the pooled block
has constant virtual value
$\widehat z=\rho z+(1-\rho)Z$.
Here $Z<\widehat z$ because $0<\rho<1$ and $z>Z$.
Also $\widehat z<w$: the chord ends at $r(k)=kw$
and starts at $r(\alpha_*)>\alpha_*w$.

The choice of cutoff also simplifies the payoff comparison.
The first change gains $q(z)X(\alpha_*)$, while pooling
replaces the contributions of the remaining two blocks by
$(a_1+a_0)F_{w,\beta}(\widehat z)$. Using
$F_{v,\beta}(z)=F_{w,\beta}(z)+(v-w)q(z)$
and $X(\alpha_*)=a_1(v-w)$, we obtain
\[
\begin{aligned}
 \Delta J_\beta
 &=q(z)X(\alpha_*)
   +(a_1+a_0)F_{w,\beta}(\widehat z)
 -a_1F_{v,\beta}(z)-a_0F_{w,\beta}(Z)\\
 &=(a_1+a_0)
 \bigl\{
 F_{w,\beta}(\widehat z)
 -\rho F_{w,\beta}(z)
 -(1-\rho)F_{w,\beta}(Z)
 \bigr\}.
\end{aligned}
\]

To show that the gain is positive, write $F=F_{w,\beta}$
and distinguish whether both original slopes lie in the
interval on which strict concavity is guaranteed. If $z\le w$,
\Cref{ass:curvature} and continuity at the endpoints
make $F$ strictly concave on $[0,w]$. Since $z>Z$ and $0<\rho<1$,
strict Jensen's inequality gives $\Delta J_\beta>0$,
including when $Z=0$ or $z=w$.

If $z>w$, we instead compare chord slopes on either side
of $w$. Strict concavity below $w$ makes the chord from
$Z$ to $w$ steeper than $F'(w)$.
For $w<t\le z$, however,
\[F'(t)
 =(w-t)q'(t)-(1-\beta)q(t)
 <-(1-\beta)q(w)=F'(w).\]
Integrating over $(w,z)$ shows that the chord from
$w$ to $z$ has slope strictly below $F'(w)$.
The slope of the chord from $Z$ to $z$ is a weighted
average of these two chord slopes, and is therefore
smaller than the slope from $Z$ to $w$.
Since both chords start at $(Z,F(Z))$, the former
lies below the latter on $(Z,w]$.
At $\widehat z\in(Z,w)$, strict concavity places
$F(\widehat z)$ above the chord from $Z$ to $w$,
and hence above the chord from $Z$ to $z$.
Thus $\Delta J_\beta>0$ in this case as well,
without requiring concavity above $w$.

Both modifications preserve monotonicity and the bounds
$[a,1]$, while the choice of $\alpha_*$ preserves the mean.
The first modification leaves $r$ unchanged everywhere,
and the pooling lemma changes it only on
$[\alpha_*,k]\subset[\ell,k]$, completing the proof.
\end{proof}

\subsubsection{Exclusion and a nondegenerate regular upper component}
\label{app:binary-continuation-existence}

The optimal value is strictly positive. Choose
$0<\epsilon<\min\{m-a,1-m,m/3\}$ and put equal
probabilities on $m-\epsilon$ and $m+\epsilon$.
This distribution has mean $m$ and support in $[a,1]$,
so it is feasible. Its lower virtual value is
$m-3\epsilon>0$. The higher type receives rent
$2\epsilon q(m-3\epsilon)$, giving consumer surplus
$\epsilon q(m-3\epsilon)>0$. Since profit is
nonnegative, this policy gives $J_\beta>0$.

\begin{lemma}[The posterior value of the excluded group]
\label{lem:binary-exclusion-floor}
Under \Cref{ass:cost}, suppose an optimal posterior
distribution for a binary prior excludes a positive
mass of buyers. Their posterior means equal $a$
almost surely. Moreover, the bottom $B$ of the trading
support satisfies $B>a$, so the excluded atom is
separated from the trading support by a positive gap.
\end{lemma}

\begin{proof}
Since the optimal value is positive, some buyers trade.
Monotonicity of quality therefore makes the trading ranks
$(0,s)$ for some $0<s<1$, up to null sets.

The cutoff $s$ must be a contact: if $r(s)>f_v(s)$,
the revenue curve would be affine on a neighborhood of
$s$, making virtual value and quality constant across
a cutoff that separates trade from exclusion.
Thus $B:=r(s)/s=v(s-)$ is the bottom of the trading
support. Virtual values are positive almost everywhere
on $(0,s)$, so $r(0)=0$ implies $r(s)>0$ and $B>0$.
They are nonpositive after $s$, giving
$a\le r(1)\le r(s)=sB<B$.
In particular, $B>a$.

Let $G_A$ denote the posterior distribution conditional
on trade, with mean $M$ and objective value
$J=J_\beta(G_A)$. Because $s$ is a contact, its
concavified revenue curve is
$r_A(t)=r(st)/s$ for $0\le t\le1$.
Its slopes are $r'(st)$, so normalization preserves
virtual values and payoffs per trading buyer.
Excluded buyers contribute zero, and hence the
original objective is $sJ>0$.

Write $\bar a$ for the mean posterior value among
excluded buyers. Then
$m=(1-s)\bar a+sM$, with
$a\le\bar a\le B\le M$.
In fact, $\bar a<M$, since equality would force both
conditional distributions to be concentrated at $B>0$,
in which case the seller would supply $q(B)>0$ to
everyone. Hence $m<M$.

Suppose $\bar a>a$. Lowering the excluded distribution
to an atom at $a$ allows us to increase the trading share
while preserving the overall mean:
\[
 G'=(1-s')\delta_a+s'G_A,
 \qquad
 s'=\frac{m-a}{M-a}\in(s,1).
\]
Here $s'>s$ because
$m-a=s(M-a)+(1-s)(\bar a-a)$, while $s'<1$
because $m<M$. The new distribution has mean $m$
and support in $[a,1]$, so it is feasible.

To determine the seller's response, rescale the
active curve to $s'r_A(t/s')$ on $[0,s']$ and
join its endpoint $(s',s'B)$ to $(1,a)$ by a chord.
Since $s'B>sB\ge a$, this chord has strictly negative
slope. It lies above the lower raw curve $at$ on
$(s',1)$ and joins the positive active slopes
concavely. The resulting curve is therefore a
concave majorant of the new raw revenue curve.

Any concave majorant must dominate the rescaled $r_A$
on the active part, by the definition of $r_A$.
On the remaining part, the raw revenue values
$s'B$ at $s'$ and $a$ at $1$ force it to dominate the
joining chord. This identifies the seller's response
to $G'$.

The atom at $a$ is excluded, and payoffs per trading
buyer are unchanged. The new objective is therefore
$s'J>sJ$, contradicting optimality.
Hence $\bar a=a$. Since excluded posterior values
are bounded below by $a$, they must equal $a$
almost surely. Trading posterior values are at least $B>a$, so $(a,B)$
is a nonempty gap in posterior support.

On the excluded ranks, the concavified revenue curve is
the chord joining $(s,sB)$ to $(1,a)$.
This chord dominates the raw revenue curve $at$,
joins the trading curve concavely, and is forced
by the two endpoint contacts.
\end{proof}

\begin{lemma}[A contraction of the top two atoms that raises $J_\beta$]
\label{lem:top-atom-refinement}
Under \Cref{ass:cost}, fix $0\le\beta<1$.
Suppose the top two faces are adjacent and pure, with
posterior values $v<p$ and respective virtual values
$z<p$ and $p$. There is a modification of these atoms
that preserves their total probability and mean,
contracts the posterior distribution in convex order,
and strictly increases $J_\beta$.
The concavified revenue curve below the two atoms
is unchanged. The lower virtual value $z$ may be
zero or negative.
\end{lemma}

\begin{proof}
We create a small atom just below the top value by pooling
buyers from the two original atoms at their mean. Its mass
will be chosen to make the objective increase.

Normalize the combined probability of the two atoms
to one, and denote their masses by $\rho$ at $v$
and $\tau=1-\rho$ at $p$. Then
$z=(v-\tau p)/\rho$.
Fix $\max\{z,0\}<\zeta<p$, to be chosen close to $p$,
and set $\lambda=\tau/(p-\zeta)$.
For small $h>0$, create an atom of mass $\lambda h$
at $p-h$. Take a mass $\lambda h^2/(p-v)$ from
the atom at $v$ and the remaining mass
$\lambda h[1-h/(p-v)]$ from the atom at $p$.
The fraction drawn from $v$ is $h/(p-v)$, so
the mean of the pooled buyers is $p-h$.
Leaving all other buyers at their original posterior
values gives probabilities
\[
 \rho_h=\rho-\frac{\lambda h^2}{p-v},\qquad
 \eta_h=\lambda h,\qquad
 \tau_h=\tau-\lambda h+\frac{\lambda h^2}{p-v}
\]
at $v,p-h,p$, respectively.
For sufficiently small $h$, these probabilities
are positive and $v<p-h<p$.
Pooling buyers at their mean preserves probability and mean
and contracts the posterior distribution in convex order.
The modification is therefore feasible under any prior
for which the original distribution was feasible.

The three revenue chords have slopes
\[
 z_{v,h}=\frac{v-(1-\rho_h)(p-h)}{\rho_h},\qquad
 z_{p-h,h}=\zeta-\frac{h^2}{p-v},\qquad
 z_{p,h}=p.
\]
As $h\downarrow0$, these converge to $z,\zeta,p$,
which are strictly ordered. Hence the chords form
a concave curve in upper-tail ranks for small $h$.
Each chord meets the raw revenue curve at two
consecutive valuation cutoffs and lies above it
between them.

The pair's total probability and lowest value $v$
are unchanged, so the contact at its lower end
is unchanged as well. Since
$\left.\partial_h z_{v,h}\right|_{h=0}=\tau/\rho>0$,
the new lowest slope exceeds $z$ for small $h$.
The slope below the pair is at most $z$, so the
new curve remains concave at the join.
Any concave majorant must lie above the chord
endpoints and all the unchanged contacts below
the pair. By concavity, it must then lie above
the constructed curve everywhere. This curve
is therefore the new concavified revenue curve.

The contribution to the objective from ranks
outside the pair is unchanged. With the pair's
probability normalized to one, its new contribution is
\[
 J_{\beta,h}
 =\rho_h F_{v,\beta}(z_{v,h})
  +\eta_h F_{p-h,\beta}(z_{p-h,h})
  +\tau_h\beta\pi(p).
\]
If $z>0$, we can compute the first-order gain from the right derivative
at $h=0$. Using $\rho'_0=0$, $\eta'_0=\lambda$,
$\tau'_0=-\lambda$, and $\pi'=q$ gives
\[
 \frac1\tau
 \left.\frac{\dd J_{\beta,h}}{\dd h}\right|_{h=0+}
 =(v-z)q'(z)+q(\zeta)-(1-\beta)q(z)
 -\frac{\beta}{p-\zeta}\int_\zeta^p q(t)\,\dd t.
\]
As $\zeta\uparrow p$, the right-hand side converges to
$(v-z)q'(z)+(1-\beta)[q(p)-q(z)]>0$.
Here $v-z=\tau(p-v)/\rho>0$, and $q$ is strictly
increasing at positive virtual values.

If $z\le0$, we obtain a lower bound on the gain without
differentiating the lower atom's contribution. The atom
at $v$ originally contributes zero, and its new contribution
is nonnegative because $z_{v,h}<v$ and both $q$ and $\pi$
are nonnegative. Dropping this contribution gives
\[
 \liminf_{h\downarrow0}
 \frac{J_{\beta,h}-J_{\beta,0}}{\tau h}
 \ge
 q(\zeta)
 -\frac{\beta}{p-\zeta}\int_\zeta^p q(t)\,\dd t.
\]
The right-hand side converges to
$(1-\beta)q(p)>0$ as $\zeta\uparrow p$.
This argument includes $z=0$ without requiring
differentiability of $q$ at zero.

In either case, choosing $\zeta$ sufficiently close to $p$
and then $h>0$ sufficiently small gives a strict gain.
Multiplying by the pair's original probability preserves
its sign, proving the claim for the full posterior distribution.
\end{proof}

\begin{lemma}[A nondegenerate regular upper component]
\label{lem:binary-nonempty-continuation}
Every binary optimum has a regular upper component
of positive probability. After normalizing its
probability to one, its virtual-value profile
$\varphi$ in lower rank is positive and strictly
increasing on $(0,b_0)$ for some $b_0\in(0,1]$.
If $b_0<1$, the profile is constant on $[b_0,1]$.
The corresponding posterior quantile is continuous
and strictly increasing on $[0,b_0)$.
\end{lemma}

\begin{proof}
We first identify a regular upper component of positive
probability, then show that it cannot consist of a single atom.

By \Cref{lem:active-face-purity}, every positive-intercept
trading face contains a single posterior atom.
There is at most one such face: otherwise,
\Cref{lem:binary-successor} gives an adjacent pair whose
slopes and posterior values strictly decrease along
upper-tail ranks. Then \Cref{lem:binary-adjacent-elimination}
gives a strict improvement, contradicting optimality.

Let $(0,s)$ be the trading ranks. The positive optimal
value established above implies $s>0$.
If a positive-intercept trading face
$I=[\ell,h]$ exists, then $\ell>0$, since a face
starting at zero has zero intercept.
The face must end at the trading cutoff, so $h=s$.
Indeed, if $h<s$, there could be no face interior
in the trading interval $(h,s)$:
another positive-intercept trading face has been
ruled out, and a zero-intercept face cannot follow
$I$. Almost every rank in this interval would
therefore satisfy $v=U_r$, contradicting
\Cref{thm:no-reentry-rigorous}.

Removing the excluded ranks and this possible
trading atom leaves $(0,\ell)$ if the atom exists,
and $(0,s)$ otherwise. In either case, the remaining component has positive
probability and no positive-intercept face. Every remaining
face therefore has zero intercept and zero capacity, so
\eqref{eq:face-value-bounds} gives
$v=U_r$ almost everywhere on the component.
Its lower boundary is a contact. Conditioning on
this component and normalizing its probability
therefore gives a regular posterior distribution.

Any affine part of its revenue curve must belong
to a zero-intercept face beginning at upper-tail
rank zero. In normalized lower ranks, this becomes
a terminal constant part of $\varphi$.
Let $b_0\in[0,1]$ be its starting rank, setting
$b_0=1$ if there is no such part.
On $(0,b_0)$, equality of virtual values at two
distinct ranks would, by monotonicity, force an
affine interval between them. Since no such interval
exists, $\varphi$ is strictly increasing there.
Its monotone representative is positive at every
interior rank because the component consists
entirely of trading buyers.

To rule out $b_0=0$, suppose the regular component is
a single posterior atom. If it contains all buyers,
the posterior distribution is $\delta_m$.
The feasible two-point perturbation considered above
changes the objective by
\[
 \epsilon q(m-3\epsilon)
 +\frac{\beta}{2}
 \{\pi(m-3\epsilon)+\pi(m+\epsilon)-2\pi(m)\}
 =(1-\beta)\epsilon q(m)+o(\epsilon).
\]
The expansion uses $\pi'=q$ and continuity of $q$.
The gain is strictly positive for sufficiently small $\epsilon$
because $m>0$ and $\beta<1$.

If there are buyers outside the component, its
single top atom has another atom immediately below
it: the lower trading atom if one exists, or the
excluded atom at $a$ identified by
\Cref{lem:binary-exclusion-floor}.
In either case, \Cref{lem:top-atom-refinement}
gives a feasible strict improvement.
Thus $b_0>0$.

The claimed quantile properties now follow from
\eqref{eq:continuation-quantile}, which expresses $Q(u)$
as the average of $\varphi$ over $(u,1)$.
Since $\varphi$ is strictly increasing, this average
is strictly greater than $\varphi(u)$ for $u<b_0$.
Consequently,
$(1-u)Q'(u)=Q(u)-\varphi(u)>0$ for almost every $u\in(0,b_0)$.
The formula for $Q$ also implies absolute continuity
on $[0,c]$ for every $c<1$. Hence $Q$ is continuous
and strictly increasing on $[0,b_0)$.
\end{proof}

\subsection{The separating distribution}
\label{app:binary-separation-proof}

To derive the Euler equation, we vary the upper
component while holding its probability, lower
support point, mean, and endpoint virtual values
fixed. No adjustment to any lower group is then
needed. We subsequently prove smoothness and derive
the terminal conditions.

\subsubsection{The two constraints and their multipliers}
\label{app:binary-multipliers}

Normalize the component to mass one and let
$w(u)=-\log(1-u)$. Its lower support point and mean
are linear functions of $\varphi$:
\[
 B(\varphi)=Q_\varphi(0)
 =\int_0^1\varphi(u)\,\dd u,
 \qquad
 M(\varphi)=\int_0^1w(u)\varphi(u)\,\dd u.
\]
The expression for $M$ follows from integrating
\eqref{eq:continuation-quantile} and reversing the order
of integration. With the component's probability fixed, holding
$B$ and $M$ fixed preserves the lower revenue
contact and the overall posterior mean.
Let $J(\varphi)$ denote the weighted objective
of the normalized upper distribution.
The first-variation formula and remainder bound
in \eqref{eq:binary-cited-first-variation} apply
to bounded perturbations supported where $\varphi$
is bounded away from zero. They require only
\Cref{ass:cost}, so we need not assume that the
optimal profile is differentiable.

\begin{lemma}[The two multipliers and interior compensation]
\label{lem:binary-two-multipliers}
For the upper component in
\Cref{lem:binary-nonempty-continuation}, there are unique
$\theta,\xi\in\R$ satisfying
\begin{equation}
 H_\beta(u)-\theta w(u)+\xi=0
 \quad\text{a.e. on }(0,b_0).
 \label{eq:binary-Euler-ae}
\end{equation}
In any nonempty open interval $I\subset(0,b_0)$,
we can choose two bounded directions supported
in $I$ with linearly independent effects on $(B,M)$.
Any fixed linear combination of these directions
can be added to $\varphi$ with a sufficiently small
coefficient of either sign without violating
monotonicity or the bounds $[0,1]$.
In particular, there is a direction $\eta$ supported
in $I$ satisfying
\begin{equation}
 DB[\eta]=0,\qquad DM[\eta]=1,\qquad DJ[\eta]=\theta.
 \label{eq:binary-remote-mean}
\end{equation}
Along this direction,
$B(\varphi+\epsilon\eta)=B(\varphi)$ and
$M(\varphi+\epsilon\eta)=M(\varphi)+\epsilon$
for every sufficiently small $\epsilon$ of either sign.
\end{lemma}

\begin{proof}
Following \citet[Lemma~A.15]{kyriazis2026markets},
consider directions $\eta_h=h\circ\varphi$, where
$h\in C_c^1((0,1))$ and $\eta_h$ has compact support
in $(0,b_0)$. For sufficiently small $\epsilon$
of either sign, $z\mapsto z+\epsilon h(z)$ is
increasing and maps $[0,1]$ onto itself.
Hence $\varphi+\epsilon\eta_h$ is nondecreasing,
takes values in $[0,1]$, and leaves the endpoint
virtual values unchanged. This construction does
not require $\varphi$ to be smooth.

If $DB[\eta_h]=DM[\eta_h]=0$, linearity implies
that $B$ and $M$ remain unchanged along
$\varphi+\epsilon\eta_h$. The reconstructed quantile
has the same lower support point and mean, and
no adjustment to any lower group is needed.
The distribution is therefore feasible for all
sufficiently small $\epsilon$ of either sign,
so optimality implies $DJ[\eta_h]=0$.

Choose two distinct continuity points of $\varphi$
in $I$. Take nonnegative smooth functions $h_1,h_2$,
each positive at the corresponding virtual value
and supported in a sufficiently small neighborhood
of it. Since $\varphi$ is strictly increasing and
continuous at the chosen ranks, the functions
$\eta_i=h_i\circ\varphi$ have positive integrals
and disjoint compact supports in $I$.
Label them so that the support of $\eta_1$ lies
to the left of that of $\eta_2$. Write
\[
 b_i=DB[\eta_i]=\int_0^1\eta_i(u)\,\dd u>0,
 \qquad
 m_i=DM[\eta_i]=\int_0^1w(u)\eta_i(u)\,\dd u.
\]
The ratio $m_i/b_i$ is the average of $w$ weighted
by $\eta_i$. Since the supports are ordered and $w$
is strictly increasing, $m_1/b_1<m_2/b_2$.
Thus $(b_1,m_1)$ and $(b_2,m_2)$ are linearly
independent. Any linear combination of the two
directions is again of the form $h\circ\varphi$.
A sufficiently small perturbation along it,
of either sign, therefore preserves monotonicity
and the bounds.

Fix one such pair. The equations
$DJ[\eta_i]=\theta m_i-\xi b_i$, $i=1,2$,
uniquely determine $\theta$ and $\xi$.
For any direction $\eta_h$, choose the unique
linear combination of $\eta_1,\eta_2$ with the
same effects on $B$ and $M$. Subtracting it from
$\eta_h$ gives a direction with $DB=DM=0$,
and hence $DJ=0$. It follows that
$DJ[\eta_h]=\theta DM[\eta_h]-\xi DB[\eta_h]$.
The first-variation formula then implies that
$H_\beta-\theta w+\xi$ has zero integral against
every admissible $h\circ\varphi$.

For $0<s<t<b_0$, choose smooth functions
$0\le h_n\le1$, compactly supported in
$(\varphi(s),\varphi(t))$, that converge to one
throughout that interval. Since $\varphi$ is strictly increasing,
$h_n(\varphi(u))\to1_{(s,t)}(u)$.
The residual is locally integrable, so dominated
convergence yields
$\int_s^t[H_\beta(u)-\theta w(u)+\xi]\,\dd u=0$.
Since $s,t$ are arbitrary, this proves
\eqref{eq:binary-Euler-ae}. Since $w$ is strictly increasing,
the two multipliers are also unique.

For the stated mean adjustment, construct the pair within
any prescribed $I$ and set
$\eta=(b_1\eta_2-b_2\eta_1)/(b_1m_2-b_2m_1)$.
Then $DB[\eta]=0$ and $DM[\eta]=1$, so the multiplier
identity gives $DJ[\eta]=\theta$. Since $B$ and $M$ are
linear, the stated changes hold along $\varphi+\epsilon\eta$.
\end{proof}

\subsubsection{Regularity, the value of mean, and the top atom}
\label{app:binary-terminal}

\begin{lemma}[Interior smoothness and dynamics]
\label{lem:binary-interior-dynamics}
Let $Q=Q_\varphi$ and
$A(u)=\xi+\int_0^u[q(\varphi(t))-\theta]/(1-t)\,\dd t$.
On $(0,b_0)$, the functions $Q,A,\varphi$ are continuously
differentiable, with $Q>\varphi>0$. They satisfy
\eqref{eq:continuation-euler} and \eqref{eq:continuation-Q},
together with $A'=[q(\varphi)-\theta]/(1-u)$.
The expression for $\varphi'$ in
\eqref{eq:continuation-phi} has a strictly positive
denominator. The next lemma establishes $\varphi'>0$.
\end{lemma}

\begin{proof}
The tail-average formula gives $Q>\varphi>0$, and
$A,Q$ are locally absolutely continuous.
By \eqref{eq:binary-Euler-ae}, the Euler equation
holds almost everywhere. At fixed $(A,Q)$, its
left-hand side is strictly decreasing in virtual
value $z$, since \Cref{ass:curvature} gives
$(Q-z)q''(z)-(2-\beta)q'(z)<0$ for $0<z<Q$.
The continuity and implicit-function argument in
\citet[proof of Proposition~A.17]{kyriazis2026markets}
therefore applies locally. It rules out jumps in
$\varphi$ and makes the Euler equation hold everywhere,
with $\varphi$ locally a $C^1$ function of $(A,Q)$.

The right-hand sides of
$A'=[q(\varphi)-\theta]/(1-u)$ and
$Q'=(Q-\varphi)/(1-u)$ are continuous, so $A$ and
$Q$ are $C^1$. Since $\varphi$ is locally a $C^1$
function of $(A,Q)$, it is $C^1$ as well.
Differentiating the Euler equation gives
\eqref{eq:continuation-phi}. Its denominator is
strictly positive by the inequality above.
\end{proof}

\begin{lemma}[Quality and the multiplier on mean]
\label{lem:shadow-quality}
For every $u\in(0,b_0)$, $q(\varphi(u))\ge\theta$.
Consequently, $\varphi'(u)>0$ throughout that interval.
\end{lemma}

\begin{proof}
Fix $u_0\in(0,b_0)$ and choose $h>0$ small enough
that $u_0+h<b_0$ and $\varphi(u_0+h)<Q(u_0)$.
Both endpoints are contacts because the upper
distribution is regular. After reversing ranks,
we can apply \Cref{lem:exact-chord-collapse}.
Pooling the posterior values on $[u_0,u_0+h]$
at $Q(u_0)$ leaves the concavified revenue curve
unchanged outside the block and reduces the mean by
$D_h=\int_{u_0}^{u_0+h}[Q(u)-Q(u_0)]\,\dd u>0$.
On the block,
$0<\varphi(u)\le\varphi(u_0+h)<Q(u_0)$,
so \Cref{lem:binary-chord-cost} gives a loss
in the weighted objective of at most
$q(\varphi(u_0+h))D_h$.

To compensate for this loss of mean, choose a direction
$\eta$ from \eqref{eq:binary-remote-mean}, supported in a
fixed interval contained in $(0,u_0)$. Apply the compensation
first, replacing $\varphi$ by $\varphi+D_h\eta$.
This preserves $B$, raises the mean by $D_h$, and changes
the objective by $\theta D_h+o(D_h)$. Since $\eta$ is supported below $u_0$,
\eqref{eq:continuation-quantile} implies that $Q$
and the revenue curve remain unchanged on
$[u_0,u_0+h]$.

Pooling this block then restores the prior mean
without violating monotonicity or the support
bounds. The upper component retains its bottom
contact and endpoint slopes, so no adjustment
to any lower group is required. Optimality therefore requires
$0\ge
 [\theta-q(\varphi(u_0+h))]D_h+o(D_h)$.
Dividing by $D_h>0$ and letting $h\downarrow0$ gives
$\theta\le q(\varphi(u_0))$.

This bound makes the numerator in \eqref{eq:continuation-phi}
strictly positive, since $Q-\varphi>0$ and $q'>0$.
The denominator is positive by \Cref{ass:curvature},
and hence $\varphi'>0$.
\end{proof}

\begin{lemma}[The efficiently served top atom]
\label{lem:binary-top-fit}
The regular upper component has $b_0=b\in(0,1)$
and $Q=\varphi=1$ on $[b,1]$. Moreover,
$\lim_{u\uparrow b}(A(u),Q(u),\varphi(u))
 =((1-\beta)\bar q,1,1)$.
\end{lemma}

\begin{proof}
We first show that the upper component ends in an atom.
Since quality is strictly increasing,
\Cref{lem:shadow-quality} implies that
$q(\varphi)-\theta$ is bounded below by a positive
constant near rank one. Thus $A(u)\to+\infty$ as
$u\uparrow1$. If $b_0=1$, this contradicts the Euler
bound $A<(1-\beta)q(\varphi)\le(1-\beta)\bar q$.
Hence $b:=b_0\in(0,1)$, and reconstruction gives
$Q=\varphi=p$ on $[b,1]$ for some $p\in(0,1]$.

To show that this atom is at the upper support bound,
suppose $p<1$ and choose $t>b$ so that
$A-(1-\beta)q(p)>0$ on $[t,1)$.
As in \citet[proof of Lemma~A.11]{kyriazis2026markets},
raise $\varphi$ by $\epsilon$ on this tail.
Use two fixed interior directions from
\Cref{lem:binary-two-multipliers} to offset the
changes in $B$ and $M$. The required coefficients are $O(\epsilon)$, and
both directions are supported strictly inside
$(0,b)$. For sufficiently small $\epsilon$, the
perturbed profile remains nondecreasing and within
$[0,1]$, with no change needed to any lower group.
The interior Euler equation gives the payoff
effect of the two compensating adjustments.
Adding this to the effect of the increase on
$[t,1)$ gives
\[
 \epsilon\int_t^1[A(u)-(1-\beta)q(p)]\,\dd u
 +o(\epsilon)>0.
\]
The integral is finite because $A$ grows at most
logarithmically. This contradiction proves $p=1$.

It remains to establish continuous fit at the atom.
Let $p_0=\varphi(b-)>0$. If $p_0<1$, continuity of $Q$
and the separating Euler equation
give
$A(b)=(1-\beta)q(p_0)-(1-p_0)q'(p_0)
      <(1-\beta)\bar q$.
Thus $A-(1-\beta)\bar q<0$ on a short initial
interval of the top atom. Lowering $\varphi$ on that interval, with the same interior
compensation, is feasible because $p_0<1$ and gives a
strictly positive first-order gain. This is the continuous-fit
argument in
\citet[proof of Proposition~A.17]{kyriazis2026markets}.
Hence $p_0=1$. Taking the limit in the Euler equation
now gives $A(b)=(1-\beta)\bar q$.
\end{proof}

\subsubsection{Proof of \texorpdfstring{\Cref{prop:separating-conditions}}{Proposition~\ref{prop:separating-conditions}}}
\label{app:separating-completion}

\begin{proof}
The preceding lemmas establish the multipliers,
smoothness, quality bound, and terminal conditions.
It remains to prove uniqueness of the separating
path for a fixed $\theta$.

To make the state equations autonomous, set $\tau=-\log(1-u)$,
with dots denoting derivatives with respect to $\tau$. The Euler equation determines
virtual value locally as a $C^1$ function $z=Z(A,Q)$.
This also holds at the terminal state
$((1-\beta)\bar q,1,1)$, where the derivative of
the Euler equation's left-hand side with respect
to $z$ is $-(2-\beta)q'(1)<0$.
Since $c$ is $C^3$ near $\bar q<\bar Q$, the inverse
$q=(c')^{-1}$ extends to a neighborhood of one,
allowing the implicit-function theorem to apply
there.

The state equations become
\[
 \dot A=q(Z(A,Q))-\theta,\qquad
 \dot Q=Q-Z(A,Q).
\]
The right-hand side is locally Lipschitz along the separating
path and near the terminal state, so the backward-uniqueness
argument in \citet[proof of Proposition~A.18]{kyriazis2026markets}
applies: two solutions with the same
$\theta$, traced backward from
$((1-\beta)\bar q,1)$, coincide for as long as both
are defined. The Euler equation then recovers the same virtual values.
Only positive virtual values are used, so no smooth extension
at zero is needed. Different lower attachments may still
select different stopping points on this common path. Its duration
$T=-\log(1-b)$ determines the top atom's probability
within the normalized upper component:
$1-b=e^{-T}$.
\end{proof}

\subsection{Lower attachment}
\label{app:binary-completion}

\subsubsection{The continuation's value and its lower attachment}
\label{app:binary-attachment}

\begin{lemma}[The upper distribution and the lower quality gap]
\label{lem:strict-attachment}
Consider the regular upper component of an optimum,
normalized to mass one. Let $B,M,J$ denote its lower
support point, mean, and weighted payoff when screened
on its own. Write $\bar z=\varphi(0)$, and let $\theta$
be its mean multiplier. If $\bar z>0$,
\begin{equation}
 J-\theta(M-B)=(B-\bar z)^2q'(\bar z)
             +\beta(B-\bar z)q(\bar z)+\beta\pi(\bar z).
 \label{eq:reduced-continuation-value}
\end{equation}
If a trading atom with virtual value $z$ precedes
this component, then
$q(z)\le\theta\le q(\bar z)$ and $z<\bar z$.
\end{lemma}

\begin{proof}
To prove the identity, write the mean and payoff
of the normalized tail above rank $u$ as
\[
 M_u=\frac{\int_u^1Q(t)\,\dd t}{1-u},\qquad
 J_u=\frac{\int_u^1\{[Q(t)-\varphi(t)]q(\varphi(t))
                       +\beta\pi(\varphi(t))\}\,\dd t}{1-u}.
\]
Set $W=Q-\varphi$, $K=J_u-\theta(M_u-Q)$, and
$F=W^2q'(\varphi)+\beta Wq(\varphi)+\beta\pi(\varphi)$.
In log rank, the tail averages satisfy
$\dot M_u=M_u-Q$ and
$\dot J_u=J_u-Wq(\varphi)-\beta\pi(\varphi)$.
Together with $\dot Q=W$ and
\eqref{eq:continuation-phi}, these identities give
$\dot K-K=\dot F-F
 =-W[q(\varphi)-\theta]-\beta\pi(\varphi)$.
The difference $K-F$ therefore solves $\dot K-\dot F=K-F$.
Since $K=F=\beta\pi(1)$ at the top atom, this difference
vanishes throughout separation. Evaluating at the bottom
and using continuity when $\bar z>0$ gives
\eqref{eq:reduced-continuation-value}.

Suppose now that a trading atom precedes the continuation.
We compare the return from raising its posterior values
with the cost of restoring the mean in the continuation.
By \Cref{prop:mean-accounting}, partial regularization
of its face can raise the overall mean by a small
$\epsilon>0$ and the objective by $q(z)\epsilon$,
without changing profit. To restore the mean when the
continuation has probability $s$, apply an interior direction
from \eqref{eq:binary-remote-mean} with step $-\epsilon/s$.
This direction preserves $B$ and the endpoint virtual values,
so the trading face and its attachment to the continuation
are unchanged. It restores the overall mean and changes
the objective by $-\theta\epsilon+o(\epsilon)$.
Optimality therefore gives $q(z)\le\theta$, hence $\bar z\ge z>0$.

To rule out equality, suppose $z=\bar z$, so that
$\theta=q(z)$. We will transfer probability from the atom
to the continuation while preserving their combined mean.
Normalize the atom and continuation together to mass one,
keeping their combined probability fixed. Let $v$ be the atom's value
and $\rho$ the continuation's share. The revenue
contact gives $v=(1-\rho)z+\rho B$, with $B>v>z$.

Keeping $v,B$ fixed, increase the continuation's share
to $\rho+t$. Preserve the contact and mean by setting
\[
 z_t=\frac{v-(\rho+t)B}{1-\rho-t},
 \qquad
 d(t)=-\frac{t(M-v)}{\rho+t},
\]
and replace $\varphi$ by $\varphi+d(t)\eta$, using
an interior direction from
\eqref{eq:binary-remote-mean}.
The new mean is
$(1-\rho-t)v+(\rho+t)[M+d(t)]=(1-\rho)v+\rho M$.
The continuation retains its lowest posterior value $B$
and its endpoint virtual values. Since $z'_0=-(B-z)/(1-\rho)<0$, for small
$t>0$ the atom's virtual value remains positive and lies
strictly below the continuation's bottom virtual value.
The lower revenue contact is unchanged, so the modification
is feasible even if $v=a$ or an excluded group lies below
the atom.

To evaluate the gain from this transfer, write the resulting
conditional objective as
\[
 (\rho+t)J(\varphi+d(t)\eta)
 +(\rho+t)(B-v)q(z_t)
 +\beta(1-\rho-t)\pi(z_t).
\]
The second term is the additional rent generated
by the atom's contract. Using $DJ[\eta]=\theta=q(z)$,
the derivative at $t=0$ is
$J-\theta(M-B)-\rho(B-z)^2q'(z)
 -\beta\{(B-z)q(z)+\pi(z)\}$.
By \eqref{eq:reduced-continuation-value}, this equals
$(1-\rho)(B-z)^2q'(z)>0$, contradicting optimality.
\end{proof}

\subsubsection{Exclusion and a lower trading atom}
\label{app:binary-exclusive}

\begin{lemma}[Exclusion and distorted inclusion]
\label{thm:notrade-trading-exclusive}
An optimal posterior distribution for a binary prior cannot
contain both a positive-mass excluded group and a lower
trading atom.
\end{lemma}

\begin{proof}
Suppose both groups exist. The preceding structural results
place the trading atom between the excluded atom at $a$
and the regular continuation. Let $v>a$ and $z>0$ be the
trading atom's posterior value and virtual value, and let
$B>v$ be the continuation's lower support point.
Write $\sigma\in(0,1)$ for the probability of trade and
$\rho$ for the share of trading buyers in the continuation.
The revenue contacts give
$\rho=(v-z)/(B-z)\in(0,1)$ and $Z=(a-\sigma v)/(1-\sigma)\le0$,
where $Z$ is the excluded group's virtual value.
If $Z=0$, the two lower faces are adjacent and pure,
with positive intercepts, so
\Cref{lem:binary-adjacent-elimination} gives a strict
improvement. Hence $Z<0$.

Fix the normalized continuation, with mean $M$ and
weighted payoff $J$ when screened on its own.
The mean and weighted payoff conditional on trade are
$\widehat M=(1-\rho)v+\rho M$ and $\widehat J=\rho J+\rho(B-v)q(z)+\beta(1-\rho)\pi(z)$.
The prior mean determines the trading probability as
$\sigma=(m-a)/(\widehat M-a)$, making the full objective
$(m-a)\widehat J/(\widehat M-a)$.
Define
$\lambda=\widehat J/(\widehat M-a)>0$,
and hold this value of $\lambda$ fixed below.

We first show that $\lambda\ge q(z)$.
By \Cref{prop:mean-accounting}, partial regularization
of the trading face raises $\widehat M$ by a small
$t>0$ and $\widehat J$ by $q(z)t$.
Choose the filling to retain positive mass at $v$, keeping
the conditional lower support point unchanged. Restore the
prior mean by adjusting the trading probability to
$\sigma(t)=(m-a)/(\widehat M+t-a)$. The excluded virtual value is then
$(a-\sigma(t)v)/(1-\sigma(t))$, which remains negative
for small $t$ because $Z<0$.
The derivative of the full objective is
$\sigma[q(z)-\lambda]$. Optimality makes it nonpositive,
so $\delta:=\lambda-q(z)\ge0$.

Now vary $v,z$, keeping the continuation fixed and
adjusting $\sigma$ to preserve the mean.
The inequalities $a<v<B$ and $Z<0<z$ are strict,
as is $z<\varphi(0)$ by \Cref{lem:strict-attachment}.
Thus all sufficiently small changes in $(v,z)$ preserve
the posterior ordering, the revenue contacts, and the
ordering of virtual values, making both signs feasible.

To derive the local optimality conditions using the original
payoff ratio as a fixed parameter, put $K=J-\lambda(M-B)$ and define
\begin{equation*}
 \mathcal A(v,z)
 =\frac{v-z}{B-z}
       [K+(B-v)(q(z)-\lambda)]
   +\beta\frac{B-v}{B-z}\pi(z)=\widehat J-\lambda(\widehat M-v).
\end{equation*}
For each nearby $(v,z)$, the full objective satisfies
$\sigma\widehat J-\lambda(m-a)
 =\sigma[\mathcal A(v,z)-\lambda(v-a)]$.
The bracket vanishes at the original distribution. Because
$\sigma>0$, the full objective can have a local maximum
there only if the bracket does as well. Hence
$\mathcal A_v=\lambda$, $\mathcal A_z=0$, and
$D^2\mathcal A$ must be negative semidefinite.

Differentiating with $\lambda$ fixed gives
$\mathcal A_{vv}=2\delta/(B-z)$.
Since $\delta\ge0$, the necessary condition
$\mathcal A_{vv}\le0$ forces $\delta=0$.
At this point, using $\mathcal A_z=0$, the mixed
derivative is
$\mathcal A_{vz}=(B-v)q'(z)/(B-z)>0$.
With $\mathcal A_{vv}=0$, the Hessian's determinant is
therefore $-\mathcal A_{vz}^{\,2}<0$, contradicting negative
semidefiniteness. Thus the two groups cannot coexist.
\end{proof}

\subsection{Completing the proof of \texorpdfstring{\Cref{thm:binary-structure}}{Theorem~\ref{thm:binary-structure}}(i)}
\label{app:main-topology-proof}

\begin{proof}
Existence follows from Section~\ref{sec:model}, so fix an
optimal posterior distribution.
By \Cref{lem:active-face-purity}, every face with
positive slope and positive intercept contains a
single posterior atom.
\Cref{thm:no-reentry-rigorous,lem:binary-successor,lem:binary-adjacent-elimination}
show that at most one such face can occur and that
it lies below the regular upper component. These
arguments also cover countably many faces.
By \Cref{lem:binary-exclusion-floor}, any excluded
mass is concentrated at $a$, with a positive gap
between $a$ and the trading support.
\Cref{lem:binary-nonempty-continuation} establishes
a regular upper component of positive probability
with a nonempty separating interval.

By \Cref{prop:separating-conditions}, posterior
means, virtual values, and qualities are continuous
and strictly increasing on this separating interval.
It ends at a positive atom at one, served at efficient
quality $q(1)=\bar q$. The normalized quantile $Q$
is $C^1$ with $Q'>0$ on the separating interior.
If the upper component has probability $s>0$, the
posterior density there is
$s/Q'(Q^{-1}(y))$, which is positive and continuous.

By \Cref{thm:notrade-trading-exclusive}, exclusion and a
lower trading atom cannot coexist, leaving the three stated
forms for the lower part.
For a lower trading face $[\ell,h]$, the positive
intercept makes $U_r$ strictly decreasing along
the face. Its atom has value $U_r(h)<U_r(\ell)$, with $U_r(\ell)$
the bottom of the separating support, giving the required
support gap.
Finally, \Cref{lem:strict-attachment} places the
atom's virtual value strictly below the limiting
virtual value at the start of separation, giving
the strict quality gap.
\end{proof}

\section{Payoff-neutral learning and high prior means}\label{app:payoff-neutral-high-mean}

\subsection{Payoff-neutral learning}
\label{app:neutral-proofs}

\paragraph{Screening with positive initial marginal cost.}
\label{app:neutral-screening-foundation}
Allow $c'(0)=d\in[0,1)$, retaining the other conditions in
\Cref{ass:cost}. A posterior type has positive gains from
trade precisely when $y>d$, while the seller supplies
positive quality precisely when its ironed virtual value
exceeds $d$. The payoff formulas
\eqref{eq:seller-profit}--\eqref{eq:consumer-surplus}
remain valid, with $q(z)=\pi(z)=0$ for $z\le d$.
For a fixed posterior distribution, an appropriate shift
and rescaling reduces the screening problem to one with
zero initial marginal cost, treating unproductive types
as zero-valued buyers. This also gives weak
continuity of both payoffs from
\citet[Proposition~A.7]{kyriazis2026markets}.

\subsubsection{Reduction to a regular active distribution}
\label{app:neutral-regularization}

After concavification increases the mean, we restore its original value by adjusting the mass on the payoff-neutral type and the active distribution.

\begin{lemma}[A regular active market with a fixed mean]
\label{lem:neutral-regularization}
\label{lem:neutral-original-law}
Let $c^0$ satisfy the baseline technology conditions except that
$(c^0)'(0)=d\in[0,1)$, and let $0\le a\le d$ and $a<m<1$.
For $0\le\beta<1$, every feasible distribution is weakly improved for
$\CS+\beta\Pi$ by a feasible distribution
$(1-\lambda)\delta_a+\lambda G_A$, where $G_A$ is regular
and has positive trade almost surely. If its mean is $M_A$, then
$M_A\ge m$ and $\lambda=(m-a)/(M_A-a)$. Its payoffs are
$\lambda\bigl(\CS(G_A),\Pi(G_A)\bigr)$. Every original optimum
already has this form.
\end{lemma}
\begin{proof}
Normalize the lowest value to zero:
\[
 X=\frac{Q-a}{1-a},\qquad
 c^e(x)=\frac{c^0(x)-ax}{1-a},\qquad
 \mu=\frac{m-a}{1-a}.
\]
Since $c^0(x)\ge dx\ge ax$, the transformed cost is nonnegative,
strictly convex, and has initial marginal cost
$\delta=(d-a)/(1-a)\ge0$. The transfer change
$t=ax+(1-a)t^e$ multiplies utility and profit by $1-a$
and preserves incentives and the seller's choice. The transformed
quality and profit functions satisfy $q^e=\pi^e=0$
on $(-\infty,\delta]$.

\paragraph{Regularization with the same mean.}
Let $R$ be the concavified revenue of $X$, and set
$X^r=R/(1-u)$ and $\psi=-R'_+$. Concavity of $R$ and $R(1)=0$ imply that $X^r$
is nondecreasing. Since the raw revenue curve
lies below $R$ and $R\le1-u$, we also have
$X\le X^r\le1$.
The tail integral of $X^r$ decreases continuously
from a value at least $\mu$ to zero. We can
therefore choose $\ell$ such that
$\int_\ell^1X^r(u)\,\dd u=\mu$.
Set $X^*=0$ below $\ell$ and $X^*=X^r$ above it.
This quantile has mean $\mu$.

For $\ell>0$, its concavified revenue curve is
the chord from $(0,0)$ to $(\ell,R(\ell))$,
followed by $R$. Concavity of $R$ and $R(0)\ge0$
imply that the chord's slope is at least
$R'_+(\ell)$, so the join is concave.
Every concave majorant must lie above this chord
and above $R$ on $[\ell,1]$, where the new raw
curve equals $R$. If $\ell=0$, the concavified
revenue curve is simply $R$.

Let $h=\max\{\ell,\sup\{u:\psi(u)\le\delta\}\}$,
taking the supremum of the empty set to be zero.
The new menu supplies zero quality below $h$
and keeps the old virtual values above it. Define
\[
 T(v)=\int_v^1(X^*-X)\,\dd u\ge0,
 \qquad T(0)=T(1)=0.
\]
The sign follows from $X^r\ge X$ for $v\ge\ell$, and from
equal means and $X\ge0$ otherwise. For the old quality schedule
$w=q^e(\psi)$, integration by parts gives
\[
 E:=\int_0^1(X^*-X)w\,\dd u=\int_{[0,1]}T\,\dd w\ge0.
\]
Here $w$ is bounded and nondecreasing, and continuity of $T$
makes the integration by parts valid even when $w$ has jumps.
The payoff formulas give
\begin{equation}
 \Delta\CS^e=E+B_0,\qquad \Delta\Pi^e=-D_0,
 \quad
 B_0=\int_0^h\psi q^e(\psi)\,\dd u,
 \quad D_0=\int_0^h\pi^e(\psi)\,\dd u.
 \label{eq:neutral-equality-accounting}
\end{equation}
Since $\pi^e(z)=zq^e(z)-c^e(q^e(z))$ and $c^e\ge0$,
$0\le D_0\le B_0$, so the objective rises by at least
$E+(1-\beta)B_0\ge0$.

\paragraph{Concentrating the residual.}
In original units, a virtual-value profile equal to some $v\in(d,1)$
on a short initial interval and one thereafter has positive consumer
surplus and mean above $m$. Mixing it with $a$ gives a feasible distribution with a positive
objective, proving the claim if the constructed distribution has zero
objective.

Suppose, then, that the constructed distribution has positive
objective. Normalize its positive-trade tail, with mean $M_e$,
payoffs $(C_e,P_e)$, and share $s=1-h$. The excluded values
lie between zero and the active values, so $sM_e\le\mu\le M_e$.
We can therefore replace them by zero and preserve the mean
by raising the active share to $\lambda=\mu/M_e\ge s$. If $\lambda<1$, the atom at zero has virtual value
$-\lambda B_e/(1-\lambda)\le0$, where $B_e$ is the lowest
posterior value in the active distribution, and remains
excluded. Changing the active share from $s$ to $\lambda$
scales revenue and the length of the active rank
interval by the same factor $\lambda/s$.
The revenue curve therefore has the same slopes,
so virtual values are unchanged. Payoffs per
trading buyer remain $(C_e,P_e)$, and total
payoffs become $\lambda(C_e,P_e)$. Since
$\lambda\ge s$, the designer's objective is
at least as high as before.
If $\lambda=1$, only the active distribution remains.
Returning to the original units gives
$M_A=a+(1-a)M_e$ and $\lambda=(m-a)/(M_A-a)$.

\paragraph{Equality identifies every optimum.}
At an optimum, the objective is positive and neither
of the two changes can raise it. Since $\beta<1$,
\eqref{eq:neutral-equality-accounting} forces $E=B_0=D_0=0$.
Positive quality below $h$ would give $\psi q^e(\psi)>0$,
so the original quality was already zero there. Above $h$ it is
positive and unchanged. Hence,
\[
 0=E=\int_h^1(X^r-X)w\,\dd u
 \quad\Longrightarrow\quad X=X^r\quad\text{a.e. on }(h,1).
\]
The original active distribution was therefore regular.
Since its conditional objective is strictly positive,
$\lambda>s$ would strictly improve the objective. Optimality
thus requires $\lambda=s$, giving $sM_e=\mu$. The excluded
buyers therefore contribute zero to the overall mean. Their
posterior values are nonnegative, so they equal zero almost
surely, or $a$ in the original units.
\end{proof}

\subsubsection{Selecting the active distribution}
\label{app:neutral-active-analysis}
\label{app:neutral-mean-selection}

Maintain \Cref{ass:cost}, including $c'(0)=0$. We use
\Cref{ass:neutral-curvature} only for uniqueness. Let $\Phi$ be
the set of nondecreasing, right-continuous $[0,1]$-valued profiles,
identified almost everywhere. A profile determines a regular distribution by
\[
 Q_\varphi(u)=\frac1{1-u}\int_u^1\varphi(v)\,\dd v,
 \qquad M(\varphi)=\int_0^1Q_\varphi(u)\,\dd u.
\]
The first equation reconstructs the quantile function, and hence
the distribution, from a virtual-value profile $\varphi$.
We call this \emph{reconstruction}.
Write $J_\beta(\varphi)$ for its weighted payoff from
\eqref{eq:seller-profit}--\eqref{eq:consumer-surplus}, with
$0\le\beta<1$. Helly's theorem implies that $\Phi$ is sequentially
compact in $L^1$. The formula for $Q_\varphi$ and
dominated convergence give continuity of $M$
and $J_\beta$.

For a given price $\theta\in\R$ per unit of mean, let
\begin{equation}
 F_\beta(\theta)=\max_{\varphi\in\Phi}
       L_{\beta,\theta}(\varphi),\qquad
 L_{\beta,\theta}(\varphi)=J_\beta(\varphi)-\theta M(\varphi).
 \label{eq:neutral-weighted-price}
\end{equation}
The constant profiles $\varphi=0$ and $\varphi=1$
correspond to point masses at zero and one.
Any maximizer of $L_{\beta,\theta}$ also maximizes
$J_\beta$ among profiles with the same mean.
The next lemma gives a lower bound on its lowest
quality. It also shows that $J_\beta/M$ attains
its maximum and that the means of profiles
attaining it are bounded away from zero.

\begin{lemma}[Mean prices and the critical ratio]
\label{lem:neutral-price-basics}
Under \Cref{ass:cost}, fix $0\le\beta<1$ and set
\[
 \rho_\beta=\sup_{M(\varphi)>0}
                         \frac{J_\beta(\varphi)}{M(\varphi)}.
\]
Then $\beta\pi(1)<\rho_\beta\le\pi(1)$. The function $F_\beta$
is finite, convex, nonincreasing, and $1$-Lipschitz. It is positive
exactly when $\theta<\rho_\beta$ and equals zero otherwise.
Every nonzero maximizing profile at a price $\theta$ satisfies
\begin{equation}
 \theta\le
 \frac{[Q_0-\varphi_0]q(\varphi_0)+\beta\pi(\varphi_0)}{Q_0}
 \le q(\varphi_0),
 \qquad Q_0=Q_\varphi(0),\quad\varphi_0=\varphi(0).
 \label{eq:neutral-lowest-quality-price}
\end{equation}
The ratio $\rho_\beta$ is attained. Every profile attaining it
has bottom virtual value and mean at least $c'(\rho_\beta)>0$.
The unit profile is not optimal at any finite price.
\end{lemma}
\begin{proof}
A maximizer exists by compactness and continuity. Interim
participation, first-best surplus, and the convex chord bound give
\[
 0\le J_\beta(\varphi)\le\CS+\Pi
 \le\int_0^1\pi(Q_\varphi(u))\,\dd u
 \le\pi(1)M(\varphi).
\]
Thus, $\rho_\beta\le\pi(1)$. The value function $F_\beta$ is the pointwise maximum of
affine functions of $\theta$ with slopes in $[-1,0]$.
It is therefore convex, nonincreasing, and $1$-Lipschitz.

To rule out the unit profile, lower it by $\varepsilon$
on $[0,\ell)$ only.
Writing $\Lambda(\ell)=\int_0^\ell[-\log(1-u)]\,\dd u=O(\ell^2)$,
the derivative of $L_{\beta,\theta}$ at zero is
$(1-\beta)\bar q\,\ell+(\theta-\bar q)\Lambda(\ell)>0$
for small $\ell$. This change is feasible and monotone, so the unit profile
is not optimal at any finite price. At $\theta=\beta\pi(1)$,
the same change improves on value zero, proving
$\rho_\beta>\beta\pi(1)$. The sign of
$F_\beta$ follows from the ratio definition and the zero profile.

To bound the lowest quality of a nonzero maximizer, delete its
bottom $\varepsilon$ ranks and renormalize. The resulting regular quantile is
$v\mapsto Q_\varphi(\varepsilon+(1-\varepsilon)v)$. Optimality gives
$\int_0^\varepsilon
 \{(Q_\varphi-\varphi)q(\varphi)+\beta\pi(\varphi)
                  -\theta Q_\varphi\}\,\dd u
          \ge\varepsilon F_\beta(\theta)$.
Divide by $\varepsilon$ and pass to right limits. Using
$Q_0=\int_0^1\varphi>0$, $F_\beta\ge0$, and
$\pi(z)\le zq(z)$ gives \eqref{eq:neutral-lowest-quality-price}.

Finally, take positive $\theta_n\uparrow\rho_\beta$ and maximizers
$\varphi_n$. Their values are positive, and the bottom-quality bound gives
$M(\varphi_n)\ge\varphi_n(0)\ge c'(\theta_n)$. These marginal
costs are well-defined since $\rho_\beta\le\pi(1)<\bar q$.
Compactness and continuity yield a limit with mean at least
$c'(\rho_\beta)>0$ and $J_\beta-\rho_\beta M=0$, attaining the
ratio. Every ratio maximizer solves the problem at
$\rho_\beta$, so the same bound applies to its bottom virtual
value and mean.
\end{proof}

\subsubsection{Uniqueness at a price on mean}
\label{app:neutral-price}

Put $\alpha=1-\beta>0$. For feasible bounded changes $h$
supported where $\varphi$ is bounded away from zero,
\Cref{ass:cost} alone gives first variation $\int H_\beta h$
for $J_\beta$, where on those ranks
\begin{equation}
 H_\beta(u)=\int_0^u\frac{q(\varphi(v))}{1-v}\,\dd v
       +(Q_\varphi(u)-\varphi(u))q'(\varphi(u))
       -\alpha q(\varphi(u)).
 \label{eq:binary-cited-first-variation}
\end{equation}
The first-variation formula is the private-value specialization
of \citet[Lemma~A.10]{kyriazis2026markets}. The corresponding
quantile response follows from
\citet[Lemma~A.9]{kyriazis2026markets}. Here reconstruction
is linear, so
$\Delta Q(u)=(1-u)^{-1}\int_u^1h(v)\,\dd v$ and
$\|\Delta Q\|_\infty\le\|h\|_\infty$.
The mean is also linear in $\varphi$. Thus the first
variation of $L_{\beta,\theta}$ is $\int_0^1 H(u)h(u)\,\dd u$,
where $H(u)=H_\beta(u)-\theta[-\log(1-u)]$.

In this private-value setting, the first-order expansion has remainder
$o(\|h\|_1)$ as $\|h\|_\infty\to0$.
To see this, keep virtual values on the perturbed ranks
uniformly bounded away from zero. Uniform continuity of
$q'$ and $\pi'=q$ then bounds the integrated Taylor
remainders by $o(1)\|h\|_1$. Products of quantile and
quality changes contribute at most
$O(\|h\|_\infty\|h\|_1)$. Outside those ranks, quality is unchanged, so the payoff change
is linear in $\Delta Q$. The mean term is linear everywhere.
For the uniqueness analysis that follows, we additionally
impose \Cref{ass:neutral-curvature}.

\begin{lemma}[Uniqueness at a price on mean]
\label{lem:neutral-price-unique}
For every finite $\theta$, \eqref{eq:neutral-weighted-price} has at
most one nonzero maximizing profile. If one exists, it is positive
and strictly increasing on $[0,b)$ for some $b\in(0,1)$ and equals
one on $[b,1]$. Its posterior quantile and quality are continuous,
strictly increasing on the separating interval, and continuously
differentiable in its interior.

Set $\tau=-\log(1-u)$, $T=-\log(1-b)$,
$x=q(\varphi)$, and
$A(\tau)=\int_0^\tau[x(v)-\theta]\,\dd v$.
On the separating interval,
\begin{align}\label{eq:neutral-static-control}
 Q=c'(x)+c''(x)(\alpha x-A),\qquad
 \dot A=x-\theta,\qquad \dot Q=Q-c'(x), \end{align}
with $A(0)=0$, $0\le A\le\alpha x$, and
\begin{equation}
 (A(T),Q(T),x(T))=(\alpha\bar q,1,\bar q),
 \qquad x(\tau)>\theta,\quad\tau>0.
 \label{eq:neutral-boundaries}
\end{equation}
\end{lemma}
\begin{proof}
Fix a nonzero maximizer and write $Q=Q_\varphi$.

\paragraph{Positive trade and the top atom.}
By \Cref{lem:neutral-price-basics}, $\theta\le q(\varphi(0))$
and the unit profile is not optimal. If $\theta>0$, this bound makes bottom quality positive.
For $\theta\le0$, suppose instead that $\varphi=0$
on an initial interval of length $d>0$. Raise it to
$\varphi_\varepsilon=\max\{\varphi,\varepsilon\}$, and put
$d_\varepsilon=\inf\{u:\varphi(u)\ge\varepsilon\}$ and
$K_\varepsilon=\int_{d_\varepsilon}^1\varphi$.
Then $d_\varepsilon\downarrow d<1$ and $K_\varepsilon\to K_0>0$.
On the old zero ranks, $Q_\varepsilon(u)\ge K_\varepsilon/(1-u)$.
On $[d,d_\varepsilon)$, the old integrand in the objective is at most
$(1+\beta)q(\varepsilon)$. Since the mean increase is weakly beneficial and nothing
changes above $d_\varepsilon$, we obtain
\[
 L_{\beta,\theta}(\varphi_\varepsilon)-L_{\beta,\theta}(\varphi)
 \ge q(\varepsilon)\{K_\varepsilon[-\log(1-d)]-d\varepsilon
                  -(1+\beta)(d_\varepsilon-d)\}>0
\]
for small $\varepsilon$. Thus $\varphi(u)>0$ for every $u>0$.

We adapt the top-atom argument of
\citet[Lemma~A.11]{kyriazis2026markets} to the objective
$L_{\beta,\theta}$.
Let $v=\varphi(1-)>0$. If $q(v)\le\theta$, the bottom-quality
bound forces a constant profile with $\theta=q(v)$ and value
$\beta\pi(v)-vq(v)<0$, below that of the zero profile.
Hence $q(v)>\theta$, so $H\to+\infty$ near one: its integral
term diverges while the remaining terms stay bounded.
On a fixed high tail with $H\ge1$, raise $\varphi$ to
$\min\{\varphi+\varepsilon,1\}$. This is monotone and, unless
the tail already equals one, has first-order gain at least the
change's positive $L^1$ norm. The remainder is of smaller order, so the gain is positive,
a contradiction. Thus $\varphi=1$ on some $[b,1]$, with
$0<b<1$ because the unit profile is not optimal. Reconstruction
gives continuous $Q>\varphi$ before $b$ and
$(1-u)Q'=Q-\varphi>0$ a.e., so $Q$ is strictly increasing there
\citep[Lemma~A.12]{kyriazis2026markets}.

\paragraph{Separation and the Euler equation.}
On a positive flat interval $\varphi=\gamma<1$,
\[
 (1-u)H'=q(\gamma)-\theta+(Q-\gamma)q'(\gamma)>0,
\]
since $q(\gamma)\ge\theta$ and $Q>\gamma$.
The following perturbation adapts the splitting and clipping
argument of \citet[Lemma~A.13]{kyriazis2026markets}.
On a maximal such interval, lower the first half by $\varepsilon$
and raise the second by $\varepsilon$. Because the halves have
equal length and $H$ is strictly increasing, the first-order
gain is of order $\varepsilon$ with a positive coefficient. Preserve monotonicity
by clipping adjacent lower values at $\gamma-\varepsilon$ and
upper values at $\gamma+\varepsilon$. Since the flat interval is maximal, the set of
ranks outside it that require clipping has measure
tending to zero as $\varepsilon\downarrow0$.
Each value changes by at most $\varepsilon$, so
the $L^1$ norm of these changes is $o(\varepsilon)$.
At these ranks, virtual values remain near
$\gamma>0$ and ranks stay below $b<1$, so $H$
is uniformly bounded. Clipping therefore contributes
only $o(\varepsilon)$ to the first variation.
The Taylor remainder is also $o(\varepsilon)$,
so the objective strictly increases. Thus, $\varphi$ is strictly increasing below $b$.

To obtain the Euler equation, we use the monotone variations
and interval approximation in
\citet[Lemma~A.15]{kyriazis2026markets}.
For smooth $h$ with compact support in $(0,1)$,
$z\mapsto z+\varepsilon h(z)$ preserves order and bounds for
sufficiently small $\varepsilon$ of either sign. Optimality therefore gives $\int Hh(\varphi)=0$.
For $0<s<t<b$ at which $\varphi$ is continuous, approximate
$1_{(\varphi(s),\varphi(t))}$ by bounded smooth functions with
supports in a fixed compact subinterval of $(0,1)$.
Since $\varphi$ is strictly increasing, their compositions
converge a.e. to $1_{(s,t)}$, so dominated convergence gives $\int_s^t H=0$.
These continuity points are dense, so $H=0$ a.e. on $(0,b)$.

\paragraph{Regularity and endpoint fit.}
We follow the curvature, continuity, and endpoint arguments of
\citet[Lemma~A.16 and Proposition~A.17]{kyriazis2026markets},
with the state variable $A$ now incorporating the mean term.
In quality coordinates, $H=0$ is \eqref{eq:neutral-static-control}.
Since $x$ is strictly increasing, the bottom-quality bound
gives $x(\tau)>\theta$ for $\tau>0$, and hence $A\ge0$.
The Euler equation gives $A\le\alpha x$.
The derivative of the static equation's right-hand side with
respect to quality is
$D_\alpha(x,A)=(1+\alpha)c''(x)+(\alpha x-A)c'''(x)>0$, $0\le A\le\alpha x$.
For $c'''\ge0$ this is immediate. Otherwise
\Cref{ass:neutral-curvature} gives
$D_\alpha>(1-\alpha)c''=\beta c''\ge0$, strictly even at $\beta=0$.
Continuity of $A$ and $Q$ means that both quality
limits at an interior jump would solve the same
static equation. Its right-hand side is strictly
increasing between these limits, since every
intermediate quality satisfies $A\le\alpha x$.
The two limits therefore cannot both be solutions.
Quality is continuous, and the Euler equation holds
throughout the separating interval.
By the implicit-function theorem, $x$ is locally
a $C^1$ function of $(A,Q)$. The state equations
then imply that $A$ and $Q$ are $C^1$, and hence
so is $x$.
Differentiation yields
\[
 \dot x=\frac{c''(x)\{(1+\alpha)x-A-\theta\}}
                    {D_\alpha(x,A)}>0,
\]
using $A\le\alpha x$ and $x>\theta$. Also $\dot Q=Q-c'(x)>0$.

To check the fit with the top atom, suppose $p=\varphi(b-)<1$
and lower the initial top block $[b,b+\delta)$
to $1-\varepsilon>p$. This is feasible, and its one-sided
variation gives $\int_b^{b+\delta}H\ge0$.
On the top atom $H=A-\alpha\bar q$, so $\delta\downarrow0$
gives $A(T)\ge\alpha\bar q$. The separating Euler limit instead
gives $A(T)=\alpha q(p)-(1-p)q'(p)<\alpha\bar q$.
Hence $p=1$ and \eqref{eq:neutral-boundaries} holds.

It remains to show that bottom quality is positive. The curvature
bound implies $q(z)/q'(z)\to0$ as $z\downarrow0$, even if $c''$
is unbounded.
For any $\zeta>2$, integrating
$(\log c'')'>-\zeta/x$ on $[x,s]\subset[x,2x]$ gives
$c''(s)\ge2^{-\zeta}c''(x)$, and hence
\[
 0\le2^{-\zeta}xc''(x)\le c'(2x)-c'(x)\longrightarrow0.
\]
Now $q(z)/q'(z)=xc''(x)$ at $z=c'(x)$.
If $\varphi(0)=0$, Euler and $A\ge0$ would give
$Q-\varphi\le\alpha q(\varphi)/q'(\varphi)\to0$, contradicting
$Q(0)=\int_0^1\varphi>0$. Bottom quality is therefore positive, and the control extends
continuously to both endpoints.

\paragraph{Uniqueness.}
We use the backward-trajectory argument of
\citet[Proposition~A.18]{kyriazis2026markets} for the state
system obtained here.
The $C^1$ implicit control makes the autonomous state system
locally Lipschitz. Because $\bar q<\bar Q$, this system extends
through the top state and has a unique backward trajectory from
$(\alpha\bar q,1)$.
Two maximizing paths of lengths $T_1<T_2$ would satisfy
$A_2(T_2-T_1)=A_1(0)=0$, contradicting
$\dot A=x-\theta>0$ at positive ranks. The durations, and hence the profiles, must coincide. Thus the
necessary conditions identify the nonzero maximizer whenever one
exists. Compactness already ensures existence of a maximizer.
\end{proof}
\Needspace{5\baselineskip}
\begin{lemma}[The mean selected by a price]
\label{lem:neutral-price-means}
Let $\rho_\beta$ be as in \Cref{lem:neutral-price-basics}.
For $\theta<\rho_\beta$ the maximizer
$\varphi_{\beta,\theta}$ is unique and nonzero. Its mean is
continuous and strictly decreasing in $\theta$, with limits
$\lim_{\theta\to-\infty}M(\varphi_{\beta,\theta})=1$ and $\lim_{\theta\uparrow\rho_\beta}M(\varphi_{\beta,\theta})=m^c_\beta\in(0,1)$.
At $\theta=\rho_\beta$, there is exactly one nonzero maximizer
$\varphi^c_\beta$, of mean $m^c_\beta$, in addition to the zero
profile.
\end{lemma}
\begin{proof}
Lemmas~\ref{lem:neutral-price-basics}--\ref{lem:neutral-price-unique}
give the asserted maximizers. By compactness and continuity, the unique maximizing profile
depends continuously on prices below $\rho_\beta$.
For $\theta_1<\theta_2<\rho_\beta$, write
$M_i=M(\varphi_{\beta,\theta_i})$. Adding the two optimality
inequalities gives $(\theta_2-\theta_1)(M_1-M_2)\ge0$.
If the means were equal, uniqueness would make the same profile
optimal at both prices. Its Euler equations would then differ by
$(\theta_2-\theta_1)[-\log(1-u)]$ on a nonempty separating
interval, a contradiction. Thus, the mean strictly decreases.
Optimality of $\varphi_{\beta,\theta}$ relative to the unit profile gives
$(-\theta)\{1-M(\varphi_{\beta,\theta})\}
 \le J_\beta(\varphi_{\beta,\theta})-J_\beta(1)$, $\theta<0$.
Bounded payoffs give the limit one as $\theta\to-\infty$.
Since the mean strictly decreases with $\theta$, its limit as
$\theta\uparrow\rho_\beta$ is at most $M(\varphi_{\beta,\theta_0})<1$
for any fixed $\theta_0<\rho_\beta$.
Equation~\eqref{eq:neutral-lowest-quality-price} bounds this limit below by
$c'(\rho_\beta)>0$. Every subsequential profile limit is therefore a nonzero maximizer
at $\rho_\beta$. Uniqueness identifies this limit as
$\varphi^c_\beta$ and its mean as $m^c_\beta$.
\end{proof}

\subsubsection{Proof of \texorpdfstring{\Cref{prop:neutral-participation}}{Proposition~\ref{prop:neutral-participation}}}
\label{app:neutral-completion}

\begin{proof}
Fix $\beta<1$ and put
\[
 \widehat c(x)=\frac{c(x)-dx}{1-d},\qquad
 \kappa=\frac{d-a}{1-d},\qquad y=\frac{m-a}{1-d}.
\]
The normalized cost satisfies the same curvature condition.
A profile $\varphi\in\Phi$ induces original values
$d+(1-d)Q_\varphi$, mean $M_A=d+(1-d)M(\varphi)$, and
payoff $(1-d)\widehat J_\beta(\varphi)$.

By \Cref{lem:neutral-original-law}, it remains to choose the active distribution:
\begin{equation}
 (m-a)\max_{\kappa+M(\varphi)\ge y}
       \frac{\widehat J_\beta(\varphi)}{\kappa+M(\varphi)}.
 \label{eq:neutral-weighted-ratio}
\end{equation}
We can allow profiles with excluded buyers without
changing the value. For any profile satisfying
the constraint, the same rescaling gives a feasible
mixture with $\delta_a$. Conversely, the lemma
shows that every optimum of the original problem
can be written in this form. If $\kappa=0$, the
zero profile is infeasible. If $\kappa>0$, it
yields zero payoff and cannot attain the positive
optimal value.

\paragraph{The critical ratio.}
Let $\widehat F_\beta$ and $\rho_\beta$ denote the price value
and critical price for $\widehat c$.
Write $B_\kappa(\theta)=\widehat F_\beta(\theta)-\kappa\theta$.
If $\kappa>0$, this function is continuous and strictly decreasing,
positive at zero and negative at $\rho_\beta$. It therefore has a
unique zero, which we denote by $\rho$. For $\kappa=0$, set
$\rho=\rho_\beta$. Then
$\widehat J_\beta\le\rho(\kappa+M)$, with equality at exactly
one nonzero profile $\varphi^c$.
Put $m_a^c=d+(1-d)M(\varphi^c)\in(d,1)$.
For $m\le m_a^c$, the critical profile is feasible and uniquely
maximizes the ratio, with active share $(m-a)/(m_a^c-a)$.

\paragraph{Above the critical mean.}
For $m>m_a^c$, \Cref{lem:neutral-price-means} selects a
unique price $\theta_m<\rho$ and profile of mean
$y-\kappa=(m-d)/(1-d)$. Its ratio $r$ satisfies
$r-\theta_m=B_\kappa(\theta_m)/y>0$.
To compare this profile with any alternative satisfying
$\kappa+M\ge y$, use optimality at the selected price:
$r(\kappa+M)-\widehat J_\beta(\varphi)
 \ge(r-\theta_m)(\kappa+M-y)\ge0$.
Equality requires the unique maximizer of the priced problem,
with $M=y-\kappa$, so the active share must equal one. \Cref{lem:neutral-price-unique}
gives the active structure, and equality in
\Cref{lem:neutral-original-law} identifies every original
optimum. The distribution induced by $\varphi^c$ is
$G_{a,\beta}^c$, with $m_a^c=m_{a,\beta}^c$.
The following lemma establishes the claimed value-function properties.
\end{proof}

\begin{lemma}[The value of a payoff-neutral prior]
\label{prop:neutral-weighted-value}
Under the hypotheses of \Cref{prop:neutral-participation}, let
$\rho_{a,\beta}:=J_\beta(G_{a,\beta}^c)/(m_{a,\beta}^c-a)>0$.
The value $V_{a,\beta}$ is continuously differentiable on $(a,1)$,
equals $\rho_{a,\beta}(m-a)$ for $m\le m_{a,\beta}^c$, and is
strictly concave for $m>m_{a,\beta}^c$.
\end{lemma}
\begin{proof}
Use the preceding notation, with $\rho=\rho_{a,\beta}$.
For $\theta\le\rho$, we have $B_\kappa(\theta)\ge0$ and
$\widehat J_\beta\le\theta(\kappa+M)+B_\kappa(\theta)$.
Using \eqref{eq:neutral-weighted-ratio} and $y/(\kappa+M)\le1$ gives
$V_{a,\beta}(m)\le(1-d)\{\theta y+B_\kappa(\theta)\}
=(1-d)\widehat F_\beta(\theta)+\theta(m-d)$.
The critical price attains this bound below $m_a^c$, and
$\theta_m$ attains it above. The value is therefore the lower
envelope of these affine bounds:
\[
 V_{a,\beta}(m)=\inf_{\theta\le\rho}
 \{(1-d)\widehat F_\beta(\theta)+\theta(m-d)\}.
\]
For $m>m_a^c$ and sufficiently small $\varepsilon>0$,
the two supporting lines give
\[
 \theta_{m+\varepsilon}
 \le
 \frac{V_{a,\beta}(m+\varepsilon)-V_{a,\beta}(m)}
      {\varepsilon}
 \le \theta_m.
\]
Continuity of $\theta_m$ in $m$ then implies
$V_{a,\beta}'(m)=\theta_m$. Since $\theta_m$
strictly decreases with $m$, the value is strictly
concave above $m_a^c$. As $m\downarrow m_a^c$,
the derivative tends to $\rho$, which is also
the slope of the linear part $\rho(m-a)$
below $m_a^c$.
\end{proof}

\subsubsection{Proof of \texorpdfstring{\Cref{thm:neutral-frontier}}{Proposition~\ref{thm:neutral-frontier}}: the single-market frontier}
\label{app:neutral-frontier-proof}

\begin{proof}
Fix $a\le d=c'(0)$ and $a<m<1$ under the hypotheses of
\Cref{prop:neutral-participation}.

The first-best bound and convexity of $\pi$, with $\pi(a)=0$, give
\begin{equation}
 \CS(G)+\Pi(G)\le\E_G[\pi(Y)]
        \le\frac{m-a}{1-a}\pi(1)=P_1.
 \label{eq:seg-neutral-surplus-bound}
\end{equation}
The chord joining $(a,0)$ to $(1,\pi(1))$ lies
strictly above $\pi$ on $(a,1)$, because $\pi$
is zero on $[a,d]$ and strictly convex above $d$.
Equality therefore requires $G=H$.
Under $H$, type $a$ receives zero quality and
type one receives quality $q(1)$ at price $q(1)$.
The resulting payoffs are $(0,P_1)$, so $H$
uniquely maximizes both total surplus and profit.

The maximizing distribution is unique for every
$\beta\in[0,1]$, by \Cref{prop:neutral-participation}
and the argument above for $\beta=1$.
If $\beta_n\to\beta$, compactness of the feasible
set and continuity of the objective imply that
every subsequential limit of $G_{m,\beta_n}^*$
is optimal at $\beta$. Uniqueness therefore gives
weak convergence to $G_{m,\beta}^*$.
Continuity of payoffs then gives a continuous
path $(C_\beta,P_\beta)$. By
\eqref{eq:weight-payoff-monotonicity}, consumer
surplus is nonincreasing and profit is nondecreasing
along this path.

Every feasible pair satisfies
$C+\beta P\le C_\beta+\beta P_\beta$
for $0\le\beta\le1$.
For $p\in(P_0,P_1]$, continuity gives a weight
$\beta>0$ with $P_\beta=p$.
Any feasible pair with $P\ge p$ therefore satisfies
\[
 C\le C_\beta+\beta(p-P)\le C_\beta.
\]
The pair $(C_\beta,p)$ attains this bound.
Equality requires $P=p$ and optimality at weight
$\beta$. For $p\le P_0$, the buyer-optimal pair
attains the maximal consumer surplus $C_0$.
More generally, the path over $[0,\bar\beta]$
covers the profit interval $[P_0,P_{\bar\beta}]$.

Since no feasible profit exceeds $P_1$, every
feasible pair is weakly dominated by a pair
on the path. For $\beta>0$, a Pareto improvement
would raise $C+\beta P$, contradicting optimality.
At $\beta=0$, uniqueness rules out another
distribution with the same consumer surplus
and higher profit. The path is therefore the
entire frontier.
Any distribution yielding $(C_\beta,P_\beta)$
also maximizes the corresponding weighted objective
and must therefore equal $G_{m,\beta}^*$.
Its structure is given by
\Cref{prop:neutral-participation} for $\beta<1$,
while $G_{m,1}^*=H$.
\end{proof}

\subsection{High prior means and the efficient frontier}
\label{app:high-mean-proof}

\subsubsection{Proof of \texorpdfstring{\Cref{prop:high-mean}}{Proposition~\ref{prop:high-mean}}}
\label{app:high-mean-characterization-proof}

\begin{proof}
For $\beta\le\bar\beta$, the curvature expression equals
its value at $\bar\beta$ plus
$(\bar\beta-\beta)[c''(x)]^2$. Thus,
\Cref{ass:curvature} holds throughout $[0,\bar\beta]$.
We first solve an auxiliary problem in which
the mean is free to vary. Its solutions give
a threshold above which the result holds for
every $\beta\in[0,\bar\beta]$.

\paragraph{The auxiliary problem.}
Fix $\beta\in[0,\bar\beta]$. For $\theta<0$, consider
\begin{equation}
 \max_{G\in\Delta([0,1])}
 \{J_\beta(G)-\theta\E_G[Y]\}.
 \label{eq:high-mean-price}
\end{equation}
The term $-\theta\E_G[Y]$ favors distributions
with higher means. A maximizer exists by weak
compactness and payoff continuity.
Regularization preserves profit and changes
the objective by
\[
 \int_0^1[Q^r(u)-Q(u)]
             [q(\phi_Q(u))-\theta]\,\dd u.
\]
Here $Q^r\ge Q$ and $q(\phi_Q)-\theta>0$,
so the gain is strictly positive unless
$Q=Q^r$ almost everywhere. Every maximizer
is therefore regular.

A maximizer cannot have mean zero, since
$\delta_0$ gives zero while $\delta_1$ gives
$\beta\pi(1)-\theta>0$. Nor can it have mean
one, since the perturbation of the unit profile
in \Cref{lem:neutral-price-basics} strictly
improves on $\delta_1$ for every finite $\theta$.
That perturbation requires only \Cref{ass:cost}
and $\beta<1$.

At its own mean, a maximizer of
\eqref{eq:high-mean-price} also maximizes
$J_\beta$, so \Cref{thm:binary-structure}(i)
applies with $a=0$. Regularity rules out the
excluded atom and the lower trading atom
allowed by that theorem. The distribution
therefore has a nondegenerate separating
interval followed by an efficiently served
top atom. Its virtual-value profile $\varphi$
is positive at interior ranks, and its lower
support point is
$B=\int_0^1\varphi(u)\,\dd u>0$.

\paragraph{Uniqueness for each multiplier.}
The interior variations constructed in
\Cref{lem:binary-two-multipliers} need no
compensating adjustments here, because neither
the mean nor the lower support point is fixed.
Writing $w(u)=-\log(1-u)$, the first variation is
$\int[H_\beta(u)-\theta w(u)]h(u)\,\dd u$.
The interval approximation in that lemma gives
the Euler equation with the given $\theta$
and $\xi=0$.

By \Cref{prop:separating-conditions},
$A(0)=0$, $A(T)=(1-\beta)\bar q$, and
$\dot A=q(\varphi)-\theta>0$.
For fixed $\theta$, the state equations have
a unique solution traced backward from
$((1-\beta)\bar q,1)$, and the Euler equation
determines $\varphi$ along it.
Two maximizing solutions cannot have different
durations. If $T_1<T_2$, backward uniqueness
would imply $A_2(T_2-T_1)=A_1(0)=0$.
This is impossible because $A_2$ starts at zero
and is strictly increasing.
The solution and its duration determine the
full profile, including the top atom's
probability $e^{-T}$.
Thus \eqref{eq:high-mean-price} has a unique
maximizer $G_\theta$ with mean $m_\theta$.
We omit $\beta$ from the notation for now.

\paragraph{Mean and lower support as the multiplier varies.}
The feasible set in \eqref{eq:high-mean-price}
is fixed and weakly compact, and the objective
is jointly continuous in the distribution and
the parameters. Any limit of maximizers must
therefore be optimal at the limiting parameters.
Since the maximizer is unique throughout
$[0,\bar\beta]\times(-\infty,0)$, it is weakly
continuous in $(\beta,\theta)$. Its mean is
continuous as well.

For fixed $\beta$ and $\theta_1<\theta_2<0$,
adding the two optimality inequalities gives
$(\theta_2-\theta_1)(m_{\theta_1}-m_{\theta_2})\ge0$.
Equal means would make each distribution maximize both
objectives. Uniqueness would make them the same distribution.
Its two Euler equations,
both with $\xi=0$, would then imply
$(\theta_2-\theta_1)w(u)=0$ on a nonempty separating
interval, a contradiction.
Thus, $m_\theta$ is strictly decreasing in $\theta$.

Since $G_\theta$ maximizes \eqref{eq:high-mean-price}
and $\delta_1$ gives $\beta\pi(1)-\theta$, we have
$J_\beta(G_\theta)-\theta m_\theta
\ge\beta\pi(1)-\theta$, or equivalently,
$(-\theta)(1-m_\theta)
 \le J_\beta(G_\theta)-\beta\pi(1)$.
The fact that payoffs are bounded therefore implies
$m_\theta\to1$ as $\theta\to-\infty$.
We also need the lower support point to approach one.
Integrating $\dot A\ge-\theta$ over the separating path gives
$T_\theta\le(1-\beta)\bar q/(-\theta)$.
Since $\varphi$ equals one on the top atom and
is nonnegative elsewhere,
\begin{equation}
 B_\theta=\int_0^1\varphi(u)\,\dd u
 \ge e^{-T_\theta}
 \ge\exp\!\left(-\frac{(1-\beta)\bar q}{-\theta}\right)
 \longrightarrow1.
 \label{eq:high-mean-bottom-bound}
\end{equation}

\paragraph{A threshold common to all weights.}
Choose $\theta_0<0$ sufficiently negative that
$\exp[-\bar q/(-\theta_0)]>a$.
By \eqref{eq:high-mean-bottom-bound}, the lower
support point then exceeds $a$ for every
$\beta\in[0,\bar\beta]$ and $\theta\le\theta_0$.
Restore the weight in the notation. Continuity and
$a<m_{\beta,\theta_0}<1$ then give the common mean threshold
$\bar m:=\max_{\beta\in[0,\bar\beta]}m_{\beta,\theta_0}\in(a,1)$.

Fix $m\in(\bar m,1)$ and $\beta\in[0,\bar\beta]$.
The mean $m_{\beta,\theta}$ is continuous and strictly decreasing
in $\theta$, and tends to one as $\theta\to-\infty$.
There is therefore a unique $\theta<\theta_0$ with
$m_{\beta,\theta}=m$.
This distribution has lowest posterior value
strictly above $a$. It uniquely maximizes $J_\beta$
among distributions on $[0,1]$ with mean $m$,
since any other maximizer at that mean would
also solve \eqref{eq:high-mean-price}.
It remains feasible when posterior values are
restricted to $[a,1]$, so the unique optimum
and its payoff are unchanged. This proves the
stated structure and
$V_{a,\beta}(m)=V_{0,\beta}(m)$.

For fixed $m$, the feasible set $\mathcal G_{a,m}$
is weakly compact and does not depend on $\beta$.
Joint continuity of the objective and uniqueness
therefore imply weak continuity of the optimal
distribution in $\beta$. Its payoffs are continuous
as well. By \eqref{eq:weight-payoff-monotonicity},
consumer surplus is weakly decreasing and profit
is weakly increasing.
\end{proof}

\subsubsection{Proof of Corollary~\ref{cor:high-mean-frontier}}
\label{app:high-mean-frontier-proof}

\begin{proof}
By \Cref{prop:high-mean}, the optimal payoff path
$(C_\beta,P_\beta)$ is continuous, with profit
weakly increasing in $\beta$.
For every $p\in(P_0,P_{\bar\beta}]$, there is
therefore a positive weight $\beta$ with
$P_\beta=p$. For any feasible distribution $G$ with
$\Pi(G)\ge p$, optimality gives $\CS(G)
 \le C_\beta+\beta[p-\Pi(G)]
 \le C_\beta$.
The distribution $G_{m,\beta}^*$ attains this bound.
Equality requires $\Pi(G)=p$ and optimality of $G$ at weight
$\beta$, so uniqueness gives $G=G_{m,\beta}^*$. For $p\le P_0$,
the unique buyer-optimal distribution meets the profit requirement
and attains $C_0$.

For $\beta>0$, increasing either payoff without
reducing the other would raise $C+\beta P$,
contradicting optimality. The buyer-optimal pair
is also efficient, since uniqueness rules out
another distribution with consumer surplus $C_0$
and higher profit.
Conversely, at every profit in
$[P_0,P_{\bar\beta}]$, the path attains the highest
feasible consumer surplus. An efficient pair
at that profit must therefore lie on the path.
Thus, the path gives exactly this part of
the frontier, and each pair is generated by a
unique distribution. When $\bar\beta=0$, the
path consists only of the buyer-optimal pair.
\end{proof}

\paragraph{Profit weight approaching one.}
For $a>0$, no single threshold below one makes
the value equality in \Cref{prop:high-mean}
hold at all higher means for every $\beta<1$.
To see why uniformity fails, fix $0<a<m<1$ and put
\[
 D:=\frac{1-m}{1-a}\{a\pi(1)-\pi(a)\}>0.
\]
The strict inequality follows from strict convexity
of $\pi$ and $\pi(0)=0$.
For $0\le\beta\le1$ and every
$G\in\mathcal G_{a,m}$, the first-best surplus
bound and the chord joining $\pi(a)$ and $\pi(1)$
give
\[
 J_\beta(G) \le\TS(G) \le\E_G[\pi(Y)]\le\frac{1-m}{1-a}\pi(a)
      +\frac{m-a}{1-a}\pi(1)=m\pi(1)-D.\]
For the binary prior on $\{0,1\}$ with mean $m$,
full revelation yields
$(\CS,\Pi)=(0,m\pi(1))$.
Hence $V_{0,\beta}(m)\ge\beta m\pi(1)$, and
$V_{a,\beta}(m)<V_{0,\beta}(m)$ whenever $(1-\beta)m\pi(1)<D$.
For each fixed $m\in(a,1)$, this condition holds once $\beta$
is close enough to one. Thus, the value equality in
\Cref{prop:high-mean} fails at that mean, and no single mean
threshold below one can make it hold for all $\beta<1$ when $a>0$.

\subsubsection{The high-mean condition at a fixed weight}
\label{app:fixed-weight-high-mean}

Fix $0\le\beta<1$ and maintain \Cref{ass:cost,ass:curvature}.
 We first allow the lowest primitive
value to be zero. This allows the designer to vary the probability
of an active market without changing the contracts traded in it.
We use the set $\Phi$ of regular virtual-value profiles and the
mean-price problem \eqref{eq:neutral-weighted-price}. Write
$\rho_\beta$ for the maximal ratio of weighted payoff to mean
in \Cref{lem:neutral-price-basics}.

\begin{lemma}[The active market at a fixed weight]
\label{lem:fixed-weight-critical-market}
There is a unique nonzero profile attaining $\rho_\beta$. Write
$G_{0,\beta}^c$ for its posterior distribution and
$m_{0,\beta}^c$ for its mean. For
$m_{0,\beta}^c\le m<1$, the optimum at $(0,m)$ is unique,
has full participation and no lower atom, and its lowest posterior
value is at least $\min\supp G_{0,\beta}^c$.
\end{lemma}

\begin{proof}
At price $\theta$, the best regular profile gives the value
$F_\beta(\theta)=\max_{\Phi}(J_\beta-\theta M)\ge0$, since the
zero profile is feasible. The same bound applies to any posterior
distribution on $[0,1]$. Indeed, at its own mean $m\in(0,1)$,
\Cref{lem:neutral-original-law} gives a weak improvement of the
form $(1-\lambda)\delta_0+\lambda G_A$, where $G_A$ is regular,
all its buyers trade, and $\lambda\E_{G_A}[Y]=m$. The priced
payoff of this mixture is
$\lambda(J_\beta(G_A)-\theta\E_{G_A}[Y])\le F_\beta(\theta)$.
At means zero and one the only distributions are themselves regular.
Consequently, a regular price maximizer also maximizes the objective
over all distributions at its own mean.

Consider a nonzero maximizing profile $\varphi$. Its mean is below
one by \Cref{lem:neutral-price-basics}. Part~(i) of
\Cref{thm:binary-structure} applies at this mean. Regularity rules
out a lower atom followed by a support gap, since the raw revenue
curve would then have an interior jump and could not be concave.
Thus $\varphi$ has a separating interval followed by the top atom.
The priced problem imposes neither a mean constraint nor a bottom
contact constraint on interior variations. The first variation
\eqref{eq:binary-cited-first-variation} therefore gives the
separating Euler equation with the specified mean price $\theta$
and $\xi=0$. More explicitly, on each compact interior interval,
both signs of a smooth perturbation are feasible because
$\varphi'>0$ there by \Cref{prop:separating-conditions}. Hence the
first variation of $J_\beta-\theta M$ vanishes there. In the time
variable $\tau=-\log(1-u)$, this identifies
$A(0)=0$ and $\dot A=q(\varphi)-\theta$.

By \eqref{eq:neutral-lowest-quality-price},
$q(\varphi(0))\ge\theta$. Since quality rises strictly along the
separating interval, $\dot A>0$ at positive separating ranks.
By \Cref{prop:separating-conditions}, for this $\theta$ the
trajectory traced back from $(A,Q)=((1-\beta)\bar q,1)$ is
unique as long as it is separating. If two maximizing trajectories
had different durations, the longer would reach $A=0$ at the
positive rank corresponding to the bottom of the shorter. But $A$
rises strictly from $A(0)=0$, so this is impossible. Their durations
and profiles therefore coincide. By \Cref{lem:neutral-price-basics}, a nonzero
maximizer exists at $\theta=\rho_\beta$ and attains the ratio.
The uniqueness just proved identifies $G_{0,\beta}^c$.

\paragraph{From prices to means.}
For $\theta<\rho_\beta$, $F_\beta(\theta)>0$, so the price
maximizer is unique and nonzero. With uniqueness established,
the argument in the proof of \Cref{lem:neutral-price-means}
applies here as well. The mean decreases strictly and
continuously from one to $m_{0,\beta}^c$ as $\theta$ rises
from $-\infty$ to $\rho_\beta$.

For $m>m_{0,\beta}^c$, choose the price whose optimizer has mean
$m$. This profile attains $F_\beta(\theta)>0$. Every optimum at
$(0,m)$ must attain the same priced payoff and must itself have
the mixture form in \Cref{lem:neutral-original-law}. The mixture
bound in the first paragraph forces $\lambda=1$ and the active
profile must be the unique price maximizer. At $m=m_{0,\beta}^c$
the priced value is zero. Here an optimal mixture must use the
ratio-maximizing active distribution, since $\lambda>0$ and
its priced payoff must also be zero. Its conditional mean already
equals $m$, so again $\lambda=1$.

\paragraph{The lowest posterior value.}
It remains to compare the support of the chosen distribution with
that of the critical market. For $f,g\in\Phi$, their pointwise
maximum and minimum belong to $\Phi$. We claim that
\[
 J_\beta(f\vee g)+J_\beta(f\wedge g)\ge J_\beta(f)+J_\beta(g).
\]
In the payoff formula the terms depending on only one rank
cancel. The remaining term for profile $f$ is
\[
 \int_0^1 Q_f(u)q(f(u))\,\dd u
 =\int_{0<u<v<1}\frac{q(f(u))f(v)}{1-u}\,\dd u\,\dd v .
\]
For each pair $u<v$, taking maxima and minima pairs the larger
of $q(f(u)),q(g(u))$ with the larger of $f(v),g(v)$.
The sum of the two products can only rise, which proves the claim.

Let $f$ maximize at a price $\theta_1$ and $g$ at a higher price
$\theta_2$. Compare $f$ with $f\vee g$ at the lower price and
$g$ with $f\wedge g$ at the higher price. Since
$M(f)=\int_0^1[-\log(1-u)]f(u)\,\dd u$, the sum of the two
payoff changes is at least
\[
 (\theta_2-\theta_1)
 \int_0^1[-\log(1-u)](g(u)-f(u))_+\,\dd u .
\]
Neither change can be positive by optimality. It follows that
$f\ge g$ almost everywhere. Apply this with $g$ the
nonzero critical profile and $f$ the optimizer for
$m>m_{0,\beta}^c$. Reconstruction gives
$Q_f(0)=\int_0^1f\ge\int_0^1g=Q_g(0)$, as required.
\end{proof}

If $0<a<m<1$, $m\ge m_{0,\beta}^c$, and
$a\le\min\supp G_{0,\beta}^c$, the unique optimum at $(0,m)$
therefore has support in $[a,1]$. For a binary prior this
support restriction and the mean condition are sufficient for
feasibility. Since every distribution feasible at $(a,m)$ is
also feasible at $(0,m)$, the same optimum is unique at $(a,m)$.

\section{Participation with a productive lowest value}
\label{app:general-participation-proof}

Fix $0\le\beta<1$ and maintain \Cref{ass:cost,ass:curvature}.
We first establish full participation near the productive lower
value and exclusion at a higher mean. We then prove parts~(ii)
and~(iii) of \Cref{thm:binary-structure} and locate the lower trading
atom at low means before completing the proof of
\Cref{thm:general-participation}.

\subsection{Full participation at low means}
\label{app:general-exclusion}

\subsubsection{Boundary conditions for an excluded market}
\label{app:exclusion-boundary-preliminaries}

\begin{lemma}[Strict exclusion]
\label{lem:strict-optimal-exclusion}
Suppose an optimal posterior distribution at $0<a<m<1$
excludes a positive mass of buyers. The excluded atom
at $a$ has strictly negative virtual value.
For the regular continuation, normalized to probability
one, let $B,M,J$ denote its lower support point, mean,
and weighted payoff. Let
$z=\lim_{u\downarrow0}\varphi(u)$ be its bottom virtual
value, and let $\theta,\xi$ be its multipliers.
Then $z>0$ and
\[
 \xi=0,\qquad J=\theta(M-a)>0,\qquad
 B-z=(1-\beta)\frac{q(z)}{q'(z)},\qquad
 z-a\ge\beta\frac{\pi(z)}{q(z)}\ge0.
\]
\end{lemma}

\begin{proof}
By \Cref{thm:binary-structure}(i), the optimum consists
of an atom at $a$ and a regular continuation with
probability $s\in(0,1)$. Its lower support point
satisfies $B>a$, and the excluded virtual value is
$\nu=(a-sB)/(1-s)\le0$.
The full objective is $sJ$.
By \Cref{lem:binary-two-multipliers}, interior
variations can change $B$ and $M$ independently
while keeping the endpoint virtual values fixed.
Their first-order effect on $J$ is
$\theta\,\dd M-\xi\,\dd B$.

We first obtain a bound that will allow us to
perturb a zero virtual value without assuming
that $q'(0)$ exists.
In virtual-value coordinates, \Cref{ass:curvature}
gives $(2-\beta)q'(t)-(1-t)q''(t)>0$, and hence
\[
 \frac{\dd}{\dd t}
 \bigl[(1-t)^{2-\beta}q'(t)\bigr]<0.
\]
Fix $t_0\in(0,1)$ and set
$k=(1-t_0)^{2-\beta}q'(t_0)>0$.
Then $q'(t)\ge k$ for $0<t\le t_0$.
Integrating and using continuity of $q$ at zero, with $q(0)=0$,
also gives $q(t)\ge kt$ on this interval.
These bounds remain valid if $q'$ is unbounded
near zero.

\paragraph{The excluded virtual value is strictly negative.}
Suppose $\nu=0$. Raising $B$ with $s$ and $M$
fixed makes $\nu$ negative. The lower atom
remains excluded, and its virtual value stays
below those in the continuation.
The first-order change in the full objective
is $-s\xi\,\dd B$, so optimality requires
$\xi\ge0$.

We must then have $z>0$. If $z=0$, as
$u\downarrow0$ we would have $Q(u)\to B>0$,
$\varphi(u)\to0$, and $A(u)\to\xi\ge0$.
But the Euler equation requires
$(Q-\varphi)q'(\varphi)
 =(1-\beta)q(\varphi)-A$.
The lower bound on $q'$ keeps the left-hand
side bounded away from zero, while the
right-hand side tends to $-\xi\le0$.
This is impossible.

Keep $a$ and $B$ fixed and change the lower
atom's virtual value from zero to $\epsilon$.
For either sign of sufficiently small $\epsilon$,
set
\[
 s_\epsilon=\frac{a-\epsilon}{B-\epsilon},
 \qquad
 M_\epsilon=a+\frac{m-a}{s_\epsilon}.
\]
These choices give the required virtual value
at the lower atom and preserve the prior mean.
Use an interior direction from
\eqref{eq:binary-remote-mean} to change the
continuation's mean to $M_\epsilon$, keeping
$B$ and its endpoint virtual values fixed.
The support gap remains $B-a>0$, and for small
$|\epsilon|$ the lower atom's virtual value
is still below $z$. Thus the construction is
feasible for both signs of $\epsilon$.

Let $J_\epsilon$ denote the continuation's
conditional weighted payoff. The full objective is
\[
 S(\epsilon)
 =s_\epsilon J_\epsilon
  +s_\epsilon(B-a)q(\epsilon)
  +\beta(1-s_\epsilon)\pi(\epsilon).
\]
The second term is the additional rent that
buyers in the continuation receive from access
to the lower contract.

The first-variation formula implies that
$s_\epsilon J_\epsilon$ is differentiable at zero.
If its derivative is negative, choosing a small
negative $\epsilon$ increases the objective,
because the other two terms are zero for
$\epsilon\le0$. If the derivative is nonnegative,
take $\epsilon>0$.
The bound $q(\epsilon)\ge k\epsilon$, together with
nonnegativity of profit, gives
\[
 S(\epsilon)-S(0)
 \ge s_\epsilon(B-a)k\epsilon+o(\epsilon)>0
\]
for sufficiently small $\epsilon$, because
$s_\epsilon\to s>0$.
Both cases contradict optimality, so $\nu<0$.

\paragraph{The boundary conditions.}
Since $\nu<0$, small changes of either sign in
$B$, holding $s,M$ fixed, preserve exclusion
and the attachment. Optimality therefore gives
$\xi=0$.

Next, change the continuation's probability to
$s+t$ and its mean to
$a+(m-a)/(s+t)$, keeping $B$ fixed.
This preserves the prior mean, and exclusion
remains strict for both signs of small $t$.
The derivative of the full objective at $t=0$
is $J-\theta(M-a)$, so
$J=\theta(M-a)$.
A feasible two-point spread around $m$ has positive
consumer surplus and nonnegative profit, so the optimal
objective $sJ$ is positive. This gives $J>0$ and $\theta>0$.

Taking the bottom limit in \Cref{lem:shadow-quality} gives
$\theta\le q(z)$. Because $\theta>0$ and $q$ is continuous with $q(0)=0$,
we have $z>0$. With $A(0)=\xi=0$, the bottom Euler equation
yields
$(B-z)q'(z)=(1-\beta)q(z)$.
Substituting this identity and
$J=\theta(M-a)$ into
\eqref{eq:reduced-continuation-value} gives
$\theta(B-a)=(B-z)q(z)+\beta\pi(z)$.
Finally, $\theta\le q(z)$ and $B>a$ imply
$(z-a)q(z)\ge\beta\pi(z)$.
Dividing by $q(z)>0$ proves the last inequality.
\end{proof}

\subsubsection{Low-mean participation}
\label{app:low-mean-comparison}
For the next result, fix $0<a<1$ and define
\begin{equation*}
 D_c(a):=\min_{q(a)\le x\le\bar q}x c''(x)>0,
 \qquad
 \eta_{c,\beta}(a):=
 \frac{a(1-\beta)D_c(a)}{a+(1-\beta)D_c(a)}>0.
\end{equation*}
\begin{lemma}[Full participation near a productive lower value]
\label{lem:full-trade-near-floor}
Every optimal posterior distribution has full participation
whenever $a<m<\min\{a+\eta_{c,\beta}(a),1\}$.
\end{lemma}

\begin{proof}
Since $q(a)>0$ and $c''$ is continuous and positive
on $[q(a),\bar q]$, the minimum defining $D_c(a)$
is strictly positive. Suppose an optimum excludes a positive mass of buyers.
Let $s\in(0,1)$ be its trading probability, and let
$B,M,z$ denote the lower support point, mean, and
bottom virtual value of its normalized trading
distribution.
By \Cref{lem:strict-optimal-exclusion}, $z\ge a$.
Thus $q(z)\in[q(a),\bar q]$, and the boundary
condition in that lemma gives
\[
 B=z+(1-\beta)\frac{q(z)}{q'(z)}
  =z+(1-\beta)q(z)c''(q(z))
  \ge a+(1-\beta)D_c(a).
\]
The lowest posterior value among buyers who
trade is therefore bounded away from $a$.

The excluded atom also has strictly negative
virtual value, so $sB>a$ and hence $s>a/B$.
Using the mean constraint and $M\ge B>a$, we obtain
\[
 m-a=s(M-a)
 >\frac{a(M-a)}{B}
 \ge\frac{a(B-a)}{B}.
\]
The last expression increases in $B$, so substituting
the lower bound on $B$ gives
\[
 m-a>
 \frac{a(1-\beta)D_c(a)}
      {a+(1-\beta)D_c(a)}
 =\eta_{c,\beta}(a).
\]
This is a contradiction, so every
optimum has full participation.
\end{proof}

\subsection{Exclusion at a higher mean}
\label{app:general-reversal}

We show that there is a mean at which every
optimum excludes a positive share of buyers,
both when $a=0$ and when $a>0$ is sufficiently
small.

\Needspace{5\baselineskip}
\begin{lemma}[Small means with a neutral lowest value]
\label{lem:small-neutral-priors}
Under \Cref{ass:cost}, there are $m_H\in(0,1)$
and constants $e_0,p_0>0$ such that every optimal
distribution at $(a,m)=(0,m_H)$ has the form
$G=(1-s)\delta_0+sG_A$, where $G_A$ is regular,
$1-s\ge e_0$, and $\Pi(G)\ge p_0$.
\end{lemma}

\begin{proof}
By \Cref{lem:neutral-price-basics}, the maximal
ratio $\rho=\rho_\beta$ is attained.
Every profile attaining it satisfies
$\varphi(0)\ge z_*:=c'(\rho)>0$ and
$M(\varphi)\ge z_*$.
Choose any $m_H\in(0,z_*)$.

By \Cref{lem:neutral-original-law}, every optimum
has a regular trading component and all remaining
mass at zero. Writing $M_A$ for the component's
mean, we have $sM_A=m_H$ and
\[
 J_\beta(G)
 =sJ_\beta(G_A)
 =m_H\frac{J_\beta(G_A)}{M_A}
 \le\rho m_H.
\]
Conversely, take any regular distribution attaining the
ratio $\rho$ and mix it with an atom at zero,
using trading probability $m_H/M_A$.
This mixture is feasible and attains $\rho m_H$.
Since the upper bound is attained, the active component
of every optimum must itself attain the ratio $\rho$.
Its virtual values are therefore at least $z_*$,
so $M_A\ge z_*$ and $\Pi(G_A)\ge\pi(z_*)$.
Since also $M_A\le1$, we obtain
$1-s\ge1-m_H/z_*>0$ and $\Pi(G)=s\Pi(G_A)\ge m_H\pi(z_*)>0$.
Take $e_0=1-m_H/z_*$ and $p_0=m_H\pi(z_*)$.
\end{proof}

\begin{lemma}[A small productive lower value]
\label{lem:small-productive-floor}
Maintain \Cref{ass:cost,ass:curvature}, and choose
$m_H$ as in \Cref{lem:small-neutral-priors}.
There are $a_0\in(0,m_H)$ and constants $e_1,p_1>0$
such that, for every $a\in(0,a_0)$, every optimum
at $(a,m_H)$ has an excluded atom of probability
at least $e_1$ and profit at least $p_1$.
\end{lemma}

\begin{proof}
For any sequence $a_n\downarrow0$ and corresponding
optimizers $G_n$ at $(a_n,m_H)$, compactness and
\Cref{lem:seg-payoff-continuity} give a subsequence
converging weakly to an optimum $G_0$ at $(0,m_H)$.
By \Cref{lem:small-neutral-priors},
$G_0(\{0\})\ge e_0$ and $\Pi(G_0)\ge p_0$.

Payoff continuity implies that every optimum
for sufficiently small $a>0$ has profit at
least $p_1:=p_0/2$.
Otherwise, take a sequence of optima with profits
below $p_0/2$ as $a_n\downarrow0$. Passing to a
subsequence gives an optimum $G_0$ at $(0,m_H)$
with $\Pi(G_0)\le p_0/2$, contrary to the bound
$\Pi(G_0)\ge p_0$ in the preceding lemma.

To rule out full participation, we bound profit in terms of
the lowest posterior value. Consider any posterior distribution
$G$ under which everyone trades, and let $b=\min\supp G$.
Its ironed virtual values are nonnegative almost
everywhere. Endpoint contact gives
$\int_0^1\phi_Q(u)\,\dd u=R_Q(0)-R_Q(1)=b$.
Convexity of $\pi$ and $\pi(0)=0$ imply
$\pi(z)\le z\pi(1)$ for $0\le z\le1$, so
\[
 \Pi(G)
 =\int_0^1\pi(\phi_Q(u))\,\dd u
 \le\pi(1)\int_0^1\phi_Q(u)\,\dd u
 =b\pi(1).
\]

Suppose fully participating optima existed along
a sequence $a_n\downarrow0$.
The uniform profit bound would place their entire
posterior support in $[b_*,1]$, where
$b_*:=p_1/\pi(1)>0$.
Any weak limit would still be supported on
$[b_*,1]$ and would therefore have no mass at
zero. But it must be optimal at $(0,m_H)$,
where every optimum has a positive atom at
zero. Hence every optimum has an excluded
atom for all sufficiently small $a>0$.

For these optima with exclusion, \Cref{thm:binary-structure}(i) gives
the form
$G=(1-s)\delta_a+sG_A$, where $G_A$ is a regular
trading distribution. Let $B$ be its lower
support point. Applying the preceding profit
bound to $G_A$ gives
$p_1\le\Pi(G)=s\Pi(G_A) \le sB\pi(1)\le B\pi(1)$.
Thus, $B\ge b_*>0$, so the trading distribution
is supported on $[b_*,1]$ for every sufficiently
small $a>0$.

Take any sequence of these optima with
$a_n\downarrow0$. Along a subsequence,
$s_n\to s$ and $G_{A,n}$ converges weakly to a
distribution $G_A$ supported on $[b_*,1]$.
The full limit is
$G_0=(1-s)\delta_0+sG_A$ and is optimal at
$(0,m_H)$. Since $G_A$ has no mass at zero,
$1-s=G_0(\{0\})\ge e_0$.
For all sufficiently small $a>0$, every optimum
must therefore exclude a share of at least
$e_1:=e_0/2$. Otherwise, we could take
$a_n\downarrow0$ and optima with excluded shares
below $e_0/2$. Along a convergent subsequence,
we would obtain $1-s\le e_0/2$, a contradiction.
Reduce $a_0$ if necessary so that both the profit
bound and the bound on exclusion hold for every
optimum whenever $0<a<a_0$.
\end{proof}

\subsection{Sufficient conditions for a basic product and for exclusion}
\label{app:binary-regime-conditions}

Fix $0\le\beta<1$ and maintain \Cref{ass:cost,ass:curvature}.
We use part~(i) of \Cref{thm:binary-structure} and the necessary
conditions for the separating interval in \Cref{prop:separating-conditions}.
All probabilities and ranks within an upper distribution are
conditional on belonging to that distribution.
The first issue is whether a sequence of separating distributions
can approach a point mass. Each has a positive atom at one, but
the probability of that atom need not be bounded away from zero.

\begin{lemma}[A separating distribution cannot concentrate at an interior value]
\label{lem:regime-no-interior-collapse}
Consider regular upper distributions with the reconstruction
\eqref{eq:continuation-quantile} and the separating equations, quality
bound, and terminal conditions of \Cref{prop:separating-conditions}.
No sequence of these distributions converges weakly to $\delta_d$
for $0<d<1$.
\end{lemma}

\begin{proof}
Suppose a sequence converges to $\delta_d$.
Write $Q_n,\varphi_n,A_n$ for the corresponding quantile functions,
virtual-value profiles, and cumulative terms. Use log rank
$\tau=-\log(1-u)$ and write a dot for a derivative in $\tau$.
The probability of the top atom tends to zero, since $d<1$.
Its starting log rank therefore tends to infinity. In particular,
each fixed interval of positive log ranks eventually lies entirely
within the separating region.

On every such compact interval, both $Q_n$ and $\varphi_n$ converge
uniformly to $d$. For $Q_n$ this follows from weak convergence to a
point mass and monotonicity. For $\varphi_n$, take any subsequence.
Its bounded monotone profiles have a further subsequence converging
almost everywhere to a monotone profile. By reconstruction and
dominated convergence, the integral of this limiting profile over
every rank interval $[u,1]$ is $(1-u)d$. The profile must therefore
equal $d$ almost everywhere. Monotonicity then gives uniform
convergence on compact interior intervals, and the argument applies
to every subsequence.

The Euler equation \eqref{eq:continuation-euler} now gives
$A_n\to(1-\beta)q(d)$ uniformly on those intervals. Integrating
$\dot A_n=q(\varphi_n)-\theta_n$ over any one of them gives
$\theta_n\to q(d)$. For large $n$ let $z_n^*=c'(\theta_n)$, so that
$q(z_n^*)=\theta_n$ and $z_n^*\to d$. The quality bound and strict
separation imply $z_n^*\le\varphi_n<Q_n$. Thus the ratio
\[
 r_n=\frac{\varphi_n-z_n^*}{Q_n-z_n^*}
\]
lies between zero and one.

We show that $r_n$ would have to increase by more than one on a
fixed interval. In log rank \eqref{eq:continuation-Q} gives
$\dot Q_n=Q_n-\varphi_n$, so
\[
 \dot r_n=\frac{\dot\varphi_n}{Q_n-z_n^*}-r_n(1-r_n).
\]
To evaluate the first term, the mean value theorem gives
$q(\varphi_n)-q(z_n^*)=(q'(d)+o(1))(\varphi_n-z_n^*)$ uniformly
on any fixed compact interval of positive log ranks. Since
$\varphi_n-z_n^*$ and $Q_n-\varphi_n$ are nonnegative and add to
$Q_n-z_n^*$, the numerator in the log-rank equation for $\dot\varphi_n$
is $(q'(d)+o(1))(Q_n-z_n^*)$. Its denominator tends uniformly to
$(2-\beta)q'(d)>0$, since $q''$ is bounded near $d$. Consequently,
\eqref{eq:continuation-phi} gives
\[
 \dot r_n=\frac{1}{2-\beta}-r_n(1-r_n)+o(1)
 \ge \frac14-o(1),
\]
where we used $1/(2-\beta)\ge1/2$ and $r_n(1-r_n)\le1/4$.
Integrating from $\tau=1$ to $\tau=6$ gives an increase of at least
$5/4-o(1)$, contradicting $0\le r_n\le1$.
\end{proof}

\subsubsection{Proof of \texorpdfstring{\Cref{thm:binary-structure}}{Theorem~\ref{thm:binary-structure}}(ii)}
\label{app:regime-trading-bounds}

We construct two bounds from the separating distributions in
\Cref{prop:separating-conditions}. The first bounds the primitive
minimum when exclusion is optimal. The second bounds the prior mean
when there is no lower atom.

Consider any normalized regular distribution whose separating interval
and top atom satisfy the conditions in
\Cref{prop:separating-conditions}, including $q(\varphi)\ge\theta$.
Write $B=Q(0)$, $\zeta=\varphi(0)$, $M=\int_0^1Q(u)\,\dd u$,
and $A(0)=\xi$, taking right limits at zero where necessary.
Let $\mathcal S$ be the set of these distributions with $\xi\ge0$.
For $G\in\mathcal S$, write $M(G)$ for its mean. Let $\mathcal E$
consist of the pairs $(d,G)$ such that $G\in\mathcal S$, $0<d<B$,
$\zeta>0$, $\theta>0$, and
\begin{equation}
 \xi=0,\qquad \theta(B-d)=F_{B,\beta}(\zeta).
 \label{eq:regime-exclusion-boundary}
\end{equation}
Here $d$ is a possible value of the primitive minimum. Membership
in either set requires only the stated conditions on a separating
distribution.

An optimum without a lower atom belongs to $\mathcal S$. Indeed,
the variations in \Cref{lem:binary-two-multipliers} can raise $B$
at fixed $M$, even when $B=a$. Since the resulting marginal change
in the objective is $-\xi\,\dd B$, optimality requires $\xi\ge0$.
If an optimum excludes buyers at $a>0$, its normalized continuation
$G$ satisfies $(a,G)\in\mathcal E$. To see the second equality in
\eqref{eq:regime-exclusion-boundary}, recall from
\Cref{lem:strict-optimal-exclusion} that
$J=\theta(M-a)$ and $(B-\zeta)q'(\zeta)=(1-\beta)q(\zeta)$.
Substitution in \eqref{eq:reduced-continuation-value} gives
\[
 \theta(B-a)=J-\theta(M-B)
   =(B-\zeta)q(\zeta)+\beta\pi(\zeta)=F_{B,\beta}(\zeta).
\]
The set $\mathcal E$ is nonempty by \Cref{lem:small-productive-floor},
whose proof uses the structural part of \Cref{thm:binary-structure}.

\begin{lemma}[Ordering the two bounds]
\label{lem:regime-ordered-bounds}
Define
\[
 \overline a_E=\sup_{(d,G)\in\mathcal E}d,\qquad
 \underline m_S=\inf_{G\in\mathcal S}M(G).
\]
Then $0<\overline a_E<\underline m_S<1$.
\end{lemma}

\begin{proof}
Compare any $(d,G_E)\in\mathcal E$ with any $G_S\in\mathcal S$.
We first show that the bottom virtual value $\zeta_S$ of $G_S$
is at least $d$. Use subscripts $E$ and $S$ for the quantities associated
with $G_E$ and $G_S$, respectively.

For the exclusion continuation, $\xi_E=0$ makes
$F'_{B_E,\beta}(\zeta_E)=0$ by the bottom Euler equation.
Here the prime denotes differentiation in virtual value.
Strict concavity therefore makes $\zeta_E$ the maximizer of
$F_{B_E,\beta}$ on $[0,B_E]$. In particular,
\[
 \theta_E(B_E-d)
 =F_{B_E,\beta}(\zeta_E)
 \ge F_{B_E,\beta}(d)
 =(B_E-d)q(d)+\beta\pi(d).
\]
Thus $q(\zeta_E)\ge\theta_E\ge q(d)$. If
$\theta_S\ge\theta_E$, then $q(\zeta_S)\ge\theta_S\ge q(d)$ as
well, proving $\zeta_S\ge d$.

Suppose, then, that $\theta_S<\theta_E$. We show that
$\zeta_S\ge\zeta_E$. Since the virtual value increases strictly on
each separating interval, we can write $P_i(z)$ for the posterior
value on path $i$ at virtual value $z$, with $P_i(1)=1$. Dividing
\eqref{eq:continuation-Q} by \eqref{eq:continuation-phi} gives
\[
 \frac{\dd P}{\dd z}
 =D_\theta(z,P)
 :=\frac{(P-z)\big[(2-\beta)q'(z)-(P-z)q''(z)\big]}
             {q(z)-\theta+(P-z)q'(z)}.
\]
At $P>z$ on the exclusion path, the numerator is strictly positive
by \eqref{eq:rent-curvature}. The denominator is positive and falls
as $\theta$ rises from $\theta_S$ to $\theta_E$. Hence
$D_{\theta_S}(z,P_E(z))<D_{\theta_E}(z,P_E(z))$ below the top.
Tracing the paths backward from their common terminal value
therefore gives $P_S(z)>P_E(z)$ wherever both are defined below one.

For the terminal point, note first that
$\theta_E\le q(\zeta_E)<q(1)$, so both equations extend to $z=1$.
On any common interval $[z_0,1]$ with
$z_0\ge\zeta_E$, the difference $h=P_S-P_E$ solves
$h'=k(z)h-e(z)$, $h(1)=0$. Here $e(z)$ is the positive difference
$D_{\theta_E}(z,P_E(z))-D_{\theta_S}(z,P_E(z))$ for $z<1$,
and $k$ is the divided difference of $D_{\theta_S}(z,\cdot)$
between $P_S$ and $P_E$, using the derivative at equality.
This coefficient is bounded because $q(z)-\theta_S\ge
\theta_E-\theta_S>0$ and $z_0>0$. Integration gives
$h(z)=\int_z^1 e(t)\exp[-\int_z^t k(v)\,\dd v]\,\dd t>0$
for $z<1$.

If $\zeta_S<\zeta_E$, both paths reach $z=\zeta_E$. The Euler
equation at this common virtual value gives
\[
 A_S=(1-\beta)q(\zeta_E)-
          [P_S(\zeta_E)-\zeta_E]q'(\zeta_E)
 <
 (1-\beta)q(\zeta_E)-
          [B_E-\zeta_E]q'(\zeta_E)=0.
\]
Yet $A_S$ starts at $\xi_S\ge0$ and increases weakly, since
$\dot A_S=q(\varphi_S)-\theta_S\ge0$ in log rank.
This contradiction gives $\zeta_S\ge\zeta_E\ge d$.

It follows that the bottom virtual value of every distribution in
$\mathcal S$ is at least $\overline a_E$. Its bottom posterior value
is at least as large, so $M(G)\ge\overline a_E$ for every such $G$.
Nonemptiness of $\mathcal E$ gives $\overline a_E>0$, and any
one of its nondegenerate continuations gives $\underline m_S<1$.
Suppose $\underline m_S=\overline a_E$. Take $G_n\in\mathcal S$
with $M(G_n)\to\overline a_E$. These distributions are supported
on $[\overline a_E,1]$, so for each $\epsilon>0$ the probability
above $\overline a_E+\epsilon$ is at most
$[M(G_n)-\overline a_E]/\epsilon$. They therefore converge to a
point mass at $\overline a_E\in(0,1)$, contrary to
\Cref{lem:regime-no-interior-collapse}.
\end{proof}

Now take $\overline a_E<a<m<\underline m_S$. An optimum with
exclusion would yield a member of $\mathcal E$ with $d=a$, and an
optimum without a lower atom would itself belong to $\mathcal S$
and have mean $m$. The two bounds rule out these possibilities.
By \Cref{thm:binary-structure}(i), every optimum has a lower
trading atom and full participation.

\paragraph{Full participation and the lowest posterior value.}
For $0<t\le1$, write $h(t)=t(1-\log t)$, with $h(0)=0$.
The function $h$ is continuous and strictly increasing on $[0,1]$.
For $0<m<1$, let $b(m)$ be the unique solution of $h(b(m))=m$.

Suppose all buyers trade and let $v=\min\supp G$. Then
$\phi_Q>0$ almost everywhere, so $R_Q$ is strictly decreasing.
Endpoint contact gives $R_Q(0)=v$ and $R_Q(1)=0$, hence $v>0$.
For $0<u<1$, the raw and concavified revenue curves satisfy
$(1-u)Q^+(u)\le R_Q(u)<v$.
Since $m<1$, we also have $v<1$. Integrating the upper bound
$\min\{1,v/(1-u)\}$ on $Q$, which is strict almost everywhere
on $(0,1-v)$, gives
\begin{equation}
 m<\int_0^{1-v}\frac{v}{1-u}\,\dd u+v
   =h(v).
 \label{eq:binary-access-mean}
\end{equation}
In particular, full participation requires $v>b(m)$. At a prior
in the basic-product region of part~(ii) satisfying $m\ge h(a)$,
the lower trading atom must therefore lie strictly above $a$.
If $h(\overline a_E)<\underline m_S$, continuity lets us choose
$a>\overline a_E$ close enough that $h(a)<\underline m_S$.
Any $m\in[h(a),\underline m_S)$ then satisfies both conditions
in part~(ii), since $h(a)>a$.

\subsubsection{Proof of \texorpdfstring{\Cref{thm:binary-structure}}{Theorem~\ref{thm:binary-structure}}(iii)}
Let $\rho_\beta$ be the maximal payoff per unit of mean among
regular active markets, as in \Cref{lem:neutral-price-basics}.
Recall that $\mathcal R_\beta(y)=\max_{0\le z\le y}F_{y,\beta}(z)$,
where $F_{y,\beta}(z)=(y-z)q(z)+\beta\pi(z)$.
We show that every optimum has an excluded atom at $a$ if
\begin{equation}
 0<a<m<\underline m_S,\qquad
 \frac{a}{b(m)}+
 \frac{\mathcal R_\beta(m)}{\rho_\beta m}\le1 .
 \label{eq:binary-exclusion-certificate}
\end{equation}
The bound $m<\underline m_S$ rules out an optimum without a lower atom.
The second inequality lets us compare serving a lower trading atom
with assigning more buyers to the upper market.

First consider any regular, fully participating distribution
$G_A$ with mean $M>m$ and lower support point $B$. In this
argument suppose $a\le b(m)$, as required by
\eqref{eq:binary-exclusion-certificate}. The full-participation
bound gives $B>b(m)$ if $M<1$. If $M=1$, then $B=1>b(m)$. Put probability $s=(m-a)/(M-a)$ on
$G_A$ and the remaining probability on $a$. This gives a feasible
posterior distribution at $(a,m)$. We claim that $sB>a$.
Indeed, the upper market's raw revenue curve is decreasing from
$B$, so the raw revenue curve of the mixture is at most $sB$
on its upper ranks. If $sB\le a$, the mixture's raw revenue
curve would be at most $a$ everywhere, giving
$Q(u)\le\min\{1,a/(1-u)\}$ almost everywhere. This inequality
is strict on $(0,\min\{1-s,1-a\})$, where $Q(u)=a$.
Integrating would give $m<h(a)$, whereas $a\le b(m)$ implies
$h(a)\le m$.

The chord joining $(0,a)$ to $(1-s,sB)$ thus slopes upward.
It dominates the raw revenue $(1-u)a$ on the lower ranks,
and the upper raw revenue curve is concave and decreasing.
The resulting curve is a concave majorant. It is the least such
majorant because every concave majorant must dominate the chord
between its endpoints. The lower atom consequently
receives zero quality, and the upper ranks retain the virtual
values of $G_A$. By \eqref{eq:seller-profit} and
\eqref{eq:consumer-surplus}, the payoff of the mixture is
$sJ_\beta(G_A)$.

Next take a regular market $G_c$ attaining $\rho_\beta$ and
write $M_c$ for its mean. By \Cref{lem:neutral-price-basics},
its bottom virtual value is positive, and $0<M_c<1$.
It is optimal at $(0,M_c)$. Indeed, by
\Cref{lem:neutral-regularization}, an optimum with $a=0$ consists
of a regular trading market with probability $s$ and a possible
atom at zero. Its payoff is at most $\rho_\beta$ times the
overall mean, and $G_c$ attains this bound at $M_c$. Its lowest
posterior valuation is positive, so it remains optimal at
$(a',M_c)$ for any sufficiently small $a'>0$. As a regular,
fully participating optimum it has no lower trading atom:
such an atom would induce strict concavification by
\Cref{thm:binary-structure}(i) and the lower revenue contact.
Consequently $M_c\ge\underline m_S>m$. Applying the preceding
construction with $G_A=G_c$ shows that the optimal payoff at
$(a,m)$ satisfies
\begin{equation}
 \frac{V_{a,\beta}(m)}{m-a}
 \ge\frac{J_\beta(G_c)}{M_c-a}
 =\rho_\beta\frac{M_c}{M_c-a}>\rho_\beta .
 \label{eq:binary-exclusion-benchmark}
\end{equation}

Suppose now that an optimum serves a lower trading atom. Denote
its posterior value by $v$, its virtual value by $z$, and the
probability of the regular upper market $G_A$ by $s$. Write
$M$ for the mean of $G_A$. The strict support gap and full
participation give $M>m>v>b(m)\ge a$. The upper ranks have the
same virtual values as when $G_A$ is screened on its own.
Each buyer in $G_A$ also receives rent $(B-v)q(z)$ from the lower
contract, where $B=\min\supp G_A$. The revenue contact gives
$s(B-v)=(1-s)(v-z)$. Adding this rent and the weighted profit
from the lower atom therefore gives
$V_{a,\beta}(m)=sJ_\beta(G_A)+(1-s)F_{v,\beta}(z)$ and $m-a=s(M-a)+(1-s)(v-a)$.
We bound the lower atom's contribution per unit of mean above
$a$. Since $0<z\le v<m$ and $zq(z)-\beta\pi(z)\ge0$, we have
\[
 \frac{F_{v,\beta}(z)}v
 =q(z)-\frac{zq(z)-\beta\pi(z)}v
 \le \frac{F_{m,\beta}(z)}m
 \le \frac{\mathcal R_\beta(m)}m.
\]
Condition \eqref{eq:binary-exclusion-certificate}, together
with $a>0$ and $v>b(m)$, implies
\[
 \frac{F_{v,\beta}(z)}{v}
 \le \rho_\beta\left(1-\frac{a}{b(m)}\right)
 < \rho_\beta\left(1-\frac{a}{v}\right).
\]
Thus $F_{v,\beta}(z)/(v-a)<\rho_\beta$. But
$V_{a,\beta}(m)/(m-a)$ is a weighted average of this ratio and
$J_\beta(G_A)/(M-a)$, with strictly positive weights. By
\eqref{eq:binary-exclusion-benchmark}, the upper market's ratio
must strictly exceed the average. Mixing this same $G_A$ with an
atom at $a$, as above, attains
$(m-a)J_\beta(G_A)/(M-a)>V_{a,\beta}(m)$, a contradiction.
By \Cref{thm:binary-structure}(i), only exclusion at $a$ remains.

\paragraph{The bound stated in the theorem.}
For $0\le z\le m$, the bounds $\pi(z)\le zq(z)$ and
$q(z)\le q(m)$ imply $F_{m,\beta}(z)\le m q(m)$.
By \Cref{lem:neutral-price-basics},
$0<\rho_\beta\le\pi(1)<\bar q$, so we may take
\[
 \bar m_E=\min\{1/2,\underline m_S,c'(\rho_\beta/2)\}>0.
\]
If $0<a<m<\bar m_E$, then $2a<1$ and
$\mathcal R_\beta(m)/(\rho_\beta m)\le q(m)/\rho_\beta<1/2$.
If also $m\ge h(2a)$, then $2a\le b(m)$ and
$a/b(m)\le1/2$. Condition \eqref{eq:binary-exclusion-certificate}
now applies. At each such fixed $m>0$, the condition $h(2a)\le m$
holds for all sufficiently small $a>0$, since $h(0)=0$.
The exclusion condition \eqref{eq:binary-exclusion-certificate}
can hold even when $m<h(2a)$. At $a=b(m)/2$ there is slack in
\eqref{eq:binary-exclusion-certificate}. Increasing $a$ slightly
preserves that inequality but gives $h(2a)>m$.

\subsubsection{The basic product when the prior mean is close to \texorpdfstring{$a$}{a}}

\begin{lemma}[A trading atom at the lowest primitive value]
\label{lem:regime-atom-near-floor}
For every $0<a<1$ there is $\eta_a^{B}\in(0,1-a)$ such that every
optimum has a lower trading atom at $a$ whenever $a<m<a+\eta_a^{B}$.
\end{lemma}

\begin{proof}
Fix $a$ and consider any sequence $m_n\downarrow a$ and corresponding
optima. By \Cref{lem:full-trade-near-floor}, everyone trades for large
$n$. Write $s_n$ for the probability of the regular upper distribution
and $M_n$ for its mean, taking $s_n=1$ if there is no lower atom.
Every posterior value is at least $a$, so
$s_n(M_n-a)\le m_n-a$. If $s_n$ stayed bounded away from zero along
a subsequence, its conditional mean would tend to $a$. Since its entire
support lies in $[a,1]$, the conditional distribution would converge
to $\delta_a$, contrary to \Cref{lem:regime-no-interior-collapse}.
Hence $s_n\to0$. In particular, there must be a lower trading atom
for all large $n$.

It remains to locate this atom. Suppress the index $n$ and write $v,z$
for its posterior and virtual values, and $B,M,J$ for the lower support point,
mean, and payoff of the upper distribution screened on its own.
The revenue contact gives $v=(1-s)z+sB$. Holding $s$ and $B$ fixed,
the full mean and payoff are
\[
 (1-s)v+sM,\qquad
 sJ+s(B-v)q(z)+\beta(1-s)\pi(z).
\]
Suppose $v>a$. Vary $z$ by a small amount $t$, which changes $v$ by
$(1-s)t$. Use \eqref{eq:binary-remote-mean} to change $M$ by
$-(1-s)^2t/s$. This restores the mean exactly, preserves $B$ and the
endpoint virtual values, and changes the conditional payoff to first
order at the rate $\theta$ per unit of conditional mean. Both signs of
$t$ are feasible. The lower support bound is slack, the atom trades,
and its virtual value is strictly below the bottom virtual value of
the upper distribution by \Cref{lem:strict-attachment}.

Differentiating the displayed payoff and dividing by $1-s$ gives
$(1-s)\theta=s(B-z)q'(z)-s q(z)+\beta q(z)$.
Since $\theta\ge q(z)$ by \Cref{lem:strict-attachment}, an atom
with $v>a$ must therefore satisfy
$(1-\beta)q(z)\le s(B-z)q'(z)$.
If $v_n>a$ along a subsequence, then $a\le v_n\le m_n$ gives
$v_n\to a$, and the contact identity with $s_n\to0$ gives $z_n\to a$.
The right-hand side of this inequality tends to zero, since $B_n\le1$
and $q'$ is continuous near $a>0$. The left-hand side tends to
$(1-\beta)q(a)>0$. This proves that the lower support bound binds
for every optimum once the mean is sufficiently close to $a$.
\end{proof}

\subsection{Proof of \texorpdfstring{\Cref{thm:general-participation}}{Theorem~\ref{thm:general-participation}}}
\label{app:general-completion}

\begin{proof}
By \Cref{lem:small-productive-floor}, there are
$m_H\in(0,1)$ and constants $a_0,e_1,p_1>0$
such that, whenever $0<a<a_0$, every optimum
at $(a,m_H)$ has participation at most $1-e_1$
and profit at least $p_1$.
We will choose a lower mean close enough to $a$ that every optimum
has a lower trading atom at $a$ and total surplus is below $p_1$.

For any feasible distribution $G$ on $[a,1]$
with mean $m$, individual rationality gives
$\Pi(G)\le\TS(G)$.
Total surplus at posterior value $y$ is at most
$\pi(y)$, and convexity bounds $\pi$ by the chord
joining its values at $a$ and $1$. Thus,
\begin{equation}
 \Pi(G)\le\TS(G)\le\E_G[\pi(Y)]
 \le\pi(a)+\frac{m-a}{1-a}[\pi(1)-\pi(a)].
 \label{eq:general-small-mean-payoff-bound}
\end{equation}

Since $\pi$ is continuous and $\pi(0)=0$, choose
$0<\bar a<\min\{a_0,m_H\}$ such that
$\pi(a)<p_1$ for every $a\in(0,\bar a)$.
Fix such an $a$ and choose
\[
 0<\eta_a<
 \min\!\left\{
  \eta_a^{B},\;
 m_H-a,\;
 \frac{(1-a)[p_1-\pi(a)]}{\pi(1)-\pi(a)}
 \right\}.
\]
All three bounds are positive.

For any $m_L\in(a,a+\eta_a)$, the first bound and
\Cref{lem:regime-atom-near-floor} give a lower trading atom at $a$
and full participation in every optimum.
The second gives $m_L<m_H$, while the third, together with
\eqref{eq:general-small-mean-payoff-bound}, makes total surplus
strictly smaller than $p_1$ for every feasible distribution
at $(a,m_L)$.

Set $\varepsilon=e_1$.
For every pair of optimal distributions $G_L$
and $G_H$ at the two means,
$T(G_L)=1$ and $T(G_H)\le1-\varepsilon$, while
$\Pi(G_L)\le\TS(G_L)
 <p_1\le\Pi(G_H)\le\TS(G_H)$.
Thus, participation falls by at least $\varepsilon$,
while both profit and total surplus strictly increase.
\end{proof}

\section{The quadratic switch from inclusion to exclusion}
\label{app:quadratic-proofs}

Throughout, $\beta=0$, $c(x)=x^2/2$, and $\bar Q>1$, so positive
virtual values equal quality. We identify the inclusion and
exclusion candidates, then compare their surpluses near a
prior where they tie.

\subsection{The candidate distributions}
\label{app:quadratic-preliminaries}

For $t\in(0,\infty]$, put
$T=2\chi=2\arctan t$, $b=1-e^{-T}$,
$\gamma(t)=e^{-\arctan t}/\sqrt{1+t^2}$, and
$h=(1-\theta)\gamma(t)$.
Expressions at $t=\infty$ are interpreted by limits.

\Needspace{5\baselineskip}
\begin{lemma}[Quadratic continuations and lower atoms]
\label{lem:mult-quadratic-smooth-tail}
\label{lem:mult-quadratic-atom-matching}
Every optimal normalized continuation has $\theta<1$
and some $t\in(0,\infty]$. Writing
$\eta(u)=\arctan t+\tfrac12\log(1-u)$, its profiles are
\begin{align*}
 \varphi(u)&=\theta+(1-\theta)e^{-\eta(u)}\cos\eta(u),
 \\
 Q(u)&=\theta+(1-\theta)e^{-\eta(u)}
                 [\cos\eta(u)+\sin\eta(u)]
 \quad(0\le u<b),
\end{align*}
with $Q=\varphi=1$ on $[b,1]$.
Conversely, these profiles give a feasible regular
continuation whenever $\theta<1$ and $\varphi(0)\ge0$.
Its bottom virtual and posterior values, mean, initial Euler state, and surplus satisfy
\[
 \begin{gathered}
 \varphi(0)=\theta+h,\qquad B=\theta+h(1+t),\qquad
 \xi=\theta+h(1-t),\\
 M=\theta+2ht+(1-\theta)e^{-T},\qquad
 K:=C-\theta(M-B)=h^2t^2,
 \end{gathered}
\]
and $\dd C=\theta\,\dd M-\xi\,\dd B$.

If an optimum has a lower trading atom with value $w$,
virtual value $z$, and continuation probability $r$, then
$t$ is finite, $t>2$, and
\begin{equation}
 r=\frac{t-2}{t+2},\qquad
 z=\theta-h,\qquad w=\theta+h(t-3).
 \label{eq:mult-atom-matching}
\end{equation}
Put $j(t)=\gamma(t)(t-2)(t+3)/(t+2)$ and
$\vartheta_I(t)=j(t)/(1+j(t))$.
Then $\theta=\vartheta_I(t)$ if $w>a$, whereas
$\theta\ge\vartheta_I(t)$ if $w=a$.
Along the feasible family with $\theta=\vartheta_I(t)$,
full consumer surplus $\mathcal C$ and mean $m$ satisfy $\dd\mathcal C=\theta\,\dd m$.
\end{lemma}

\begin{proof}
The quadratic equations in \Cref{prop:separating-conditions}
are $\dot Q=Q-\varphi$ and $2\dot\varphi=Q-\theta$.
Solving backward from $Q=\varphi=1$ gives the profiles
and $\xi=2\varphi(0)-B$. They increase strictly on the separating
interval only if $\theta<1$. By \Cref{lem:shadow-quality},
$\varphi\ge\theta$, hence $T\le\pi$.
Conversely, the displayed profiles increase to one,
satisfy reconstruction, and remain nonnegative when
$\varphi(0)\ge0$.
Integration gives $M,C$. Differentiation gives the payoff
identity and, for finite $t$ at fixed $\theta$,
\begin{equation*}
 B_t=-\frac{2ht}{1+t^2}<0,
 \qquad K_B=h(t-1).
\end{equation*}

For an attached atom, contact gives $r=(w-z)/(B-z)$.
The full mean and surplus are $m=w+r(M-w)$ and
$\mathcal C=r[C+(B-w)z]$.
To account for the mean compensation, subtract $\theta(m-w)$
from full surplus to obtain
\[
 \mathcal A(w,z,B)
 =\frac{w-z}{B-z}[K(B)+(B-w)(z-\theta)].
\]
Keep $w,\theta$ fixed and vary $z$ or $B$.
Strict attachment permits either sign, and the interior direction
in \eqref{eq:binary-remote-mean} restores the mean at marginal
surplus $\theta$. Optimality therefore gives
$\mathcal A_z=\mathcal A_B=0$, or
\begin{equation*}
 \begin{gathered}
 K-(B-w)(\theta-z)=(w-z)(B-z),\\
 w-z=K_B-(\theta-z).
 \end{gathered}
\end{equation*}
Substitution gives $(\theta-z)^2=h^2$.
Since $\theta\ge z$, we obtain $z=\theta-h$ and
\eqref{eq:mult-atom-matching}. The condition $w>z$ requires $t>2$.
At $T=\pi$, small duration changes in either direction
remain feasible by strict attachment and continuity.
In the limit, the same equations would give
$z=\theta=\varphi(0)$, contradicting the strict quality gap.

Finally, varying $w$ and compensating its mean effect
has derivative $\mathcal A_w-\theta=(1-\theta)j(t)-\theta$.
Optimality makes this derivative zero when $w>a$ and
nonpositive when $w=a$, giving the last conditions.
Along the feasible interior family, the attachment conditions
and $\dd C=\theta\,\dd M-\xi\,\dd B$ give
$\dd\mathcal C=\theta\,\dd m$. Terms involving $\dd\theta$
cancel because $\mathcal A_\theta=-(m-w)$.
\end{proof}

\subsection{The two candidates}
\label{app:quadratic-global-bounds}

Put $k=e^{-\pi/2}$, $a_0=k/(1+k)$, and
$m_0=k(k+3)/(1+k)$. Define
\begin{equation}
 \theta_E(t)=\frac{\gamma(t)(t-1)}{1+\gamma(t)(t-1)},
 \qquad a_E(t)=\theta_E(t)\left[1-\frac1{(t-1)^2}\right],
 \qquad t>2.
 \label{eq:quad-exclusion-parameters}
\end{equation}
The formula for $\theta_E$ also applies for $0<t\le2$.
These formulas impose the exclusion conditions $\xi=0$ and
$C=\theta(M-a)$ from \Cref{lem:strict-optimal-exclusion}.
Write $B_E(t),M_E(t)$ for the continuation's lower support point and mean.
At prior $(a_E(t),m)$, its trading probability is
$s_E=(m-a_E)/(M_E-a_E)$ and its consumer surplus is
$\theta_E(t)(m-a_E(t))$, provided the lower atom is excluded.

For inclusion, use the matched distribution at $\theta=\vartheta_I(t)$.
Write $m_I(t)$ and $\mathcal C_I(t)$ for its mean and consumer
surplus, obtained from
$\mathfrak m=(1-r)w+rM$ and $\mathfrak c=r[C+(B-w)z]$.
The matching conditions and the continuation differential give
$\mathcal C_I'(t)=\vartheta_I(t)m_I'(t)$.

\begin{lemma}[The possible optima near the switch]
\label{lem:quadratic-global-comparison}
The functions $a_E,m_I,\vartheta_I$ strictly increase on $(2,\infty)$,
with $a_E$ ranging from zero to $a_0$ and $m_I$ from
$-\gamma(2)$ to $m_0$.
Let $0<a<a_0$ and $a<m<m_0$.
If the inclusion candidate at $m_I(t_I)=m$ is feasible with
$w_I>a$ and $z_I>0$, and the exclusion candidate at $a_E(t_E)=a$ is feasible
with a positive excluded share and strictly negative lower virtual
value, then every optimum is one of these two distributions.
The larger of their consumer surpluses equals $V_a(m)$.
If the two surpluses are equal, these are the only optimal distributions.
\end{lemma}

\begin{proof}
\emph{Regular full participation and exclusion.}
At a regular optimum, increasing $B$ at fixed $M$ is feasible,
even if $B=a$, so $\xi\ge0$ and $\theta\ge\theta_E(t)$.
The mean increases with $\theta$. Along $\xi=0$, it is
\[
 M_E(\chi)=\frac{3\sin\chi-\cos\chi+e^{-\chi}}
                   {e^\chi+\sin\chi-\cos\chi},
\]
and decreases to $m_0$ at $\chi=\pi/2$.
Indeed, its derivative has the sign of
$e^\chi\cos\chi-\sin\chi\cosh\chi-1$. This expression
starts at zero and decreases because $\tan\chi>\tanh\chi$
makes its derivative negative.
Thus regular full participation requires $m\ge m_0$.

With exclusion, \Cref{lem:strict-optimal-exclusion} gives
$\xi=0$, $\theta>0$, and $C=\theta(M-a)$, hence
\eqref{eq:quad-exclusion-parameters} with $t>2$.
Here
$\theta_E'=2\gamma(t)/\{(1+t^2)[1+(t-1)\gamma(t)]^2\}>0$,
so $a_E$ increases from zero to $a_0$.
The endpoint $t=\infty$ requires $a=a_0$.
Consequently $a<a_0$ identifies a unique exclusion candidate,
with $M_E>m_0$.

\emph{Inclusion.}
To express the mean as an affine function of the multiplier,
write $\gamma=\gamma(t)$ and put
$w_0=(t-3)\gamma$, $p=r\gamma^2(1+t^2)$, and
$A=w_0+j+p$, where $j$ is defined above.
Then $\mathfrak m=A+\theta(1-A)$ and
$m_I=(A+j)/(1+j)$.
We have $A<1$: it is the convex combination of
$w_0<1$ and $2\gamma t+\gamma^2(1+t^2)<1$.
The latter inequality is equivalent to
$\sinh\chi>\sin\chi$.
Differentiation gives
\[
 w_0'=\frac{2(t+2)\gamma}{1+t^2}>0,\quad
 j'=\frac{10(t^2+2t+2)\gamma}{(1+t^2)(t+2)^2}>0,\quad
 p'=\frac{2(t^2+6)\gamma^2}{(t+2)^2}>0.
\]
Thus $A'>0$, $\vartheta_I'=j'/(1+j)^2>0$, and
$m_I'=A'/(1+j)+(1-A)j'/(1+j)^2>0$.
The stated limits follow from $t\gamma(t)\to k$.

An interior trading atom must therefore have $t=t_I$.
An atom at $w=a$ would instead require $\theta\ge\vartheta_I(t)$,
so $m\ge m_I(t)$ and $t\le t_I$.
At a fixed mean, however, the lower atom's value is
\[
 w_m(t)=m-(1-m)\frac{j+p}{1-A},
\]
which strictly decreases because $j+p$ increases and $1-A$
decreases. Thus $w_m(t)\ge w_m(t_I)=w_I>a$, contradicting
an atom at the lower bound.
Existence and \Cref{thm:binary-structure}(i) now leave only
the two feasible candidates, proving the result.
\end{proof}

\subsection{Proof of \texorpdfstring{\Cref{prop:participation-switch}}{Proposition~\ref{prop:participation-switch}}}
\label{app:local-quadratic-regime-selection}

\begin{proof}
Write $\mathcal C_I(t)$ for the inclusion candidate's
consumer surplus. For the exclusion candidate indexed
by $u$, the difference in consumer surplus is
\[
 \Delta(t,u)
 =\theta_E(u)[m_I(t)-a_E(u)]-\mathcal C_I(t).
\]
The payoff differentials established above give
$\mathcal C_I'=\vartheta_I m_I'$ and
$\theta_E a_E'=(M_E-a_E)\theta_E'$.
Consequently,
\[
 \frac{\dd}{\dd t}\Delta(t,t)
 =(\theta_E-\vartheta_I)m_I'
   +\theta_E'(m_I-M_E).
\]
Using $j=\gamma(t)(t-1)-4\gamma(t)/(t+2)$,
the candidate formulas give, as $t\to\infty$,
\[
 \begin{gathered}
 \theta_E-\vartheta_I=
 \frac{4k}{(1+k)^2t^2}+O(t^{-3}),\qquad
 m_I'=\frac{2k}{t^2}+O(t^{-3}),\\
 \theta_E'=\frac{2k}{(1+k)^2t^3}+O(t^{-4}),\qquad
 m_I-M_E=-\frac{4k}{t}+O(t^{-2}).
 \end{gathered}
\]
The leading terms cancel, so
$\frac{\dd}{\dd t}\Delta(t,t)=O(t^{-5})$.
Since $\Delta(t,t)\to0$, integration gives
$\Delta(t,t)=O(t^{-4})$. Holding $t$ fixed, we now vary $u$:
\[
 \Delta_u(t,u)=\theta_E'(u)[m_I(t)-M_E(u)]
 =-\frac{8k^2}{(1+k)^2t^4}+O(t^{-5})
\]
uniformly for bounded $u-t$.
Thus a sufficiently large fixed $L$ gives
$\Delta(t,t-L)>0>\Delta(t,t+L)$ for all sufficiently large $t$.
The intermediate value theorem then supplies
$u_t=t+O(1)$ with $\Delta(t,u_t)=0$.

We now check feasibility at the prior $a=a_E(u_t)$, $m=m_I(t)$.
For large $t$, $0<a<a_0<m<m_0$.
The inclusion atom has probability
$1-r=4/t+O(t^{-2})$, value $w_I$ with
$w_I-a\to a_0>0$, and virtual value $z_I\to a_0>0$.
Its support and virtual-value gaps are $4h>0$ and $2h>0$.
For exclusion,
\[
 M_E(u_t)-m=\frac{4k}{t}+O(t^{-2}),\qquad
 s=1-\frac{4(1+k)}{(k+2)t}+O(t^{-2}).
\]
Thus $0<s<1$, and the lower virtual value is negative
because its numerator $a-sB_E$ tends to $-a_0<0$.
The continuation's bottom virtual value and its support
gap both tend to $a_0>0$.
Both candidates are therefore feasible, and
\Cref{lem:quadratic-global-comparison} identifies
them as the only optima.

The multiplier difference remains
$\theta_E(u_t)-\vartheta_I(t)
=4k/[(1+k)^2t^2]+O(t^{-3})>0$.
The continuation multipliers and durations differ by $O(t^{-2})$,
so the conditional profits differ by $O(t^{-2})$. Their common
limit is
$P^\circ=a_0^2(8+\pi)/4$, obtained by integrating
$\varphi^2/2$ under the limiting continuation.
Including the group probabilities and the lower atom's profit,
full profits satisfy
\[
 P_E-P_I
 =(s-r)P^\circ-(1-r)\frac{a_0^2}{2}+O(t^{-2})
 =\frac{a_0^2(4+\pi-2k)}{(k+2)t}+O(t^{-2})>0.
\]

Fix one such prior. Since $m_I'>0$ and $a_E'>0$,
the two candidates vary smoothly on a neighborhood.
Their strict feasibility inequalities persist there.
Let $t_I(m),t_E(a)$ be the local inverses and put
$\rho_E(a)=\theta_E(t_E(a))$.
\Cref{lem:quadratic-global-comparison} therefore gives
$V_a(m)=\max\{C_I(m),C_E(a,m)\}$, with
$C_I(m)=\mathcal C_I(t_I(m))$ and
$C_E(a,m)=\rho_E(a)(m-a)$.
The inclusion atom remains above $a$, so the inclusion
candidate's value is independent of $a$.
Write $s_E=(m-a)/(M_E(t_E(a))-a)$.
The payoff differentials give
\[
 C_I'(m)=\vartheta_I(t_I(m)),\qquad
 (C_E)_m=\rho_E(a),\qquad
 (C_E)_a=-\rho_E(a)(1-s_E)<0.
\]
At the tie, $C_I-C_E$ therefore decreases in $m$ and increases
in $a$. The implicit-function theorem gives the increasing
curve $m^\dagger(a)$ and \eqref{eq:participation-boundary-slope}.
On a smaller rectangle, inclusion is uniquely optimal
below this curve, exclusion uniquely optimal above it,
and only the two candidates are optimal on it.
Continuity of the candidates' payoffs and persistence of the
strict profit ordering give the stated one-sided increases in
profit and total surplus. Participation falls from one to
$s_E<1$, while consumer surplus remains continuous as the maximum
of the two continuous values.
\end{proof}

\subsection{Proof of Corollary~\ref{cor:inclusion-convexity}}
\label{app:quadratic-convexity-proof}

\begin{proof}
Since $\vartheta_I'>0$ and $m_I'>0$,
the preceding derivative identity gives
\[C_I''(m)=\vartheta_I'(t_I(m))/m_I'(t_I(m)) >0.\]
On the inclusion side, $V_a(m)=C_I(m)$.
\end{proof}

\section{Market segmentation}
\label{app:segmentation}

For binary priors, we prove the results on valuation-independent segmentation and on the joint design of learning and valuation-dependent segmentation.
\subsection{Implementation and continuity}
\label{app:seg-preliminaries}

\subsubsection{Attainable payoffs and the number of segments}
\label{app:seg-geometry}

For fixed cost and prior $H$, $\mathcal G(H)$ contains the
point mass at the prior mean and is a weakly closed convex
subset of $\Delta([0,1])$, hence is weakly compact.
Payoff continuity
\citep[Proposition~A.7]{kyriazis2026markets}
makes $\mathcal P(H)$ compact. Since mixtures of feasible
distributions remain feasible and profit varies continuously
along them, the attainable profit set
$\mathcal I_H:=\{\Pi(G):G\in\mathcal G(H)\}$
is a compact interval.

By \eqref{eq:seg-payoff-hull}, the convex-hull
argument in \citet[proof of Theorem~3]{kyriazis2026markets}
shows that $\mathcal P^{\mathrm{seg}}(H)$ is compact
and every boundary payoff, exposed or not, can
be implemented with at most two segments.
In both payoff sets, compactness ensures that
consumer surplus attains a maximum when profit
is fixed or bounded below, provided the
constraint is feasible. These conclusions also hold with positive
initial marginal cost, using the continuity extension in
Appendix~\ref{app:neutral-screening-foundation}.

\subsubsection{Optimal state-independent segmentation}
\label{app:seg-support-gaps}

For $p\in\mathcal I_H$, define the maximal consumer surplus
under a minimum-profit requirement by
\[
 \widehat\Gamma_H(p)
 :=\max_{\substack{G\in\mathcal G(H)\\\Pi(G)\ge p}}\CS(G),
 \qquad
 \widehat\Gamma_H^{\mathrm{seg}}(p)
 :=\max_{\substack{(C,P)\in\mathcal P^{\mathrm{seg}}(H)\\P\ge p}}C.
\]
For a binary prior, $\widehat\Gamma_H=\widehat\Gamma_m^a$.
For each finite $\beta\ge0$, let
\[
 \mathcal O_H(\beta):=\argmax_{G\in\mathcal G(H)}J_\beta(G),
 \qquad
 \mathscr P_H(\beta):=\{\Pi(G):G\in\mathcal O_H(\beta)\}.
\]
Both sets are nonempty and compact. Unlike $\mathcal I_H$,
however, the profits attained by maximizers,
$\mathscr P_H(\beta)$, need not form an interval.

\begin{proposition}[Optimal state-independent segmentation]
\label{prop:independent-segmentation-characterization}
\label{lem:seg-supported-profit-gap}
\begin{enumerate}
 \setlength{\itemsep}{4pt}
 \setlength{\parsep}{0pt}

 \item[(i)] \emph{Weighted objectives.}
 For every finite $\beta\ge0$, the maximum of $C+\beta P$
 under state-independent segmentation is $W_H(\beta)$.
 A segmentation attains this value if and only if every
 positive-probability segment uses a distribution in
 $\mathcal O_H(\beta)$.

 \item[(ii)] \emph{Supported payoffs and profit gaps.}
 Fix a finite $\beta>0$ and let
 $\underline P_\beta:=\min\mathscr P_H(\beta)$ and
 $\overline P_\beta:=\max\mathscr P_H(\beta)$.
 The segmented payoff pairs maximizing $C+\beta P$ form
 the line segment
 \[
  \left\{\bigl(W_H(\beta)-\beta p,p\bigr):
       p\in[\underline P_\beta,\overline P_\beta]\right\}.
 \]
 Every pair on this line segment is Pareto efficient.
 All can be implemented using the same minimum- and
 maximum-profit optimizers, varying only their probabilities.
 For each such $p$,
 \[
  \widehat\Gamma_H^{\mathrm{seg}}(p)=W_H(\beta)-\beta p,
  \qquad
  \widehat\Gamma_H(p)=\widehat\Gamma_H^{\mathrm{seg}}(p)
  \ \Longleftrightarrow\ p\in\mathscr P_H(\beta).
 \]
 One segment suffices at profits in $\mathscr P_H(\beta)$.
 At the remaining profits, two segments are necessary and
 sufficient, and segmentation gives strictly more consumer
 surplus than any single-market policy earning at least $p$.

 \item[(iii)] \emph{The minimum-profit boundary.}
 On $\mathcal I_H$, $\widehat\Gamma_H^{\mathrm{seg}}$
 is the least concave nonincreasing majorant of
 $\widehat\Gamma_H$. Any strict gain in part~(ii)
 therefore implies that $\widehat\Gamma_H$ is nonconcave.
\end{enumerate}
\end{proposition}

\begin{proof}
For any state-independent segmentation,
$\sum_k\lambda_kJ_\beta(G_k)
 \le\sum_k\lambda_kW_H(\beta)=W_H(\beta)$.
A single-market optimizer attains this bound.
Equality holds if and only if every positive-probability
segment is itself a maximizer, proving part~(i).

For part~(ii), each maximizing segment has payoff
$\bigl(W_H(\beta)-\beta\Pi(G_k),\Pi(G_k)\bigr)$.
Their weighted averages lie on the line segment
in part~(ii). Conversely, if the endpoint profits
differ, choose an optimizer at each endpoint and
assign probability
\[
 \lambda(p)=\frac{\overline P_\beta-p}
                  {\overline P_\beta-\underline P_\beta}
\]
to the one with lower profit and the remaining
probability to the other. This attains the stated
pair with profit $p$. If the endpoint profits
are equal, one segment suffices.
By part~(i), every implementation must use
maximizing distributions in all segments with
positive probability. Since both weights in
the objective are positive, every pair on the
line segment is Pareto efficient.

For a profit requirement $p$, every segmented
pair with $P\ge p$ satisfies
$C\le W_H(\beta)-\beta P\le W_H(\beta)-\beta p$.
The proposed segmentation attains this bound.
Because $\beta>0$, a single market attains it
if and only if it earns exactly $p$ and uses
a distribution in $\mathcal O_H(\beta)$.
Equivalently, $p\in\mathscr P_H(\beta)$.
The single-market maximum is attained, so it
must be strictly below the bound whenever
$p\notin\mathscr P_H(\beta)$.
At those profits, two segments are necessary
and sufficient. At profits in
$\mathscr P_H(\beta)$, one segment suffices.

For part~(iii), mixing two segmented optima proves
concavity of $\widehat\Gamma_H^{\mathrm{seg}}$.
It is nonincreasing because raising the profit requirement
shrinks the feasible set, and it majorizes
$\widehat\Gamma_H$ because one-segment policies are feasible.

Let $g$ be any concave nonincreasing
majorant of $\widehat\Gamma_H$ on $\mathcal I_H$.
For a segmentation with segment profits $P_k$ and
aggregate profit $P=\sum_k\lambda_kP_k\ge p$,
\[
 \sum_k\lambda_k\CS(G_k)
 \le\sum_k\lambda_kg(P_k)
 \le g(P)\le g(p).
\]
The first inequality uses
$\CS(G_k)\le\widehat\Gamma_H(P_k)\le g(P_k)$.
The others use concavity and monotonicity, respectively.
Maximizing over segmentations proves
$\widehat\Gamma_H^{\mathrm{seg}}(p)\le g(p)$.

Finally, part~(ii) shows that, for any
$p\in[\underline P_\beta,\overline P_\beta]
       \setminus\mathscr P_H(\beta)$,
$\widehat\Gamma_H(p)$ lies strictly below the chord
joining its values at $\underline P_\beta$ and
$\overline P_\beta$.
It therefore cannot be concave.
\end{proof}

\subsubsection{Changing the primitive prior}
\label{app:seg-prior-continuity}

We first show that a small change in the prior permits a
comparably small adjustment to any feasible posterior
distribution. Let $d_1$ denote the first Wasserstein distance:
the minimum of $\E|X-\widetilde X|$ over joint distributions
with the specified marginals. On $[0,1]$, this distance
metrizes weak convergence
\citep[Theorem~6.9]{villani2009}.
For any priors $H,\widetilde H$ and any $G\preceq_\cx H$,
there exists $\widetilde G\preceq_\cx\widetilde H$ such that
\begin{equation}
 d_1(G,\widetilde G)\le d_1(H,\widetilde H).
 \label{eq:seg-prior-projection}
\end{equation}
To construct this approximation, take an implementation of $G$ under $H$:
$Y\sim G$, $V\sim H$, and $\E[V\mid Y]=Y$.
Choose a coupling of $H$ and $\widetilde H$ attaining
$d_1(H,\widetilde H)$, and draw $\widetilde V$ from its
conditional distribution given $V$, independently of $Y$
conditional on $V$. Set
$\widetilde Y=\E[\widetilde V\mid Y]$ and let
$\widetilde G$ be its distribution.
Jensen's inequality ensures feasibility:
$\widetilde G\preceq_\cx\widetilde H$. Moreover,
$d_1(G,\widetilde G)
 \le \E|Y-\widetilde Y|
 =\E\bigl|\E[V-\widetilde V\mid Y]\bigr|
 \le \E|V-\widetilde V|
 =d_1(H,\widetilde H)$.
Thus, the construction approximates any specified feasible
posterior distribution, whether or not it is optimal.

\begin{lemma}[Continuity in the prior and profit weight]
\label{lem:seg-payoff-continuity}
Fix the maintained cost function. As $H$ varies weakly,
$\mathcal G(H)$ and $\mathcal P(H)$ vary continuously in
Hausdorff distance, using $d_1$ for posterior distributions
and Euclidean distance for payoff pairs.
For binary priors, $V_{a,\beta}(m)$ is continuous in
$(a,m,\beta)$ on $0\le a<m<1$ and finite $\beta\ge0$.
The optimizer correspondence
$\mathcal O_B(a,m,\beta)
 :=\argmax_{G\in\mathcal G_{a,m}}J_\beta(G)$
has nonempty compact values and is upper hemicontinuous.
In particular, if parameters converge within this domain,
every weak limit of corresponding optimal distributions
is optimal at the limiting parameters.
\end{lemma}

\begin{proof}
Applying \eqref{eq:seg-prior-projection} in both directions
bounds the Hausdorff distance between $\mathcal G(H)$
and $\mathcal G(\widetilde H)$ by $d_1(H,\widetilde H)$,
proving continuity of the feasible sets.

The payoff map $G\mapsto(\CS(G),\Pi(G))$ is continuous
by \citet[Proposition~A.7]{kyriazis2026markets}, with
the extension to positive initial marginal cost given
in Appendix~\ref{app:neutral-screening-foundation}.
Since $\Delta([0,1])$ is compact, this continuity is
uniform. The payoff images $\mathcal P(H)$ therefore
also vary continuously in Hausdorff distance.

Binary priors depend continuously on $(a,m)$, and
$J_\beta(G)$ is jointly continuous in $(G,\beta)$.
The feasible sets vary continuously and are
nonempty and compact, so Berge's maximum theorem
gives continuity of $V_{a,\beta}(m)$.
The optimizer correspondence $\mathcal O_B$
has nonempty compact values and is upper
hemicontinuous. Thus, along any parameter
sequence converging within the stated domain,
every weak limit of optimal distributions
is optimal at the limiting parameters.
\end{proof}

\subsection{Gains near the participation crossing}
\label{app:seg-gain-proof}

The next lemma shows that, at nearby lower means, a positive
weight on profit can make both low- and high-profit distributions
globally optimal. No optimizer has profit in a fixed interval
between the two groups.

To this end, fix the maintained cost function and $a<m_*<1$. Suppose
the buyer-optimal payoff set at $m_*$ is
$\{(C_*,P_L^*),(C_*,P_H^*)\}$, where $P_L^*<P_H^*$.
For every $m<m_*$ sufficiently close to $m_*$, suppose
the buyer-optimal payoff pair is unique and its profit
converges to $P_L^*$ as $m\uparrow m_*$.

\begin{lemma}[A gap in maximizing profits]
\label{lem:seg-separated-supports}

There exist $\eta>0$ and
$P_L^*<p_-<p_+<P_H^*$ such that, for every
$m\in(m_*-\eta,m_*)$, some $\beta_m\in(0,1)$
admits global maximizers with profits below $p_-$
and above $p_+$, but none with profit in $[p_-,p_+]$.
\end{lemma}

\begin{proof}
Choose $P_L^*<p_-<p_+<P_H^*$.
By \Cref{lem:seg-payoff-continuity}, upper hemicontinuity
of the optimizer correspondence and continuity of profit
give a neighborhood of $m_*$ and $\bar\beta\in(0,1)$
such that
\begin{equation}
 \{\Pi(G):G\in\mathcal O_B(a,m,\beta)\}
 \cap[p_-,p_+]=\varnothing
 \quad\text{for nearby }m,\quad
 0\le\beta\le\bar\beta.
 \label{eq:seg-uniform-profit-gap}
\end{equation}
Otherwise, optimizers at parameters converging to
$(a,m_*,0)$ would have a buyer-optimal limit with
profit in $[p_-,p_+]$, contrary to the hypothesis.

At mean $m_*$, every optimizer for $\beta>0$ has
profit at least $P_H^*$. Indeed, its consumer surplus
cannot exceed $C_*$, and it must weakly outperform
the buyer-optimal pair $(C_*,P_H^*)$.
Upper hemicontinuity at $(a,m_*,\bar\beta)$ therefore
places all nearby optimizers above $p_+$.
At weight zero, the assumed convergence of buyer-optimal
profits places all optimizers below $p_-$ for means
just below $m_*$. Choose $\eta>0$ small enough that
both conclusions and \eqref{eq:seg-uniform-profit-gap}
hold on $(m_*-\eta,m_*)$.

To find a weight with global maximizers in both profit groups, fix a
mean $m$ in this interval. Define the best objective values
within the two groups by
\[
 w_-(\beta)
 :=\max_{\substack{G\in\mathcal G_{a,m}\\\Pi(G)\le p_-}}
       J_\beta(G),
 w_+(\beta)
 :=\max_{\substack{G\in\mathcal G_{a,m}\\\Pi(G)\ge p_+}}
       J_\beta(G).
\]
Both feasible sets are compact and nonempty, as they
contain optimizers at weights zero and $\bar\beta$,
respectively. Since these sets are fixed as $\beta$ varies and profit is bounded,
both value functions are Lipschitz continuous.
By \eqref{eq:seg-uniform-profit-gap}, every global
optimizer belongs to one of these two sets. Hence
$V_{a,\beta}(m)=\max\{w_-(\beta),w_+(\beta)\}$
for $0\le\beta\le\bar\beta$.
At weight zero all optimizers belong to the first
group, whereas at $\bar\beta$ all belong to the
second. Since the restricted maxima are attained,
$w_-(0)>w_+(0)$ and $w_-(\bar\beta)<w_+(\bar\beta)$.
By continuity, there is a $\beta_m\in(0,\bar\beta)$
such that
$w_-(\beta_m)=w_+(\beta_m)=V_{a,\beta_m}(m)$.
Any optimizer of either restricted problem is
therefore globally optimal. By
\eqref{eq:seg-uniform-profit-gap}, the first
problem's optimizers earn less than $p_-$,
and the second's earn more than $p_+$.
No global optimizer has profit in $[p_-,p_+]$.
\end{proof}

\paragraph{The participation regimes persist.}
Under quadratic cost, the two profit ranges correspond to
distorted inclusion and exclusion. Payoff continuity alone
does not establish this correspondence, since positive
qualities can converge to zero. The next lemma uses the
structure of optimal distributions to rule this out.

\begin{lemma}[Persistence of the two participation regimes]
\label{lem:seg-participation-persistence}
Under quadratic cost, fix a crossing $(a_*,m_*,0)$ from
\Cref{prop:participation-switch}, and let $P_I^*<P_E^*$
be the profits of its inclusion and exclusion optima.
Fix $P_I^*<p_-<p_+<P_E^*$. In a sufficiently small neighborhood of $(a_*,m_*,0)$,
every optimum with $0\le\beta<1$ either has distorted
inclusion and profit below $p_-$, or has exclusion and
profit above $p_+$. The excluded share in the latter
case is uniformly bounded away from zero.
\end{lemma}

\begin{proof}
By \Cref{prop:participation-switch}, the only optimal
distributions at the crossing are $G_I^*$ and $G_E^*$.
The inclusion distribution has a positive lower atom strictly
above $a_*$, separated from its upper support by a gap.
The exclusion distribution is
$G_E^*=(1-s_E)\delta_{a_*}+s_EG_E^A$, where
$0<s_E<1$ and the active distribution has bottom
value $B_E>a_*$. Strict exclusion gives
$a_*<s_EB_E$ by \Cref{lem:strict-optimal-exclusion}.

By compactness of $\Phi$ and continuity of
reconstruction, the regular distributions obtained
from $\Phi$ form a weakly compact set
(Appendix~\ref{app:neutral-active-analysis}).
Their quantiles extend continuously to $[0,1]$,
so each has connected support. Both $G_I^*$ and
$G_E^*$ have a support gap and therefore lie
outside this set.

Take a sequence of optimal distributions whose
parameters converge to $(a_*,m_*,0)$. Passing to
a subsequence, we may assume that the distributions
converge weakly. By \Cref{lem:seg-payoff-continuity},
their limit is either $G_I^*$ or $G_E^*$.

First suppose the limit is $G_I^*$. There cannot
be a subsequence of regular optima without a lower
atom, since a weak limit of regular distributions
is regular. Nor can there be a subsequence with
exclusion. Along such a subsequence, the probability
at $a_n\to a_*$ must tend to zero, since $G_I^*$
has no mass at $a_*$. The conditional distributions
among trading buyers would then also converge to
$G_I^*$. These distributions are regular, giving
the same contradiction. By the binary
characterization, the optima must therefore have
distorted inclusion for all sufficiently large
indices.

Next suppose the limit is $G_E^*$, and assume along
a subsequence that every buyer trades. Write
$G_n=(1-r_n)\delta_{w_n}+r_nG_{\varphi_n}$,
$B_n=\int_0^1\varphi_n(u)\,\dd u$ and $w_n\le B_n$.
When there is no lower atom, take $r_n=1$ and
$w_n=B_n$. Otherwise, revenue contact and the
nonnegative lower virtual value $z_n$ give
$w_n=(1-r_n)z_n+r_nB_n\ge r_nB_n$.

Passing to a further subsequence yields
\[
 G_E^*=(1-r)\delta_w+rG_\varphi,
 \qquad
 w\le B:=\int_0^1\varphi(u)\,\dd u,
 \qquad w\ge rB.
\]
Here $r=0$ would make the limit a point mass, while
$r=1$ would make it regular. Thus, $0<r<1$.
Since $w\le B$, the lowest support point is
$w=a_*$. All mass above the gap must come from $G_\varphi$.
Since its support is connected, it cannot also
contain $a_*$. Comparing the probabilities at
$a_*$ gives $r=s_E$. It follows that
$G_\varphi=G_E^A$ and $B=B_E$.
But the inequality $w\ge rB$ would then imply
$a_*\ge s_EB_E$, contrary to strict exclusion.
Thus any sequence of optima converging to $G_E^*$
must have exclusion for all sufficiently large indices.

To bound the excluded share away from zero, consider
a sequence of optima with exclusion,
$(1-s_n)\delta_{a_n}+s_nG_{\varphi_n}$,
converging to $G_E^*$. By compactness, we can pass
to a subsequence along which both $s_n$ and the
regular components converge. The preceding support
argument gives $s_E$ as the limit of $s_n$.
Since this holds for every convergent subsequence,
$s_n\to s_E$.

Payoff continuity now gives the claimed profit
bounds for the two participation regimes in a
sufficiently small neighborhood. We can further
restrict this neighborhood so that every optimum
with exclusion leaves at least
$(1-s_E)/2>0$ of buyers unserved.
\end{proof}

\subsection{Gains with full participation}
\label{app:seg-full-participation}

Throughout this subsection, let $c(x)=\kappa x^\nu/\nu$, where
$\kappa>0$, $\nu\ge2$, and $\bar Q>\kappa^{-1/(\nu-1)}$.
Put $r=1/(\nu-1)$ and $A=\kappa^{-r}$. Then
\[
 q(z)=Az^r,\qquad \pi(z)=\frac{Az^{r+1}}{r+1},\qquad
 \frac{c(q(z))}{q(z)}=\frac z\nu,\quad0<z\le1.
\]
In particular, $q$ is concave and the maintained curvature condition
holds at every $\beta<1$.

\begin{lemma}[Full participation under power costs]
\label{lem:seg-always-full-participation}
If $1/\nu\le a<m<1$, every maximizer of $J_\beta$
over $\mathcal G_{a,m}$ has full participation for
every finite $\beta\ge0$.
\end{lemma}

\begin{proof}
First suppose $0\le\beta<1$. An optimum with exclusion
would have a regular active distribution with bottom
posterior value $B<1$ and bottom virtual value $z>0$.
By \Cref{lem:strict-optimal-exclusion},
\[
 B=z\frac{r+1-\beta}{r},
 \qquad
 a\le z\left(1-\frac{\beta}{r+1}\right)
   =\frac{B}{\nu}<\frac1\nu,
\]
contrary to the assumed lower bound on $a$.

For $\beta\ge1$, we show that any feasible distribution
with exclusion admits a garbling that strictly raises
both profit and total surplus. Let $s<1$ be its trading
share, $e$ the mean posterior value among excluded
buyers, and $Y_A$ the posterior conditional on trade.
The seller can earn $\pi(a)>0$ by offering quality
$q(a)$ at price $a q(a)$, so $s>0$.
Monotonicity of quality gives $a\le e\le m<1$.

The trading cutoff $u_0=1-s$ is a revenue contact.
Otherwise it would lie inside an interval of strict
concavification, where virtual value and quality are
constant, and could not separate trade from exclusion.
Conditioning on trade therefore gives the revenue
envelope
$R_A(u)=R_Q(u_0+su)/s$, with the same virtual values
$Z_A>0$ as on the original trading ranks.
The original profit and total surplus are consequently
$s$ times the corresponding conditional payoffs.

To construct the garbling, retain each trading signal and
replace every excluded signal by an independent draw from the
conditional distribution of trading signals.
The resulting posterior distribution is that of
$\widetilde Y=(1-s)e+sY_A$.
This is a garbling of the original information policy
and hence remains feasible.
Concavification commutes with positive scaling and
addition of an affine function, so the new revenue
envelope and virtual values are
$\widetilde R(u)=(1-s)e(1-u)+sR_A(u)$ and $\widetilde Z=(1-s)e+sZ_A>0$.
Thus, the garbling serves everyone.

It remains to compare payoffs. For an original active
pair $0<z\le y\le1$, define
$w_t=e+t(z-e)$, $y_t=e+t(y-e)$, and
$S_t=\pi(w_t)+(y_t-w_t)q(w_t)$.
At $t=1$ these give the original conditional payoffs.
At $t=s$ they give the new payoffs.
We show that $\pi(w_t)/t$ and $S_t/t$ strictly decrease
on $[s,1]$. These ratios account for the original trading
share $s$ in the comparison.

Since $e<1$, we have $0<w_t<1$ for $t<1$.
For such $t$, the power cost formulas give
\[
 t^2\frac{\dd}{\dd t}\frac{\pi(w_t)}{t}
 =q(w_t)\left(\frac{w_t}{\nu}-e\right)<0,
\]
where the strict sign follows from $e\ge1/\nu$.
For total surplus,
\[
 t^2\frac{\dd}{\dd t}\frac{S_t}{t}
 =(w_t-e)\{q(w_t)+(y_t-w_t)q'(w_t)\}
   -\pi(w_t).
\]
This is strictly negative when $w_t\le e$.
When $w_t>e$, use $y_t-w_t\le1-w_t$,
$e\ge r/(r+1)$, and the power-cost formulas to obtain
\[
 t^2\frac{\dd}{\dd t}\frac{S_t}{t}
 \le A w_t^{r-1}
 \left[
 \left(w_t-\frac{r}{r+1}\right)
       \{w_t+r(1-w_t)\}
       -\frac{w_t^2}{r+1}
 \right]
 =-\frac{Ar^2}{r+1}w_t^{r-1}(1-w_t)^2<0.
\]
Hence $\pi(w_s)>s\pi(z)$ and $S_s>sS_1$.
Averaging over the conditional active distribution
shows that the garbling strictly raises both profit
and total surplus. It therefore strictly raises
$J_\beta=\TS+(\beta-1)\Pi$ for every $\beta\ge1$,
ruling out exclusion at an optimum.
\end{proof}

\begin{proposition}[Revelation and no learning under power costs]
\label{prop:seg-power-full-participation}
Fix $1/\nu<a<2/(\nu+1)$. For every $m>a$ sufficiently
close to $a$, there is a weight $\beta^*\in(1,2)$
at which the only maximizers of $J_{\beta^*}$ over
$\mathcal G_{a,m}$ are full revelation $H$ and no
learning $\delta_m$. Full revelation yields strictly
positive consumer surplus and strictly less profit
than no learning. These conclusions continue to hold under small changes in
$(a,m)$, with $\beta^*$ allowed to vary, and hence
hold on a nonempty open set of priors. At these
priors, every optimum has full participation for
every finite $\beta\ge0$.
\end{proposition}

\begin{proof}
By \Cref{lem:seg-always-full-participation}, every
optimizer has full participation. We first compare
such distributions with a family supported on two
posterior values.

For $a\le b\le m$, define
\begin{equation*}
 \tau_b=\frac{m-b}{1-b},\qquad
 z_b=\frac{b-\tau_b}{1-\tau_b},\qquad
 F_b=(1-\tau_b)\delta_b+\tau_b\delta_1.
\end{equation*}
The distribution $F_b$ is feasible because its mean
is $m$ and its support lies in $[a,1]$.
The endpoints are $F_a=H$ and $F_m=\delta_m$.
When $z_b>0$, its virtual values are $z_b$ and one,
and its payoffs are
\begin{equation*}
 C_b=(m-b)q(z_b),\qquad
 P_b=(1-\tau_b)\pi(z_b)+\tau_b\pi(1),\qquad
 P_N=\pi(m).
\end{equation*}

\paragraph{Reduction to two posterior values.}
Let $G$ be any fully participating distribution with
mean $m$ and lowest posterior value $b$.
If $b=m$, then $G=\delta_m$. Suppose $b<m$,
$z_b>0$, $1<\beta<2$, and
\begin{equation}
 (b-z_b)q'(z_b)
 \le(\beta-1)[q(1)-q(z_b)].
 \label{eq:seg-two-posterior-condition}
\end{equation}
We show that $J_\beta(G)\le J_\beta(F_b)$,
with equality only for $G=F_b$.

Full participation gives $0<\phi_Q(u)\le Q(u)$
almost everywhere, while the revenue endpoints imply
$\int_0^1\phi_Q(u)\,\dd u=b$. To use this identity in
bounding the objective, define
\[
 \lambda=(b-z_b)q'(z_b)+(\beta-1)q(z_b),
 \qquad
 g(y)=\max_{0\le z\le y}
 \{(y-z)q(z)+\beta\pi(z)-\lambda z\}.
\]
For each posterior value, the expression maximized in $g$
is strictly concave in $z$: its second derivative is
$(y-z)q''(z)-(2-\beta)q'(z)<0$ for $0<z<y$.
The choice of $\lambda$ makes $z_b$ its unique
maximizer at $y=b$, while
\eqref{eq:seg-two-posterior-condition} makes one
its unique maximizer at $y=1$.

We next establish strict convexity
in the posterior value. Concavity of $q$ and $\beta<2$ imply
$\lambda>(\beta-1)q(b)$, whereas
\eqref{eq:seg-two-posterior-condition} gives
$\lambda\le(\beta-1)q(1)$.
There is therefore a unique $y_0\in(b,1]$ satisfying
$\lambda=(\beta-1)q(y_0)$.
The maximizing virtual value $z(y)$ is positive
and interior for $b\le y<y_0$, and equals $y$
for $y\ge y_0$. Implicit differentiation and the
envelope theorem give
\[
 g''(y)=
 \begin{cases}
 \displaystyle
 \frac{[q'(z(y))]^2}
 {(2-\beta)q'(z(y))-(y-z(y))q''(z(y))}>0,
       & b<y<y_0,\\[6pt]
 \beta q'(y)>0,
       & y_0<y<1.
 \end{cases}
\]
The values and first derivatives agree at $y_0$,
where $g'(y_0)=q(y_0)$. Together with the positive second
derivatives, this proves that $g$ is strictly convex.

Using $\int_0^1\phi_Q=b$ and bounding $g$ by the
chord joining its values at $b$ and one gives
\[
 J_\beta(G)
 \le \lambda b+\E_G[g(Y)]
 \le \lambda b+(1-\tau_b)g(b)+\tau_b g(1)
 =J_\beta(F_b).
\]
The last equality uses
$(1-\tau_b)z_b+\tau_b=b$.
Strict convexity makes the second inequality strict
unless $G$ is supported on $\{b,1\}$, in which case
its mean requires $G=F_b$.

\paragraph{Equal values under full revelation and no learning.}
For $t\in(0,1)$, write
$\bar q_t=[\pi(1)-\pi(t)]/(1-t)$, and set
$\beta_0=q(a)/[2q(a)-\bar q_a]$.
Since $a>1/\nu$, we have
$(1-a)q'(a)<q(a)$. Since $q$ is increasing and concave,
\[
 q(a)<\bar q_a
 \le q(a)+\tfrac12(1-a)q'(a)
 <\tfrac32q(a),
\]
so $1<\beta_0<2$.
Writing $b=m-h$ gives
$z_b=m-2h-h^2/(1-m)$.
For small $h\ge0$, uniformly for $m$ near $a$,
\[
 C_{m-h}=h q(m)+O(h^2),\qquad
 P_N-P_{m-h}
 =h[2q(m)-\bar q_m]+O(h^2).
\]
Taking $h=m-a$ in these expansions shows that, for $m>a$ sufficiently
close to $a$, full revelation has $C_a>0$ and
$P_a<P_N$. The weight
\[
 \beta^*:=\frac{C_a}{P_N-P_a}
\]
then makes $F_a$ and $F_m$ yield the same objective
value, and $\beta^*\to\beta_0$ as $m\downarrow a$.

Uniformly over $b\in[a,m]$, we also have
$z_b\to a>0$ and $b-z_b\to0$.
Thus \eqref{eq:seg-two-posterior-condition} holds
strictly for every $b\in[a,m]$ when $\beta$ is
close to $\beta_0$ and $m>a$ is sufficiently
close to $a$. Every fully participating distribution
$G$ therefore satisfies $J_\beta(G)\le J_\beta(F_b)$,
where $b$ is its lowest posterior value.

\paragraph{Ruling out interior members of the family.}
The payoff formula
$f(b;m,\beta):=J_\beta(F_b)$ has a $C^2$ extension
near $(a,a,\beta_0)$. Differentiation gives
\[
 f_{bb}(a;a,\beta_0)
 =2(\beta_0-1)
   \left[2q'(a)-\frac{q(a)}{1-a}\right]>0,
\]
where the strict inequality is equivalent to
$a<2/(\nu+1)$.
Thus, for nearby means and weights, $f$ is strictly
convex in $b$ on $[a,m]$.
At $\beta^*$, full revelation and no learning give
the same objective value. Strict convexity implies
that every $F_b$ with $a<b<m$ gives a strictly lower
value. The bound for fully participating distributions
and \Cref{lem:seg-always-full-participation} therefore
give $J_{\beta^*}(G)\le\beta^*\pi(m)$ for every
feasible $G$, with equality if and only if
$G\in\{H,\delta_m\}$.

For each such prior, $z_b>0$, $f_{bb}>0$, and
the strict inequality in
\eqref{eq:seg-two-posterior-condition} hold uniformly
over $b\in[a,m]$. By continuity, these inequalities
remain strict at nearby priors, as do $C_a>0$ and
$P_N>P_a$. The formula for $\beta^*$ is continuous
in $(a,m)$, so this weight also remains in $(1,2)$.
The same proof therefore applies on an open
neighborhood of the prior. Choose this neighborhood
small enough that $a>1/\nu$ throughout.
By \Cref{lem:seg-always-full-participation}, every
optimum then has full participation at every finite
nonnegative weight.
\end{proof}

\subsection{Proof of \texorpdfstring{\Cref{thm:segmentation-gain}}{Theorem~\ref{thm:segmentation-gain}}}

\begin{proof}
By \Cref{lem:seg-supported-profit-gap}, it suffices to
find a positive weight and two global optimizers with
profits $P_L<P_H$, together with
$P_L<p_-<p_+<P_H$, such that no optimizer has profit
in $[p_-,p_+]$. Assigning the two distributions to separate segments
and varying only the segment probabilities gives the
claimed Pareto-efficient pairs and strict gains in
consumer surplus. At $P_L$ and $P_H$,
$\widehat\Gamma^a_m$ equals the supporting line.
On $[p_-,p_+]$, it lies strictly below that line,
and hence below the chord joining these two points.
Thus $\widehat\Gamma^a_m$ is not concave.
It remains to find open sets of priors with
such profit gaps.

\paragraph{Weights above one.}
For every power cost in the theorem,
\Cref{prop:seg-power-full-participation} gives a
nonempty open set of priors at which the only
maximizers, at some $\beta^*\in(1,2)$, are full
revelation $H$ and no learning $\delta_m$.
Their profits are distinct, so any nondegenerate
closed interval strictly between them contains
no maximizing profit. The same proposition gives
full participation under every optimum at every
finite nonnegative weight throughout this open set.

\paragraph{Weights below one.}
First take $c(x)=x^2/2$. Fix $a\in\mathcal A$ and
$m_*=m^\dagger(a)$.
By \Cref{prop:participation-switch}, the crossing
has two buyer-optimal distributions, with strictly
higher profit under exclusion. The unique inclusion
optimum just below $m_*$ converges to the lower-profit
crossing pair. Hence
\Cref{lem:seg-separated-supports} gives the required
profit gap at a weight in $(0,1)$ for every mean
sufficiently close to and below $m_*$.
In that construction, choose the auxiliary upper
weight $\bar\beta<1$ and the mean neighborhood
small enough that
\Cref{lem:seg-participation-persistence} also applies.
The low-profit optima then have distorted inclusion,
and the high-profit optima have exclusion.

To obtain an open set of priors, fix one such lower
mean $m_0$. At $(a,m_0)$, all optimizers at weight
zero have profit below $p_-$, all optimizers at
$\bar\beta$ have profit above $p_+$, and no optimizer
has profit in $[p_-,p_+]$ at any weight in
$[0,\bar\beta]$.

By \Cref{lem:seg-payoff-continuity}, small changes
in $a$ and $m$ preserve the strict profit inequalities
at weights zero and $\bar\beta$. The profit gap
also remains valid for every weight in $[0,\bar\beta]$.
Otherwise, there would be priors approaching $(a,m_0)$
and optima at weights in this interval whose profits
lie in $[p_-,p_+]$. By compactness, we could take
limits of both the weights and the distributions.
The same lemma would then give an optimum at
$(a,m_0)$ with profit in the gap, a contradiction.

Fix a nearby prior and let $w_-$ and $w_+$ be the
restricted values from \Cref{lem:seg-separated-supports}.
Since $w_--w_+$ is continuous, positive at zero,
and negative at $\bar\beta$, the two values coincide
at some weight in $(0,\bar\beta)$. At this weight,
every maximizer of either restricted problem is
globally optimal. Those attaining $w_-$ earn less
than $p_-$, while those attaining $w_+$ earn more
than $p_+$.

The result therefore holds on an open neighborhood
of $(a,m_0)$. We can choose this neighborhood small
enough that the lower-profit optima have distorted
inclusion and the higher-profit optima have exclusion,
by \Cref{lem:seg-participation-persistence}.

Finally, replacing $c(x)=x^2/2$ with
$c(x)=\kappa x^2/2$ divides consumer surplus and
profit by $\kappa$. The optimal distributions
and weights are unchanged, so the result holds
for every $\kappa>0$.
\end{proof}

\subsection{State-dependent segmentation}
\label{app:state-dependent}

\subsubsection{Efficiency and full surplus extraction}
\label{app:sd-full-extraction}

\begin{proof}[Proof of \Cref{prop:sd-full-extraction}]
Let $K$ denote the public segment and $Y$ the buyer's posterior
mean after observing her segment and private signal. For every
joint policy and $\beta\ge1$,
\begin{equation}
 \CS+\beta\Pi
 =\TS+(\beta-1)\Pi
 \le\beta\TS
 \le\beta\E[\pi(Y)]
 \le\beta\E_H[\pi(V)].
 \label{eq:sd-full-extraction-bound}
\end{equation}
The first inequality follows from individual rationality,
the second bounds total surplus by first-best surplus,
and the last follows from Jensen's inequality,
since $Y=\E[V\mid Y,K]$. The function $\pi$ is strictly
convex because $\pi'=q$ is strictly increasing.
Full public revelation attains the bound: the seller supplies
$q(V)$ and charges $Vq(V)$, leaving zero rent.

Every maximizing policy must attain equality
throughout \eqref{eq:sd-full-extraction-bound}.
Equality in Jensen's inequality requires
$Y=V$ almost surely. Equality in the first-best
bound then requires $x_K(V)=q(V)$ almost surely.

Suppose a segment with positive probability contains
two distinct positive valuations. Write its support
as $w_1<\cdots<w_L$, with probabilities $p_i>0$.
Let $w_i$ be the lowest positive valuation and put
$S_{i+1}:=\sum_{\ell>i}p_\ell>0$.
Lower its quality from $q(w_i)$ to $q(w_i)-\epsilon$,
leaving all other qualities unchanged. For sufficiently
small $\epsilon>0$, qualities remain nondecreasing,
since any lower valuation receives zero quality and
$q(w_i)<q(w_{i+1})$.

For a given quality schedule, the seller leaves buyers
only the rents required by incentive compatibility.
After the quality reduction, adjusting transfers
accordingly lowers every higher type's rent by
$(w_{i+1}-w_i)\epsilon$.
At the original efficient quality, $c'(q(w_i))=w_i$,
so the loss in total surplus is
$o(\epsilon)$. Profit therefore rises by
$S_{i+1}(w_{i+1}-w_i)\epsilon+o(\epsilon)>0$.
This contradicts seller optimality, so every segment
with positive probability contains at most one
positive primitive valuation.

Conversely, suppose $Y=V$ and each segment contains
at most one positive valuation. If that valuation
is $v$, the seller can offer $(0,0)$ and
$(q(v),vq(v))$ and earn the segment's entire
first-best surplus as profit. A segment containing
only zero has zero payoffs. Since no menu can yield
profit above first-best surplus, every seller-optimal
menu must attain that surplus and leave buyers
zero rent. These conditions are therefore sufficient.
Every maximizing joint policy satisfies
$\CS=0$ and $\Pi=\E_H[\pi(V)]$, including at
$\beta=1$.

If $a>0$, each segment contains a single primitive valuation,
so its public label reveals $V$. If zero is possible,
it may share a segment with one positive valuation, provided
the buyer still learns her valuation privately.
\end{proof}

\subsubsection{Implementing conditional priors}
\label{app:sd-implementation}

For a finite prior $H=\sum_i p_i\delta_{v_i}$, assign a buyer
in state $v_i$ to segment $k$ with probability
$\lambda_kH_k(\{v_i\})/p_i$.
The identity $\sum_k\lambda_kH_k=H$ ensures that these
probabilities sum to one in every state. Bayes' rule then
gives segment probability $\lambda_k$ and conditional prior
$H_k$. Within each segment, a private information policy
can induce any $G_k\preceq_\cx H_k$, proving sufficiency
of \eqref{eq:sd-prior-feasibility}.
Conversely, every joint policy satisfies the prior identity
by the law of total probability and
$G_k\preceq_\cx H_k$ by Bayes plausibility within each segment.
For a binary prior, write $p=(m-a)/(1-a)$ and
$p_k=(m_k-a)/(1-a)$.
Any finite split satisfying $\sum_k\lambda_km_k=m$
can be implemented by
\begin{equation}
 \Pr(K=k\mid V=1)=\frac{\lambda_kp_k}{p},\qquad
 \Pr(K=k\mid V=a)=\frac{\lambda_k(1-p_k)}{1-p}.
 \label{eq:sd-assignment-kernel}
\end{equation}
These probabilities are nonnegative and sum to one
in each state because $\sum_k\lambda_kp_k=p$.
Within segment $k$, a posterior distribution is feasible
precisely when it is supported on $[a,1]$ and has mean
$m_k$. Optimizing private learning at each conditional
mean and then choosing the finite split therefore gives
the least concave majorant in
\eqref{eq:sd-concavification}
\citep{kamenicagentzkow2011}.

\subsubsection{Splitting an interior lower atom}
\label{app:sd-local-proof}

\begin{lemma}[A local split at an interior lower atom]
\label{lem:sd-interior-atom-split}
Fix $0\le\beta<1$ and a binary prior on $\{a,1\}$ with
mean $m$. Under \Cref{ass:cost,ass:curvature}, suppose
a single-market maximizer has a lower trading atom at
a posterior value strictly above $a$. Then a two-segment
state-dependent policy strictly raises both $J_\beta$
and profit relative to that distribution. The conditional
means can be chosen arbitrarily close to $m$, and both
segments retain full participation. In particular, at
$\beta=0$, the policy strictly raises both consumer
surplus and profit.
\end{lemma}

\begin{proof}
Write the specified optimum as
$G^*=s\delta_v+(1-s)G_A$, where $s\in(0,1)$ is the
lower atom's probability, $a<v<B:=\min\supp G_A$,
and $G_A$ is the normalized regular upper distribution.
Let $z>0$ be the atom's virtual value, $\bar z$ the
bottom virtual value of $G_A$, and $\theta$ its mean
multiplier. Write $M,C_A,P_A$ for the mean, consumer
surplus, and profit of $G_A$.
By \Cref{lem:strict-attachment},
$z<\bar z$ and $q(z)\le\theta\le q(\bar z)$.
In particular, $P_A>\pi(z)$.

\paragraph{Varying the lower atom.}
Keep $G_A$ fixed and, for $(r,\zeta)$ near $(s,z)$, define
\[
 v(r,\zeta)=r\zeta+(1-r)B,\qquad
 G_{r,\zeta}=r\delta_{v(r,\zeta)}+(1-r)G_A.
\]
The formula for $v(r,\zeta)$ makes the revenue chord
from $(0,v(r,\zeta))$ to $(r,(1-r)B)$ have slope
$-\zeta$. For $(r,\zeta)$ sufficiently close to $(s,z)$,
we still have $a<v(r,\zeta)<B$ and
$0<\zeta<\bar z$. The chord therefore joins the
rescaled upper curve concavely and forms the lower
part of the new revenue envelope.

The seller supplies quality $q(\zeta)$ to buyers
at the lower atom and leaves qualities within
$G_A$ unchanged. Relative to a market with
distribution $G_A$ alone, each buyer in the upper
component receives additional rent
$(B-v(r,\zeta))q(\zeta)$. The mean and payoffs
are therefore
\[
\begin{aligned}
 \mu(r,\zeta)
   &=(1-r)M+rB-r^2(B-\zeta),\\
 C(r,\zeta)
   &=(1-r)C_A+r(1-r)(B-\zeta)q(\zeta),\\
 P(r,\zeta)
   &=(1-r)P_A+r\pi(\zeta).
\end{aligned}
\]
Keep $\theta$ fixed and write $J=C+\beta P$
and $L=J-\theta\mu$. At $(r,\zeta)=(s,z)$,
the mean and payoffs are those of $G^*$.

\paragraph{Stationarity at the original distribution.}
The distributions $G_{r,\zeta}$ need not have
mean $m$. To use the optimality of $G^*$,
we first restore the mean by adjusting the
upper component. By
\Cref{lem:binary-two-multipliers}, a perturbation
at interior ranks can change its conditional mean
by $\delta$, holding $B$ and the endpoint virtual
values fixed. Its conditional weighted payoff
changes by $\theta\delta+o(|\delta|)$.
Choosing $\delta=[m-\mu(r,\zeta)]/(1-r)$
restores the overall mean to $m$.

The corrected objective is therefore
$J(r,\zeta)+\theta[m-\mu(r,\zeta)]
+o(|r-s|+|\zeta-z|)$.
All support and attachment inequalities remain strict,
so both signs of small parameter changes are feasible.
Optimality of $G^*$ gives
$L_r(s,z)=L_\zeta(s,z)=0$.

\paragraph{The case $\theta>q(z)$.}
Put $\ell=B-z>0$ and
$D=\mu_r(s,z)=B-M-2s\ell<0$.
Keep $\zeta=z$ and use the two distributions
$G_{s+\epsilon,z}$ and $G_{s-\epsilon,z}$.
Their means are $m\pm D\epsilon-\ell\epsilon^2$,
which lie on opposite sides of $m$ for small
$\epsilon>0$. Assign them probabilities
\[
 \lambda_+=\frac12+\frac{\ell\epsilon}{2D},
 \qquad \lambda_-=1-\lambda_+.
\]
These probabilities average the two conditional means to $m$.

At fixed $z$, $L$ is quadratic in $r$, with
$L_{rr}=2\ell[\theta-q(z)]$, while $P$ is affine.
Using $L_r(s,z)=0$, the aggregate gains are
\begin{equation*}
 \Delta J=\ell[\theta-q(z)]\epsilon^2>0,
 \qquad
 \Delta P=\frac{\ell\epsilon^2}{D}
                  [\pi(z)-P_A]>0.
\end{equation*}
The profit gain is positive because both $D$ and
$\pi(z)-P_A$ are negative. The split reduces the
aggregate probability of the lower atom, which earns
less profit per buyer than the upper component.

\paragraph{The case $\theta=q(z)$.}
Changing the atom's probability alone now gives no
gain in the designer's objective. We also reduce its virtual
value as its probability increases.
Stationarity in $\zeta$ gives
$(1-s)\ell q'(z)=(1-\beta)q(z)$, and hence
$L_{rr}(s,z)=0$ and
$L_{r\zeta}(s,z)=-s\ell q'(z)<0$.

To use this negative cross derivative, fix a small $h>0$
and consider $(r,\zeta)=(s+t,z-ht)$. Denote the resulting
functions by $\mu_h,J_h,L_h,P_h$. At $t=0$,
\[
 \mu_h'(0)=D-hs^2<0,\qquad L_h'(0)=0,\qquad
 L_h''(0)=2hs\ell q'(z)+h^2L_{\zeta\zeta}(s,z).
\]
Thus $L_h''(0)>0$ for all sufficiently small $h>0$.

The means at $t=\epsilon$ and $t=-\epsilon$ again
lie on opposite sides of $m$. The probability on the
first distribution that preserves the aggregate mean is
\[
 \lambda_+
 =\frac{m-\mu_h(-\epsilon)}
        {\mu_h(\epsilon)-\mu_h(-\epsilon)}
 =\frac12-\frac{\mu_h''(0)}{4\mu_h'(0)}
                  \epsilon+o(\epsilon).
\]
Since $L_h'(0)=0$, Taylor expansion gives
\[
 \Delta J
=\frac12L_h''(0)\epsilon^2+o(\epsilon^2)\qquad\text{and}\qquad
 \Delta P
   =\frac12\left[
       P_h''(0)
       -\frac{\mu_h''(0)}{\mu_h'(0)}P_h'(0)
     \right]\epsilon^2+o(\epsilon^2).
\]
The coefficient of $\epsilon^2$ in $\Delta P$
is continuous in $h$ and equals
$\ell[\pi(z)-P_A]/D>0$ at $h=0$.
Thus both quadratic coefficients are positive
for sufficiently small $h>0$. Fix such an $h$.
For all sufficiently small $\epsilon>0$, both
the objective and profit strictly increase.

In either case, the two posterior distributions
are supported on $[a,1]$, and their means average
to $m$ under the chosen segment probabilities.
They can therefore be implemented by the assignment
in \eqref{eq:sd-assignment-kernel} and private
learning within each segment. Both segments have
full participation because their lowest virtual
values remain positive. Their means approach $m$
as $\epsilon\downarrow0$.
At $\beta=0$, the objective is consumer surplus,
so the construction strictly increases both
consumer surplus and profit.
\end{proof}

\begin{proof}[Proof of \Cref{prop:state-dependent-gain}]
Apply \Cref{lem:sd-interior-atom-split} at $\beta=0$.
\end{proof}

\subsubsection{A common multiplier across segments}
\label{app:sd-joint-price}

Throughout Appendices~\ref{app:sd-joint-price}--\ref{app:sd-joint-proof},
fix $0\le\beta<1$ and maintain \Cref{ass:cost,ass:curvature}
at that weight. We first consider $a>0$ and return to $a=0$
below.

Write $\mu(G)=\E_G[Y]$. For $\theta\in\R$, define
\begin{equation}
 L_\theta(G):=J_\beta(G)-\theta[\mu(G)-a],\qquad
 \Psi_a(\theta):=\max_{G\in\Delta([a,1])}L_\theta(G).
 \label{eq:sd-priced-problem}
\end{equation}
Weak compactness and payoff continuity ensure that
the maximum is attained. Every maximizer of $L_\theta$
also maximizes $J_\beta$ among distributions with the
same mean. Thus \Cref{thm:binary-structure}(i) and
\Cref{prop:separating-conditions} apply whenever
that mean lies in $(a,1)$.

By \eqref{eq:sd-assignment-kernel}, public assignment
can implement any finite weighted average of the
pairs $(\mu(G),J_\beta(G))$. These pairs form a
compact set, so their convex hull is also compact.
The joint value $V^{\mathrm{SD}}_{a,\beta}$ is the
upper boundary of this hull and is attained at
every mean. At an interior mean $m$, a supporting
line with slope $\theta$ gives
\begin{equation}
 V^{\mathrm{SD}}_{a,\beta}(m)
 =\theta(m-a)+\Psi_a(\theta).
 \label{eq:sd-supporting-price}
\end{equation}

Since segment means average to $m$, the objective
under any segmentation equals
$\theta(m-a)+\sum_k\lambda_kL_\theta(G_k)$.
Each $L_\theta(G_k)$ is at most $\Psi_a(\theta)$.
Attaining \eqref{eq:sd-supporting-price} therefore
requires the distribution in every segment with
positive probability to maximize $L_\theta$.
Conversely, any finite segmentation using only
these maximizers is optimal at its aggregate mean.

Among the maximizers of $L_\theta$, choose one
with the smallest mean and one with the largest.
Both exist by compactness. Since an optimal
segmentation at $m$ uses only $L_\theta$-maximizers,
$m$ lies between these two means. Assigning the
two chosen distributions to segments with suitable
probabilities therefore attains the optimum at $m$.
If the two means coincide, one segment suffices.
At $m=a$ or $m=1$, use $\delta_a$ or $\delta_1$
alone. Since the argument does not require $\beta<1$, an
optimum with at most two segments exists for
every finite $\beta\ge0$.

We will also use the pointwise maximization in
\eqref{eq:sd-rent-function} to choose the basic quality.
For $0<y\le1$, write $f_y=F_{y,\beta}$.
This function is continuous on $[0,y]$ and strictly
concave, since, for $0<z<y$,
\[
 f_y'(z)
   =(y-z)q'(z)-(1-\beta)q(z)\qquad\text{and}\qquad
 f_y''(z)
   =(y-z)q''(z)-(2-\beta)q'(z)<0.
\]
Its maximizer cannot be zero because $f_y(0)=0$ and
$f_y(z)>0$ for $0<z<y$. Nor can it equal $y$, since
$f_y'(y)=-(1-\beta)q(y)<0$.
Thus $z_{y,\beta}^{\dagger}$ is unique and lies in
$(0,y)$. No derivative at zero is needed.

Implicit differentiation of the first-order condition
and the envelope theorem give
\begin{equation*}
\begin{aligned}
 \frac{\dd z_{y,\beta}^{\dagger}}{\dd y}
 &=
 \frac{q'(z_{y,\beta}^{\dagger})}
 {(2-\beta)q'(z_{y,\beta}^{\dagger})
  -(y-z_{y,\beta}^{\dagger})q''(z_{y,\beta}^{\dagger})}
 >0,\\
 \mathcal R_\beta'(y)
 &=q(z_{y,\beta}^{\dagger}).
\end{aligned}
\end{equation*}
Since both $q$ and $z_{y,\beta}^{\dagger}$ are strictly
increasing, $\mathcal R_\beta$ is strictly convex on
$(0,1]$. Moreover, the interior maximizer gives
$\mathcal R_\beta(a)>f_a(a)=\beta\pi(a)$ for $a>0$.
We use this strict inequality below when considering
the degenerate market at $a$.

\subsubsection{Attaching and serving the lower group}
\label{app:sd-joint-lower}

A \emph{regular upper distribution} has a nondegenerate
separating interval and an efficiently served atom at one,
with no lower atom. We first describe the seller's response
when a trading atom at $a$ is added to such a distribution.

Let $F$ have bottom posterior $B>a$, bottom virtual value
$\bar z>0$, mean $M_F$, and payoffs $C_F,P_F$.
For $0<z<\min\{a,\bar z\}$, define
\[
 r=\frac{a-z}{B-z}\in(0,1),\qquad
 G=(1-r)\delta_a+rF.
\]
The identity $a=(1-r)z+rB$ makes the chord from $(0,a)$
to $(1-r,rB)$ have slope $-z$. This chord lies above
the atom's raw revenue and joins the rescaled revenue
curve of $F$ concavely because $z<\bar z$.
Its endpoints are contacts, so any concave majorant
must lie above it. The chord and the rescaled upper
curve therefore form the concavified revenue curve of $G$.

Buyers at $a$ receive quality $q(z)$ and zero rent.
Buyers in the upper component retain their original
qualities, but each receives an additional rent
$(B-a)q(z)$. Hence
\begin{align}
 \mu(G)&=a+r(M_F-a),\qquad
 \CS(G)=rC_F+r(B-a)q(z),\notag\\
 \Pi(G)&=rP_F+(1-r)\pi(z),\qquad
 L_\theta(G)=rL_\theta(F)+(1-r)f_a(z).
 \label{eq:sd-lower-splice}
\end{align}
The last identity uses $r(B-a)=(1-r)(a-z)$.
It expresses $L_\theta(G)$ as a weighted average of
$L_\theta(F)$ and the pointwise contribution $f_a(z)$.

\begin{lemma}[Participation and the lower atom]
\label{lem:sd-priced-structure}
Let $a>0$ and $\theta\in\R$. Every maximizer of
\eqref{eq:sd-priced-problem} other than $\delta_a$
consists of a regular upper distribution and, possibly,
a lower trading atom at $a$. All maximizers have full
participation. The upper component's mean multiplier
equals $\theta$.
\end{lemma}

\begin{proof}
Every maximizer with mean in $(a,1)$ satisfies the
binary characterization. Moreover, a mean-preserving
split that strictly raises the objective cannot
start from a maximizer of $L_\theta$: the split's average
$L_\theta$ would exceed $\Psi_a(\theta)$, although
each component is bounded above by $\Psi_a(\theta)$.
Thus \Cref{lem:sd-interior-atom-split} rules out a
lower trading atom strictly above $a$.

We next rule out exclusion. Suppose a maximizer has
the form $G=(1-s)\delta_a+sF$.
By \Cref{lem:strict-optimal-exclusion},
$sB>a$ and $\bar z\ge a>0$.
Choose $0<z<a$ and set
\[
 r=\frac{a-z}{B-z}<\frac aB<s,\qquad
 G_L=(1-r)\delta_a+rF,\qquad G_H=F,\qquad
 \lambda=\frac{1-s}{1-r}\in(0,1).
\]
Then $\lambda G_L+(1-\lambda)G_H=G$.
The split leaves the aggregate posterior distribution
unchanged. Previously excluded buyers now trade in
the segment with distribution $G_L$, together with
some buyers from $F$. The remaining buyers face
a separate menu for $F$ alone.

The original payoffs are $sC_F$ and $sP_F$.
Applying \eqref{eq:sd-lower-splice} to the new segments
gives
\begin{equation}
 \Delta\CS=(1-s)(a-z)q(z)>0,\qquad
 \Delta\Pi=(1-s)\pi(z)>0.
 \label{eq:sd-exclusion-improvement}
\end{equation}
Since the aggregate mean is unchanged, this also raises
average $L_\theta$, contradicting optimality.

It remains to consider the endpoint means.
At mean one, the only distribution is $\delta_1$.
For $0<d<1-a$ and small $\epsilon>0$, replace it by
a distribution assigning probability $\epsilon$ to
$1-\epsilon d$ and the remaining probability to one.
Its mean is $1-\epsilon^2d$ and its lower virtual
value is $1-d$. The change in $L_\theta$ is
\[
 \epsilon\{d q(1-d)+\beta[\pi(1-d)-\pi(1)]\}
 +\epsilon^2d[\theta-q(1-d)].
\]
The expression in braces, divided by $d$, converges
to $(1-\beta)q(1)>0$ as $d\downarrow0$.
Choose $d$ to make it positive, then choose
$\epsilon$ sufficiently small. Thus $\delta_1$
cannot maximize $L_\theta$ for any finite $\theta$.
At mean $a$, the only distribution is $\delta_a$,
which serves everyone at quality $q(a)>0$.
This proves the structural and participation claims.

Finally, let $\vartheta$ be the upper component's
mean multiplier from \Cref{prop:separating-conditions}.
The direction in \eqref{eq:binary-remote-mean}
changes its conditional mean by $t$, preserves its lowest posterior value and endpoint
virtual values, and changes its conditional weighted
payoff by $\vartheta t+o(|t|)$.
Hold the component probabilities and any lower atom
fixed. The derivative of $L_\theta$ is then the upper
component's probability times $\vartheta-\theta$.
Small increases and decreases in the upper component's
mean are both feasible. Optimality requires
$\vartheta=\theta$.
\end{proof}

\subsubsection{Uniqueness of the lower attachment at a common price}
\label{app:sd-joint-uniqueness}

Fix $\theta$. By
\Cref{lem:sd-priced-structure,prop:separating-conditions},
the upper component of every nondegenerate maximizer
of $L_\theta$ satisfies the same system.
Using $\tau=-\log(1-u)$, with dots denoting
derivatives with respect to $\tau$, this system is
\begin{equation}
 \begin{gathered}
 A+(Q-\varphi)q'(\varphi)-(1-\beta)q(\varphi)=0,\qquad
 \dot A=q(\varphi)-\theta,\qquad \dot Q=Q-\varphi,\\
 (A,Q,\varphi)=((1-\beta)\bar q,1,1)
 \quad\text{at the start of the top atom}.
 \end{gathered}
 \label{eq:sd-common-trajectory}
\end{equation}
Along the separating region, $Q>\varphi>0$ and
$q(\varphi)\ge\theta$.

Solving the system backward from these terminal values
gives a unique path, which every continuation follows. Since $Q$ increases strictly, its value
$B$ identifies where a continuation begins.
Let $F_B$ denote the corresponding upper tail,
including the top atom and normalized to have mass one.
Write $\zeta_B$ and $A_B$ for the values of $\varphi$
and $A$ where $Q=B$. Let $M_B$ be the mean of $F_B$
and set $J_B:=J_\beta(F_B)$.
Define $S_\theta(B):=J_B-\theta(M_B-a)$.
Raising $B$ therefore truncates the same distribution
from below and renormalizes the remaining tail.

\paragraph{The value of a continuation.}
Every feasible continuation satisfying
\eqref{eq:sd-common-trajectory}, whether or not it is
optimal, satisfies
\begin{equation}
 J_B-\theta(M_B-B)
 =f_B(\zeta_B)-(B-\zeta_B)A_B
 =(B-\zeta_B)^2q'(\zeta_B)
   +\beta(B-\zeta_B)q(\zeta_B)+\beta\pi(\zeta_B).
 \label{eq:sd-tail-value}
\end{equation}
The values $\zeta_B$ and $A_B$ are understood as limits from
within the separating region when necessary.

To prove the first equality, let $M$ be the mean of
the normalized tail beginning at the current rank
and $J$ the value of $J_\beta$ for that distribution. Differentiating these conditional
averages gives $\dot M=M-Q$ and
$\dot J=J-f_Q(\varphi)$.
To compare the two sides, put $W=Q-\varphi$ and define
$K=J-\theta(M-Q)$ and $T=f_Q(\varphi)-WA$.
The static equation gives $f_Q'(\varphi)=-A$.
Together with the two differential equations, this yields
\[
 \dot K=K-f_Q(\varphi)+\theta W,\qquad
 \dot T=W(\theta-A)=T-f_Q(\varphi)+\theta W.
\]
Thus $K-T$ satisfies $\frac{\dd}{\dd\tau}(K-T)=K-T$.
At the start of the top atom, the remaining distribution
is $\delta_1$, so $K=T=\beta\pi(1)$.
Hence $K=T$ throughout the continuation.
Substituting the static equation for $A$ gives the
second equality in \eqref{eq:sd-tail-value}.

Along the separating trajectory,
$\dot S_\theta=\dot J-\theta\dot M
=K-f_Q(\varphi)=-WA$.
Dividing this identity and the equation for $\dot A$
by $\dot Q=W>0$ gives
\begin{equation}
 S_\theta'(B)=-A_B,\qquad
 \frac{\dd A_B}{\dd B}
 =\frac{q(\zeta_B)-\theta}{B-\zeta_B}.
 \label{eq:sd-stopping-derivatives}
\end{equation}

\paragraph{Regular candidates.}
Consider a maximizer with no lower atom.
The directions in \Cref{lem:binary-two-multipliers}
vary its bottom posterior independently of its mean.
The first-order change in $L_\theta$ is
$-\xi\,\dd B=-A_B\,\dd B$.
Small increases and decreases in $B$ are feasible
when $B>a$, while small increases are feasible
when $B=a$. Optimality therefore requires
\begin{equation}
 A_B=0,\quad B>a\qquad\text{and}\qquad
 A_a\ge0,\quad B=a.
 \label{eq:sd-regular-bottom}
\end{equation}
We call a feasible regular continuation a
\emph{regular candidate} at $\theta$ if it satisfies
\eqref{eq:sd-common-trajectory},
$q(\varphi)\ge\theta$, and
\eqref{eq:sd-regular-bottom}.

Every regular candidate has a positive bottom
virtual value. Reconstruction first gives
$B=\int_0^1\varphi(u)\,\dd u>0$.
Suppose instead that $\zeta_B=0$.
Since $f_B(0)=0$ and $f_B$ is strictly concave and
positive on $(0,B)$, there is an $\eta>0$ such that
$f_B'(z)\ge\eta$ for all sufficiently small $z>0$.
Because $Q\ge B$ and $q'>0$, the static equation then gives
$A=-f_Q'(\varphi)\le-f_B'(\varphi)\le-\eta$
near the bottom. Taking the bottom limit contradicts
$A_B\ge0$ in \eqref{eq:sd-regular-bottom}.
This argument requires no derivative of $q$ at zero.

\paragraph{Candidates with a lower atom.}
For $0<z<a$, define
\[
 g_a(z):=(1-\beta)q(z)-(a-z)q'(z)=-f_a'(z),
 \qquad
 g_a'(z)=(2-\beta)q'(z)-(a-z)q''(z)>0.
\]
By \eqref{eq:sd-lower-splice}, the distribution
$(1-r)\delta_a+rF_B$, with $r=(a-z)/(B-z)$,
has value
\[\ell(B,z)=rS_\theta(B)+(1-r)f_a(z).\]

At a maximizer with a lower atom,
$B>a$ and $0<z<\min\{a,\zeta_B\}$.
These strict inequalities allow small increases and
decreases in $z$. To vary $B$, shorten or extend the
continuation at its bottom and renormalize.
The strict support and quality gaps persist under
small changes. Moreover,
\eqref{eq:continuation-phi} gives a strictly positive
derivative of virtual value at the bottom, even if
$q(\zeta_B)=\theta$, because
$(B-\zeta_B)q'(\zeta_B)>0$.
A short backward extension therefore remains increasing.

Both $B$ and $z$ can consequently vary locally in
either direction, so $\ell_B=\ell_z=0$.
Using \eqref{eq:sd-stopping-derivatives},
\begin{align*}
 \ell_B
 =-\frac r{B-z}
   \{S_\theta(B)-f_a(z)+(B-z)A_B\},\qquad
 \ell_z
 =\frac{B-a}{(B-z)^2}
   \{-(S_\theta(B)-f_a(z))+(B-z)f_a'(z)\}.
\end{align*}
The first-order conditions are therefore equivalent to
\begin{equation}
 A_B=g_a(z),\qquad
 S_\theta(B)-f_a(z)+(B-z)A_B=0.
 \label{eq:sd-atom-bottom}
\end{equation}

An \emph{atom candidate} is a feasible distribution
$(1-r)\delta_a+rF_B$ with
$B>a$ and $0<z<\min\{a,\zeta_B\}$, satisfying
\eqref{eq:sd-atom-bottom}, whose upper continuation
satisfies \eqref{eq:sd-common-trajectory} and
$q(\varphi)\ge\theta$.
Every nondegenerate maximizer of $L_\theta$ is
therefore either a regular candidate or an atom candidate.

\begin{lemma}[Uniqueness within each candidate class]
\label{lem:sd-price-uniqueness}
Fix $\theta$. For $a>0$, there is at most one regular
candidate and at most one atom candidate. If a regular
candidate has bottom $a$, no atom candidate exists at
the same $\theta$. Moreover, if $L_\theta$ has a
nondegenerate maximizer, then $\delta_a$ is not a
maximizer. Regular-candidate uniqueness also holds
when $a=0$.
\end{lemma}

\begin{proof}
As in \Cref{prop:separating-conditions}, the differential
system has a unique solution when solved backward from
its terminal condition. All candidates therefore follow
the same trajectory. For any candidate,
\eqref{eq:continuation-phi} and
$q(\varphi)\ge\theta$ imply that virtual values
increase strictly. If two candidates have distinct bottoms, the
continuation with the lower bottom consequently satisfies
$q(\varphi)>\theta$ above that bottom.
By \eqref{eq:sd-stopping-derivatives}, $A_B$ therefore
increases strictly between any two distinct
candidate bottoms.

\paragraph{Regular candidates.}
A regular candidate requires $A_B=0$ when $B>a$
and $A_a\ge0$ when $B=a$.
Two regular candidates with bottoms $B_1<B_2$
would require $A_{B_1}\ge0$ and $A_{B_2}=0$.
But $A_{B_2}>A_{B_1}\ge0$, a contradiction. A common bottom determines the same normalized
distribution along the common trajectory, including
the mass of its top atom.

This also proves uniqueness when $a=0$: reconstruction
gives $B>0$ for every regular candidate, so each must
satisfy $A_B=0$.

\paragraph{Atom candidates.}
At a given bottom $B$, the equation $g_a(z)=A_B$
has at most one solution because $g_a$ is strictly increasing.
The share $r=(a-z)/(B-z)$ is then also determined,
so two atom candidates with the same bottom coincide.

Suppose instead that two atom candidates have bottoms
$B_1<B_2$ and lower virtual values $z_1,z_2$.
But $A_{B_1}<A_{B_2}$, so $z_1<z_2$.
For $B\in[B_1,B_2]$, define
$z(B)=g_a^{-1}(A_B)\in[z_1,z_2]\subset(0,a)$.
This function is continuously differentiable because
$g_a'>0$. Evaluate the second condition in
\eqref{eq:sd-atom-bottom} along these values:
$H(B)=S_\theta(B)-f_a(z(B))+(B-z(B))A_B$.
Using $S_\theta'(B)=-A_B$ and
$f_a'(z(B))=-A_B$, differentiation gives
\begin{equation}
 H'(B)=(B-z(B))\frac{\dd A_B}{\dd B}>0,
 \qquad B_1<B<B_2.
 \label{eq:sd-atom-residual-monotonicity}
\end{equation}
Here $B-z(B)>0$ because $z(B)<a<B$.
Both candidates require $H(B_1)=H(B_2)=0$,
but $H$ is strictly increasing between these points.
The function $H$ is well defined throughout the
interval. Feasibility of the intermediate mixtures
is not required.

\paragraph{A regular candidate with bottom $a$.}
Suppose such a candidate exists together with an
atom candidate whose bottom is $B_*>a$.
The static equation at the regular bottom gives
$g_a(\zeta_a)=A_a$, with $0<\zeta_a<a$.
The same inverse construction defines
$z(B)=g_a^{-1}(A_B)$ on $[a,B_*]$, starting at
$z(a)=\zeta_a$ and ending at the atom candidate's
lower virtual value.

By \eqref{eq:sd-tail-value},
$S_\theta(a)=f_a(\zeta_a)-(a-\zeta_a)A_a$,
so $H(a)=0$. The derivative calculation in
\eqref{eq:sd-atom-residual-monotonicity} then gives
$H(B_*)>0$, contradicting the atom candidate's
bottom condition.

\paragraph{The degenerate distribution $\delta_a$.}
Its value is $L_\theta(\delta_a)=f_a(a)=\beta\pi(a)$.
To show that every nondegenerate maximizer has a
strictly higher value, we use the strict tangent
inequality
$f_a(t)+(a-t)f_a'(t)>f_a(a)$ for $0<t<a$,
which follows from strict concavity of $f_a$.

For a regular maximizer with bottom $a$,
\eqref{eq:sd-tail-value} and
$A_a=-f_a'(\zeta_a)$ give
$L_\theta(F_a)
 =f_a(\zeta_a)+(a-\zeta_a)f_a'(\zeta_a)>f_a(a)$.
For an atom maximizer,
\eqref{eq:sd-atom-bottom} gives
$S_\theta(B)-f_a(z)=(B-z)f_a'(z)$.
Since $r=(a-z)/(B-z)$, we have
$\ell(B,z)=f_a(z)+(a-z)f_a'(z)>f_a(a)$.

Finally, consider a regular maximizer with $B>a$.
The condition $A_B=0$ and the static equation imply
$\zeta_B=z_{B,\beta}^{\dagger}
>z_{a,\beta}^{\dagger}$.
Set $z=z_{a,\beta}^{\dagger}$ and
$r=(a-z)/(B-z)$.
The inequalities $0<z<\min\{a,\zeta_B\}$ make
$(1-r)\delta_a+rF_B$ feasible.
Since $F_B$ maximizes $L_\theta$,
\eqref{eq:sd-lower-splice} implies
$\Psi_a(\theta)
 \ge r\Psi_a(\theta)+(1-r)\mathcal R_\beta(a)$.
Because $r<1$, it follows that
$\Psi_a(\theta)\ge\mathcal R_\beta(a)>f_a(a)$.
Thus $\delta_a$ cannot be optimal whenever a
nondegenerate maximizer exists.
\end{proof}

\subsubsection{Proof of \texorpdfstring{\Cref{thm:joint-consumer-segmentation}}{Theorem~\ref{thm:joint-consumer-segmentation}}}
\label{app:sd-joint-proof}

\begin{proof}
The argument leading to \eqref{eq:sd-supporting-price}
shows that an optimum exists and can be implemented
with at most two segments. Fix $a>0$ and $m\in(a,1)$,
and choose a supporting multiplier $\theta$ at $m$.
In every optimal policy, the distribution in each
segment with positive probability maximizes $L_\theta$.

Since $m>a$, at least one of these distributions
differs from $\delta_a$. It is nondegenerate by
\Cref{lem:sd-priced-structure}, so
\Cref{lem:sd-price-uniqueness} implies that
$\delta_a$ cannot also maximize $L_\theta$.
The posterior distribution in every segment used
in an optimum therefore consists of a regular
upper component and, possibly, a lower trading
atom at $a$. All buyers in these segments trade,
proving part~(i).

\paragraph{The two segments share the same upper distribution.}
Suppose two distinct distributions maximize
$L_\theta$. By \Cref{lem:sd-price-uniqueness},
one is a regular distribution $G_A$ with bottom
$B>a$, and the other has a lower atom at $a$.
Write $\Psi=\Psi_a(\theta)$.

To bound the common value from below, start with $G_A$.
The condition $A_B=0$ gives
$\zeta_B=z_{B,\beta}^{\dagger}
>z_{a,\beta}^{\dagger}$.
We can therefore add a trading atom at $a$ with
virtual value $z_{a,\beta}^{\dagger}$, using the
upper share
$r=(a-z_{a,\beta}^{\dagger})/
   (B-z_{a,\beta}^{\dagger})$.
By \eqref{eq:sd-lower-splice}, the resulting
distribution has value
$r\Psi+(1-r)\mathcal R_\beta(a)$.
Since this cannot exceed $\Psi$, we obtain
$\Psi\ge\mathcal R_\beta(a)$.

For the reverse bound, write the atom maximizer as
$(1-\widehat r)\delta_a+\widehat r F$, with lower
virtual value $z$ and $0<\widehat r<1$.
The same identity gives
$\Psi
 =\widehat r L_\theta(F)+(1-\widehat r)f_a(z)
 \le\widehat r\Psi+(1-\widehat r)\mathcal R_\beta(a)$.
Hence $\Psi\le\mathcal R_\beta(a)$, and therefore
 $\Psi=\mathcal R_\beta(a)$.
Since $0<\widehat r<1$, the preceding equality requires
$L_\theta(F)=\Psi$ and $f_a(z)=\mathcal R_\beta(a)$.
Uniqueness of the maximizer of $f_a$ then gives
\begin{equation}
 \Psi_a(\theta)=\mathcal R_\beta(a),\qquad
 z=z_{a,\beta}^{\dagger},\qquad
 L_\theta(F)=\Psi_a(\theta).
 \label{eq:sd-coexistence-equality}
\end{equation}
Thus $F$ is itself a regular maximizer.
Regular-candidate uniqueness implies $F=G_A$,
and revenue contact gives $\widehat r=r$.
The two distributions are therefore precisely
$G_L$ and $G_H=G_A$ in \eqref{eq:sd-joint-markets}.

Let $M=\mu(G_A)$. Since $A_B=0$,
\eqref{eq:sd-tail-value} gives
$J_\beta(G_A)-\theta(M-B)=\mathcal R_\beta(B)$.
Subtracting
$J_\beta(G_A)-\theta(M-a)=\mathcal R_\beta(a)$
yields
\begin{equation}
 \theta
 =\frac{\mathcal R_\beta(B)-\mathcal R_\beta(a)}{B-a}>0,
 \qquad
 q(z_{a,\beta}^{\dagger})
 <\theta<q(z_{B,\beta}^{\dagger}).
 \label{eq:sd-secant-condition}
\end{equation}
The inequalities follow from strict convexity of
$\mathcal R_\beta$ and
$\mathcal R_\beta'(y)=q(z_{y,\beta}^{\dagger})$.
In particular, the basic product has strictly lower
quality than the bottom of the separating region.

By \eqref{eq:sd-lower-splice}, buyers with posterior
types in $G_A$ receive the same qualities in both
segments, but pay less in $G_L$ by
$(B-a)q(z_{a,\beta}^{\dagger})$.
Moreover, $L_\theta(G)\le\mathcal R_\beta(a)$ for
every $G\in\Delta([a,1])$.
The line $\theta(m-a)+\mathcal R_\beta(a)$ therefore
gives an upper bound on the designer's objective
at every mean, both in a single market and under
segmentation, with equality at the means of $G_L$
and $G_A$.

\paragraph{The interval of strict gains.}
There is at most one multiplier at which two distinct
maximizers coexist. Suppose otherwise that
$\theta_1<\theta_2$ both have this property.
Choose a maximizer $G_2$ at $\theta_2$, whose mean
exceeds $a$. Then,
$\Psi_a(\theta_1)
 \ge L_{\theta_1}(G_2)
 =\Psi_a(\theta_2)
   +(\theta_2-\theta_1)[\mu(G_2)-a]
 >\Psi_a(\theta_2)$.
This contradicts \eqref{eq:sd-coexistence-equality},
which makes both values equal to $\mathcal R_\beta(a)$.

When this multiplier exists, denote it by
$\theta_{a,\beta}$. Its only maximizing distributions
are $G_L$ and $G_A$, with means
$a<m_-=a+r(M-a)<m_+=M<1$.
For every $m\in[m_-,m_+]$, assigning probability
$\lambda(m)=(M-m)/(M-m_-)$ to $G_L$ and the remainder
to $G_A$ attains the supporting line.
This proves \eqref{eq:sd-joint-value}.

Only $G_L$ and $G_A$ maximize
$L_{\theta_{a,\beta}}$, and neither has mean
$m\in(m_-,m_+)$. A single-market optimizer at
such a mean therefore gives strictly less than
the joint optimum. Every optimal segmentation
must use only these two distributions, with
total probabilities $\lambda(m)$ and
$1-\lambda(m)$, respectively. At either endpoint,
only the corresponding distribution is used.

Conversely, suppose an optimal segmentation
strictly improves on the single-market optimum.
It must use at least two distinct distributions,
since using the same distribution in every
segment gives the single-market payoff.
All distributions used must maximize $L_\theta$
for a common supporting multiplier. Hence
$\theta=\theta_{a,\beta}$, and the distributions
are $G_L$ and $G_A$. Both must have positive
probability, so the aggregate mean lies in
$(m_-,m_+)$. There can therefore be no strict
gain outside this interval. If no multiplier
has two distinct maximizers, segmentation
cannot strictly improve the objective.
At the degenerate prior means, equality follows
because every segment has the same degenerate prior.
This proves parts~(ii)--(iii), including uniqueness
of the pair, multiplier, and interval.

\paragraph{The order of optimization.}
Fix $m\in(m_-,m_+)$ and put $\theta=\theta_{a,\beta}$.
Every joint optimum uses only $G_L$ and $G_A$, so its
aggregate posterior distribution is
\[
 \overline G=(1-s)\delta_a+sG_A,
 \qquad s=\frac{m-a}{M-a}\in(r,1).
\]
The lower atom has virtual value
\[
 z=\frac{a-sB}{1-s}
   <z_{a,\beta}^{\dagger}
   <z_{B,\beta}^{\dagger},
\]
where $z_{B,\beta}^{\dagger}$ is the bottom virtual
value of $G_A$. The lower chord therefore joins the
upper curve concavely, including when $z\le0$, and
\eqref{eq:sd-lower-splice} gives
$J_\beta(\overline G)
 =sJ_\beta(G_A)+(1-s)F_{a,\beta}(z)$.
For $G_A$, the multipliers are $\theta$ and
$\xi=A_B=0$.

Suppose $\overline G$ were privately optimal.
Uniqueness in \Cref{lem:binary-two-multipliers}
identifies its upper-component multipliers with
$\theta$ and $\xi=0$.
If $z>0$, the same lemma permits variations in $B$
that keep $M$ and the endpoint virtual values fixed.
Keep $s$ fixed as well, so the aggregate mean remains
$m$. Since $B>a$ and $0<z<z_{B,\beta}^{\dagger}$,
sufficiently small variations of either sign preserve
the support and virtual-value ordering.
The first variation of the upper payoff is
$-\xi\,\dd B$, while $\dd z=-s\,\dd B/(1-s)$.
Stationarity of the displayed objective therefore requires
$0=-s\{\xi+\partial_zF_{a,\beta}(z)\}$.
This is impossible: $\xi=0$, whereas strict concavity
of $F_{a,\beta}$ gives
$\partial_zF_{a,\beta}(z)>0$ for
$0<z<z_{a,\beta}^{\dagger}$.

If $z\le0$, \Cref{lem:strict-optimal-exclusion}
rules out $z=0$ and requires
$J_\beta(G_A)=\theta(M-a)$.
But \eqref{eq:sd-coexistence-equality} gives
\[
 J_\beta(G_A)-\theta(M-a)=\mathcal R_\beta(a)>0.
\]
Thus $\overline G$ is not privately optimal.

To establish a strict gap for the supremum, allow
segmentations to be arbitrary probability measures
on $\Delta([a,1])$ whose aggregate belongs to
$\mathcal G_{a,m}$ and attains $V_{a,\beta}(m)$.
The one-market policy using a private optimum belongs
to this set. Since aggregation and $J_\beta$ are
continuous, the set is closed in the weakly compact
space $\Delta(\Delta([a,1]))$, and hence compact.
The designer's objective is continuous on this set,
so its maximum is attained.

The bound $L_\theta(G)\le\mathcal R_\beta(a)$ applies
to every segment distribution. After averaging, it
bounds the objective by
$\theta(m-a)+\mathcal R_\beta(a)
 =V^{\mathrm{SD}}_{a,\beta}(m)$.
Equality would require the segmentation to be supported
on $\{G_L,G_A\}$, the only maximizers of $L_\theta$.
The mean constraint would then force its aggregate
to equal $\overline G$, which is not privately optimal.
The attained maximum, and therefore the supremum
over finite segmentations, is strictly below the joint
value. This proves part~(iv).

\paragraph{The payoff-neutral case.}
Set $a=0$. The two-point perturbation in the proof
of \Cref{lem:sd-priced-structure} also applies here
and rules out $\delta_1$ for every finite $\theta$.
Every maximizer other than $\delta_0$ therefore has
an interior mean and is a single-market optimum
at that mean. By \Cref{lem:neutral-regularization},
it has the form
$G=(1-s)\delta_0+sF$, with $0<s\le1$ and $F$
regular and fully trading. Its mean and both
payoffs are $s$ times those of $F$, so
$\Psi_0(\theta)=L_\theta(G)=sL_\theta(F)
 \le s\Psi_0(\theta)$ and $\Psi_0(\theta)\ge0$.
If $s<1$, these inequalities imply
$\Psi_0(\theta)=L_\theta(F)=0$.
If $s=1$, then $F=G$.
Thus $F$ is a maximizer in either case.

The binary characterization makes $F$ a regular
upper distribution. Reconstruction gives its bottom
$B>0$. The interior variations used above identify
its mean multiplier as $\theta$ and give $A_B=0$.
It is therefore a regular candidate, and the
$a=0$ assertion in \Cref{lem:sd-price-uniqueness}
makes it unique.

At a supporting multiplier for an interior mean,
at least one positive-probability segment in an optimal
segmentation has positive mean. Hence this common regular distribution $F$
exists, and every segment can be written as
$G_k=(1-s_k)\delta_0+s_kF$, allowing $s_k=0$.
Write $\bar s=\sum_k\lambda_ks_k$ for their average
trading share. The single-market distribution
$G=(1-\bar s)\delta_0+\bar sF$ has the original
mean and is feasible under the original binary prior.
Adding excluded zero types leaves the active
virtual values unchanged, so its payoffs satisfy
$(\CS(G),\Pi(G))
 =\bar s(\CS(F),\Pi(F))
 =\sum_k\lambda_k(\CS(G_k),\Pi(G_k))$.
Thus, every joint optimum can be matched by a single
market, proving
$V^{\mathrm{SD}}_{0,\beta}(m)=V_{0,\beta}(m)$.
\end{proof}

\subsection{Proof of Corollary~\ref{cor:seg-neutral}}
\label{app:seg-neutral-proof}
\label{app:state-dependent-neutral}

\begin{proof}
Consider any state-dependent segmentation. Write
$m_k=\E_{G_k}[Y]$ for segment $k$'s mean and $(C,P)$
for aggregate payoffs. For $0\le\beta<1$,
\Cref{prop:neutral-weighted-value} gives concavity of
$V_{a,\beta}$ in mean, including its continuous extension
to $[a,1]$. Since $\sum_k\lambda_km_k=m$,
\[
 C+\beta P
 =\sum_k\lambda_kJ_\beta(G_k)
 \le\sum_k\lambda_kV_{a,\beta}(m_k)
 \le V_{a,\beta}(m)
 =C_\beta+\beta P_\beta.
\]
At $\beta=1$, averaging the bound
\eqref{eq:seg-neutral-surplus-bound}, which is linear
in the conditional mean, gives $C+P\le P_1$.
Thus segmented payoffs satisfy the path's supporting
inequalities for every $0\le\beta\le1$. In particular,
$C\le C_0$ and, since consumer surplus is nonnegative,
$P\le P_1$.

For any $p\in(P_0,P_1]$, continuity of the path in
\Cref{thm:neutral-frontier} gives a weight
$\beta\in(0,1]$ with $P_\beta=p$. Every segmented pair
with $P\ge p$ then satisfies
$C\le C_\beta-\beta(P-p)\le C_\beta$.
The single-market pair $(C_\beta,p)$ attains
the bound. For $p\le P_0$, the buyer-optimal pair
$(C_0,P_0)$ meets the profit requirement and
attains maximal consumer surplus $C_0$.
Thus, segmentation cannot raise consumer surplus
under any feasible profit requirement. Taking
$p=P$ for any segmented pair $(C,P)$ also shows
that it is weakly dominated by a point on the
single-market path.

Path points with $\beta>0$ remain Pareto efficient.
At the buyer-optimal endpoint, suppose a segmented
pair has $C=C_0$. The supporting inequalities give
$\beta(P-P_\beta)\le C_\beta-C_0\le0$ for every
$0<\beta\le1$. Hence $P\le P_\beta$, and letting
$\beta\downarrow0$ gives $P\le P_0$. The
buyer-optimal endpoint therefore remains efficient,
so segmentation leaves the Pareto frontier
unchanged. State-independent assignment is a
special case, so the same conclusions apply.
\end{proof}

\subsection{Proof of Corollary~\ref{cor:seg-high-mean}}
\label{app:seg-high-mean-proof}

\begin{proof}
Fix $m\in(\bar m,1)$ and $\beta\in[0,\bar\beta]$
as in \Cref{prop:high-mean}. Its proof constructs
$\theta<0$ such that $G_{m,\beta}^*$ uniquely maximizes
$J_\beta(G)-\theta\E_G[Y]$ over $\Delta([0,1])$.
Consequently,
\begin{equation}
 J_\beta(F)
 \le V_{a,\beta}(m)
       +\theta\bigl(\E_F[Y]-m\bigr)
 \qquad\text{for every }F\in\Delta([a,1]),
 \label{eq:high-mean-global-support}
\end{equation}
with equality if and only if $F=G_{m,\beta}^*$.
The bound holds for every segment, even if its
mean is at or below $\bar m$.

For any feasible state-dependent segmentation, write
$m_k=\E_{G_k}[Y]$. Bayes plausibility gives
$\sum_k\lambda_km_k=m$, so averaging
\eqref{eq:high-mean-global-support} yields
\[
 \sum_k\lambda_kJ_\beta(G_k)
 \le V_{a,\beta}(m)
      +\theta\left(\sum_k\lambda_km_k-m\right)
 =V_{a,\beta}(m).
\]
The single-market optimum attains this bound, proving
\eqref{eq:seg-high-mean-value}. Only $G_{m,\beta}^*$ attains the bound, so equality
after averaging requires every segment with positive
probability to use this distribution. Each such
segment has mean $m$ and therefore the original
prior, since a binary prior on $\{a,1\}$ is
determined by its mean. This holds under either
assignment rule.

For $p\in(P_0,P_{\bar\beta}]$, continuity of the path
in \Cref{cor:high-mean-frontier} gives a weight
$\beta\in(0,\bar\beta]$ with $P_\beta=p$.
The bound just proved implies
$C\le C_\beta-\beta(P-p)\le C_\beta$
whenever $P\ge p$, and the single-market path point
attains this bound. For $p\le P_0$, the weight-zero
bound gives $C\le C_0$, attained by the buyer optimum.

Positive-weight path points remain Pareto efficient.
If $C=C_0$, equality in the weight-zero bound requires
every segment to use $G_{m,0}^*$, so aggregate profit
is $P_0$. The buyer-optimal pair is therefore also Pareto
efficient. Any segmented pair with profit in
$[P_0,P_{\bar\beta}]$ that is not on the path
gives less consumer surplus than the path point
with the same profit. The path is therefore the
segmented Pareto frontier over this profit interval. When
$\bar\beta=0$, it consists only of the buyer-optimal pair.
\end{proof}

\section{Optimal learning with general finite priors}
\label{app:finite-priors}

The arguments used for binary priors remain feasible under a
finite prior when all altered posterior means stay within the
same interval between adjacent primitive valuations and the
mean is preserved. Throughout, maintain $0\le\beta<1$ and
\Cref{ass:cost,ass:curvature}.

\subsection{Feasibility and the order of trading atoms}
\label{app:finite-feasibility}
\label{app:finite-order}
\label{app:finite-localization}

Write $C_F(t):=\int(y-t)_+F(\dd y)$ for the stop-loss
transform, also called the call function. We refer to $C_F(t)$ as the \emph{call at $t$}. As shown in
Section~\ref{subsec:finite-constraints}, once support and mean
are fixed, feasibility requires checking this function only
at the primitive valuations. The following observation
therefore allows us to apply the binary arguments locally.

\Needspace{5\baselineskip}
\begin{lemma}[Preserving the information constraints]
\label{lem:finite-interval-feasibility}
Suppose two quantiles agree outside a measurable rank set
$E$, both take values in $[v_i,v_{i+1}]$ on $E$, and have
the same integral over $E$. Their distributions have the
same mean and the same value of every $D_j$. In particular,
if one is feasible under the prior, so is the other.
\end{lemma}

\begin{proof}
For $j\le i$, the function $(y-v_j)_+$ equals $y-v_j$
throughout $[v_i,v_{i+1}]$. Its integral over $E$ is therefore
unchanged because the two quantiles have the same integral
there. For $j\ge i+1$, the function is zero throughout this
interval. Since the quantiles also agree outside $E$, their
means and all $D_j$ coincide. Feasibility then follows from
\eqref{eq:finite-feasibility}.
\end{proof}

By \Cref{lem:active-face-purity}, every nonterminal trading
face $I=[\ell,h]$ contains a single posterior atom at
$U_r(h)$, followed in increasing posterior value by the
support gap $(U_r(h),U_r(\ell))$.
The constructions in
\Cref{lem:fixed-envelope-fillings,lem:exact-chord-collapse,lem:binary-adjacent-elimination}
keep altered posterior values between the relevant contact
values and determine the new concavified revenue curve.
We can therefore use their objective comparisons under a finite
prior whenever the entire mean-preserving change, including
compensation, stays within one primitive interval.
By \Cref{lem:finite-interval-feasibility}, every information
constraint then remains unchanged.

\begin{lemma}[Order within a primitive interval]
\label{lem:finite-band-order}
Under \Cref{ass:cost,ass:curvature}, an optimum has at most
one nonterminal trading face whose atom and entire following
gap lie in a given primitive interval $[v_i,v_{i+1}]$.
Below such a face in posterior value, the set of regular
active contacts with values in the same interval has
measure zero.
\end{lemma}

\begin{proof}
Fix $[L,H]=[v_i,v_{i+1}]$ and a nonterminal trading face
$I_A=[\ell_A,h_A]$ with
$[U_r(h_A),U_r(\ell_A)]\subset[L,H]$.
Its slope $z_A$ and intercept are positive. By
\Cref{lem:active-face-purity}, it is filled by the constant
$U_r(h_A)$, with $d_{I_A}=\kappa_{I_A}>0$.

First suppose that regular active contacts below $I_A$,
with values in $[L,H]$, form a set $E$ of positive measure.
Discard face endpoints, quantile discontinuities, and
nondifferentiability points of $r$.
Contacts at $L$ or $H$ can also be discarded: a
positive-measure set of contacts at one value would imply
a zero-intercept face, which can occur only at the top.
Choose a left density point $\gamma$ of the remaining set.
Concavity and maximality of $I_A$ give
$0<z(\gamma)<z_A<U_r(h_A)$ and $z(\gamma)<U_r(\gamma)\in(L,H)$.
The contact selection in the proof of
\Cref{thm:no-reentry-rigorous} gives
$\alpha\in E$, arbitrarily close to $\gamma$, with
$|E\cap(\alpha,\gamma)|>0$, such that all posterior
values on $[\alpha,\gamma]$ remain in $[L,H]$ and
$0<z(s)<\min\{z_A,U_r(\gamma)\}$ almost everywhere.
On this block,
$U_r'=(z-U_r)/s<0$ almost everywhere.
Thus $v(s)=U_r(s)>U_r(\gamma)$ on
$E\cap(\alpha,\gamma)$.

Pool the block at $b=U_r(\gamma)$.
Monotonicity gives $v\ge b$ throughout the block,
and the strict inequality on the contact set gives
$D=\int_\alpha^\gamma[v(s)-b]\,\dd s>0$. By \Cref{lem:binary-chord-cost}, the loss in the objective
is strictly below $q(z_A)D$.
Take the block short enough that $D<\kappa_{I_A}$.
We can then partially regularize $I_A$ to raise
the mean by $D$ and the objective by $q(z_A)D$.
This gain exceeds the loss from pooling.
Together, the two changes preserve the mean and
keep all altered posterior values in $[L,H]$.
They are therefore feasible by
\Cref{lem:finite-interval-feasibility} and strictly
improve the objective, contradicting optimality.

Now suppose two trading faces have their atoms
and entire following gaps in $[L,H]$.
All slopes between them are positive, and all
contact values lie in $[L,H]$. Any interval
between them containing no trading face would
consist almost everywhere of regular active
contacts, which the first part rules out.

If the upper face $I_A$ has no immediate lower
successor, trading faces must accumulate at $h_A$.
Following the construction in the proof of
\Cref{lem:binary-successor}, pool a sufficiently
short block just below $I_A$ that contains two
lower faces and satisfies $U_r>z_A$.
This lowers the mean by $0<D<\kappa_{I_A}$,
with a loss in the objective strictly smaller than
$q(z_A)D$. Partial regularization of $I_A$
restores the mean and gives a gain exceeding
this loss. All altered values remain in $[L,H]$,
so the combined change is feasible and improves
the objective. Trading faces therefore cannot
accumulate at $h_A$.

There must consequently be an immediate lower
successor. It must begin at $h_A$, since any
interval between the two faces would contain
regular active contacts on a set of positive
measure, which we have already ruled out.
The faces are therefore adjacent, and
\Cref{lem:binary-adjacent-elimination} gives
a strict improvement. \end{proof}

We count nonterminal trading faces by their following support gaps.
For a nonterminal trading face $I=[\ell,h]$ with
$r(s)=A+zs$, the following gap has positive width because
\[
 U_r(h)=z+\frac{A}{h}
 <z+\frac{A}{\ell}=U_r(\ell),
 \qquad A>0,\quad \ell<h.
\]
By \Cref{lem:finite-band-order}, at most $N-1$ such faces
have their atom and entire gap within a single primitive
interval. Every remaining gap must contain an interior
primitive valuation strictly between its endpoints.
Otherwise both endpoints would lie in one closed
primitive interval. Since the gaps are disjoint, they
cannot contain the same primitive valuation. There are
only $N-2$ interior primitive valuations, so there are
at most $N-2$ further faces. Altogether, there are at
most $2N-3$ nonterminal trading faces.

Remove these faces and the possible zero-intercept face
at the top. On each remaining open interval of active
ranks, the quantile equals the continuous function $U_r$.
Virtual value is strictly increasing in lower rank
$u=1-s$ since, by monotonicity, equality at two distinct ranks
would make it constant between them and hence define
another trading face.

\subsection{Separation and the highest type}
\label{app:finite-separation}
\label{app:finite-local-equations}
\label{app:finite-top}

With $u$ denoting rank in the full posterior distribution,
$Q$ its quantile, and $\varphi=\phi_Q$, the following lemma
establishes the local separating conditions
\begin{equation}
 A_i+(Q-\varphi)q'(\varphi)-(1-\beta)q(\varphi)=0,
 \qquad
 A_i'=\frac{q(\varphi)-\theta_i}{1-u},
 \qquad
 Q'=\frac{Q-\varphi}{1-u}.
 \label{eq:finite-local-equations}
\end{equation}
Here $\theta_i$ is the local mean multiplier on a
separating range with posterior values in
$(v_i,v_{i+1})$.

\begin{lemma}[Separation inside a primitive interval]
\label{lem:finite-separating}
Under \Cref{ass:cost,ass:curvature}, let $G$ maximize $J_\beta$
over $\mathcal G(H)$ and let $Q$ be its quantile. Let $I$ be an
open interval of regular active ranks outside the affine faces,
with $Q(I)\subset(v_i,v_{i+1})$. On $I$, $Q$ and
$\varphi=\phi_Q$ are $C^1$ and satisfy
\eqref{eq:finite-local-equations} for constants
$\theta_i,\xi_i$, where $\xi_i$ is the integration
constant in $A_i$. Moreover,
\begin{equation}
 q(\varphi)\ge\theta_i,\qquad
 \varphi'=
 \frac{q(\varphi)-\theta_i+(Q-\varphi)q'(\varphi)}
 {(1-u)\{(2-\beta)q'(\varphi)-(Q-\varphi)q''(\varphi)\}}>0.
 \label{eq:finite-local-phi}
\end{equation}
\end{lemma}

\begin{proof}
Fix a compact interval $J=[\alpha,\gamma]$ inside $I$
and put $w(u)=-\log(1-u)$.
To construct variations preserving the endpoint contacts and mean,
use directions $h=k\circ\varphi$, where
$k\in C_c^1((0,1))$ is chosen so that $h$ has compact
rank support in the interior of $J$.
For sufficiently small $|\varepsilon|$, the map
$z\mapsto z+\varepsilon k(z)$ is increasing and
maps $[0,1]$ into itself. Hence
$\varphi+\varepsilon h$ remains nondecreasing
and takes values in $[0,1]$.
Choose $h$ so that
$\int_J h(u)\,\dd u=0$ and
$\int_J w(u)h(u)\,\dd u=0$.

Reconstruct the revenue curve on $J$ from these
virtual values, keeping its value at $\gamma$
fixed.
The quantile change is
\[
 \Delta Q(u)=\frac{\varepsilon}{1-u}
                  \int_u^\gamma h(t)\,\dd t,
 \qquad u\in J.
\]
Since $h$ is chosen so that $\int_J h(u)\,\dd u=0$, the contact at $\alpha$
is also preserved, so the revenue curve and quantile
can remain unchanged outside $J$. Using this condition
and $\int_J w(u)h(u)\,\dd u=0$, we obtain
\[
 \int_J\Delta Q(u)\,\dd u
 =\varepsilon\int_J[w(t)-w(\alpha)]h(t)\,\dd t=0.
\]
The perturbation therefore preserves both endpoint
contacts and the mean. For sufficiently small
$|\varepsilon|$, all altered posterior values remain
in $(v_i,v_{i+1})$, so every $D_j$ is unchanged by
\Cref{lem:finite-interval-feasibility}.

The reconstructed revenue curve is concave because
the perturbed virtual values remain nondecreasing.
On $J$, it equals the new raw revenue curve.
Outside $J$, the original curve and its contacts
are unchanged. Any concave majorant must therefore
lie above the reconstructed curve, which is thus
the new concavified revenue curve.

Since $Q$ and $\varphi$ are unchanged outside $J$,
the calculation in
\eqref{eq:binary-cited-first-variation} gives the
first variation $\int_J H_J(u)h(u)\,\dd u$, where
\[
 H_J(u)=\int_\alpha^u
          \frac{q(\varphi(t))}{1-t}\,\dd t
       +(Q(u)-\varphi(u))q'(\varphi(u))
       -(1-\beta)q(\varphi(u)).
\]
The construction in \Cref{lem:binary-two-multipliers}
gives two directions that allow the contact at
$\alpha$ and the mean to be adjusted independently.
For changes preserving both, optimality requires
the first variation to vanish. The same argument
as in that lemma, including the approximation of
indicators of rank intervals, gives constants
$\theta_J,\xi_J$ satisfying
$H_J-\theta_Jw+\xi_J=0$ almost everywhere on $J$.

To compare the equations on different compact
subintervals, write their cumulative integrals
from the same reference rank and adjust the
additive constants accordingly. On an overlap
of positive length, subtract the two equations.
The terms involving $Q$ and $\varphi$ cancel.
For multiplier pairs $(\theta,\xi)$ and
$(\widetilde\theta,\widetilde\xi)$, we obtain
$(\theta-\widetilde\theta)w(u)=\xi-\widetilde\xi$
almost everywhere on the overlap. Since $w$
is not constant there, $\theta=\widetilde\theta$
and then $\xi=\widetilde\xi$. Comparing overlapping
subintervals therefore gives one mean multiplier
$\theta_i$ and one integration constant $\xi_i$
for $A_i$ throughout $I$. These give
\eqref{eq:finite-local-equations} almost everywhere
on $I$.

On this active interval, $\varphi>0$, and
reconstruction gives $Q>\varphi$. The integral
formulas show that $Q$ and $A_i$ are locally
absolutely continuous. As in
\Cref{lem:binary-interior-dynamics}, curvature
rules out two distinct virtual values solving
the Euler equation at the same $(A_i,Q)$.
A jump in $\varphi$ is therefore impossible.
Thus $\varphi$ is continuous, and the Euler
equation holds everywhere. The implicit-function
theorem then expresses $\varphi$ locally as a
$C^1$ function of $(A_i,Q)$. The differential
equations for $A_i$ and $Q$ now have continuous
right-hand sides, so all three functions are $C^1$.

It remains to show that $q(\varphi)\ge\theta_i$.
Fix $u_0\in I$ and pool a short contact block
around it. Restore the lost mean on a fixed,
disjoint subinterval of $I$, using the two
directions constructed above. This adjustment
preserves the endpoint contacts and has a
first-order payoff effect of $\theta_i$ per
unit of mean restored. Together, the two changes
preserve the mean and keep all altered posterior
values in $(v_i,v_{i+1})$. Every information
constraint is therefore unchanged.

The calculation in the proof of
\Cref{lem:shadow-quality} shows that, as the
pooled block shrinks to $u_0$, the loss in the
objective per unit of mean removed tends to
$q(\varphi(u_0))$. If $\theta_i$ were larger,
restoring the mean elsewhere would yield more
than pooling costs, contradicting optimality.
Hence $q(\varphi(u_0))\ge\theta_i$.
Since $u_0$ was arbitrary, the inequality holds
throughout $I$. Differentiating the Euler equation
gives \eqref{eq:finite-local-phi}. Its denominator
is positive by curvature. Its numerator is
positive because $q(\varphi)\ge\theta_i$,
$Q>\varphi$, and $q'>0$.
\end{proof}

\begin{corollary}[Local mean multipliers across primitive intervals]
\label{cor:finite-mean-multipliers}
Under the hypotheses of \Cref{lem:finite-separating}, suppose an
optimum has separating ranges in $(v_i,v_{i+1})$ and
$(v_k,v_{k+1})$, with $i<k$. Their local mean multipliers satisfy
$\theta_i\le\theta_k$. If $D_j(G)<\kappa_j$ for every $i<j\le k$,
then $\theta_i=\theta_k$.
\end{corollary}

\begin{proof}
Choose a compact rank interval $J=[\alpha,\gamma]$ inside either
separating range, and write $\ell$ for its primitive interval
index. Ranks are measured in the full distribution. The two
directions constructed in the proof of \Cref{lem:finite-separating}
have linearly independent effects on the bottom contact and the
mean. With $w(u)=-\log(1-u)$, take a linear combination $h$ such
that
\[
 \int_J h(u)\,\dd u=0,\qquad
 \int_J w(u)h(u)\,\dd u=1.
\]
For small $|t|$, the perturbation $\varphi+th$ preserves
monotonicity, since $h$ is a smooth function of $\varphi$ supported
in the interior of the block. Reconstruction from the right endpoint
fixes both revenue contacts by the first equality. The perturbation
vanishes near the endpoints, so the regular block joins the original
envelope with the same endpoint slopes. Its posterior values stay
in $(v_\ell,v_{\ell+1})$.
Its contribution to the full mean changes by exactly
$t\int_J[w(u)-w(\alpha)]h(u)\,\dd u=t$. The local first-variation
condition $H_J=\theta_\ell w-\xi_J$ in the same proof gives
$\int_J H_J h=\theta_\ell$. Thus the change in the objective on
this block is $t\theta_\ell+o(|t|)$.

Now add $\varepsilon>0$ units of mean on the lower block and
remove the same amount on the upper block. For $i<j\le k$, the
two blocks lie on opposite sides of $v_j$, so $D_j$ falls by exactly
$\varepsilon$. At thresholds below both blocks, the changes
in $D_j$ cancel. Above both blocks, each block
contributes zero. The mean is unchanged and no
$D_j$ increases, so the transfer is feasible
by \eqref{eq:finite-feasibility}. Its first-order gain is
$\varepsilon(\theta_i-\theta_k)$, which gives $\theta_i\le\theta_k$
by optimality. If all the intervening constraints are slack, a
sufficiently small transfer in the opposite direction is feasible
as well, giving equality.
\end{proof}

\begin{lemma}[Continuous fit at a primitive value]
\label{lem:finite-primitive-fit}
Let $\bar u$ lie in the interior of a regular active
component of an optimal posterior distribution, with
$Q(\bar u)=v_j$ for an interior primitive value $v_j$.
Suppose $Q$ is continuous and strictly increasing through
$v_j$, without an atom or support gap.
Then $\varphi(\bar u-)=\varphi(\bar u+)$.
Consequently, quality and the derivative of the posterior
quantile are continuous across $v_j$:
\begin{equation}
 q(\varphi(\bar u-))=q(\varphi(\bar u+)),
 \qquad
 Q'(\bar u-)=Q'(\bar u+).
 \label{eq:finite-primitive-quality-fit}
\end{equation}
The one-sided limits of $A_{j-1}$ and $A_j$ also agree.
The right-hand rank derivatives of virtual value and quality
are weakly smaller than their respective left-hand derivatives.
Both hold with equality if the information constraint
at $v_j$ is slack.
\end{lemma}

\begin{proof}
Write $\varphi_-=\varphi(\bar u-)$ and
$\varphi_+=\varphi(\bar u+)$.
Monotonicity, interior activity, and regular reconstruction
give $0<\varphi_-\le\varphi_+<v_j$.
The last inequality is strict because virtual value
increases on the regular interval above $\bar u$.
Suppose, towards a contradiction, that
$\varphi_-<\varphi_+$.

For sufficiently small $h>0$, the block
$I_h=[\bar u-h,\bar u+h]$ lies in the regular component,
its posterior values lie in $(v_{j-1},v_{j+1})$, and
$\sup_{I_h}\varphi<b_h:=Q(\bar u-h)<v_j$.
Pool all posterior values on this block at $b_h$.
By \Cref{lem:exact-chord-collapse}, expressed in lower
ranks, the new concavified revenue curve is the chord
through the block's endpoint contacts and is unchanged
elsewhere. Its virtual value on the block and the
reduction in mean are
\[
 \bar\varphi_h=\frac1{2h}\int_{I_h}\varphi(u)\,\dd u,
 \qquad
 D_h=\int_{I_h}[Q(u)-b_h]\,\dd u>0.
\]
On either side of $\bar u$,
\Cref{lem:finite-separating} gives
$Q'=(Q-\varphi)/(1-u)$.
This derivative is uniformly bounded nearby.
Together with continuity of $Q$, this gives
$D_h=O(h^2)$.

Since the quantile and virtual values outside $I_h$
are unchanged, the gain in the objective is
\begin{align*}
 \Delta J_h
 &=2hF_{b_h,\beta}(\bar\varphi_h)
   -\int_{I_h}F_{Q(u),\beta}(\varphi(u))\,\dd u
   \notag\\
 &=
 \left[
  2hF_{b_h,\beta}(\bar\varphi_h)
  -\int_{I_h}F_{b_h,\beta}(\varphi(u))\,\dd u
 \right]
 -\int_{I_h}[Q(u)-b_h]q(\varphi(u))\,\dd u.
\end{align*}
The magnitude of the last term is bounded by $\bar qD_h=O(h^2)$.
The bracketed term is positive by strict concavity.
Dividing it by $2h$ and letting $h\downarrow0$ gives
\[
 F_{v_j,\beta}\!\left(
       \frac{\varphi_-+\varphi_+}{2}\right)
 -\frac12F_{v_j,\beta}(\varphi_-)
 -\frac12F_{v_j,\beta}(\varphi_+)>0.
\]
Thus, for all sufficiently small $h$, pooling raises
the objective by at least a positive multiple of $h$,
while reducing the mean by only $O(h^2)$.

We restore the lost mean below $v_j$.
Fix a compact regular subinterval below $\bar u$,
with posterior values in $(v_{j-1},v_j)$ and disjoint
from $I_h$ for all small $h$.
The compensating directions constructed in the proof of
\Cref{lem:finite-separating}, using
\Cref{lem:binary-two-multipliers}, allow its mean to
increase by exactly $D_h$ while preserving its endpoint
contacts and monotonicity.
The first variation along this fixed direction is finite,
so its effect on the objective is $O(D_h)=O(h^2)$.
For small $h$, all posterior values changed by this
compensation remain in $(v_{j-1},v_j)$.

We check the combined change at each primitive threshold.
The call at each primitive threshold below $v_j$ is unchanged since all affected posterior values remain above the threshold, and the two mean changes cancel. Pooling lowers the call at $v_j$, while compensation leaves it unchanged because it affects only posterior values below $v_j$.
At thresholds above $v_j$, both changes occur below
$v_{j+1}$ and have no effect.
Support and the overall mean are also preserved.
The combined change is therefore feasible and retains
a strictly positive gain in the objective for small $h$,
contradicting optimality.

Hence $\varphi_-=\varphi_+=:\varphi_*$.
Since $q$ is continuous, quality is continuous as well.
The formula for $Q'$ gives
$Q'(\bar u\pm)=(v_j-\varphi_*)/(1-\bar u)$,
proving \eqref{eq:finite-primitive-quality-fit}.
One-sided limits in the static equation
\eqref{eq:finite-local-equations} give
\[
 A_{j-1}(\bar u-)=A_j(\bar u+)
 =(1-\beta)q(\varphi_*)
   -(v_j-\varphi_*)q'(\varphi_*).
\]
Take one-sided limits in \eqref{eq:finite-local-phi}.
The denominator has the same positive limit on both sides.
The difference between the right-hand and left-hand limits
of the numerator is $\theta_{j-1}-\theta_j$.
By \Cref{cor:finite-mean-multipliers}, this difference is
nonpositive and is zero when $D_j(G)<\kappa_j$. Multiplication
by $q'(\varphi_*)>0$ gives the same for the rank
derivatives of quality.

\end{proof}

\begin{lemma}[Separation and efficient trade at the top]
\label{lem:finite-terminal}
Under \Cref{ass:cost,ass:curvature}, every optimum has a
regular upper component consisting of a nonempty separating
interval followed by an atom of positive mass at one.
Virtual value converges to one at the upper end of separation,
and the atom receives the efficient quality $q(1)$.
\end{lemma}

\begin{proof}
At the no-learning distribution $\delta_m$, every internal
information constraint is strictly slack:
$D_j(\delta_m)=(m-v_j)_+<\kappa_j$ for $2\le j<N$.
The inequality is strict because the prior puts positive
mass on both sides of $v_j$.
Thus, for sufficiently small $\epsilon>0$, the
equal-probability spread to $m-\epsilon$ and $m+\epsilon$
is feasible and remains in $[a,1]$.
For $N=2$, there are no internal constraints to check.
The calculation in the proof of
\Cref{lem:binary-nonempty-continuation} gives a gain in the objective
of $(1-\beta)\epsilon q(m)+o(\epsilon)>0$.
Hence $\delta_m$ is not optimal, and the optimal value
is strictly positive. In particular, some buyers trade.

The finite bound on nonterminal trading faces leaves
a regular upper tail, which could at this stage
consist only of a single atom.
Suppose this component has no separating interval.
If another posterior type trades, the top two faces must
be adjacent pure trading atoms.
The contraction in \Cref{lem:top-atom-refinement} then
strictly raises $J_\beta$ and preserves feasibility
under the original prior.

If the top atom is the only trading type, there must be
excluded mass, since no learning has already been ruled out.
Write $s$ for the top atom's probability and $p$ for its
value. The original payoffs are $(0,s\pi(p))$.
Pool all excluded types at their conditional mean.
This is a contraction in convex order, so it preserves
feasibility. The top atom still has virtual value $p$, so the
new distribution gives profit of at least
$s\pi(p)$ and nonnegative consumer surplus.
Pooling therefore cannot lower $J_\beta$.
Applying \Cref{lem:top-atom-refinement} then gives
a strict gain, even if the lower atom has a
nonpositive virtual value. This contradicts
optimality. The regular upper component must
therefore contain a nonempty separating interval.

Let $p=\sup\supp G$. Choose a contact rank in
the separating interval with posterior value
$B<p$, taking $B$ above every primitive value
below $p$. Condition on posterior values above
$B$, and let $u$ denote rank within this tail.
Its separating part lies in a single primitive
interval, so \Cref{lem:finite-separating} gives
multipliers $\theta,\xi$ and the local Euler
equations.

Suppose separation continues to rank one.
Since $q(\varphi)$ is strictly increasing and
never below $\theta$, the difference
$q(\varphi)-\theta$ is bounded away from zero
for ranks sufficiently close to one.
The equation $A'=[q(\varphi)-\theta]/(1-u)$ then gives
$A\to+\infty$ as $u\uparrow1$.
But the static Euler equation gives
$A<(1-\beta)q(\varphi)\le(1-\beta)q(p)$,
a contradiction.
Thus the upper component ends in an atom of positive
mass at $p$. Write $b\in(0,1)$ for its starting rank,
so $Q=\varphi=p$ on $[b,1]$.

Suppose $p<1$.
Since $q(p)>\theta$, the upward terminal variation in
the proof of \Cref{lem:binary-top-fit} gives a strictly
positive first-order gain after two interior directions
restore the bottom contact and mean.
Choose both compensating directions in the separating
part above $B$.
The total change in virtual values has zero integral,
so posterior values below $B$ remain unchanged.

It remains to check the information constraints.
At every primitive threshold below $p$, all altered
posterior values remain above the threshold.
Mean preservation therefore leaves its call unchanged.
At each internal threshold $v_j\ge p$, the original
constraint is strictly slack:
$D_j(G)=0<\kappa_j$.
For sufficiently small perturbations, these
constraints remain slack by continuity. Since
$p<1$, we can also keep all posterior values
below one. The proposed improvement is therefore
feasible, contradicting optimality. Thus $p=1$.

Finally, consider the downward variation at the
start of the top atom used in the proof of
\Cref{lem:binary-top-fit}. Choose both compensating
directions in the separating region above
$v_{N-1}$. The combined change preserves the
mean and keeps all altered posterior values
in $[v_{N-1},1]$. By
\Cref{lem:finite-interval-feasibility}, every
information constraint is therefore unchanged.
The one-sided calculation in that proof gives
a strict gain if $\varphi(b-)<1$, again
contradicting optimality.
Hence virtual value converges to one along separation
and equals one on the terminal atom, which receives
quality $q(1)$.
\end{proof}

\subsection{Proof of \texorpdfstring{\Cref{thm:finite-structure}}{Theorem~\ref{thm:finite-structure}}}
\label{app:finite-assembly}

\begin{proof}
Existence follows from Section~\ref{subsec:screening}.
Fix an optimal distribution.
Since quality is nondecreasing in rank, excluded buyers,
if present, occupy a single lower interval.

A constant positive quality over a rank interval
corresponds to an affine face of the revenue curve.
By \Cref{lem:active-face-purity}, every nonterminal
trading face contains a single posterior value.
A zero-intercept face, which can occur only at the top,
also has a constant posterior value.
Thus distinct trading posterior values receive
distinct qualities.

For a nonterminal trading face $I=[\ell,h]$ in upper-tail
ranks, purity places its atom at $w=U_r(h)$.
Its positive intercept gives $y^+:=U_r(\ell)>w$,
and the following support gap is $(w,y^+)$.
A gap between trading posterior values makes both
the quantile and raw revenue jump. Since the
concavified curve is continuous, it must lie
strictly above raw revenue on an interval on
the lower-value side of the gap. This interval
is part of a trading face, whose buyers share
one posterior value. The gap therefore begins
at that atom.

By \Cref{lem:finite-band-order}, each primitive
interval contains at most one trading atom
together with its entire following gap, and
that atom lies below any separation in the
interval. The counting argument in
Appendix~\ref{app:finite-feasibility} then gives
at most $2N-3$ nonterminal trading atoms.

Support gaps are disjoint as open intervals.
A crossing gap starting in $[v_i,v_{i+1})$
must end strictly above $v_{i+1}$. If two such
gaps began at $w_1<w_2$ in this interval,
the atom at $w_2$ would lie inside the first
gap. Thus at most one crossing gap can start
in each $[v_i,v_{i+1})$.
A gap ending at $v_j$ does not cross that
primitive value. If there is an atom at $v_j$,
it belongs to the next face, and the two faces
share an endpoint contact.

Remove the finitely many trading faces, including
the terminal face. On the remaining active ranks,
the quantile equals the continuous regularized
quantile, so the remaining open rank intervals
are regular. Suppose two separated regular
regions had posterior values in the same open
primitive interval $(v_i,v_{i+1})$. A trading
face between them would have its atom and
entire following gap in that interval. This
would place the atom above separation in the
same interval, contrary to
\Cref{lem:finite-band-order}.
Each open primitive interval therefore contains
at most one separating portion.
Every nondegenerate separating region contains
such a portion, so there are at most $N-1$
connected separating regions.

Within each open primitive interval,
\Cref{lem:finite-separating} gives $C^1$ posterior
and virtual-value profiles and strictly increasing
quality. With $u$ denoting rank in the full posterior
distribution,
$Q'=(Q-\varphi)/(1-u)>0$.
The inverse quantile therefore has a positive continuous
derivative, giving the posterior density
$g(Q(u))=1/Q'(u)$.

If two separating regions meet at a primitive
value without an atom or support gap,
\Cref{lem:finite-primitive-fit} gives continuity
of virtual value, quality, and $Q'$ there.
They thus form one separating region,
throughout which posterior values and qualities
are continuous and strictly increasing.
The mean multiplier and the one-sided derivatives
of virtual value and quality may still differ
on the two sides. Finally, \Cref{lem:finite-terminal} shows that
the highest separating interval is nonempty
and ends at an atom of positive mass at one.
Virtual value approaches one along this interval,
and buyers at the atom receive the efficient
quality $q(1)$.
\end{proof}

\section{Joint learning and valuation-dependent segmentation}
\label{app:finite-joint-design}

We prove \Cref{thm:finite-joint-design}. Fix the finite prior $H$
in \eqref{eq:primitive-prior}, write $a=v_1$, and maintain
$0\le\beta<1$ and \Cref{ass:cost,ass:curvature} at that weight.

\subsection{Aggregate feasibility, existence, and the segment bound}
\label{app:finite-joint-feasibility}

\begin{lemma}[Aggregate implementation]
\label{lem:finite-joint-implementation}
A finite collection of segment probabilities and posterior
distributions $\{(\lambda_k,G_k)\}$ is implementable under $H$
if and only if $\overline G:=\sum_k\lambda_kG_k\preceq_\cx H$.
The joint attainable payoff set is compact and convex, and
some joint optimum uses at most $N$ positive-probability segments.
\end{lemma}

\begin{proof}
In any implementation, $Y=\E[V\mid Y,K]$, so Jensen's inequality
gives $\overline G\preceq_\cx H$.

Conversely, suppose $\overline G\preceq_\cx H$. By
Section~\ref{subsec:feasible-learning}, there is a coupling
of $V\sim H$ and $Y\sim\overline G$ with $\E[V\mid Y]=Y$.
Let $\nu_y$ be the conditional distribution of $V$ given $Y=y$ and,
for each $\lambda_k>0$, set $H_k=\int\nu_y\,G_k(\dd y)$.
Since $G_k\ll\overline G$, the beliefs $\nu_y$ have mean $y$
for $G_k$-almost every $y$. Hence $G_k\preceq_\cx H_k$ and
$\sum_k\lambda_kH_k=H$. The assignment and learning policies
in Appendix~\ref{app:sd-implementation} implement this collection.

For existence, write $p=(p_1,\ldots,p_N)$ and let $\Delta_N$
be the simplex of probabilities on the primitive states.
For $h\in\Delta_N$, put $H_h=\sum_i h_i\delta_{v_i}$.
The set
$\mathcal A
 :=\{(h,\CS(G),\Pi(G)):
       h\in\Delta_N,\ G\in\Delta([a,1]),\ G\preceq_\cx H_h\}$
is compact. Indeed, the set of feasible distributions is closed, including
when some coordinates of $h$ vanish, and both screening payoffs
are weakly continuous \citep[Proposition~A.7]{kyriazis2026markets}.
All points of $\mathcal A$ satisfy $\sum_i h_i=1$.
By Carath\'eodory's theorem, every point of
$\operatorname{co}\mathcal A$ is a convex combination
of at most $N+2$ points of $\mathcal A$. The convex
hull is also compact
\citep[Propositions~1.2.1--1.2.2]{bertsekas2009}.
The points whose first $N$ coordinates equal $p$
form a nonempty compact convex set. By the
implementation argument above, projecting this
set onto the two payoff coordinates gives exactly
the joint attainable payoff set. The objective therefore attains its maximum.

Take an optimal finite policy and let $h_k$ be
the vector of conditional primitive probabilities
in segment $k$. These vectors satisfy
$\sum_k\lambda_kh_k=p$.
If more than $N$ segments have positive probability,
their vectors are linearly dependent. Choose
coefficients $\alpha_k$, not all zero, such that
$\sum_k\alpha_kh_k=0$.
Since the coordinates of each $h_k$ sum to one,
$\sum_k\alpha_k=0$.
Keep $h_k$ and $G_k$ fixed and replace $\lambda_k$
by $\lambda_k+t\alpha_k$. For sufficiently small
$t$ of either sign, the probabilities remain
nonnegative and satisfy the same prior constraint.
Optimality therefore requires
$\sum_k\alpha_kJ_\beta(G_k)=0$.
Now vary $t$ until a probability first reaches zero.
The prior and objective remain unchanged, so we
can remove that segment. Repeating this step
gives an optimal policy with at most $N$ segments.
\end{proof}

\subsection{Formulation of the joint program}
\label{app:finite-joint-program}

Let $\mathcal T_a=\{(y,z):a\le y\le1,\ 0\le z\le y\}$,
and let $\mu_Y$ be the posterior-mean marginal of
$\mu\in\Delta(\mathcal T_a)$.

\begin{proposition}[Value of the joint program]
\label{prop:finite-joint-program}
The optimal joint value satisfies
\begin{equation}
 \begin{aligned}
 W_H^{\mathrm{SD}}(\beta)
 &=\max_{\mu\in\Delta(\mathcal T_a)}
       \int F_{y,\beta}(z)\,\mu(\dd y,\dd z)\\
 &\text{subject to}\quad \mu_Y\preceq_\cx H,\\
 &\hspace{18mm}
   \int_{\{y\ge t\}}(z-t)\,\mu(\dd y,\dd z)\ge0
   \quad\text{for every }t\in[a,1].
 \end{aligned}
 \label{eq:finite-joint-LP}
\end{equation}
The maximum is attained by a measure induced by an optimal
policy with at most $N$ segments. Every optimal joint policy
induces a maximizing measure.
\end{proposition}

We first derive the screening constraints from a quality upgrade
available to the seller. We then implement an optimal wedge rule
for any fixed finite posterior distribution and approximate
arbitrary feasible distributions. Until attainment is established,
write $S$ for the supremum in \eqref{eq:finite-joint-LP}.
It is finite because the objective is continuous on the compact
set $\mathcal T_a$. Its feasible set contains $\delta_{(m,m)}$.

\paragraph{Screening constraints.}
Consider a market with posterior distribution $G$ and an optimal direct
menu $(x_G,T_G)$. Set $z_G(y)=c'(x_G(y))$.
By the screening characterization, this is the ironed virtual
value with negative values replaced by zero. Thus
$0\le z_G(y)\le y$ almost surely and
\begin{equation*}
 x_G(y)=q(z_G(y)),\qquad
 J_\beta(G)=\int F_{y,\beta}(z_G(y))\,G(\dd y).
\end{equation*}

Fix a threshold $t$. Add quality $\varepsilon$ to every contract
intended for a type $y\ge t$, and raise its transfer by
$t\varepsilon$:
\[
 x_\varepsilon(y)=x_G(y)+\varepsilon1_{\{y\ge t\}},
 \qquad
 T_\varepsilon(y)=T_G(y)+t\varepsilon1_{\{y\ge t\}}.
\]
Small positive $\varepsilon$ respects the quality bound because
$x_G\le q(1)<\bar Q$. For type $y$, the utility difference between
reporting truthfully and reporting $\widehat y$
changes by
$\varepsilon(y-t)(1_{\{y\ge t\}}-1_{\{\widehat y\ge t\}})\ge0$.
Truthful reporting therefore remains optimal.
Truthful utility increases by $\varepsilon(y-t)_+$
and is unchanged at $y=t$, so participation is
preserved even when $G$ has an atom at $t$.

The upgrade cannot raise profit. Differentiating at
$\varepsilon=0+$ gives
$\int_{[t,1]}[t-c'(x_G(y))]\,G(\dd y)\le0$, or
\begin{equation*}
 \int_{[t,1]}(z_G(y)-t)\,G(\dd y)\ge0
 \qquad(a\le t\le1).
\end{equation*}
Differentiation under the integral is valid because $c'$ is
bounded on a compact neighborhood of the optimal quality range.

Averaging the distributions of $(Y,z_{G_k}(Y))$ across segments gives a
measure $\mu$ with the same weighted payoff. It satisfies the
screening inequalities above, and its posterior marginal is
dominated by $H$ by \Cref{lem:finite-joint-implementation}.
Hence $W_H^{\mathrm{SD}}(\beta)\le S$.

\paragraph{Conditional averaging and wedges.}
For any feasible $\mu$, put $G=\mu_Y$.
Replacing $Z$ by $\E_\mu[Z\mid Y]$ preserves every constraint
and weakly raises the objective by concavity of $F_{y,\beta}$.
By strict concavity, equality requires $Z$ to
be constant conditional on $Y=y$ for
$G$-almost every $y>0$. At $y=0$, the bound
$0\le Z\le Y$ forces $Z=0$.

Write $d(y)=y-\E_\mu[Z\mid Y=y]$,
$\mathsf E(\dd y)=d(y)G(\dd y)$, and
$u_\beta(y,d)=d q(y-d)+\beta\pi(y-d)$.
The wedge $d(y)$ is the difference between posterior mean and
virtual value. $\mathsf E(A)$ is its total over types in $A$.
With $C_G(t)=\int(y-t)_+G(\dd y)$, the aggregate supremum is
equivalently
\begin{equation*}
 S=\sup_{G,\mathsf E}\int u_\beta(y,d(y))\,G(\dd y),
 \qquad d=\frac{\dd\mathsf E}{\dd G},
\end{equation*}
subject to
\begin{equation}
 G\preceq_\cx H,\qquad
 0\le\mathsf E\le yG,\qquad
 \mathsf E([t,1])\le C_G(t)
 \quad(a\le t\le1).
 \label{eq:finite-joint-wedge-constraints}
\end{equation}
The measure bound gives $0\le d(y)\le y$ almost surely.
Conversely, a feasible pair defines a feasible aggregate measure
by setting $Z=Y-d(Y)$.

\paragraph{Implementation with finitely many posterior means.}
\label{implementation_joint}
Fix
\begin{equation}
 G=\sum_{i=1}^M w_i\delta_{y_i}\preceq_\cx H,
 \qquad a\le y_1<\cdots<y_M\le1,\quad w_i>0.
 \label{eq:finite-joint-grid-learning}
\end{equation}
Holding $G$ fixed, choose the total wedge $e_i$ at each type:
\begin{equation}
 \max_e\sum_{i=1}^M w_i u_\beta(y_i,e_i/w_i)
 \label{eq:finite-joint-grid-objective}
\end{equation}
subject to
\begin{equation*}
 0\le e_i\le y_iw_i,\qquad
 \sum_{j\ge i}e_j\le\sum_{j>i}w_j(y_j-y_i),
 \quad1\le i\le M.
\end{equation*}
Zero wedges are feasible, and compactness gives a maximizer.
It suffices to check the screening constraints
at $y_1,\ldots,y_M$. For $y_i<t<y_{i+1}$,
the left-hand side equals its value at $y_{i+1}$,
while $C_G(t)\ge C_G(y_{i+1})$.
The same argument applies below $y_1$.
Above $y_M$, both sides are zero.

\begin{lemma}[Finite posterior implementation]
\label{lem:finite-joint-grid}
For the fixed distribution $G$ in
\eqref{eq:finite-joint-grid-learning}, the value of
\eqref{eq:finite-joint-grid-objective} equals the maximum of the designer's objective
over joint policies with aggregate posterior
distribution $G$. Every optimal wedge rule
can be implemented with at most $M$ segments.
Each type receives its prescribed virtual value
in every segment where it appears.
\end{lemma}

\begin{proof}
The screening inequalities and conditional averaging
show that the program's value is an upper bound
on the joint objective at $G$. To attain it, take a
maximizing wedge vector and set $d_i=e_i/w_i$
and $z_i=y_i-d_i$. For every $y_i>0$, we must have
\begin{equation}
 0\le d_i\le y_i-z_{y_i,\beta}^{\dagger}<y_i.
 \label{eq:finite-joint-wedge-cap}
\end{equation}
Otherwise, reducing $d_i$ to
$y_i-z_{y_i,\beta}^{\dagger}$ would strictly raise
the objective and relax every screening constraint
involving $e_i$, contradicting optimality.
In particular, $z_i>0$ whenever $y_i>0$.

The marginal return from adding wedge mass at a type is
\begin{equation}
 \begin{gathered}
 g_{y,\beta}(z)=(1-\beta)q(z)-(y-z)q'(z),\\
 \partial_zg_{y,\beta}
 =(2-\beta)q'(z)-(y-z)q''(z)>0,\qquad
 \partial_yg_{y,\beta}=-q'(z)<0
 \quad(0<z\le y).
 \end{gathered}
 \label{eq:finite-joint-wedge-marginal}
\end{equation}
Suppose $0<y_i<y_j$ but $z_i\ge z_j$.
Move a small amount of wedge mass from $j$ to $i$.
The pointwise bounds permit this: $z_i>0$, while
$d_j=y_j-z_j\ge y_j-y_i>0$.
Every screening constraint either stays unchanged or relaxes.
The marginal gain is positive because
$g_{y_i,\beta}(z_i)\ge g_{y_i,\beta}(z_j)>
g_{y_j,\beta}(z_j)$, a contradiction.
Thus virtual values strictly increase across positive types.
At a zero posterior mean, the virtual value is zero.

We now divide the population into segments that
implement these virtual values.
At each step, let $r_i$ be the remaining mass at $y_i$, initially
$r_i=w_i$, and write
$B_i(r)=\sum_{\ell\ge i}r_\ell(z_\ell-y_i)$
for the screening slack. We maintain $B_i(r)\ge0$ for the remaining masses.
In choosing a segment, ignore types whose remaining
mass is zero.

The highest remaining type must have zero wedge, since its
constraint is $-r_i d_i\ge0$. Choose the \emph{lowest} remaining
type with zero wedge, indexed by $h$, as the top of the next
segment. This segment will contain every remaining type below it.
Its top will therefore receive efficient quality, and every
lower type has $d_i>0$.

Give the top type unnormalized mass $b_h=1$, set $b_i=0$
outside this segment, and work downward. If $j$ is the next
remaining type above $i<h$, choose
\begin{equation}
 b_i=\frac{y_j-y_i}{d_i}\sum_{\ell\ge j}b_\ell,
 \qquad\text{so that}\qquad
 b_i d_i=(y_j-y_i)\sum_{\ell\ge j}b_\ell.
 \label{eq:finite-joint-discrete-hazard}
\end{equation}
These masses are positive and finite. Their proportions make
the discrete virtual value at $y_i$ equal to
$y_i-(y_j-y_i)\sum_{\ell\ge j}b_\ell/b_i=y_i-d_i=z_i$.
At the top it is $y_h=z_h$. After normalizing the masses,
the seller's objective is the sum of
$z_i x_i-c(x_i)$ weighted by those probabilities.
Its coordinatewise maximizers $x_i=q(z_i)$ are nondecreasing,
so they also solve the seller's screening problem.

We have $B_i(b)=0$ for $i\ge h$, since there is
no mass above $h$ and $z_h=y_h$.
For consecutive types $i<j\le h$ with positive
remaining mass,
$B_i(b)-B_j(b)=-b_i d_i+(y_j-y_i)\sum_{\ell\ge j}b_\ell=0$
by \eqref{eq:finite-joint-discrete-hazard}.
Induction down the ordered types therefore gives
$B_i(b)=0$ at every threshold occupied before
this step.

Set $\alpha=\min_{i:b_i>0}r_i/b_i$ and assign
mass $\alpha b_i$ of type $y_i$ to this segment.
Its probability is $\lambda=\alpha\sum_i b_i$,
and its conditional probabilities are
$b_i/\sum_\ell b_\ell$.
The remaining masses $r'_i=r_i-\alpha b_i$
are nonnegative, and at least one previously
positive mass becomes zero.
At every previously occupied threshold,
\[
 B_i(r')=B_i(r)-\alpha B_i(b)=B_i(r)\ge0.
\]
If no mass remains, we are done. Otherwise,
take any threshold $t$ at or below the highest
remaining posterior value. Let $y_j$ be the
smallest remaining value at or above $t$.
Then
\[
 \sum_{\ell:y_\ell\ge t}r'_\ell(z_\ell-t)
 =B_j(r')+(y_j-t)\sum_{\ell\ge j}r'_\ell\ge0.
\]
Above the highest remaining value, the left-hand
side is zero. Thus the remaining population
satisfies every screening constraint, and we
can repeat the construction.

Each step assigns all remaining mass of at least
one type, so at most $M$ segments are needed.
If zero is a posterior value, it has zero wedge
and is selected first, forming a singleton segment.
The segment probabilities sum to one, their
posterior distributions average to $G$, and each type
receives its prescribed virtual value wherever
it is assigned. Aggregate profit and consumer
surplus are therefore
$\int\pi(z(y))G(\dd y)$ and
$\int(y-z(y))q(z(y))G(\dd y)$, respectively.
Their weighted sum equals the value in
\eqref{eq:finite-joint-grid-objective}.
Since $G\preceq_\cx H$,
\Cref{lem:finite-joint-implementation} ensures
that this segmentation can be implemented through
assignment and private learning.
\end{proof}

\paragraph{Passage to unrestricted posterior means.}
\label{app:finite-joint-limit}
We approximate posterior values and wedges by conditional averages.
This preserves Bayes plausibility, but may violate the screening inequalities.
We show that reducing wedges at the highest posterior values
restores these inequalities at a cost that vanishes as the
approximation becomes finer.

\begin{lemma}[Finite approximation]
\label{lem:finite-joint-approximation}
For every pair $(G,\mathsf E)$ satisfying
\eqref{eq:finite-joint-wedge-constraints} and every $\rho>0$,
there is a feasible pair with finite posterior support whose
objective is at least $\int u_\beta(y,d(y))G(\dd y)-\rho$.
\end{lemma}

\begin{proof}
Fix $\delta>0$ and let $Y\sim G$. Partition
$\{(y,d):a\le y\le1,\ 0\le d\le y\}$ into finitely many
measurable cells of diameter at most $\delta$ in the maximum norm.
On each cell of positive probability, replace $(Y,d(Y))$ by its
conditional mean $(\widehat Y,\widehat D)$. Then
$|\widehat Y-Y|\le\delta$, $|\widehat D-d(Y)|\le\delta$ and $0\le\widehat D\le\widehat Y$.
The distribution $\widehat G$ of $\widehat Y$ has finite support and satisfies
$\widehat G\preceq_\cx G\preceq_\cx H$ by Jensen's
inequality. Define
\[
 \widehat{\mathsf E}(A)
 =\E[\widehat D\,1_{\{\widehat Y\in A\}}],\qquad
 \widehat d=\frac{\dd\widehat{\mathsf E}}{\dd\widehat G}.
\]
Thus, $0\le\widehat{\mathsf E}\le y\widehat G$. If different cells
have the same posterior mean, $\widehat d$ averages their wedges.
Concavity in the wedge and uniform continuity of $u_\beta$ give
\[
 \int u_\beta(y,\widehat d(y))\,\widehat G(\dd y)
 \ge\E[u_\beta(\widehat Y,\widehat D)]
 =\int u_\beta(y,d(y))\,G(\dd y)+o(1)
 \quad(\delta\downarrow0).
\]

For every $t\in[a,1]$, averaging and the original
screening constraint give
\begin{equation}
 \widehat{\mathsf E}([t,1])
 =\E[d(Y)1_{\{\widehat Y\ge t\}}]
 \le\mathsf E([\max\{a,t-\delta\},1])
 \le C_G(t)+\delta
 \le C_{\widehat G}(t)+2\delta.
 \label{eq:finite-joint-rounding-tail}
\end{equation}
The last two inequalities use the fact that $(y-t)_+$ is
$1$-Lipschitz in both arguments. Remove
$\min\{2\delta,\widehat{\mathsf E}([a,1])\}$ units of wedge mass,
starting at the highest posterior values and removing only part
of the last atom if necessary. The remaining measure $\mathsf E'$
satisfies, at every threshold,
\[
 \mathsf E'([t,1])
 =\bigl(\widehat{\mathsf E}([t,1])-2\delta\bigr)_+
 \le C_{\widehat G}(t).
\]
Also $0\le\mathsf E'\le\widehat{\mathsf E}\le y\widehat G$, so
$(\widehat G,\mathsf E')$ is feasible.

To bound the loss from this correction, take $0\le d'\le d\le y$
and put $z=y-d$ and $z'=y-d'$. Since $q$ is nondecreasing and
$\pi(z')-\pi(z)\ge q(z)(z'-z)$,
\[
 u_\beta(y,d)-u_\beta(y,d')
 \le(1-\beta)q(z)(d-d')
 \le(1-\beta)\bar q(d-d').
\]
Apply this inequality with $d=\widehat d(y)$ and
$d'=\frac{\dd\mathsf E'}{\dd\widehat G}(y)$.
The integral of their difference is the wedge mass
removed, which is at most $2\delta$. The correction
therefore costs at most $2(1-\beta)\bar q\delta$. Together with the preceding $o(1)$
bound, this makes the total loss smaller than $\rho$ for sufficiently
small $\delta$.
\end{proof}

\begin{proof}[Proof of \Cref{prop:finite-joint-program}]
We have already shown that
$W_H^{\mathrm{SD}}(\beta)\le S$.
To prove the reverse inequality, fix a feasible
aggregate measure $\mu$ and $\rho>0$.
Replace $Z$ by $\E_\mu[Z\mid Y]$, which cannot
lower the objective. By
\Cref{lem:finite-joint-approximation}, the resulting
wedge pair can be approximated by a feasible pair
with finite posterior support, at a loss of at
most $\rho$. For this finite posterior distribution,
maximize over wedges and apply
\Cref{lem:finite-joint-grid}. This gives an implementable policy
with value at least $\int F_{y,\beta}(z)\mu(\dd y,\dd z)-\rho$.
Letting $\rho\downarrow0$ and taking the supremum over $\mu$
proves $S\le W_H^{\mathrm{SD}}(\beta)$.

By \Cref{lem:finite-joint-implementation}, there
is an optimal joint policy with at most $N$
segments. Since $W_H^{\mathrm{SD}}(\beta)=S$,
the aggregate measure induced by any optimal
joint policy attains $S$. This proves the value
equality in \Cref{thm:finite-joint-design} and
the upper bound on the number of segments in
part~(iii).
\end{proof}

\subsection{Common quality, participation, and ordering}
\label{app:finite-joint-allocation}

The aggregate measure of an optimal policy maximizes
the program. Equality in Jensen's inequality therefore
requires $Z=z(Y)$ almost surely. Otherwise, averaging
virtual values conditional on the posterior mean
would strictly raise the objective. Thus, the quality rule
$x(y)=q(z(y))$ is the same in every segment, up to
null sets.

If $z(y)<z_{y,\beta}^{\dagger}$ on a set of positive
probability, raising it to $z_{y,\beta}^{\dagger}$
would strictly raise the objective and reduce the left-hand
side of every affected screening constraint, as in
\eqref{eq:finite-joint-wedge-cap}.
Hence $z_{y,\beta}^{\dagger}\le z(y)\le y$ for
$\overline G$-almost every $y>0$.
Applying $q$ gives the stated quality bounds.
Since $z_{y,\beta}^{\dagger}>0$, every positive
posterior type trades. At zero, virtual value
and quality are both zero.

Suppose the claimed ordering fails, so that
$y<y'$ and $z(y)\ge z(y')$ on a set of positive
$\overline G\otimes\overline G$ measure.
The bounds just proved give
$0<z(y')\le z(y)\le y<y'$, and
\eqref{eq:finite-joint-wedge-marginal} implies
\[
 g_{y,\beta}(z(y))
 \ge g_{y,\beta}(z(y'))
 >g_{y',\beta}(z(y')).
\]
Moving wedge mass from the higher posterior mean
to the lower one therefore has a positive marginal
payoff. Both changes are possible, since
$z(y)>0$ and
$d(y')=y'-z(y')\ge y'-y>0$.

Choose such a pair with both points in the support
of the distribution of $(Y,z(Y))$. Take small neighborhoods
of these points in $(y,z)$-space, and let $L,U$
be the sets of posterior means whose points
$(y,z(y))$ lie in the respective neighborhoods.
Both sets have positive probability. By continuity
of $g$ for $z>0$, we can choose the neighborhoods
small enough that
$\sup L<\inf U$ and
$\inf_L g_{y,\beta}(z(y))>\sup_U g_{y,\beta}(z(y))$.
We can also keep $z$ bounded away from zero on
both sets and $d$ bounded away from zero on $U$.

Let $G_L,G_U$ be the posterior distributions conditional
on these sets. For sufficiently small $\tau>0$,
replace $\mathsf E$ by
$\mathsf E+\tau G_L-\tau G_U$, keeping
$\overline G$ fixed.
The pointwise bounds remain satisfied.
Moreover, $G_L([t,1])\le G_U([t,1])$ at every
threshold, so no screening constraint has a
larger left-hand side. The first-order
gain is
\[
 \tau\left\{\int g_{y,\beta}(z(y))\,G_L(\dd y)
             -\int g_{y,\beta}(z(y))\,G_U(\dd y)\right\}>0.
\]
The remainder is $o(\tau)$ because the affected
virtual values stay bounded away from zero.
This contradicts optimality. We conclude that
$z(y)<z(y')$ for
$\overline G\otimes\overline G$-almost every
pair $y<y'$.

We can choose a nondecreasing version of $z$.
For each rational $c\in[0,1]$, the set
$\{y:z(y)>c\}$ agrees, up to a $\overline G$-null
set, with an upper interval $I_c$.
Otherwise, some cutoff would leave positive
mass from this set below positive mass from
its complement, contrary to the ordering just
proved. Define
$\widetilde z(y)
 =\sup\{c\in\mathbb Q\cap[0,1]:y\in I_c\}$,
taking the value zero when the set is empty.
Since each $I_c$ is an upper interval,
$\widetilde z$ is nondecreasing. There are only
countably many rational levels, so
$\widetilde z=z$ almost surely.
We may therefore use this version of $z$.
Since $q$ is strictly increasing, qualities
satisfy the same pairwise ordering, proving
part~(i).

Consider two segments $k,\ell$ and a connected
interval of posterior values contained in both
separating ranges. Learning in each segment
must be optimal for its conditional prior.
After the rescaling described in the next
subsection,
\Cref{lem:finite-separating,lem:finite-primitive-fit}
give positive densities and continuous quality
rules on this interval, including at primitive
values. The quality rules agree almost surely
by the common-quality result. Positive densities
and continuity then imply that they agree
throughout the interval.
The envelope formula gives locally absolutely
continuous utility functions with
$U_k'=U_\ell'=x$ almost everywhere.
Hence $U_k-U_\ell$ is constant.
Since $t_k(y)=yx(y)-U_k(y)$, the difference in
transfers is the opposite constant.
This constant may change across a gap in
either segment's support.

\subsection{Posterior atoms and revelation}
\label{app:finite-joint-atoms}

Write $G=\overline G$ and $\mathsf E=dG$ for an optimal policy.
A distorted atom at $w>0$ has mass $r>0$ and wedge $0<d(w)<w$,
since every positive posterior type trades.
We first consider $a<w<1$. Writing $z=y-d$, direct differentiation
gives
\[
 \begin{aligned}
 (u_\beta)_{yy}&=d q''(z)+\beta q'(z),\\
 (u_\beta)_{yd}&=(1-\beta)q'(z)-d q''(z),\\
 (u_\beta)_{dd}&=-(2-\beta)q'(z)+d q''(z),
 \end{aligned}
 \qquad \det D^2u_\beta=-q'(z)^2<0.
\]
To construct a split that raises the objective, put $d=d(w)$, $z=w-d$,
and $\tau=-(u_\beta)_{yd}/(u_\beta)_{dd}$ at $(w,d)$.
Replace the atom by equal masses at
$(w-\varepsilon,d-\tau\varepsilon)$ and
$(w+\varepsilon,d+\tau\varepsilon)$.
Small $\varepsilon>0$ keeps both pairs in $0<d<y$ and
preserves the posterior mean and total wedge. First-order changes
cancel, while curvature gives
\begin{equation*}
 (u_\beta)_{yy}+2\tau(u_\beta)_{yd}
                  +\tau^2(u_\beta)_{dd}
 =\frac{q'(z)^2}{(2-\beta)q'(z)-d q''(z)}>0.
\end{equation*}
The gain in the objective is $r\varepsilon^2/2$ times the
expression above, plus $o(\varepsilon^2)$.
If either new posterior value already has mass,
average the wedges there using the corresponding
masses as weights. This preserves all constraints
and weakly raises the objective by concavity.

The posterior spread weakly raises $C_G(t)$
at every threshold. The wedge mass
$\mathsf E([t,1])$ is unchanged outside
$(w-\varepsilon,w+\varepsilon]$.
For $w-\varepsilon<t\le w$, it falls by
$r(d-\tau\varepsilon)/2$.
For $w<t\le w+\varepsilon$, it rises by
$r(d+\tau\varepsilon)/2$.

Write $S(t)=C_G(t)-\mathsf E([t,1])$ for the
original slack. As $t$ passes $w$, the atom's
contribution $rd$ drops out of
$\mathsf E([t,1])$, while $C_G$ is continuous.
Hence $S(w+)=S(w)+rd\ge rd$.
The added wedge mass, by comparison, tends
to $rd/2$.
Thus, $\inf_{w<t\le w+\varepsilon}S(t)>r(d+\tau\varepsilon)/2$
for small $\varepsilon$, so the screening constraints
remain satisfied at these thresholds.

If $w$ is not primitive, choose $\varepsilon$ so both new
posterior means stay between the same adjacent primitive values.
Their mean is unchanged, so every $D_j$ is unchanged.
The resulting feasible improvement contradicts optimality.
If $w=v_j$ is interior and $D_j(G)<\kappa_j$, choose
$\varepsilon$ smaller than the distances to the neighboring
primitive values. Only $D_j$ changes, increasing by
$r\varepsilon/2$. For small $\varepsilon$ this is less than
its slack, giving the same contradiction. Hence every distorted
atom is primitive, and an interior one requires $D_j(G)=\kappa_j$.

At a binding constraint, Jensen's inequality gives
$\E[(V-v_j)_+\mid Y,K]\ge(Y-v_j)_+$ with zero expected slack.
Equality therefore holds almost surely. On $Y=v_j$, it excludes
$V>v_j$. The posterior mean $v_j$ then implies $V=v_j$.
When $Y=a$ or $Y=1$, the bounds $a\le V\le1$
force $V=Y$ almost surely. Every distorted atom
therefore reveals the buyer's primitive valuation.

For efficient atoms, apply the private-learning results within
each segment. Its conditional prior is supported on a subset
of the primitive values, and its learning must be optimal for
that prior. With a degenerate prior, buyers already know their valuation.
For a nondegenerate prior with highest value $b>0$, divide
valuations, transfers, and cost by $b$.
This scales both payoffs by $1/b$ and gives
$\widetilde q(z)=q(bz)$, with
$(2-\beta)\widetilde q'(z)-(y-z)\widetilde q''(z)
 =b\{(2-\beta)q'(bz)-(by-bz)q''(bz)\}>0$, $0<z<y\le1$.
The other cost assumptions also hold, and $q(b)<\bar Q$.
After omitting primitive values of zero conditional probability,
\Cref{thm:finite-structure} therefore applies.
Every nonterminal trading atom is strictly distorted:
if its value is $w$, its mass is $p$, and the mass $\sigma$
above it begins at $B>w$, revenue contact gives virtual value
$w-\sigma(B-w)/p<w$.
The only efficient positive atom is consequently the terminal
atom at the conditional prior's highest value, which it reveals.

An aggregate atom is an atom in some positive-probability segment.
Common quality ensures that, if it is efficient, it is efficient
in every segment placing mass there. The preceding argument
therefore gives revelation in all such segments. An atom at zero
also reveals its value by the support bound.
Every aggregate atom is thus primitive and revealing.
Each segment's nonterminal trading atoms lie below its highest
primitive value, so there are at most $N-1$ of them.
This proves part~(ii).

\subsection{A global bound and its equality conditions}
\label{app:finite-joint-certificate}

We will use the following bound to construct priors that require
all $N$ segments. Let $k:[a,1]\to\R_+$ be continuous and
nondecreasing, with $k(a)=0$, and define
\[
 K(y)=\int_a^y k(t)\,\dd t,\qquad
 r_{k,\beta}(y)=\max_{0\le d\le y}
                    \{u_\beta(y,d)-k(y)d\}.
\]

\begin{lemma}[An upper bound for joint design]
\label{lem:finite-joint-certificate}
If $h$ is continuous and convex on $[a,1]$ and
$h\ge K+r_{k,\beta}$, then
\begin{equation*}
 W_H^{\mathrm{SD}}(\beta)\le\int h(v)\,H(\dd v).
\end{equation*}
For a feasible policy with aggregate posterior distribution $G$, let
$d(y)$ be its conditional mean wedge and put $\mathsf E=dG$.
The policy attains the bound if and only if wedges are common
conditional on posterior mean, $d(y)$ maximizes
$u_\beta(y,d)-k(y)d$ almost surely, $h=K+r_{k,\beta}$ on $\supp G$,
$\int h\,\dd G=\int h\,\dd H$, and
$\mathsf E([t,1])=C_G(t)$ for $\dd k$-almost every $t$.
\end{lemma}

\begin{proof}
By Jensen's inequality, the designer's objective
under this policy is at most
$\int u_\beta(y,d(y))G(\dd y)$, with equality exactly
when wedges are common. Pointwise maximization then gives
$u_\beta(y,d(y))\le r_{k,\beta}(y)+k(y)d(y)$.
Integrating the screening constraints against the nonnegative
Stieltjes measure $\dd k$ gives
\[
 \int k(y)\,\mathsf E(\dd y)
 =\int\mathsf E([t,1])\,\dd k(t)
 \le\int C_G(t)\,\dd k(t)
 =\int K(y)\,G(\dd y).
\]
The identities follow from Tonelli and
$\int_a^y(y-t)\,\dd k(t)=K(y)$, using $k(a)=0$.
Thus
\[
 \int u_\beta(y,d(y))\,G(\dd y)
 \le\int(K+r_{k,\beta})\,\dd G
 \le\int h\,\dd G
 \le\int h\,\dd H.
\]
The last inequality uses convex order.
All slacks are nonnegative, so the stated conditions are
necessary and sufficient for equality.
The contact condition holds on the entire support because
$r_{k,\beta}$ is continuous: writing $d=ys$ expresses it as
a maximum over the fixed compact interval $s\in[0,1]$.
\end{proof}

\subsection{Sharpness of the segment bound}
\label{app:finite-joint-sharpness}

We construct a finite family of segments attaining a common upper
bound, show that these are the only segments that attain it, and
choose a prior that requires $N$ of them.
Fix $N\ge2$, put $n=N-1$, and take $c(x)=x^2/2$,
$\bar Q>1$, and $\beta=0$. Set
\[
 \chi=\arctan3,\qquad
 \gamma_3=\frac{e^{-\chi}}{\sqrt{10}},\qquad
 b=\theta=\frac{2\gamma_3}{1+2\gamma_3},\qquad B=3b.
\]
Since $(1-\theta)\gamma_3=b/2$, the converse in
\Cref{lem:mult-quadratic-smooth-tail} gives a regular distribution
$F$ whose lowest posterior value is $B$ and lowest virtual value
is $B/2>\theta$, with an atom of mass $e^{-2\chi}>0$ at one.
Its posterior and virtual-value profiles increase continuously
to one. It has no atom at $B$, and
$b<B<1$, with $\theta=(B+b)/4$.

Choose $0<v_1<\cdots<v_n=b$ and $v_N=1$, and put $a=v_1$.
Write $z(y)$ for the virtual value under $F$ and $d(y)=y-z(y)$.
The separating equations and the positive density of $F$ give
\[
 z'(y)=\frac{y-\theta}{2d(y)},\qquad
 d(y)F(\dd y)=F((y,1])\,\dd y,
 \quad B<y<1.
\]
Set $z(1)=1$ and $d(1)=0$.
Define $k=0$ on $[a,B]$ and $k(y)=2z(y)-y$ on $[B,1]$.
It is continuous, with $k(B)=0$, $k(1)=1$, and
\begin{equation}
 k'(y)=\frac{z(y)-\theta}{y-z(y)}>0,
 \qquad B<y<1.
 \label{eq:finite-joint-sharpness-k}
\end{equation}
Thus $k$ is nondecreasing and $0\le k(y)\le y$.
If $M$ is the mean of $F$, then $B<M<1$ and $F$ is feasible
under the binary prior
$H_F=(1-\alpha_F)\delta_b+\alpha_F\delta_1$, where
$\alpha_F=(M-b)/(1-b)$.

\paragraph{Segments attaining the bound.}
To apply the bound, let $h$ interpolate $y^2/4$ linearly at $v_1,\ldots,v_n$,
and set $h(y)=b^2/4+\theta(y-b)$ for $b\le y\le1$.
For $N=2$, the latter formula defines $h$ on the whole domain.
Its slopes $(v_i+v_{i+1})/4$, followed by $\theta=(B+b)/4$,
strictly increase at each interior primitive value, so $h$ is convex.

Completing the square in the bound of
\Cref{lem:finite-joint-certificate} gives
$r_{k,0}(y)=(y-k(y))^2/4$, with unique maximizing wedge
$d^*(y)=(y-k(y))/2$.
Below $B$, $K+r_{k,0}=y^2/4$, and $h$ is its chord between
successive points of $\{v_1,\ldots,v_n,B\}$.
Above $B$, \eqref{eq:finite-joint-sharpness-k} gives
$(K+r_{k,0})'=z-dk'=\theta$.
The levels agree at $B$ because
$B^2/4=b^2/4+\theta(B-b)$.
Hence $h\ge K+r_{k,0}$, with contact set exactly
$\{v_1,\ldots,v_n\}\cup[B,1]$.
The upper bound applies under every prior on these primitive values.

For each subset $S\subseteq\{1,\ldots,n\}$, start with $F$
and add the atoms indexed by $S$ in descending order.
If the current upper distribution $L$ begins at $w>v_i$,
replace it by
\begin{equation}
 (1-r)\delta_{v_i}+rL,\qquad
 r=\frac{v_i}{2w-v_i}\in(0,1).
 \label{eq:finite-joint-sharpness-splice}
\end{equation}
The new atom's revenue chord has virtual value
$(v_i-rw)/(1-r)=v_i/2<w/2$, where $w/2$ is the bottom
virtual value of $L$. The chord therefore joins the upper
revenue curve concavely. Each added atom receives quality $v_i/2$,
and the part with distribution $F$ retains its original allocation.
Denote the resulting distribution by $G_S$, including $G_\varnothing=F$.
Write $G_S=\sum_{i\in S}m_i\delta_{v_i}+sF$ and set
$H_S=\sum_{i\in S}m_i\delta_{v_i}+sH_F$.
Then $G_S\preceq_\cx H_S$: its lower atoms reveal their states,
while $F$ is implemented under $H_F$.

These segments give equality at every step of the upper bound.
Their wedges equal $d^*$ and their support lies in the contact set.
For $B<t<1$, integrating
$d(y)F(\dd y)=F((y,1])\,\dd y$ gives
$\int_{[t,1]}d(y)F(\dd y)=C_F(t)$.
Multiplying both sides by $s$ gives the same
equality for the upper part $sF$. The measure
$\dd k$ assigns mass only to these thresholds.
Finally, $h$ is affine on $[b,1]$, so
$\int h\,\dd G_S=\int h\,\dd H_S$.
Any mixture of these segments is therefore
jointly optimal for the prior obtained by
mixing their conditional priors $H_S$ with
the same weights.

\paragraph{All segments attaining the bound.}
Now consider any policy attaining the bound.
In every segment with positive probability,
the equality conditions restrict posterior
support to $\{v_1,\ldots,v_n\}\cup[B,1]$.
They also require virtual value $v_i/2$ at
each lower atom and $z(y)$ on $[B,1]$.

The aggregate screening constraint binds for
$\dd k$-almost every threshold. Its slack is
a weighted average of the nonnegative slacks
in individual segments. The constraint must
therefore bind in every segment with positive
probability. Since $k'>0$ on $(B,1)$, this gives
equality for almost every threshold in that
interval. In each segment, the wedge mass
in $[t,1]$ is left-continuous in $t$, while
the call function is continuous. Taking limits
from below through thresholds where equality
holds gives equality at every $t\in(B,1)$.

To identify the upper tail, fix one such segment, with posterior distribution $G$
and wedge measure $\mathsf E_G=d^*G$.
Differentiating its tail equality as measures gives
$d(y)G(\dd y)=G((y,1])\,\dd y$ on $(B,1)$.
Since $d$ is locally positive and continuous, $G$ has a density
there and its survival function satisfies $S_G'=-S_G/d$.
The survival function of $F$ is positive and solves the same
equation. Their ratio is constant, so
$G|_{(B,1]}=sF$ for some $s\ge0$. The limit at one also
identifies the top atom.
If $s=0$, the segment has finite support in
$\{v_1,\ldots,v_n,B\}$. Its highest posterior then has a
positive prescribed wedge, violating the screening constraint
at that value. Thus $s>0$.
Also, letting $t\downarrow B$ gives
$\mathsf E_G((B,1])=C_G(B)$.
The constraint at $B$ forces $\mathsf E_G(\{B\})=0$,
so the prescribed wedge $B/2>0$ rules out an atom there.

It remains to determine the lower atom masses.
At an included atom $v_i$, let $p$ be its mass,
$\sigma$ the mass above it, and $w$ the bottom of that upper
support. Its quality $x_i=v_i/2$ is strictly separated from
the neighboring qualities. Small changes of either sign are
therefore feasible for the seller.
Under revenue-maximizing IC and IR transfers, increasing $x_i$ by
$\varepsilon$ raises this atom's transfer by $v_i\varepsilon$
and lowers the transfer paid by each higher type
by $(w-v_i)\varepsilon$.
Indeed, the minimal-utility extension keeps quality at $x_i$
across the empty interval $[v_i,w)$.
This also applies when $w=B$, even though there is no atom at $B$.
Seller optimality gives
$p(v_i-x_i)-\sigma(w-v_i)=0$, or
\[
 v_i-\frac{\sigma}{p}(w-v_i)=\frac{v_i}{2},
 \qquad
 \frac{\sigma}{p+\sigma}=\frac{v_i}{2w-v_i}.
\]
This is precisely \eqref{eq:finite-joint-sharpness-splice}.
Starting with $sF$ and working downward, these ratios and
normalization identify the segment as one of the $G_S$.

Equality also fixes its conditional primitive prior.
The function $h$ is affine plus positive multiples of the hinges
at $v_2,\ldots,v_n$. Equality in its convex-order bound therefore
forces the corresponding aggregate information constraints to bind.
Jensen's inequality has nonnegative slack given $(Y,K)$, so
the revelation argument in Appendix~\ref{app:finite-joint-atoms}
applies in every segment. Each lower atom reveals its state.
The lowest one does so by the support bound.
For a posterior mean in $[B,1]$, equality at the threshold $b$
excludes primitive values below $b$. Only $b$ and one remain,
with probabilities determined by the posterior mean.
When $N=2$, this follows directly from the primitive support.
Thus the conditional prior of $G_S$ must be $H_S$.

\paragraph{A prior requiring $N$ segments.}
Consider $H_F$ and the $n$ priors obtained by adding one lower atom:
\[
 H_i=(1-r_i)\delta_{v_i}+r_iH_F,\qquad
 r_i=\frac{v_i}{2B-v_i},\qquad 1\le i\le n.
\]
These $N$ priors are affinely independent.
Indeed,
$H_i-H_F=(1-r_i)(\delta_{v_i}-H_F)$, with
$1-r_i>0$.
Suppose $\sum_i\alpha_i(\delta_{v_i}-H_F)=0$.
The coordinate at one gives $\sum_i\alpha_i=0$,
since $H_F(\{1\})>0$. The coordinate at each
$v_i$ then gives $\alpha_i=0$.
Their convex hull is therefore an
$(N-1)$-dimensional simplex.
There are only $2^n$ priors $H_S$.
The convex hull of any collection of at most
$N-1$ of them has dimension at most $N-2$.
These finitely many hulls cannot cover the
simplex's relative interior.
Choose $H$ in the relative interior of
$\operatorname{co}\{H_F,H_1,\ldots,H_n\}$
but outside all these hulls.
Every primitive valuation then has positive
probability. Write $H$ as a convex combination
of $H_F,H_1,\ldots,H_n$ and use the coefficients
as segment probabilities. The corresponding
$N$ segments attain the upper bound.

Every optimal policy must attain this same
bound and therefore use segments with posterior
distributions $G_S$ and conditional priors $H_S$.
If it used fewer than $N$ distinct segments,
$H$ would be a convex combination of at most
$N-1$ priors $H_S$, contrary to its choice.
Thus, every optimum requires at least $N$
distinct segments. Together with
\Cref{lem:finite-joint-implementation}, this
proves \Cref{thm:finite-joint-design}(iii).
\subsection{The joint Pareto frontier}
\label{app:finite-joint-frontier}

Let $\mathcal P_H^{\mathrm{SD}}$ denote the set of consumer-surplus
and profit pairs attainable through joint private learning and
state-dependent segmentation.

\begin{lemma}[The joint payoff set]
\label{lem:finite-joint-frontier}
The set $\mathcal P_H^{\mathrm{SD}}$ is compact and convex.
Every Pareto-efficient pair maximizes $\CS+\beta\Pi$ for some
$\beta\in[0,1]$. Conversely, every maximizer for
$0<\beta\le1$ is Pareto efficient. At $\beta=0$, a maximizer
is Pareto efficient if and only if it has the highest profit
among all consumer-surplus maximizers.
\end{lemma}

\begin{proof}
The proof of \Cref{lem:finite-joint-implementation}
obtains $\mathcal P_H^{\mathrm{SD}}$ by projecting
a compact convex set onto the payoff coordinates.
It is therefore compact and convex.

Fix a Pareto-efficient pair $(C,P)$.
Separating $\mathcal P_H^{\mathrm{SD}}$ from
$(C,P)+\operatorname{int}\R_+^2$ gives weights
$\alpha,\gamma\ge0$, not both zero, for which
$(C,P)$ maximizes $\alpha C'+\gamma P'$ over
the attainable set.
If $\gamma\le\alpha$, then $\alpha>0$.
Dividing by $\alpha$ gives the required objective
with $\beta=\gamma/\alpha\in[0,1]$.

If $\gamma>\alpha$, put
$\overline P=\E_H[\pi(V)]$.
Every attainable pair satisfies $C'\ge0$ and
$C'+P'\le\overline P$. Hence
$\alpha C'+\gamma P'
 =\gamma(C'+P')-(\gamma-\alpha)C'
 \le\gamma\overline P$.
Equality requires $C'=0$ and $P'=\overline P$.
Full public revelation attains this pair,
so $(C,P)=(0,\overline P)$.
This pair also maximizes total surplus and is
therefore optimal at $\beta=1$.
Conversely, a Pareto improvement would raise
$\CS+\beta\Pi$ whenever $\beta>0$.
Every maximizer for $0<\beta\le1$ is therefore
Pareto efficient.
At $\beta=0$, a consumer-surplus maximizer is
efficient exactly when it has the highest profit
among all consumer-surplus maximizers.
This set is compact, so the highest profit
is attained.
\end{proof}

\end{document}